\documentclass[11pt, a4paper]{article}
\usepackage{url}
\usepackage{a4wide}
\usepackage{graphicx}
\usepackage{natbib}
\usepackage{enumitem}
\usepackage{hyperref}
\usepackage{setspace}
\usepackage{multirow, booktabs}
\usepackage{amsmath, amsthm, amssymb, amsfonts}
\usepackage{bm, bbm}
\usepackage[T1]{fontenc}
\usepackage[ruled, vlined]{algorithm2e}
\usepackage{color}
\usepackage{csquotes}
\usepackage{xr}
\usepackage[capitalise, noabbrev]{cleveref}

\allowdisplaybreaks
\DeclareMathAlphabet\mathbfcal{OMS}{cmsy}{b}{n}
\newcommand{\mbf}{\mathbf}
\newcommand{\mc}{\mathcal}
\newcommand{\bmx}{\begin{bmatrix}}
\newcommand{\emx}{\end{bmatrix}}

\newcommand{\vep}{\varepsilon}

\renewcommand{\l}{\left}
\renewcommand{\r}{\right}

\def\wh{\widehat}
\def\wt{\widetilde}
\def\wc{\widetilde}

\newcommand{\cp}{\theta}
\newcommand{\Cp}{\Theta}

\newcommand{\E}[0]{\mathsf{E}}
\newcommand{\Var}[0]{\mathsf{Var}}
\newcommand{\Cov}[0]{\mathsf{Cov}}
\newcommand{\tr}[0]{\mathsf{tr}}
\newcommand{\p}{\mathsf{P}}
\newcommand{\R}{\mathbb{R}}
\newcommand{\Z}{\mathbb{Z}}
\newcommand{\N}{\mathbb{N}}

\newcommand{\iid}{\text{\upshape iid}}
\newcommand{\nn}{\nonumber}

\newcommand{\crsc}{C_{\text{\tiny \upshape RE}}}
\newcommand{\tE}{\text{E}}
\newcommand{\tO}{\text{O}}

\theoremstyle{definition}
\newtheorem{thm}{Theorem}
\theoremstyle{definition}
\newtheorem{cor}[thm]{Corollary}
\theoremstyle{definition}
\newtheorem{lem}[thm]{Lemma}
\theoremstyle{definition}
\newtheorem{prop}[thm]{Proposition}
\theoremstyle{definition}
\newtheorem{assum}{Assumption}
\theoremstyle{remark}
\newtheorem{rem}{Remark}[section]
\theoremstyle{definition}
\newtheorem{defn}{Definition}
\theoremstyle{definition}

\theoremstyle{definition}

\newtheorem{massuminner}{Assumption}
\newenvironment{massum}[1]{%
  \renewcommand\themassuminner{#1}%
  \massuminner
}{\endmassuminner}

\graphicspath{{./figures/}}

\title{Detection and inference of changes in high-dimensional linear regression with non-sparse structures}
\author{Haeran Cho\textsuperscript{1} \and Tobias Kley\textsuperscript{2} \and Housen Li\textsuperscript{2,*}}
\date{{\footnotesize\textsuperscript{1}School of Mathematics, University of Bristol, UK \\
 \textsuperscript{2}Institute
for Mathematical Stochastics, University of Göttingen, Germany\\
\textsuperscript{*}Corresponding author. {housen.li@mathematik.uni-goettingen.de}}}

\begin{document}

\maketitle

\begin{abstract}
For data segmentation in high-dimensional linear regression settings, the regression parameters are often assumed to be {exactly} sparse segment-wise, which enables many existing methods to estimate the parameters locally via $\ell_1$-regularised maximum likelihood-type estimation and then contrast them for change point detection. Contrary to this common practice, we show that the {exact} sparsity of neither regression parameters nor their differences, a.k.a.\ differential parameters, is necessary for consistency in multiple change point detection. In fact, both statistically and computationally, better efficiency is attained by a simple strategy that scans for large discrepancies in local covariance between the regressors and the response. We go a step further and propose a suite of tools for directly inferring about the differential parameters post-segmentation, which are applicable even when the regression parameters themselves are non-sparse. Theoretical investigations are conducted under general conditions permitting non-Gaussianity, temporal dependence and ultra-high dimensionality. Numerical results from simulated and macroeconomic datasets demonstrate the competitiveness and efficacy of the proposed methods. Implementation of all methods is provided in the R package \texttt{inferchange} on GitHub.
\end{abstract}

\noindent%

\section{Introduction}

With rapid technological advancements, modern datasets are high-dimensional and massive in volume, which calls for novel statistical and computational tools.
As a prominent example, regression modelling in high dimensions 
has found numerous applications in a wide range of scientific fields, including genomics, signal processing, finance and economics, to name a few, see \citet{buhlmann2011statistics} for an overview. 
Another notable feature of modern data is the underlying heterogeneity \citep{fan2014challenges}, especially when datasets are collected in temporal (or other meaningful) order in nonstationary environments. 

We address this heterogeneity in high-dimensional regression settings by considering the following model for  observations $(Y_t, \mbf x_t), \, t \in \{1, \ldots, n\}$, with $\mbf x_t = (X_{1t}, \ldots, X_{pt})^\top \in \R^p$:
\begin{align} 
\label{eq:model}
Y_t = \l\{\begin{array}{ll}
\mbf{x}_t^\top\bm{\beta}_0 + \vep_t & \text{for } \cp_0 = 0 < t \le \cp_1, \\
\mbf{x}_t^\top\bm{\beta}_1 + \vep_t & \text{for } \cp_1 < t \le \cp_2, \\
\vdots \\
\mbf{x}_t^\top\bm{\beta}_q + \vep_t & \text{for } \cp_q < t \le n = \cp_{q+1}.
\end{array} \r.	 
\end{align}
We assume that $\E(\vep_t) = 0$ and $\E(\vep_t^2) = \sigma_\vep^2 \in (0, \infty)$ for all~$t$. 
Under this model, the joint distribution of $(Y_t, \mbf x_t)$ undergoes multiple shifts at the change points $\cp_j \in \N, \, j \in \{1, \ldots, q\}$, which are attributed to the changes in the regression parameters $\bm{\beta}_j$. We refer to the differences between the regression parameters from the adjacent segments as {\it differential parameters} and denote them by $\bm\delta_j = \bm\beta_j - \bm\beta_{j - 1} \ne \mbf 0$.

The data segmentation problem under~\eqref{eq:model} has been investigated both in multivariate (i.e.\ for fixed $p$, see e.g.\ \citealp{csorgo1997limit} and \citealp{bai1998estimating}) and, more recently, in high-dimensional settings.
In the latter, most of the contributions require estimating the (linear mixtures of) regression coefficients $\bm\beta_j$ over local intervals via Lasso-type estimators assuming that $\bm\beta_j$'s are sparse. For a non-exhaustive list of references, we refer to \cite{lee2016lasso}, \cite{leonardi2016computationally}, \cite{kaul2019detection}, \cite{wang2021statistically}, \cite{rinaldo2020localizing}, \cite{bai2022unified}, \cite{cho2022high}, \cite{xu2022change} and \cite{liu2022change, liu2024simultaneous}. 
One exception is \cite{gao2022sparse} where the sparsity is imposed directly on $\bm\delta_j$, along with a rather strong requirement that $p$ is strictly smaller than $n$. 

The high dimensionality poses computational and statistical challenges on the change point problem.
For instance, searching for a single change point over the full grid requires $O(n)$ Lasso fits \citep{lee2016lasso, leonardi2016computationally}, which becomes impractical for large $n$ and~$p$. 
To mitigate this, some methods reduce the number of expensive Lasso-type fits by searching on coarse grids \citep{cho2022high, li2023divide}, or selecting intervals systematically \citep{qian2023reliever} or adaptively \citep{kovacs2020optimistic}.  However, all these approaches still rely on locally estimating regression parameters via $\ell_1$-regularised methods, which remain computationally demanding for large datasets.

Beyond the detection of multiple change points, another important task is to infer the variables undergoing the changes, e.g.\ by constructing simultaneous confidence intervals about the differential parameters, which is particularly relevant for large $p$.
While there are procedures for testing for a change \citep{wang2022optimal, liu2024simultaneous} or deriving confidence intervals about the change point location \citep{xu2022change} under the model~\eqref{eq:model}, little effort has been made in inferring about $\bm\delta_j$ without assuming the \emph{exact} sparsity of $\bm\beta_j$. {Here, exact sparsity refers to an explicit assumption on the number of non-zero entries in the parameter vector. }

In short, there still remain fundamental challenges for the change point problem in~\eqref{eq:model} in both statistical and computational regards, which we summarise below:
\begin{enumerate}[label = (C\arabic*)]
	\item \label{q:one}
	Is the {exact} sparsity of either $\bm{\beta}_j$ or $\bm\delta_j$ necessary for achieving consistency in multiple change point detection under~\eqref{eq:model} in high dimensions? 	On a related note, is it possible to estimate~$\cp_j$ without the computationally costly estimation of either~$\bm{\beta}_j$ or~$\bm\delta_j$?
	\item \label{q:two} If $\bm\delta_j$ is sparse, is it possible to recover the differential parameters $\bm\delta_j$ with confidence statements even when the regression parameters $\bm{\beta}_j$ are possibly non-sparse?
\end{enumerate}

Motivated by these questions, we develop a suite of methods for estimation and inference under the change point model. 
On the change point detection front, we propose a statistically and computationally efficient method that requires neither the exact sparsity of $\bm\beta_j$ nor $\bm\delta_j$ and permits ultra-high dimensionality ($p = o(\exp(n))$), addressing the questions posed in \ref{q:one}. 
Addressing~\ref{q:two}, we introduce a novel procedure for directly estimating and inferring $\bm\delta_j$'s which, to the best of our knowledge, is a first such contribution. 
Below we further elaborate on the statistical challenges and introduce the proposed methods along with their novelty.

\subsection{Multiscale covariance scanning}
\label{sec:size:change}

The first question in~\ref{q:one} concerns the formulation of the detection boundary under the change point model in~\eqref{eq:model}. 
Towards this, we first consider the at-most-one-change situation.

\begin{lem}
	\label{lem:tv}
	Assume that $\mbf x_t \sim_{\iid} \mc N_p(\mbf 0, \bm\Sigma)$ with positive semi-definite $\bm\Sigma$ and $\vep_t \sim_{\iid} \mc N(0, \sigma_\vep^2)$ are independent. 
	Let $\p_0(\bm\beta)$ denote the joint distribution of $\{(\mbf x_t, Y_t): \,  t = 1, \ldots, n\}$ such that $Y_t = \mbf x_t^\top\bm\beta + \vep_t$ for all $t$. 
	We also denote by $\p_{\cp_1}(\bm\beta_0, \bm\beta_1)$ the joint distribution of $\{(\mbf x_t, Y_t): \, t = 1, \ldots, n\}$ under~\eqref{eq:model} when $q = 1$. 
	Then the total variation distance between $\p_0(\bm\beta)$ and $\p_{\cp_1}(\bm\beta_0, \bm\beta_1)$, denoted by $\mathrm{TV}\bigl( \p_0(\bm\beta),\, \p_{\cp_1}(\bm\beta_0, \bm\beta_1) \bigr)$, satisfies
	$$
	\frac{1}{100} \le \frac{\min_{\bm\beta, \bm\beta_0, \bm\beta_1: \, \bm\beta_1 - \bm\beta_0 = \bm\delta} \mathrm{TV} \bigl( \p_0(\bm\beta),\, \p_{\cp_1}(\bm\beta_0, \bm\beta_1) \bigr)}{\min\left\{ 1, \, \sqrt{\frac{\cp_1 (n - \cp_1)}{n \sigma_\vep^2}\bm\delta^\top \bm\Sigma \bm\delta}\right\}} \le \frac{3\sqrt{3}}{2}. 
	$$
\end{lem}

\cref{lem:tv} suggests that the detection boundary of the change point problem is determined by the spacing between the change points, $\Delta_j = \min(\cp_j - \cp_{j - 1}, \cp_{j + 1} - \cp_j)$,\label{ref:deltaj} and the magnitude of the changes measured by $\sigma_\vep^{-2} \bm\delta_j^\top \bm\Sigma \bm\delta_j$. 
For the latter quantity, which is closely related to the explained variance (\citealp{VerGas18, tony2020semisupervised}), a.k.a.\ heritability in genetics \citep{Mah08},
its sample analogue is not easily accessible unless additional structural assumptions are imposed on $\bm\delta_j$ and/or $\bm\Sigma$. 
We propose to circumvent this difficulty by considering the {\it covariance-weighted} differential parameter $\bm\Sigma \bm\delta_j$ instead. 

Specifically, an empirical surrogate of $\bm\Sigma \bm\delta_j$ can be obtained by screening the sequence $\{\mbf x_t Y_t \}_{t = 1}^n$ without estimating $\bm\beta_j$, $\bm\delta_j$ or $\bm\Sigma$. 
Built upon this crucial observation, we propose the {\it \underline{m}ultiscale \underline{c}ovariance \underline{scan}ning} (McScan) algorithm for the detection of multiple change points under the model~\eqref{eq:model}, which involves scanning the weighted averages of $\{\mbf x_t Y_t\}_{t = 1}^n$ over carefully chosen deterministic intervals. Notably, McScan avoids costly computations involved in locally estimating the regression or differential parameters that typically require cross validation, and thus enjoys numerical stability as well as computational efficiency with the worst case run time of order $np\log(n)$. 

\begin{figure}[h!t!]
\centering
\begin{tabular}{c c}
\includegraphics[width = .45\textwidth]{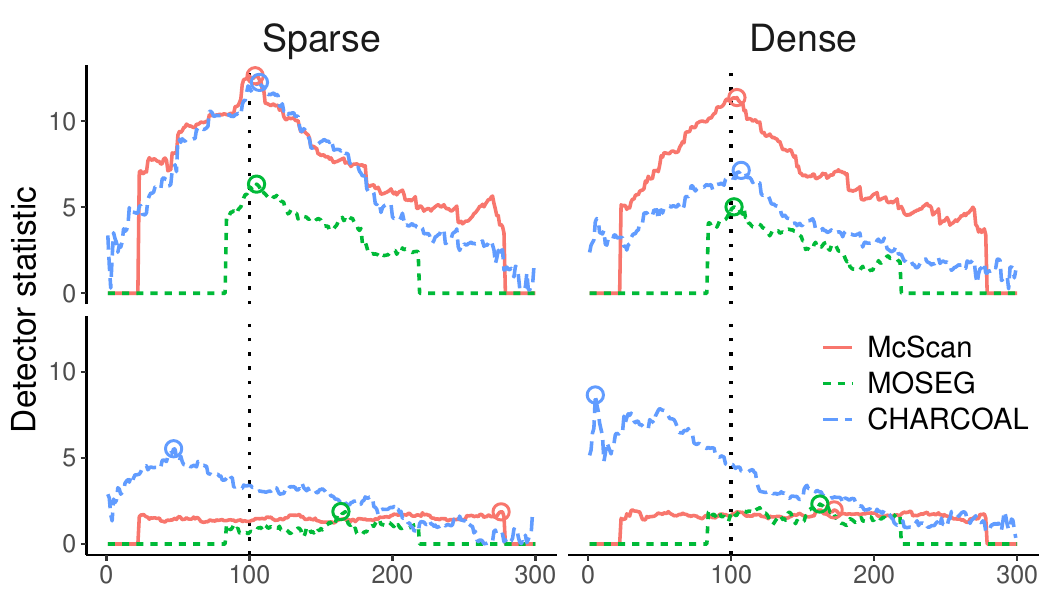}
&
\includegraphics[width = .45\textwidth]{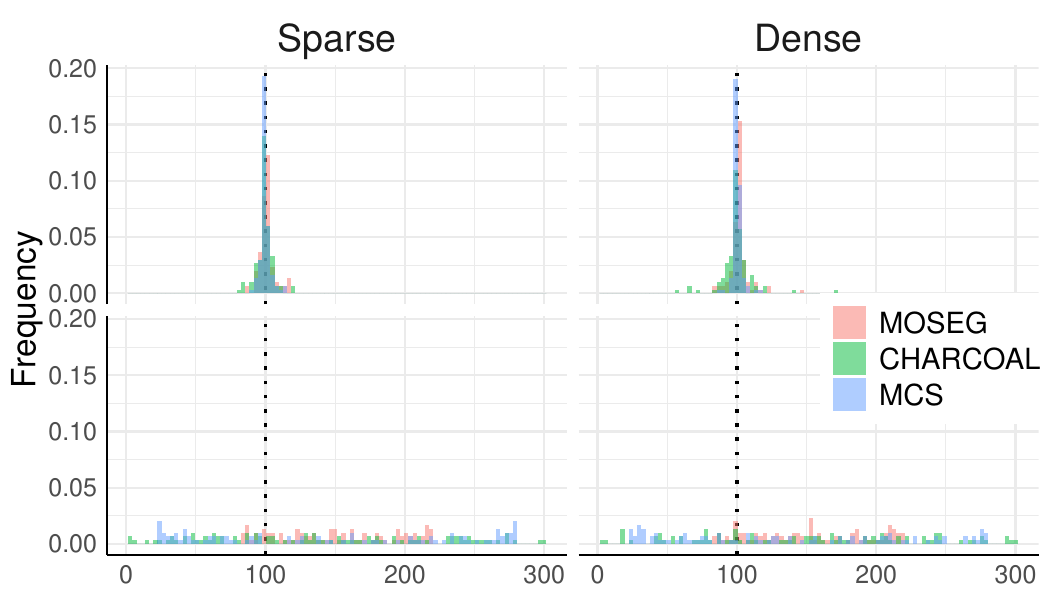}
\end{tabular}
\caption{Consider sparse ($\mathfrak{s} = 5$) and dense ($\mathfrak{s} = p = 200$) scenarios under model~\eqref{eq:model} with a single change ($q = 1$), where $\bm\beta_0 = \bm\delta_1/2$, $\bm\beta_1 = -\bm\delta_1/2$ and $\vert \bm\delta_1 \vert_0 = \mathfrak{s}$.
We vary $\vert \bm\Sigma \bm\delta_1 \vert_\infty$ from $1.41$ (top) to $0.09$ (bottom), while keeping $\vert \bm\delta_1 \vert_2 = 2$ unchanged in all scenarios. 
Left: We plot the detector statistics of McScan, MOSEG \citep{cho2022high} and CHARCOAL \citep{gao2022sparse} where for each method, the change point location $\cp_1$ is estimated by the maximiser of its detector statistic (marked by a circle). 
Right: We plot the estimated change points over $100$ realisations.
In all plots, the true $\cp_1 = 100$ is marked by the vertical dotted lines. 
Further details of the simulation setup are given in \cref{sec:sim:intro}.}
\label{fig:intro}
\end{figure}

In terms of statistical guarantees, most existing papers measure the size of changes via the sparsity-adjusted $\ell_2$-norm, $\vert \bm\delta_j \vert_0^{-1/2} \vert \bm\delta_j \vert_2$ (where $\vert\cdot\vert_d$ is the $\ell_d$ (pseudo-)norm of a vector), and the investigation into the minimax optimality is also conducted with a parameter space defined according to this measure \citep{rinaldo2020localizing}, see \cref{tab:comp} in Appendix~\ref{sec:comparison} for a comprehensive overview of the literature.
In sharp contrast, we show that McScan achieves consistency in multiple change point detection for a strictly broader class of problems with no worse rate of estimation (\cref{thm:one}). 
Our theoretical results are derived under general conditions permitting temporal dependence and non-Gaussianity and, remarkably, do not require the {exact} sparsity of either $\bm\beta_j$ or~$\bm\delta_j$. 
This follows from that McScan measures the size of change by $\vert \bm\Sigma \bm\delta_j \vert_\infty$ which, by repeated applications of H\"{o}lder's inequality, satisfies
\begin{align}
\label{eq:size:change}
\vert \bm\Sigma \bm\delta_j \vert_\infty \ge \frac{\bm\delta_j^\top \bm\Sigma \bm\delta_j}{\vert \bm\delta_j \vert_1} \ge \frac{\Lambda_{\min}(\bm\Sigma) \vert \bm\delta_j \vert_2}{\sqrt{\vert \bm\delta_j \vert_0}},
\end{align}
where $\Lambda_{\min}(\bm\Sigma)$ denotes the smallest eigenvalue of~$\bm\Sigma$. 
When $\Lambda_{\min}(\bm\Sigma)$ is bounded away from zero as commonly assumed in the literature, it indicates that measuring the size of change via $\vert \bm\Sigma \bm\delta_j \vert_\infty$ leads to a gain in statistical efficiency compared to adopting $\vert \bm\delta_j \vert_0^{-1/2} \vert \bm\delta_j \vert_2$, and the gain may become more significant in higher dimensions. 

\cref{fig:intro} empirically illustrates that indeed, the intrinsic difficulty in detecting a change is determined by $\vert \bm\Sigma \bm\delta_j \vert_\infty$ rather than by $\vert \bm\delta_j \vert_0$ and/or $\vert \bm\delta_j \vert_2$. 
The proposed McScan and two recent proposals (MOSEG, \citealp{cho2022high}, and CHARCOAL, \citealp{gao2022sparse}) 
all perform well, regardless of the sparsity $\vert \bm\delta_j \vert_0$ (kept constant in each column), when $\vert \bm\Sigma \bm\delta_j \vert_\infty$ is large, while their performance deteriorates when $\vert \bm\Sigma \bm\delta_j \vert_\infty$ is small even though $\vert \bm\delta_j \vert_2$ remains constant across all scenarios. 
Their good performance in the dense scenario is particularly surprising since such situations have explicitly been excluded from consideration in the existing literature. 
We explore these hitherto unexplored aspects of the change point problem, namely the factors determining the detectability of changes in the high-dimensional regression setting, and provide a first theoretical guarantee without imposing strict sparsity on $\bm\beta_j$ or $\bm\delta_j$.

\subsection{Post-segmentation inference}

In real-world applications, the coefficients in linear regression are often non-sparse \citep{bradic2022testability}.
Thus, it is more plausible to impose sparsity on the components of the regression parameters undergoing the shifts at each change point \citep{gao2022sparse}, while allowing for each $\bm\beta_j$ to be non-sparse.
Then, the interest lies in directly estimating and inferring about the differential parameters $\bm\delta_j$, as put forward in \ref{q:two}.

To this end, we propose two estimators for $\bm\delta_j$ which 
share connections with the literature on direct estimation of differential networks \citep{zhao2014direct, yuan2017differential} and sparse linear discriminant analysis \citep{cai2011direct}. 
However, unlike these papers, our setting brings in the additional uncertainty stemming from the change point detection step, which is fully accounted for in our theoretical investigation. 
As a representative example, under Gaussianity, the proposed estimator satisfies $\max_{1 \le j \le q} \sqrt{\Delta_j} \vert \wh{\bm\delta}_j - \bm\delta_j \vert_2 = O_P(\sqrt{\mathfrak{s} \log(p \vee n)})$ with $\Delta_j = \min(\cp_j - \cp_{j - 1}, \cp_{j + 1} - \cp_j)$ and $\mathfrak{s} = \max_{1 \le j \le q} \vert \bm\delta_j \vert_0$ (\cref{prop:dp:est}), a result comparable to those derived for high-dimensional linear regression in stationary settings (i.e.\ no changes). 
Empirically, this direct estimation approach outperforms the naive alternative of taking the difference of regression parameter estimators separately obtained from neighbouring segments (\cref{sec:sim:dpe}).

Further, we address the problem of inferring about $\bm\delta_j$ across its $p$ coordinates, by proposing a de-sparsified estimator $\wc{\bm\delta}_j$. Built on high-dimensional central limit theorems \citep{CCK23}, we derive a non-asymptotic bound on the Gaussian approximation of $\sqrt{\Delta_j} \vert \wc{\bm\delta}_j - \bm\delta_j \vert_\infty$ (\cref{thm:two}).
Accompanied by a bootstrap procedure, this enables simultaneous inference about $\delta_{ij}, \, 1 \le i \le p$, 
and provides a first solution to the thus-far unaddressed, yet important, inferential problem under the change point setting. . 

\smallskip
\textbf{Role of (non-)sparsity.} To summarise, we shed light on the different roles played by the sparsity of the differential parameters $\bm\delta_j$'s in detecting and inferring about changes under the model~\eqref{eq:model}, without requiring the exact sparsity of $\bm\beta_j$'s. For the detection of the multiple change points, the proposed McScan only requires $\vert \bm\Sigma \bm\delta_j \vert_\infty$ to be large enough without imposing any condition on $\vert \bm\delta_j \vert_0$, yet achieves consistency in multiple change point detection. 
For the estimation of changes in high dimensions, we propose to directly estimate $\bm\delta_j$ via $\ell_1$-regularisation and demonstrate that the sparsity of $\bm\delta_j$ plays a role similar to its counterpart in the standard regression setting.
Finally, for (simultaneous) inference about the coordinates of $\bm\delta_j$, we make a stronger assumption on the sparsity of $\bm\delta_j$ and additionally, (approximate) sparsity of $\Cov(\mbf x_t)^{-1}$ (\cref{rem:inference}). This is in line with the emerging literature on statistical inference under non-sparsity in high dimensions.

\smallskip
\textbf{Organisation of the paper.}
\cref{sec:cp} introduces the McScan method for the multiple change point detection problem and establishes its theoretical consistency. \cref{sec:diff} presents the post-segmentation procedures for direct estimation and inference about the differential parameters. Numerical experiments (\cref{sec:num}) and an application to a macroeconomic dataset (\cref{sec:data}) demonstrate the competitiveness of the proposed suite of tools, and \cref{sec:conc} concludes the paper. 
Further discussions on the change point problem including a literature overview, proofs of all theoretical results and additional numerical results are provided in the Appendix. An implementation of the proposed methods is provided in the R package \texttt{inferchange} available at \url{https://github.com/tobiaskley/inferchange}.

\smallskip
\textbf{Notation.}
For a positive integer $m$, we write $[m] = \{1, \ldots, m\}$. 
For some $\nu > 0$ and a random variable $X$, we write $\Vert X \Vert_\nu = \bigl(\E(\vert X \vert^\nu)\bigr)^{1/\nu}$.
For a matrix $\mbf A = [a_{ij}] \in \R^{m \times n}$, we write $\vert \mbf A \vert_0 = \sum_{i \in [m]} \sum_{j \in [n]} \mathbb{I}_{\{a_{ij} \ne 0\}}$, $\vert \mbf A \vert_2 = \sqrt{ \sum_{i \in [m]} \sum_{j \in [n]} \vert a_{ij} \vert^2}$, $\vert \mbf A \vert_1 = \sum_{i \in [m]} \sum_{j \in [n]} \vert a_{ij} \vert$ and $\vert \mbf A \vert_\infty = \max_{i \in [m]} \max_{j \in [n]} \vert a_{ij} \vert$. Also, we define $\Vert \mbf A \Vert_1 = \max_{j \in [n]} \sum_{i \in [m]} \vert a_{ij} \vert$.  
For sequences $\{a_m\}$ and $\{b_m\}$ of positive numbers, we write $a_m \lesssim b_m$ or equivalently $a_m = O(b_m)$, if $a_m \le C b_m$ for some finite constant $C > 0$. If $a_m \lesssim b_m$ and $b_m \lesssim a_m$, we write $a_m \asymp b_m$. We write $a \vee b = \max(a, b)$ for $a, b \in \R$.
By $\mbf 0$ and $\mbf I$, we denote the vector of zeros and the identity matrix, respectively, whose dimensions depend on the context.

\section{Multiscale covariance scanning for data segmentation}
\label{sec:cp}
\subsection{Methodology}
\label{sec:not}

\cref{sec:size:change} demonstrates that the sparsity of either $\bm\beta_j$ or $\bm\delta_j = \bm\beta_j - \bm\beta_{j - 1}$, is not necessary for the detection of change points, and argues that statistical and computational efficiency can be gained by accessing the covariance weighted differential parameter $\bm\Sigma \bm\delta_j$. 
For this, we observe that for $\cp_{j-1} < k < \cp_{j+1}$,
\begin{align}
\label{eq:cov:diff}
\bm\gamma_{k, \cp_{j + 1}} - \bm\gamma_{\cp_{j - 1}, k} = \min\l\{\frac{\cp_{j}-\cp_{j-1}}{k - \cp_{j-1}},\,\frac{\cp_{j+1} - \cp_{j}}{\cp_{j+1} - k}\r\}\bm\Sigma \bm\delta_j,
\end{align}
with $\bm\gamma_{a, b} = (b - a)^{- 1} \sum_{t = a + 1}^b \Cov(\mbf x_t, Y_t)$.
The magnitude of the difference in local covariances is maximised at $k = \cp_j$, which suggests that change points are detectable by examining the changes in $\Cov(\mbf x_t, Y_t)$. Thus motivated, we consider the detector statistic
\begin{align}
\label{eq:detector}
T_{s, k, e} = \sqrt{\frac{(k - s)(e - k)}{e - s}} \l\vert \wh{\bm\gamma}_{k, e} - \wh{\bm\gamma}_{s, k} \r\vert_\infty \text{ \ with \ } \wh{\bm\gamma}_{a, b} = \frac{1}{b - a} \sum_{t = a + 1}^b \mbf x_t Y_t
\end{align}
for $s < k < e$, over some interval $(s, e]$ with $0 \le s < e \le n$.
The statistic $T_{s, k, e}$ is an empirical counterpart of $\vert\bm\gamma_{k, e} - \bm\gamma_{s, k}\vert_{\infty}$ with a location-based scaling, and scanning for the maximiser of $T_{s, k, e}$ over $s < k < e$, identifies any potential change point in $(s, e]$.

\begin{rem}[Choice of $\ell_\infty$-norm]
\label{rem:norm}
The important role played by the choice of the norm aggregating information across the variables, is well-studied in the high-dimensional change point literature \citep{cho2021data}: Typically, $\ell_\infty$-norm is well-adapted to detecting a sparse change, while $\ell_2$-norm is suited for detecting a dense one. For the mean change point detection, \cite{liu2021minimax} identify a phase transition with respect to the sparsity level in the minimax detection rate and propose an accompanying optimal test, all under the assumption of spatial independence. On the other hand, in the presence of spatial dependence, they noted the challenge in deriving the minimax rate for general degree of sparsity. We also refer to \cite{horvath2022detecting} where both the asymptotic null distribution of an $\ell_2$-norm-based test statistic and its detection boundary, are shown to depend on the degree of dependence.
In the current setting under~\eqref{eq:model}, the coordinates of $\mbf x_t Y_t$ are dependent even when $X_{it}, \, i \in [p]$, are independent, with the covariance determined by (possibly) non-sparse $\bm\beta_j$. We adopt $T_{s, k, e}$ measuring the difference in local covariances via $\ell_\infty$-norm, to approximate the (appropriately scaled) $\vert \bm\Sigma \bm\delta_j \vert_\infty$. It in turn serves as a proxy for $\bm\delta_j^\top \bm\Sigma \bm\delta_j$ that determines the detectability of the corresponding change point (\cref{lem:tv}).
\end{rem}
For the detection and estimation of multiple, possibly heterogeneously spaced change points, we propose the \emph{\underline{m}ultiscale \underline{c}ovariance \underline{scan}ning} (McScan) method that searches for large discrepancies in local sample covariances between $Y_t$ and $\mbf x_t$ over a set of strategically selected intervals.
Specifically, for statistical and computational advantages, the \emph{seeded intervals} \citep{kovacs2020seeded} defined below are chosen for the purpose.

\begin{defn}[Seeded intervals]
\label{def:seed}
	The collection of seeded intervals is defined as 
	\begin{align*}
	\mathbb{M} \;=\;\bigcup_{k = 1}^{\lceil\log_2(n)\rceil}\l\{\bigl(\lfloor(i-1)r_k\rfloor, \, \lceil({i+1})r_k\rceil\bigr]: \; i = 1, \ldots, \lceil n/r_k \rceil - 1,\; r_k = n2^{-k}\r\}.
	\end{align*}
\end{defn}
By construction, the cardinality of $\mathbb{M}$ is $O(n \log(n))$.
Exploiting the deterministic and multiscale construction of $\mathbb{M}$, we can systematically zoom in the neighbourhoods of individual change points. 
Over each of the thus-generated seeded intervals $(a_\ell, b_\ell]$ in $\mathbb{M}$, McScan computes the series of detector statistics in~\eqref{eq:detector} and identifies a candidate estimator of a change point as $\mathop{\arg\max}_{a_\ell < k < b_\ell} T_{a_\ell, k, b_\ell}$. To obtain the final estimators, we adopt the narrowest-over-threshold (NOT) selection rule of \cite{baranowski2019narrowest}, originally proposed for univariate mean change detection. The NOT rule iteratively selects the shortest interval over which the local covariance difference measured by $\mathop{\max}_{a_\ell < k < b_\ell} T_{a_\ell, k, b_\ell}$,
exceeds a given threshold.
In doing so, McScan locates seeded intervals that are likely to contain one and only one change point, and thus \enquote{translates} the problem of multiple change point detection into multiple problems of single change point detection.

Specifically, the McScan algorithm proceeds in the following steps.
\begin{enumerate}[wide, label = {\bf Step~\arabic*:}]
	\setcounter{enumi}{-1}
	\item Take in the trimming parameter $\varpi_{n, p} \ge 0$ and the threshold $\pi_{n, p}$ as input arguments. 
	Set $\wh{\Cp} = \emptyset$, $\mathbb{M} = \{(a_\ell, b_\ell]\}$ the seeded intervals in \cref{def:seed} and $\mathbb{L} = [ \vert \mathbb{M} \vert ]$.
	
	\item For each $\ell \in \mathbb{L}$, if $b_\ell - a_\ell \ge 2\varpi_{n, p} + 1$, identify $k_\ell = \mathop{\arg\max}_{a_\ell + \varpi_{n, p} < k < b_\ell - \varpi_{n, p}} T_{a_\ell, k, b_\ell}$ and set $T_\ell = T_{a_\ell, k_\ell, b_\ell}$.
	If $b_\ell - a_\ell \le 2\varpi_{n, p}$, remove such $\ell$ from $\mathbb{L}$.
	
	\item Identify $\ell^\circ = \mathop{\arg\min}_{\ell \in \mathbb{L}: \, T_\ell > \pi_{n, p} } (b_\ell - a_\ell)$, set $\wh\cp = k_{\ell^\circ}$ and update $\wh{\Cp} \leftarrow \wh{\Cp} \cup \{ \wh\cp \}$. 
	
	\item Update $\mathbb{L} \leftarrow \mathbb{L} \setminus \{ \ell \in \mathbb{L}: \, \wh\cp \in (a_\ell, b_\ell] \}$. 
    
        \item Repeat Steps~2--3 until $\mathbb{L} = \emptyset$.
\end{enumerate}

\begin{rem}[Random intervals]
	It is possible to replace the seeded intervals with randomly generated intervals considered in \citet{fryzlewicz2014wild} and \citet{baranowski2019narrowest}. 
	This would still lead to consistency in multiple change point detection as in \cref{thm:one} below since, with a sufficiently large number of randomly generated intervals, we will have an interval well-suited for detecting each change point with probability tending to one (for the precise description, see~\eqref{eq:cond:M} in Appendix~\ref{sec:pf:thm:one}). However, with the deterministic seeded intervals, we can utilise their regular and recursive structure to enhance computational efficiency, which is not possible with random intervals, see \citet{kovacs2020seeded}. 
\end{rem}

\begin{rem}[Computational complexity]
	\label{rem:comp}
	In the implementation of McScan, we first compute the partial sums $\{\sum_{t= 1}^s\mbf x_t Y_t : \,  s = 1, \ldots, n\}$, which requires an $O(pn)$ runtime. Then each evaluation of $T_{a_\ell, k, b_\ell}$ in Step~1 takes an $O(p)$ runtime. Thus, combined with that the total length of seeded intervals is $O(n\log(n))$, the (worst-case) runtime of McScan is $O(pn\log(n))$. This can be slightly improved to $O(pn)$ if we replace the full grid search of McScan on each seeded interval with the optimistic search strategy proposed in \citet{kovacs2020optimistic}. In fact, in doing so, we can further achieve a runtime of order $O\big( pn \min(\Delta_{\min}^{-1} \log(n), 1) \big)$, provided that the partial sums of data are pre-computed and a lower bound $\Delta_{\min}$ on the minimum spacing between the change points is a priori known.
\end{rem}

We argue that covariance scanning is preferable to directly searching for changes in the local estimators of $\bm\beta_j$.
Firstly, this approach bypasses locally estimating the regression parameters, and thus alleviates the necessity to impose exact sparsity on $\bm\beta_j$.
Also, as discussed in Remark~\ref{rem:comp}, McScan is considerably cheaper computationally, compared to performing $\ell_1$-regularised estimation $O(n^2)$ times as in \cite{xu2022change} or $O(n)$ times as in \cite{cho2022high}, for example. 
Besides, McScan does not require the selection of regularisation parameters which adds numerical stability to its output. 
In addition to these computational benefits, 
we show the statistical efficiency of McScan in the next section.

\subsection{Theoretical properties}
\label{ss:th:pt}

We establish the consistency of $\wh{\Cp} = \bigl\{\wh\cp_j, \, j \in [\wh q]: \, \wh\cp_1 < \cdots < \wh\cp_{\wh q} \bigr\}$ returned by McScan, in estimating $\Cp = \bigl\{\cp_j, \, j \in [q]: \, \cp_1 < \cdots < \cp_q \bigr\}$ under general conditions permitting serial dependence and non-Gaussianity as well as ultra-high dimensionality (i.e.\ $p = o(\exp(n))$).
Firstly, we make the following assumption on the distribution of $\mbf Z_t = (\mbf x_t^\top, \vep_t)^\top$ which is commonly found in the relevant literature.
\begin{assum}[Distribution of $\mbf Z_t$]
	\label{assum:xe}
	\begin{enumerate}[wide, label = (\roman*)]
		\item  
		$\E(\mbf x_t) = \mbf 0$ and $\Cov(\mbf x_t) = \bm\Sigma$ for all $t$.
		\item 
		$\E(\vep_t) = 0$, $\Var(\vep_t) = \sigma_\vep^2 \in (0, \infty)$ and $\Cov(\mbf x_t, \vep_t) = \mbf 0$ for all $t$.
	\end{enumerate}
\end{assum}

In quantifying dependence in $\{\mbf Z_t\}_{t \in \Z}$, we adopt the framework of functional dependence; first introduced by \cite{Wu05} and extended to high-dimensional settings in \cite{zhang2017gaussian}, this general framework incorporates a plethora of stochastic processes, both linear and non-linear.
Let $\{\bm\xi_0^\prime, \bm\xi_t, t\in \Z\}$ be a sequence of independent and identically distributed (i.i.d.)~random elements. 
We assume that $\mbf Z_t$ admits a representation $\mbf Z_t = \mc G(\mc F_t) \in \R^{p + 1}$ with an $\R^{p+1}$-valued measurable function $\mc G$ and $\mc F_t = (\ldots, \bm\xi_{t-1}, \bm\xi_{t})$. Introduce $\mbf Z_{t, \{0\}} = \mc G(\mc F_{t, \{0\}})$ with $\mc F_{t, \{0\}} = (\ldots, \bm\xi_{-1}, \, \bm\xi_0^\prime,\, \bm\xi_{1}, \ldots, \bm\xi_{t})$ being a coupled version of~$\mc F_t$.  Then, we measure the degree of temporal and spatial dependence in $\{\mbf Z_t\}_{t \in \Z}$ by
$\Vert U_{\cdot}(\mbf a) \Vert_{\nu} = \sum_{t = 0}^\infty \zeta_{t, \nu}(\mbf a)$ with $\zeta_{t, \nu}(\mbf a) = \l\Vert \mbf a^\top \mbf Z_t - \mbf a^\top \mbf Z_{t, \{0\}} \r\Vert_\nu$,  and $\Vert{ W_{\cdot}(\mbf a, \mbf b) }\Vert_{\nu} = \sum_{t = 0}^\infty \zeta_{t, \nu}(\mbf a, \mbf b)$  with $\zeta_{t, \nu}(\mbf a, \mbf b) = \l\Vert \mbf a^\top \mbf Z_t \mbf Z_t^\top \mbf b - \mbf a^\top \mbf Z_{t, \{0\}} \mbf Z_{t, \{0\}}^\top \mbf b \r\Vert_\nu$, for $\mbf a, \mbf b \in \mathbb{B}_2(1) = \{\mbf a: \, \vert \mbf a \vert_2 \le 1\}$.
Further, we denote the dependence adjusted sub-exponential norm
of $W_t(\mbf a, \mbf b)$ by
$\Vert W_{\cdot}(\mbf a, \mbf b) \Vert_{\psi_\kappa}
= \sup_{\nu \ge 2} \nu^{-\kappa} \Vert{ W_{\cdot}(\mbf a, \mbf b) }\Vert_{\nu}$ for some $\kappa \ge 0$. With these definitions, we assume the following. 

\begin{assum}[Functional dependence]
	\label{assum:func:dep}
	There exists a constant $\Xi \in (0, \infty)$ such that {\bf one} of the following three conditions is met: 
	\begin{enumerate}[wide, label = (\roman*)]
		\item \label{assum:fd:exp} $\sup_{\mbf a, \mbf b \in \mathbb{B}_2(1)} \Vert W_{\cdot}(\mbf a, \mbf b) \Vert_{\psi_\kappa} \le \Xi$ with some $\kappa \ge 0$, 
  
		\item \label{assum:fd:gauss} $\{\mbf Z_t\}_{t \in \Z}$ is Gaussian and $\sup_{\mbf a \in \mathbb{B}_2(1)} \Vert U_{\cdot}(\mbf a) \Vert_2 \le \Xi^{1/2}$, or 
  \item \label{assum:fd:ind}
$\{\mbf Z_t\}_{t \in \Z}$ is a sequence of independent sub-Weibull random vectors of order $r \in (0, 2]$, i.e.\ $\sup_{\mbf a \in \mathbb{B}_2(1)} \sup_{\nu \ge 2} \nu^{-1/r} \Vert \mbf a^\top \mbf Z_t \Vert_\nu \le \Xi^{1/2}$. 
\end{enumerate}
\end{assum}

\cref{assum:func:dep}~\ref{assum:fd:exp}--\ref{assum:fd:gauss} permit short-range dependence \citep{wu2016performance}, with \ref{assum:fd:exp} placing a weaker assumption on the tail behaviour of $\mbf Z_t$.
They are fulfilled e.g.\ for linear processes with algebraically decaying coefficients (see Lemma~B.3 of \citealp{cho2022high}).
Condition~\ref{assum:fd:ind} assumes temporal independence while still allowing for non-Gaussianity \citep{KuCh22}. 
Let us define
\begin{align}
\psi_{n, p} &= \l\{\begin{array}{ll}
\log^{\kappa + 1/2}(p \vee n) & \text{under Assumption~\ref{assum:func:dep}~\ref{assum:fd:exp},}  \\
\sqrt{\log(p \vee n)} & \text{under Assumption~\ref{assum:func:dep}~\ref{assum:fd:gauss},}  \\
\log^{1/r}(p \vee n) & \text{under Assumption~\ref{assum:func:dep}~\ref{assum:fd:ind}.}  \\
\end{array}
\r. 
\label{eq:psi}
\end{align}
An important consequence of \cref{assum:func:dep} is that uniformly over intervals of length sufficiently large, the deviation of a quadratic form of $\mbf Z_t$ is bounded by $C_0 \psi_{n, p}$ for some constant $C_0 \in (0, \infty)$ (see \cref{lem:exp}).
We refer to \citet[Remark~4]{wu2016performance} for the optimality of the choice of $\psi_{n, p}$ in certain serially dependent cases.

Denoting the largest (resp.\ smallest) eigenvalue of $\bm\Sigma$ by $\Lambda_{\max}(\bm\Sigma)$ (resp.\  $\Lambda_{\min}(\bm\Sigma)$), \cref{assum:func:dep} implicitly places an upper bound on $\Lambda_{\max}(\bm\Sigma)$. However, in investigating the consistency of McScan, we do not require $\Lambda_{\min}(\bm\Sigma)$ to be bounded away from zero, which is distinguished from the relevant literature; see Table~\ref{tab:comp} in Appendix~\ref{sec:comparison}. Instead, \cref{assum:size} below places a lower bound directly on $\vert \bm\Sigma \bm\delta_j \vert_\infty$ such that even when $\Lambda_{\min}(\bm\Sigma)$ is close to or exactly zero, change points are detectable by McScan. 

Let us write $\bm\mu_j = \bm\beta_j + \bm\beta_{j - 1}$ for all $j \in [q]$, such that $\bm\beta_j = (\bm\mu_j + \bm\delta_j)/2$ and $\bm\beta_{j - 1} = (\bm\mu_j - \bm\delta_j)/2$.
Further, let $\Psi = \max_{j \in [q]} \Psi_j$ with $\Psi_j = 1 + \vert \bm\delta_j \vert_2 + \vert \bm\mu_j \vert_2$.
The following assumption specifies the detection boundary for McScan in terms of the size of covariance-weighted differential parameters $\bm\Sigma \bm\delta_j$, and the spacing between the change points $\Delta_j = \min(\cp_j - \cp_{j - 1}, \cp_{j + 1} - \cp_j)$.
\begin{assum}[Detection boundary]
	\label{assum:size} 
	There is a large enough constant $c_0 > 0$ depending on $\Xi$ and $\kappa$ only, such that with $\psi_{n, p}$ defined in~\eqref{eq:psi}, $\vert \bm\Sigma \bm\delta_j \vert_\infty^{2} \Delta_j \, \ge\, c_0 \Psi^2 \psi_{n, p}^2$ for all $j \in [q]$.
\end{assum}
Assumptions~\ref{assum:func:dep} and~\ref{assum:size} jointly imply a lower bound on the segment length
$\min_{j \in [p]} \Delta_j \gtrsim \psi_{n, p}^2$, from that $\vert \bm\Sigma \bm\delta_j \vert_{\infty} \le \max_{i \in [p]} \vert \bm\Sigma_{i \cdot}\vert_2 \vert\bm\delta_j\vert_2 \le \Lambda_{\max}(\bm\Sigma) \Psi_j$. 

\begin{rem}[Size of $\Psi$]
\label{rem:Psi}
Noting that $\Var(Y_t) = \sigma_\vep^2 + \bm\beta_j^\top \bm\Sigma \bm\beta_j \ge \sigma_\vep^2 + \Lambda_{\min}(\bm\Sigma) \vert \bm\beta_j \vert_2^2$ for $t \in \{\cp_j + 1, \ldots, \cp_{j + 1}\}$,
in the relevant literature, $\vert \bm\beta_j \vert_2$ and/or $\vert \bm\delta_j \vert_2$ are often assumed to be bounded (hence also $\Psi_j$), see Appendix~\ref{sec:comparison}. 
However, these quantities may grow with $p$ in high dimensions, unless $\bm\beta_j$'s and $\bm\delta_j$'s are approximately sparse.
	We thus make the dependence on $\Psi_j$ or $\Psi$ explicit in our theoretical investigation where as expected, the change point problem becomes more difficult as their values increase.
\end{rem}

\begin{thm}[Consistency of McScan]
	\label{thm:one}
	Let Assumptions~\ref{assum:xe}, \ref{assum:func:dep} and~\ref{assum:size} hold. 
	Set $\pi_{n, p}$ to satisfy $c^\prime \Psi \psi_{n, p} < \pi_{n, p} < c^{\prime\prime} \min_{j \in [q]} \vert \bm\Sigma \bm\delta_j \vert_\infty \sqrt{\Delta_j}$ for some constants $c^\prime, c^{\prime\prime} > 0$ fulfilling $(c^{\prime\prime})^{-1} c^\prime < c_0$, and also set $\varpi_{n, p} = C_1 \psi_{n, p}^2$ for some constant $C_1 > 0$, for $\psi_{n, p}$ in~\eqref{eq:psi}.
	Then, there exist some constants $c_1, c_2, c_3 \in (0, \infty)$ such that 
	$\p(\mc S_{n, p}) \ge 1 - c_2 (p \vee n)^{-c_3}$, where
	\begin{align*}
	\mc S_{n, p} = \l\{ \wh q = q \text{ \ and \ } \vert \wh\cp_j - \cp_j \vert \le 
	c_1 \vert \bm\Sigma\bm\delta_j\vert_\infty^{-2}\Psi_j^2 \psi_{n, p}^2 \text{ for every }j \in [q]\r\}.
	\end{align*}
\end{thm}

\begin{rem}[Comparison of detection boundaries and rates of estimation]
\label{rem:dlb}
For the data segmentation problem in~\eqref{eq:model}, detectability of change points is jointly determined by $\bm\delta_j$, $\Delta_j$ and~$\bm\Sigma$.
To facilitate comparison, consider the Gaussian setting and set $\Psi = O(1)$.
The state-of-the-art procedures which assume the segment-wise sparsity of $\bm\beta_j$, achieve detection consistency with probability tending to one provided that 
$\min_{j \in [q]} \underline{\sigma}^2 \mathfrak{s}_\beta^{-1} \vert \bm\delta_j \vert_2^2 \Delta_j \gtrsim \log(p \vee n)$,
with $\underline{\sigma} = \Lambda_{\min}(\bm\Sigma)$ and $\mathfrak{s}_\beta = \max_{j \in [q]} \vert \bm\beta_j \vert_0$. \cite{gao2022sparse} impose a similar condition, $\min_{j \in [q]} \vert \bm\delta_j \vert_0^{-1} \vert \bm\delta_j \vert_2^2 \Delta_j \gtrsim \log^2(p)$, requiring the sparsity of the differential parameters, in a more restrictive setting where $X_{it}, \vep_t \sim_{\iid} \mc N(0, 1)$, $\Delta_j \asymp n$, $p < n$ and $n - p \asymp n$. In contrast, thanks to~\eqref{eq:size:change}, 
Assumption~\ref{assum:size} accommodates a broader parameter space than those permitted by the detection boundaries of the existing methods, all without assuming the exact sparsity of regression or differential parameters. This demonstrates the greater generality of our results, see also Appendix~\ref{sec:comparison} for a detailed comparison and an example illustrating a case of an approximately sparse change where the gap in detection boundaries diverges with~$p$.

Similar conclusions apply to the rate of localisation attained by McScan. Our rate captures the difficulty of detecting individual changes with the multiplicative factor of $\vert \bm\Sigma \bm\delta_j \vert_\infty^{-2}$, and thus generally improves upon the rates attained by the existing methods. In fact, we establish that McScan is \emph{minimax (near-)optimal} in both the detection boundary and the estimation rate, matching the lower bounds (possibly up to log factors) for Gaussian models (see Lemmas~\ref{lem:lb} and \ref{lem:lb:id} in \cref{ss:lb}). Notably, our results also confirm that in the case of sparse changes, the logarithmic factors in the detection boundary and the estimation rate are necessary.

Finally, in the simpler situation where the covariance matrix $\bm\Sigma$ is known, one might be able to impose a slightly weaker detection boundary than \cref{assum:size}. This would suggest the possibility of accommodating broader scenarios where regression parameters are non-sparse. For related discussions in standard linear regression settings, see \citet{verzelen2018adaptive}. This opens an interesting direction for future research.  
\end{rem}

Regarding the model~\eqref{eq:model} with $q = 1$ from the two-sample testing point of view, \cite[Proposition~9]{gao2022two} showed that testing with a dense nuisance parameter is not feasible (in the minimax sense) when $p \ge \min(\cp_1, n - \cp_1)$.  Our detection consistency result in \cref{thm:one}, which allows $p \gg n$, does not contradict their claim, as their result holds conditionally on the design matrix such that when $p \ge \min(\cp_1, n - \cp_1)$, there exists some $\bm\delta_1$ that avoids detection due to singularity. In contrast, we approach the problem by considering the changes as those in the joint distribution of $(\mbf x_t, Y_t)$, and impose \cref{assum:size} which excludes the case where $\bm\Sigma \bm\delta_j \approx \mbf 0$, see also Lemma~\ref{lem:tv}.

\section{Post-segmentation estimation and inference}
\label{sec:diff}

\subsection{Direct estimation of differential parameters}
\label{sec:diff:est}

A naive approach to estimating the differential parameter $\bm\delta_j$ is to take the difference of the separately obtained estimators of $\bm\beta_{j - 1}$ and $\bm\beta_j$ which, however, requires the assumption of segment-wise sparsity. 
Instead, motivated by the observation that 
$\bm\Sigma \bm\delta_j = \bm\gamma_{\cp_j, b} - \bm\gamma_{a, \cp_j}$ for any $a$ and $b$ satisfying $\cp_{j - 1} \le a < \cp_j < b \le \cp_{j + 1}$ (see also~\eqref{eq:cov:diff}),
we propose to directly estimate $\bm\delta_j$ via $\ell_1$-regularisation. 
In doing so, we assume the sparsity of the differential parameters while $\bm\beta_j$'s may be non-sparse. 
We introduce two such estimators, both of which can be efficiently computed by convex optimisation algorithms and share the same theoretical guarantees. 
On a generic interval $(s, e]$ with $0 \le s < e \le n$, let us write $\bm\delta_{s, e}(k) = \bm\Sigma^{-1} (\bm\gamma_{k, e} - \bm\gamma_{s, k})$.
\smallskip

\textbf{LOPE: $\ell_1$-penalised estimator.} 
Consider the following quadratic function of $\mbf a \in \R^p$
\begin{align*}
\mc L(\mbf a; \bm\Sigma, \bm\gamma_{s, k}, \bm\gamma_{k, e})
=  \frac{1}{2} \mbf a^\top \bm\Sigma \mbf a - \mbf a^\top (\bm\gamma_{k, e} - \bm\gamma_{s, k})
=  \frac{1}{2}\left\vert\bm\Sigma^{1/2}\bigl( \mbf a - \bm\delta_{s, e}(k)\bigr)\right\vert_2^2 - \frac{1}{2}\bm\delta_{s, e}(k)^\top\bm{\Sigma}\bm\delta_{s, e}(k),
\end{align*}
which measures the closeness between $\mbf a$ and $\bm\delta_{s, e}(k)$, and is minimised when $\mbf a = \bm\delta_{s, e}(k)$. 
Replacing $\bm\gamma_{a, b}$ by $\wh{\bm\gamma}_{a, b}$ defined in~\eqref{eq:detector}, and $\bm\Sigma$ by $\wh{\bm\Sigma}_{s, e} = (e - s)^{-1} \sum_{t = s + 1}^e \mbf x_t \mbf x_t^\top$, 
we propose the \emph{$\underline{\ell}$-\underline{o}ne-\underline{pe}nalised} (LOPE) estimator of $\bm\delta_{s, e}(k)$, as
\begin{align}
\label{eq:lasso:est}
\wh{\bm\delta}_{s, e}(k) \;\in\; \mathop{\arg\min}_{\mbf a \in \R^p} \, \mc L\l( \mbf a; \wh{\bm\Sigma}_{s, e}, \wh{\bm\gamma}_{s, k}, \wh{\bm\gamma}_{k, e} \r) + 
\lambda \sqrt{\frac{e-s}{(k-s)(e-k)}} \vert \mbf a \vert_1
\end{align}
for some $\lambda > 0$. The problem in~\eqref{eq:lasso:est} can be re-formulated into a modified Lasso problem~as
\begin{align*}
\wh{\bm\delta}_{s, e}(k) \;\in\; \mathop{\arg\min}_{\mbf a \in \R^p} \, \frac{1}{2(e-s)}\l\vert
\bmx
\frac{e-s}{k-s}\bm Y_{s,k}\\
\frac{e-s}{e-k}\bm Y_{k,e}
\emx
- \bmx
-\mbf X_{s,k} \\
\mbf X_{k,e}
\emx
\mbf a
\r\vert_2^2  + 
\lambda \sqrt{\frac{e-s}{(k-s)(e-k)}} \vert \mbf a \vert_1,
\end{align*}
with $\bm Y_{a, b} = (Y_{a+1}, \ldots, Y_b)^{\top}$ and $\mbf X_{a, b} = [\mbf x_{a+1}, \ldots, \mbf x_b]^\top$. 
\smallskip

\textbf{CLOM: constrained $\ell_1$-minimisation estimator.} Alternatively, we set out to minimise the $\ell_1$-norm of the estimator under a constraint that the covariance weighted estimator closely approximates the difference in local covariances of $\mbf x_t$ and $Y_t$.
With the plug-in estimators as above, we propose the \emph{\underline{c}onstrained $\underline{\ell}$-\underline{o}ne \underline{m}inimisation} (CLOM) estimator 
\begin{align}
\wh{\bm\delta}_{s, e}(k) \;\in\; \mathop{\arg\min}_{\mbf a \in \R^p} \vert \mbf a \vert_1 \;\text{ \ subject to \ }\; \sqrt{\frac{(k - s)(e - k)}{e - s}} \l\vert \wh{\bm\Sigma}_{s, e} \mbf a - \wh{\bm\gamma}_{k, e} + \wh{\bm\gamma}_{s, k} \r\vert_\infty \le \lambda, 
\label{eq:direct:est}
\end{align}
where $\lambda > 0$ is a tuning parameter.
\smallskip

We propose to estimate $\bm\delta_j, \, j \in [q]$, by $\wh{\bm\delta}_j = \wh{\bm\delta}_{a_j, b_j}(\wh\cp_j)$ obtained by either LOPE in~\eqref{eq:lasso:est} or CLOM in~\eqref{eq:direct:est}, where
\begin{align}
\label{eq:ab}
& a_j = \wh\cp_j - \wh{\Delta}_j \text{ \ and \ } b_j = \wh\cp_j + \wh{\Delta}_j, \quad \text{with}
\\
& \wh{\Delta}_j = \min\l(\wh\cp_j - \l\lfloor \frac{2}{3} \wh\cp_{j - 1} + \frac{1}{3} \wh\cp_j \r\rfloor, \;\l\lceil \frac{1}{3} \wh\cp_j + \frac{2}{3} \wh\cp_{j + 1} \r\rceil - \wh\cp_j \r), \nn
\end{align}
which are chosen to isolate each change point $\cp_j$ within the interval $(a_j, b_j]$. By convention, $\wh\cp_0 = 0$ and $\wh\cp_{q+1} = n$.
To investigate the properties of $\wh{\bm\delta}_j$, we make an additional assumption on the eigenvalues of $\bm\Sigma = \Cov(\mbf x_t)$.
\begin{assum}[Spectrum of $\bm\Sigma$]
	\label{assum:xe:iii} 
	$\bm\Sigma$ is positive definite with its smallest and largest eigenvalues satisfying $\underline{\sigma} \le \Lambda_{\min}(\bm\Sigma) \le \Lambda_{\max}(\bm\Sigma) \le \bar{\sigma}$ for some $\underline{\sigma}, \bar{\sigma} \in (0, \infty)$.  
\end{assum}

We denote the sparsity of $\bm\delta_j$ by $\mathfrak{s}_j = \vert \mc S_j \vert_0$ where $\mc S_j$ is the support of $\bm\delta_j = (\delta_{1j}, \ldots, \delta_{pj})^\top$, i.e.\ $\mc S_j = \{i \in [p]: \, \delta_{ij} \ne 0\}$. 

\begin{prop}[Consistency in differential parameter estimation]
	\label{prop:dp:est}
	Suppose that all the conditions made in Theorem~\ref{thm:one} hold in addition to \cref{assum:xe:iii}. Also, let $c_1 c_0^{-1} < 1/6$ and for all $j \in [q]$, assume that
	\begin{align}
	\label{eq:prop:lasso:two}
	\Delta_j \ge 3 (2\underline{\sigma})^{-2} \bigl(64  \mathfrak{s}_j \crsc \log(p) \bigr)^{1 + 2\kappa}
	\end{align}
	for a constant $\crsc$ depending on $\kappa$ (with $\kappa = 0$ under Assumption~\ref{assum:func:dep}~\ref{assum:fd:gauss} and $\kappa = 1/r - 1/2$ under~\ref{assum:fd:ind}).
    Then, setting $\lambda = C_\lambda \Psi_j \psi_{n, p}$ with some $C_\lambda$ that depends on $\kappa$ and $\Xi$ only,
	\begin{align*}
	\l\vert \wh{\bm\delta}_j - \bm\delta_j \r\vert_2 \lesssim \frac{\Psi_j \sqrt{\mathfrak{s}_j} \psi_{n, p}}{\underline{\sigma} \sqrt{\Delta_j}} \text{ \ and \ }
	\l\vert \wh{\bm\delta}_j - \bm\delta_j \r\vert_1 \lesssim \frac{\Psi_j \mathfrak{s}_j \psi_{n, p}}{\underline{\sigma} \sqrt{\Delta_j}}
	\end{align*} 
	for all $j \in [q]$, with probability at least $1 - 2c_2 (p \vee n)^{-c_3}$ with the constants $c_2, c_3 \in (0, \infty)$ as in Theorem~\ref{thm:one}.
	All unspecified constants depend only on $\Xi$ and $\kappa$.
\end{prop}

The above results bear close resemblance to the rates attained by $\ell_1$-regularised estimators of the regression parameter in the standard high-dimensional linear model, on their dependence on the effective sample size $\Delta_j$ and the sparsity $\mathfrak{s}_j$.
In particular, under (sub-)Gaussianity, the rates match those derived for the Lasso and Dantzig selector estimators, see e.g.\ \cite{bickel2009simultaneous}.
As a by-product, papers proposing data segmentation methods under the strict sparsity of $\bm\beta_j$, report the consistency of local Lasso estimators; see e.g.\ Lemma~B.2 of \cite{cho2022high} for a comparable result.

\begin{rem}[Recovery of $\mc S_j$]
\label{rem:linf}
It is well-known that $\ell_1$-regularised estimators do not achieve consistency in the estimation of $\mc S_j$.
Using the arguments in the proof of \cref{prop:dp:est} (see \cref{pf:prop:lasso}), we obtain that with probability tending to one,
\begin{align*}
\sqrt{\Delta_j} \l\vert \wh{\bm\delta}_j - \bm\delta_j \r\vert_\infty \lesssim \psi_{n, p} \min \l( \frac{\Psi_j \sqrt{\mathfrak{s}_j}}{\underline{\sigma}}, \Vert \bm\Omega \Vert_1 \vert \bm\delta_j \vert_1 \r),
\end{align*}
where $\bm\Omega = \bm\Sigma^{-1}$ is the precision matrix of $\mbf x_t$.
Then, the support of $\bm\delta_j$ can be estimated consistently by thresholding the elements of $\wh{\bm\delta}_j$ with some threshold $\mathfrak{t}$, provided that
\begin{align*}
\min_{i \in \mc S_j} \vert \delta_{ij} \vert \, \gtrsim \, \mathfrak{t} \, \gtrsim  \, \frac{\psi_{n, p}}{\sqrt{\Delta_j}} \min\l( \frac{\Psi_j \sqrt{\mathfrak{s}_j}}{\underline{\sigma}}, \Vert \bm\Omega \Vert_1 \vert \bm\delta_j \vert_1 \r).
\end{align*}
This thresholding-based approach, which has been explored in high-dimensional linear regression \citep{van2011adaptive}, however relies on selecting $\mathfrak{t}$ that depends on many unknown parameters, and does not provide any guarantee at a given confidence level.
Towards a practical method with statistical guarantees, in what follows, we propose a procedure for the recovery of $\mc S_j$ via de-sparsification of $\wh{\bm\delta}_j$, which allows for controlling the family-wise error rate across $p$ dimensions.
\end{rem}

\subsection{Simultaneous confidence intervals for differential parameters}
\label{sec:diff:inf}

\subsubsection{De-sparsified estimator for differential parameters}
\label{sec:debiased}

Both LOPE~\eqref{eq:lasso:est} and CLOM~\eqref{eq:direct:est} estimators are biased due to $\ell_1$-regularisation, which calls for a de-sparsification step before their large sample distributions can be derived.
We extend the de-sparsified estimator, originally proposed in the context of high-dimensional linear regression \citep{zhang2014confidence, javanmard2014confidence, van2014asymptotically}, to the change point setting for simultaneous inference about the differential parameters $\bm\delta_j$.
We mention that, assuming that $\bm\beta_j$'s are sparse, \cite{liu2024simultaneous} propose an estimator of $\bm\delta_j$ formed by differencing the de-sparsified estimators of $\bm\beta_j$ and $\bm\beta_{j - 1}$.

As shown below, due to the direct estimation of $\bm\delta_j$, the distribution of its de-sparsified estimator is ``driven'' by a non-standard quadratic form of $\mbf Z_t = (\mbf x_t^\top, \vep_t)^\top$, not to mention that it depends on $a_j$, $b_j$ and $\wh\cp_j$ which are random.
To handle this,
we adopt a sample splitting strategy and partition the data into observations with even and odd indices, which we denote by $\mc D^{\tE} = \{(Y_t^{\tE}, \mbf x_t^{\tE}), \, t \in [n_0]\}$ and $\mc D^{\tO} = \{(Y_t^{\tO}, \mbf x_t^{\tO}), \, t \in [n_0]\}$, respectively, with $n_0 = \lfloor n/2 \rfloor$.
Throughout, the superscripts `E' and `O' denote that the relevant estimators are obtained using the datasets $\mc D^{\tE}$ and $\mc D^{\tO}$, respectively.
For notational simplicity, we assume that the change points $\cp_j$ are in the scale of the index sets for $\mc D^{\ell}$, $\ell \in \{\tE, \tO\}$, i.e.\ the joint distribution of $(Y^\ell_t, \mbf x^\ell_t)$ undergoes changes at $\cp_j, \, j \in [q]$. 

Applying the McScan algorithm to $\mc D^{\tE} $, we obtain a set of change point estimators $\wh{\Cp}^{\tE} = \bigl\{\wh\cp^{\tE}_j, \, j \in [\wh q]: \, \wh\cp^{\tE}_1 < \ldots < \wh\cp^{\tE}_{\wh q} \bigr\}$.
Modifying~\eqref{eq:ab} to accommodate the use of sample splitting, we identify $\wh{\Delta}^{\tE}_j$, $a^{\tE}_j$ and $b^{\tE}_j$ for each $j \in [q]$,
and generate an estimator $\wh{\bm\delta}^{\tE}_j$ of $\bm\delta_j$ by LOPE in~\eqref{eq:lasso:est} or CLOM in~\eqref{eq:direct:est} from~$\mc D^{\tE}$.
Then, we construct a bias-corrected version of $\wh{\bm\delta}^{\tE}_j$ with the adjustment motivated by the Karush--Kuhn--Tucker condition of \eqref{eq:lasso:est}, as 
\begin{align}
\label{eq:debiased}
\wc{\bm\delta}_j &= \wh{\bm\delta}^{\tE}_j - \wh{\bm\Omega}^{\tE} \l( \wh{\bm\Sigma}^{\tO}_{a^{\tE}_j, b^{\tE}_j} \wh{\bm\delta}^{\tE}_j - \wh{\bm\gamma}^{\tO}_{\wh\cp^{\tE}_j, b^{\tE}_j} + \wh{\bm\gamma}^{\tO}_{a^{\tE}_j, \wh\cp^{\tE}_j} \r)
\end{align}
where $\wh{\bm\Sigma}^{\tO}_{a^{\tE}_j, b^{\tE}_j}$, $\wh{\bm\gamma}^{\tO}_{\wh\cp^{\tE}_j, b^{\tE}_j}$ and $\wh{\bm\gamma}^{\tO}_{a^{\tE}_j, \wh\cp^{\tE}_j}$ are obtained from $\mc D^{\tO}$, and $\wh{\bm\Omega}^{\tE}$  an estimator of the precision matrix $\bm\Omega = \bm\Sigma^{-1}$ from $\mc D^{\tE}$ whose choice will be given shortly.
We can rewrite $\wc{\bm\delta}_j$ as 
\begin{align}
\label{eq:debias:decomp}
& \sqrt{\frac{\wh\Delta^{\tE}_j}{2}} ( \wc{\bm\delta}_j - \bm\delta_j ) 
= \frac{1}{\sqrt{b^{\tE}_j - a^{\tE}_j}} \wh{\bm\Omega}^{\tE} 
\sum_{t = a^{\tE}_j + 1}^{b^{\tE}_j} w^{\tE}_{j, t} \mbf x^{\tO}_t \l( \vep^{\tO}_t+ (\mbf x^{\tO}_t)^\top \bar{\bm\mu}^{\tE}_j \r) + o_P(1),
\end{align}
with location-dependent weights $w^{\tE}_{j, t}$ (see~\eqref{eq:weights}) and $\bar{\bm\mu}^{\tE}_j$ defined in Theorem~\ref{thm:two} below.
The leading term in~\eqref{eq:debias:decomp}, which determines the limiting distribution of~$\wc{\bm\delta}_j$, is non-standard compared to those appearing in the de-sparsified Lasso literature that involve the sum of $\mbf x_t \vep_t$ only.
This motivates the proposed sample splitting step for a rigorous, non-asymptotic treatment of the Gaussian approximation error.

\subsubsection{Distribution of the de-sparsified estimator}

In studying the distribution of $\wc{\bm\delta}_j$, we require stronger assumptions than those for consistency in estimating the multiple change points and differential parameters. 

\begin{massum}{\ref{assum:func:dep}$^\prime$}[Independence, sub-Gaussianity]
\label{assum:func:dep:two}
We assume that \cref{assum:func:dep}~\ref{assum:fd:ind} holds with $r = 2$ and additionally, that $\{\mbf x_t\}$ and $\{\vep_t\}$ are independent.
\end{massum}

Under Assumption~\ref{assum:func:dep:two}, we set $\psi_{n, p}  = \sqrt{\log(p \vee n)}$ (see~\eqref{eq:psi}).
Next, we make an assumption on the estimator of the precision matrix.
\begin{assum}[Precision matrix estimator]
\label{assum:omega:w}
Consider $\mbf x_t \in \R^p$, $t \in [n]$, with the common precision matrix $\bm\Omega = \Cov(\mbf x_t)^{-1}$. Define $\mathfrak{s}_{\varrho} = \max_{i \in [p]} \vert \bm\Omega_{i\cdot} \vert_{\varrho}^{\varrho} = \max_{i \in [p]} \sum_{i' \in [p]} \vert \Omega_{ii'} \vert^{\varrho}$ with $\varrho \in [0, 1)$, and let $\wh{\bm\Sigma}$ be the sample covariance of $\mbf x_t$'s. Then, there is an estimator $\wh{\bm\Omega}$, possibly asymmetric, so that on an event $\mc{O}_{n, p}$ with $\p(\mc O_{n, p}) \to 1$ as $n \to \infty$, the following~holds:
\begin{enumerate}[label = (\roman*)]
    \item \label{assum:omega:w:one} $\Vert \wh{\bm\Omega} \Vert_\infty \lesssim \Vert \bm\Omega \Vert_\infty = \Vert \bm\Omega \Vert_1$, 
    \item \label{assum:omega:w:two} $\vert \mbf I_p - \wh{\bm\Omega} \wh{\bm\Sigma} \vert_\infty  \lesssim {\Vert \bm\Omega \Vert_1 \sqrt{\log(p \vee n)/n}}$, and
    \item \label{assum:omega:w:three} $\Vert \wh{\bm\Omega} - \bm\Omega \Vert \lesssim \mathfrak{s}_\varrho \bigl({\Vert \bm\Omega \Vert_1^{\omega}\sqrt{\log(p \vee n)/n}}\bigr)^{1 - \varrho}\le 1/(2\bar{\sigma})$ for some $\varrho \in [0, 1)$ and $\omega \ge 1$.
\end{enumerate}
\end{assum}

The error rate in~\ref{assum:omega:w:three} is minimax optimal when $\omega = 1$ \citep{cai2016estimating}. 
Later in Section~\ref{sec:ci}, we consider the constrained $\ell_1$-minimisation estimator (CLIME; \citealp{cai2011constrained}) applied to the data $\mc{D}^{\tE}$ as $\wh{\bm\Omega}^{\tE}$ and verify Assumption~\ref{assum:omega:w}, 
see \cref{prop:clime:dataE} and also \cref{rem:alter} that discusses alternative estimators. 
We emphasise that the non-asymptotic result on the distribution of $\wc{\bm\delta}_j$ derived in \cref{thm:two} below, continues to hold for any estimator of $\bm\Omega$ satisfying Assumption~\ref{assum:omega:w}.

The next assumption strengthens Assumption~\ref{assum:size} on the size of changes and segment lengths~$\Delta_j$.
\begin{massum}{\ref{assum:size}$^\prime$}[Sample size and sparsity]
	\label{assum:size:two} For all $j \in [q]$, we have
	\begin{align}
	\frac{\Psi_j \max\bigl( \Vert \bm\Omega \Vert_1 \mathfrak{s}_j, \; \Vert \bm\Omega \Vert_1 \Psi \vert \bm\Sigma \bm\delta_j \vert_\infty^{-1}, \; \Psi_j^3 \log(p \vee n) \bigr) } {\sqrt{\Delta_j}} = o\l(\frac{1}{\log^{3/2}(p \vee n)}\r). 
\label{eq:assum:size:two:two} 
\end{align}
\end{massum}

\begin{rem}[Strengthened conditions]
	\label{rem:inference}
	\begin{enumerate}[label = (\roman*)]
		\item 
Under Assumption~\ref{assum:func:dep:two}, the first part of~\eqref{eq:assum:size:two:two} implies the condition~\eqref{eq:prop:lasso:two} (with $\kappa = 0$) as the former requires that $\Vert \bm\Omega \Vert_1 \cdot \mathfrak{s}_j \log(p \vee n) = o(\sqrt{\Delta_j})$.
In the case of $\Vert \bm\Omega \Vert_1 = O(1)$, this requirement parallels the assumption found in 
\cite{liu2024simultaneous} where, directly assuming that $\bm\beta_0$ and $\bm\beta_1$ are sparse, the problem of testing whether $\bm\delta_1 = \mbf 0$ or not, is studied in a single change point setting.
		
		\item 
The second part of the condition in~\eqref{eq:assum:size:two:two} places a stronger requirement on the size of changes compared to Assumption~\ref{assum:size}, namely that $\vert \bm\Sigma \bm\delta_j \vert_\infty \sqrt{\Delta_j} \gg \Vert \bm\Omega \Vert_1 \Psi_j^2 \log^{3/2}(p \vee n)$.
This condition, together with Assumption~\ref{assum:omega:w}~\ref{assum:omega:w:three}, effectively imposes (approximate) sparsity on $\bm\Omega$.
\cite{bradic2022testability} study the problem of testing for a single regression coefficient, say $\mc H_0: \, \beta_1 = b$ against $\mc H_1: \, \beta_1 = b + h$, without assuming sparsity on the vector of coefficients $\bm\beta = (\beta_1, \ldots, \beta_p)^\top$. 
Their Theorem~2 establishes the uniform non-testability when $h \lesssim \mathfrak{s}_0 \log(p) / n$, 
where $\mathfrak{s}_0 = \max_{i \in [p]} \sum_{i' \in [p]} \mathbb{I}_{\{\vert \omega_{ii'} \vert > 0\}}$ denotes the row-wise sparsity of $\bm\Omega = [\omega_{ii'}]$, which ties in with our condition on (approximate) sparsity of $\bm\Omega$.
We also note that in addressing a related but distinct problem of deriving the limit distribution of the change point estimator under~\eqref{eq:model}, \cite{xu2022change} make a condition stronger than that for consistency in change point estimation, namely $\vert \bm\delta_j \vert_2 \sqrt{\Delta_j} \gg \mathfrak{s}_\beta \log^{3/2}(p \vee n)$.
\end{enumerate}
\end{rem}

\begin{thm}[Distribution of $\wc{\bm\delta}_j$]
	\label{thm:two}
	Suppose that Assumptions~\ref{assum:xe}, \ref{assum:func:dep:two}, \ref{assum:size:two},\ref{assum:xe:iii} and \ref{assum:omega:w} hold, and set
	$\lambda = C_\lambda \Psi \sqrt{\log(p \vee n)}$ with $C_\lambda$ that depends only on $\kappa$ and $\Xi$.
	For each $j \in [q]$, let $\mbf V_j$ denote a $p$-variate random vector satisfying 
	\begin{align*}
	& \mbf V_j \vert \mc D^{\tE} \sim \mc N_p\l(\mbf 0, \frac{(\wh\cp^{\tE}_j - a^{\tE}_j)(b^{\tE}_j - \wh\cp^{\tE}_j)}{(\cp_j - a^{\tE}_j)(b^{\tE}_j - \cp_j)} \wh{\bm\Omega}^{\tE} \bm\Gamma_j \l(\wh{\bm\Omega}^{\tE}\r)^\top \r), \text{ \ where}
	\\
	& \bm\Gamma_j = \Cov\l(\mbf x^{\tO}_t\bigl(\vep^{\tO}_t + (\mbf x^{\tO}_t)^\top\bar{\bm\mu}^{\tE}_j\bigr) \big\vert \mc D^{\tE} \r) \text{ \ and \ } \bar{\bm\mu}^{\tE}_j = \frac{(b^{\tE}_j - \cp_j) \bm\beta_{j - 1} + (\cp_j - a^{\tE}_j) \bm\beta_j}{b^{\tE}_j - a^{\tE}_j}.
	\end{align*}
	Then there exists some constant $C$ that depends on $\bar{\sigma}$, $\underline{\sigma}$, $\sigma_\vep$ and $\Xi$ only, 
	such that
	\begin{multline*}
	\sup_{z \in \R}\; \l\vert \p\l( 
	\sqrt{\frac{\wh\Delta^{\tE}_j}{2}}
	\l\vert \wc{\bm\delta}_j - \bm\delta_j \r\vert_\infty \le z \Biggm\vert \mc D^{\tE} \r) - \p\l(\vert \mbf V_j \vert_\infty \le z \Big\vert \mc D^{\tE} \r) \r\vert 
	\\
	\le \; \frac{C \Psi_j \log(p \vee n) \sqrt{\log(p)}}{\sqrt{\Delta_j}} \Bigl( \Psi_j^3 \log(p) 	+ \Vert \bm\Omega \Vert_1 \max( \mathfrak{s}_j, \vert \bm\Sigma \bm\delta_j \vert_\infty^{-1} \Psi_j) \Bigr) 
	+  \p\l( (\mc O_{n, p}^{\tE})^c \r).
	\end{multline*}
\end{thm}

Theorem~\ref{thm:two} makes use of a Gaussian approximation result in high dimensions \citep{CCK23}, as well as taking into account the errors arising from the estimation of $\cp_j$, $\bm\delta_j$ and $\bm\Omega$; under Assumptions~\ref{assum:size:two} and~\ref{assum:omega:w}, the approximation error is controlled as $o(1)$.
Theorem~\ref{thm:two} allows for conducting inference about $\delta_{ij}$ with family-wise error control across $i \in [p]$, without resorting to the Bonferroni correction that does not utilise the dependence among the coordinates, and enables identifying the set $\mc S_j$ 
of the coordinates undergoing the changes.
This useful result comes 
at the price of the stronger Assumptions~\ref{assum:func:dep:two} and~\ref{assum:size:two}, in place of Assumptions~\ref{assum:func:dep} and~\ref{assum:size} that are sufficient for the estimation consistency derived in Proposition~\ref{prop:dp:est}.

The width of the simultaneous confidence intervals constructed based on Theorem~\ref{thm:two} is of the parametric rate up to a logarithmic factor.
To see this, note that $\Vert\bm\Gamma_j \Vert \lesssim \Psi_j^2$ and $\Vert \wh{\bm\Omega}^{\tE} \Vert \lesssim \underline{\sigma}^{-1}$, the latter with probability tending to one from Assumption~\ref{assum:omega:w}. Combining this with the consistency of change point estimators (Theorem~\ref{thm:one}), we have $\max_{i \in [p]} \Var(V_{j, i} \vert \mc D^{\tE}) = O_P(\Psi_j^2)$ which leads to the width of simultaneous confidence intervals being bounded as $O_P(\Psi_j \sqrt{\log(p)/\Delta_j})$.

\subsubsection{Construction of simultaneous confidence intervals}
\label{sec:ci}

For the estimation of $\bm\Omega$, we consider the CLIME \citep{cai2011constrained}, i.e.\
\begin{align}\label{eq:mod:clime}
\wh{\bm\Omega}^{\tE} \;\in\; \mathop{\arg\min}_{\mbf M \in \R^{p \times p}} \vert \mbf M \vert_1 
\text{ \ subject to \ } \sqrt{n_0} \l\vert \mbf M \wh{\bm\Sigma}^{\tE} - \mbf I_p \r\vert_\infty \le \eta
\end{align}
with some tuning parameter $\eta > 0$, where the convex optimisation problem in~\eqref{eq:mod:clime} can be solved in parallel for each row of $\wh{\bm\Omega}^{\tE}$.
The following Proposition~\ref{prop:clime:dataE} establishes that it satisfies Assumption~\ref{assum:omega:w} with $\omega = 2$ on the data $\mc D^{\tE}$.

\begin{prop}[CLIME]
\label{prop:clime:dataE}
Suppose that Assumptions~\ref{assum:xe}, \ref{assum:func:dep} and~\ref{assum:xe:iii} hold, and set $\eta = C_\eta \Vert \bm\Omega \Vert_1 \psi_{n, p}$ for some $C_\eta$ that depends only on $\kappa$ and $\Xi$. 
Then by its construction, we have
$\vert \mbf I_p - \wh{\bm\Omega}^{\tE} \wh{\bm\Sigma}^{\tE} \vert_\infty \le C_\eta \Vert \bm\Omega \Vert_1 \psi_{n, p}/\sqrt{n_0}$.
Additionally, with probability at least $1 - c_2 (p \vee n)^{-c_3}$ with the constants $c_2, c_3 \in (0, \infty)$ as in Theorem~\ref{thm:one}, we have
\begin{align*}
    \Vert \wh{\bm\Omega}^{\tE} \Vert_\infty \le \Vert \bm\Omega \Vert_1 \text{ \ and \ } \l\Vert \wh{\bm\Omega}^{\tE} - \bm\Omega \r\Vert \lesssim \mathfrak{s}_{\varrho}\,
		\l(\frac{C_\eta \Vert \bm\Omega \Vert_1^2 \psi_{n, p} }{\sqrt{n}}\r)^{1 - \varrho} \text{ \ for all \ } \varrho \in [0, 1).
\end{align*}
\end{prop}

\begin{rem}[Alternative estimators of $\bm\Omega$]
\label{rem:alter}
	Besides CLIME, we may consider its adaptive version (ACLIME, \citealp{cai2016estimating}) 
	or the one based on Gaussian graphical modelling (\citealp{liu2017tiger}, extending the nodewise regression of \citealp{MeBu06}).
 They are shown to attain the minimax optimal error rate, namely $\mathfrak{s}_\varrho (\Vert \bm\Omega \Vert_1 \log(p)/n)^{(1 - \varrho)/2}$, in matrix $L_r$-norm for $r \in [1, \infty]$ under independence and (sub-)Gaussianity.
Adopting these methods in place of CLIME does not alter the rate derived in Theorem~\ref{thm:two}, but it weakens the requirement on the (approximate) sparsity of $\bm\Omega$ in Assumption~\ref{assum:omega:w}~\ref{assum:omega:w:three} by the factor of $\Vert \bm\Omega \Vert_1^{1 - \varrho}$ with $\omega = 1$. 
At the same time, Proposition~\ref{prop:clime:dataE} shows that CLIME meets Assumption~\ref{assum:omega:w} under the more general Assumption~\ref{assum:func:dep} that permits temporal dependence and tail behaviour heavier than sub-Gaussianity, while \cite{liu2017tiger} assume Gaussianity and \cite{cai2016estimating} sub-Gaussianity.
Also, ACLIME utilises a perturbed version of the sample covariance matrix $\wh{\bm\Sigma}^{\tE}$, and thus it does not readily meet Assumption~\ref{assum:omega:w}~\ref{assum:omega:w:two}.
\end{rem}

For the estimation of $\bm\Gamma_j$, we define
\begin{align*}
\wh{\mbf U}_{j, t} = \l\{\begin{array}{ll}
\mbf x^{\tO}_t \l(Y^{\tO}_t + \frac{1}{2} (\mbf x_t^{\tO})^\top \wh{\bm\delta}^{\tE}_j \r),
& t \in \{a^{\tE}_j + 1, \ldots, \wh\cp^{\tE}_j\},
\\
\mbf x^{\tO}_t \l(Y^{\tO}_t - \frac{1}{2} (\mbf x_t^{\tO})^\top \wh{\bm\delta}^{\tE}_j \r), 
& t \in \{\wh\cp^{\tE}_j + 1, \ldots, b^{\tE}_j\},
\end{array}\r.
\end{align*}
$\bar{a}^{\tE}_j = b^{\tE}_j - \lfloor (1 - \epsilon) \wh{\Delta}^{\tE}_j \rfloor$ and $\bar{b}^{\tE}_j = a^{\tE}_j + \lfloor (1 - \epsilon) \wh{\Delta}^{\tE}_j \rfloor$, for some small constant $\epsilon \in (0, 1)$. Then, denoting by $\wh{\Cov}_{a, b}(\cdot)$ the empirical covariance operator, we estimate $\bm\Gamma_j$ by
\begin{align*}
\wh{\bm\Gamma}_j = \frac{1}{2} \l( \wh{\Cov}_{a^{\tE}_j, \bar{b}^{\tE}_j}( \wh{\mbf U}_{j, t} ) + \wh{\Cov}_{\bar{a}^{\tE}_j, b^{\tE}_j}( \wh{\mbf U}_{j, t} ) \r).
\end{align*}

\begin{prop}
	\label{prop:G}
	Suppose that Assumptions~\ref{assum:xe}, \ref{assum:func:dep:two}, \ref{assum:size:two} and \ref{assum:xe:iii} hold.
	Then with probability at least $1 - 4 c_2 (p \vee n)^{-c_3}$ with the constants $c_2, c_3 \in (0, \infty)$ as in Theorem~\ref{thm:one}, we have 
	\begin{align*}
	\sqrt{\Delta_j} \vert \wh{\bm\Gamma}_j - \bm\Gamma_j \vert_\infty 
	\lesssim \Psi_j^2 \sqrt{\mathfrak{s}_j \log(p \vee n)} \text{ \ for all \ } j \in [q].
	\end{align*}
\end{prop}

Combining Theorem~\ref{thm:two} and Propositions~\ref{prop:clime:dataE} and~\ref{prop:G}, we obtain: 
\begin{cor}	\label{cor:thm:two} Suppose that the conditions of Theorem~\ref{thm:two} hold. 
Then for $\wh{\mbf V}_j \vert \mc D^{\tE} \cup \mc D^{\tO} \sim \mc N_p(\mbf 0, \wh{\bm\Omega}^{\tE} \wh{\bm\Gamma}_j (\wh{\bm\Omega}^{\tE})^\top)$, we have 
\begin{multline*}
\sup_{z \in \R} \l\vert \p\l( \sqrt{\frac{\wh\Delta^{\tE}_j}{2}} \l\vert \wc{\bm\delta}_j - \bm\delta_j \r\vert_\infty \le z \;\bigg\vert\; \mc D^{\tE} \r) - \p\l(\vert \wh{\mbf V}_j \vert_\infty \le z \;\bigg\vert\; \mc D^{\tE} \cup \mc D^{\tO} \r) \r\vert 
\lesssim \\\frac{\Psi_j \log(p \vee n) \sqrt{\log(p)}}{\sqrt{\Delta_j}} \Bigl( \Psi_j^3 \log(p) + \Vert \bm\Omega \Vert_1 \max(\mathfrak{s}_j, \vert \bm\Sigma \bm\delta_j \vert_\infty^{-1} \Psi_j) \Bigr) 
+ \Vert \bm\Omega \Vert_1^2 \Psi_j^2 \log(p)\log(n) \sqrt{\frac{\mathfrak{s}_j \log(p \vee n)}{\Delta_j}}.
\end{multline*}
\end{cor}

Corollary~\ref{cor:thm:two} specifies the price to pay when handling the uncertainty stemming from estimating  $\bm\Gamma_j$ and $\cp_j$;
for the approximation error to be controlled as $o(1)$, Corollary~\ref{cor:thm:two} imposes an extra constraint on $\Delta_j$ compared to Theorem~\ref{thm:two} if $\Vert \bm\Omega \Vert_1^2 \sqrt{\mathfrak{s_j}} \gg \Psi_j^2$ and $\Psi_j \Vert \bm\Omega \Vert_1 \log(p) \gg \sqrt{\mathfrak{s}_j}$.
In a representative example, when $\Vert \bm\Omega \Vert_1 \lesssim 1$, $p \asymp n^\gamma$ for some $\gamma \in (0, \infty)$, and $\Psi_j \lesssim 1$, we have the rate in \cref{cor:thm:two} translates to
\begin{align*}
\sqrt{\frac{\mathfrak{s}_j \log(p)}{\Delta_j}} \max\l( \sqrt{\mathfrak{s}_j} \log(p), \log^2(p) \r) + \frac{\log^{3/2}(p)}{\vert \bm\Sigma \bm\delta_j \vert_\infty \sqrt{\Delta_j}}.
\end{align*}

Based on this, we construct a simultaneous $100(1 - \alpha)\%$ confidence interval about $\delta_{ij}$ for all $i \in [p]$ at a given confidence level $\alpha \in (0, 1)$, as
\begin{align}
\label{eq:ci}
\mc C_{ij}(\alpha) = 
\l( \wc{\delta}_{ij} - \sqrt{\frac{2}{\wh\Delta^{\tE}_j}} C_{\alpha/2, \infty}(\wh{\mbf V}_j),
\, 
\wc{\delta}_{ij} + \sqrt{\frac{2}{\wh\Delta^{\tE}_j}} C_{\alpha/2, \infty}(\wh{\mbf V}_j) \r),
\end{align}
where $C_{\alpha/2, \infty}(\wh{\mbf V}_j)$ is the upper $\alpha/2$-quantile of $\vert \wh{\mbf V}_j  \vert_\infty$ which can be approximated by sampling directly from the limit distribution in Corollary~\ref{cor:thm:two}.
Alternatively, we consider a multiplier bootstrap procedure which, involving randomly weighted sums of (centered) $\wh{\mbf U}_{j, t}$, mimics the distribution of the leading term in~\eqref{eq:debias:decomp}, see Appendix~\ref{sec:boot} for its full description.
Our numerical experiments show that this approach leads to simultaneous confidence intervals that achieve better coverage in general.
We conjecture that with arguments similar to those adopted in the proofs of \cref{prop:G} and \cref{cor:thm:two}, and by Theorem~3.1 of \cite{CCK23}, the validity of the bootstrap procedure can be established.

\section{Numerical experiments}
\label{sec:num}

\subsection{Change point estimation}
\label{sec:sim:cp}

In this section, we examine the empirical performance of the proposed McScan, and compare it with MOSEG \citep{cho2022high}, CHARCOAL \citep{gao2022sparse}, DPDU \citep{xu2022change} and VPWBS \citep{wang2021statistically}. The details on implementation as well as choice of tuning parameters for these methods are discussed in \cref{sec:add:cmp:detail}.

\subsubsection{Single change point scenarios}
\label{sec:sim:single}

We start with the estimation of a single change point (i.e.\ $q = 1$) where, to avoid the issue of model selection, it is assumed to be known that there is a single change. In this simple setup, it is sufficient to apply the covariance scanning by computing $T_{0, k, n}$ as in~\eqref{eq:detector} on the whole interval $(0, n]$ and identify the maximiser, which we continue to refer to as McScan. For DPDU and VPWBS, their implementations does not allow for setting the number of change point estimates. 
For fair comparison, we select among the estimates return by DPDU and VPWBS (including $\{0, n\}$ when no change point is detected) the one closest to $\cp_1$ as the final estimate, which leads to their localisation error truncated at $\min(\cp_1, n - \cp_1)$.
We generate the data according to the following scenarios with $\mbf x_t \sim_{\iid} \mc N_p(\mbf 0, \bm\Sigma)$ and independently $\varepsilon_t \sim_{\iid} \mathcal{N}(0, 1)$.
Fixing $n = 300$, in the main text, we consider the case where $\cp_1 = 75$; see Appendix~\ref{sec:add:cp:m1} for the results when $\cp_1 = 150$.

\paragraph{(M1) Isotropic Gaussian design with sparse coefficients.} 
We set $\bm\Sigma = \mbf I_p$, $p = 200$ 
and $\bm{\beta}_0 = -\bm{\beta}_1 = \rho \bm\delta$ with $\rho \in \{1, 2\}$, where $\bm\delta = (\delta_1, \ldots, \delta_p)^\top$ with $\delta_i = 0$ if $i \notin \mc S \subseteq [p]$.
We sample the set $\mc S$ with $\mathfrak{s} = \vert \mc S \vert \in \{5, 10, 50\}$ uniformly from $[p]$, and draw $(\delta_i, \, i \in \mc S)$ uniformly from a unit sphere of dimension $\mathfrak{s}$. 
This is a canonical setup that is commonly investigated in the literature. 
\cref{fig:can_err} shows that the performance of all methods becomes better for larger change size $\rho = \vert \bm\Sigma \bm\delta \vert_\infty = \vert \bm\delta \vert_\infty$ while worse for larger $\mathfrak{s}$, since the effective size of change for McScan is measured by $\Psi_1^{-1} \vert\bm\Sigma \bm\delta_1 \vert_{\infty} \ge  \mathfrak{s}^{-1/2} (2\rho)/(2\rho + 1)$; similar observation applies to the performance of the competitors also.
Overall, it is clear to see that McScan outperforms the other methods in nearly all cases in terms of estimation accuracy. 
Computationally, McScan is the fastest, followed by CHARCOAL, MOSEG, DPDU and then VPWBS, see \cref{fig:can_tm}.

\begin{figure}[h!t!b!]
	\centering
	\includegraphics[width = .8\textwidth]{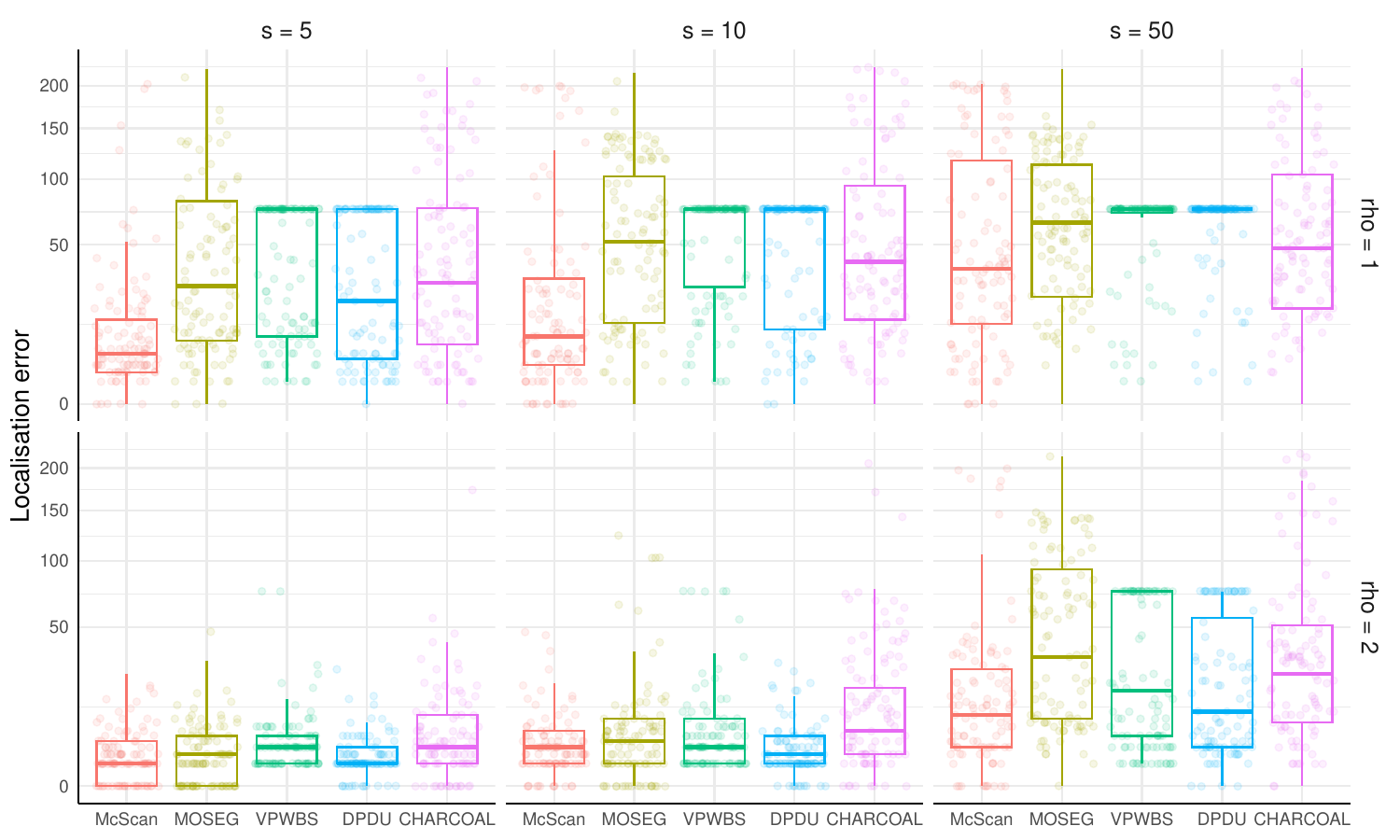}
	\caption{Localisation errors in (M1). For each method, the localisation errors $\vert\wh{\cp}_1 - \cp_1\vert$ over 100 repetitions are jittered in dots with a low intensity, and the overall performance is summarised as a boxplot (obtained from 1000 repetitions, outliers not shown), for $\cp_1 = 75$ and varying $\rho \in \{1, 2\}$. The $y$-axis is in the square root scale.}
	\label{fig:can_err}
\end{figure}

\begin{figure}[h!t!b!]
	\centering
	\includegraphics[width = .8\textwidth]{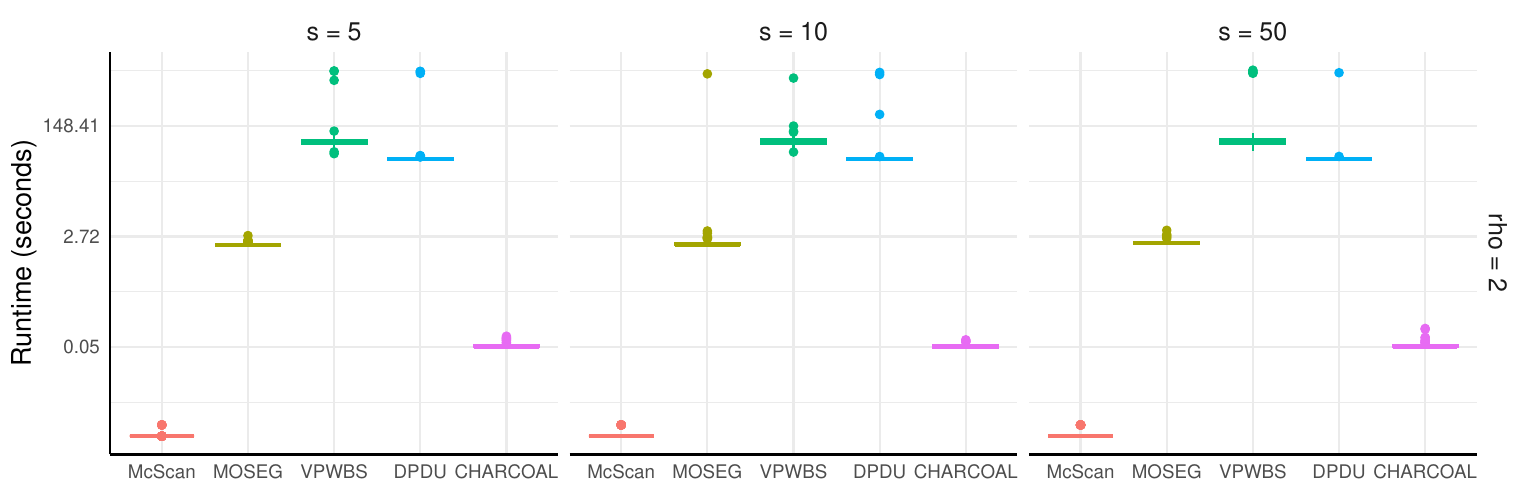}
	\caption{Runtimes in (M1) as boxplots (over 100 repetitions) when $\rho = 2$, recorded on a laptop with Apple M1 Pro chip. Runtimes for $\rho = 1$ are similar and thus omitted. The $y$-axis is in the logarithm scale.}
	\label{fig:can_tm}
\end{figure}

\paragraph{(M2) Toeplitz Gaussian design with dense coefficients.} 
We set $\bm\Sigma = [\gamma^{\vert i - j \vert}]_{i, j = 1}^p$ with varying $\gamma$. For each realisation, we generate $\bm\beta_0 = \bm\mu - \bm\delta/2$ and $\bm\beta_1 = \bm\mu + \bm\delta/2$, where $\bm\delta$ has $\mathfrak{s} = 5$ non-zero elements taking values from $\{1, - 1\}$ at random locations, and $\bm\mu = \nu \cdot \bm\mu_\circ / \sqrt{p}$ for $\bm\mu_\circ \sim \mc N_p(\mbf 0, \mbf I_p)$. 
The parameter $\nu$ governs the size of $\Psi_1$. 
\cref{fig:toep_err} shows the results obtained with $p \in \{200,\, 400\}$, $\gamma \in \{0.6,\, 0.9\}$ and $\nu \in \{0.5, 1, 2\}$. All considered methods tend to perform worse as $p$ and $\nu$ increases.
An increase in the parameter $\gamma$, which determines the degree of cross-correlations in $\mbf x_t$, does not have an adverse influence on McScan.
This may be accounted for by that $\bm\Sigma$ enters into the detection boundary of McScan in \cref{assum:size}, both in measuring the size of change (via $\vert \bm\Sigma \bm\delta_j \vert_\infty$) and the level of noise (via $\Xi$ and $\kappa$ defined in \cref{assum:func:dep}), an observation that has not been made in the existing literature. 
In all settings, McScan is among the best, and its edge over the other methods becomes more significant when $p$ is larger. VPWBS and DPDU have relatively similar performance (though DPDU being slightly better), and both outperform MOSEG and CHARCOAL except for when $\nu = 2$ and $\gamma = 0.6$. In that case (top right panel), CHARCOAL performs the best,  
but it is only applicable when $p < n$ unlike the other methods. 

\begin{figure}[h!t!b!]
	\centering
	\includegraphics[width = .8\textwidth]{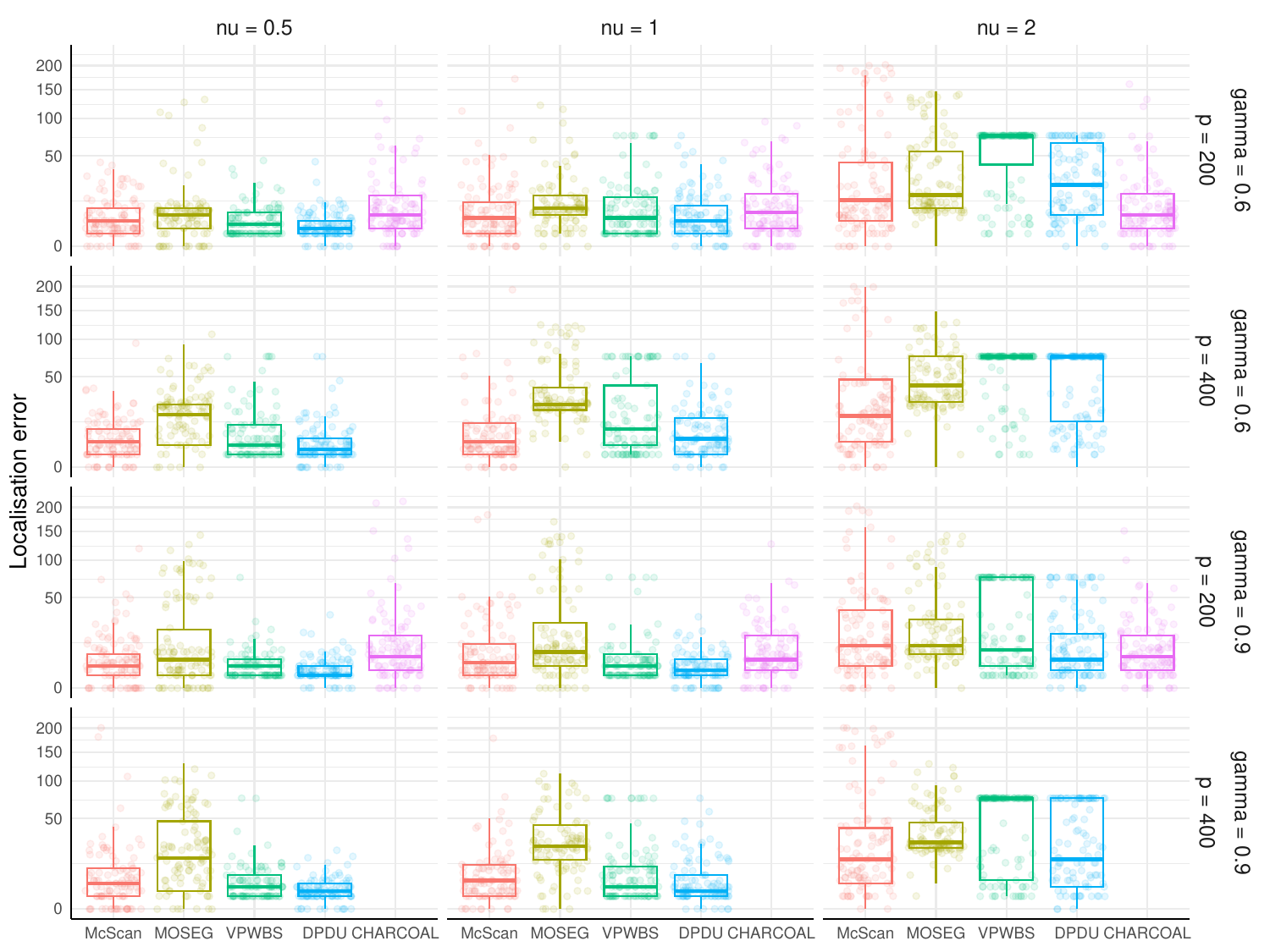}
	\caption{Localisation errors in (M2).  For each method, the localisation errors $\vert\wh{\cp}_1 - \cp_1\vert$ over 100 repetitions are jittered in dots with a low intensity, and the overall performance is summarised as a boxplot (obtained from 1000 repetitions, outliers not shown), for $\cp_1 = 75$ and varying $p \in \{200, 400\}$ and $\gamma \in \{0.6, 0.9\}$. The $y$-axis is in the square root scale. There is no result from CHARCOAL for $p = 400$ as it is not applicable when $p > n = 300$. }
	\label{fig:toep_err}
\end{figure}

{We further evaluate the performance of McScan in the presence of model mis-specification when $\Cov(\mbf x_t)$ varies over time, in Appendix~\ref{sec:sim:mask}. The results show that McScan remains competitive in detecting changes in regression parameters compared to other methods. Notably, when variations in \(\bm\beta_t\) and \(\Cov(\mbf x_t)\) offset each other so that \(\Cov(\mbf x_t) \bm\beta_t\) is constant, all methods struggle to detect any changes.}

\subsubsection{Multiple change point scenarios}
\label{sec:sim:multi}

Consider now the problem of detecting multiple change points under the model~\eqref{eq:model} with $q > 1$. We explore two approaches for setting the threshold $\pi_{n, p}$ for McScan.

\begin{description}
\item \emph{Fixed threshold.} As shown in \cref{thm:one}, McScan attains the consistency with a fixed threshold $\pi_{n, p} = c_\pi \psi_{n, p}$ for a properly chosen constant~$c_\pi$ which, however, depends on unknown quantities such as $\Psi$. 
To remedy this, we first standardise the data so that $\{X_{it}Y_t\,:\, t \in [n]\}$ has a unit \enquote{variability} for every $i \in [p]$. Towards this, we divide the $i$-th variable $\mbf X_{i\cdot}$ by the median absolute deviation of $\{(Y_{t+1}X_{i(t+1)} - Y_{t}X_{it})/\sqrt{2}\,:\,t\in[n-1]\}$. 
We performed extensive numerical experiments with $\pi_{n, p} = c_\pi \sqrt{\log(np)}$ for a grid of values for $c_\pi$, which suggests the choice of $c_\pi = 1.9$; see Appendix~\ref{sec:add:cp:m3} for an excerpt from such experiments where we consider $c_\pi \in \{1.7, 1.9, 2.1\}$.
	
\item \emph{Automatic threshold.} We compute the solution path of McScan, namely, the collection of all distinct sets of change point estimators obtained with different values for $\pi_{n,p}$. Next, we order the distinct solutions such that the number of detected change points increases, and assign to each individual solution a score defined as the maximal value of our detector statistic \eqref{eq:detector} over the seeded intervals containing no change point estimators thus far detected, which measures the evidence of any undetected change points. We then select the \enquote{elbow} point on the solution path as the final solution: The elbow point is defined using the average of the slopes of the two connecting line segments at each point, as where the average slop starts to decrease in absolute value for the first time, see \cref{fig:solution_path} for an illustration. The computation of solution path as well as the elbow point does not impose additional burden. To see this, McScan with the automatic threshold selection requires an additional runtime of $O\bigl(n^2\log (n)\bigr)$ in the worst case (cf.\ \citealp{kovacs2020seeded}), which matches that of McScan with a fixed threshold if $p \gtrsim n$, see \cref{rem:comp}.
\end{description}

\begin{figure}[h!t!b!]
	\centering
 	\includegraphics[width = .8\textwidth]{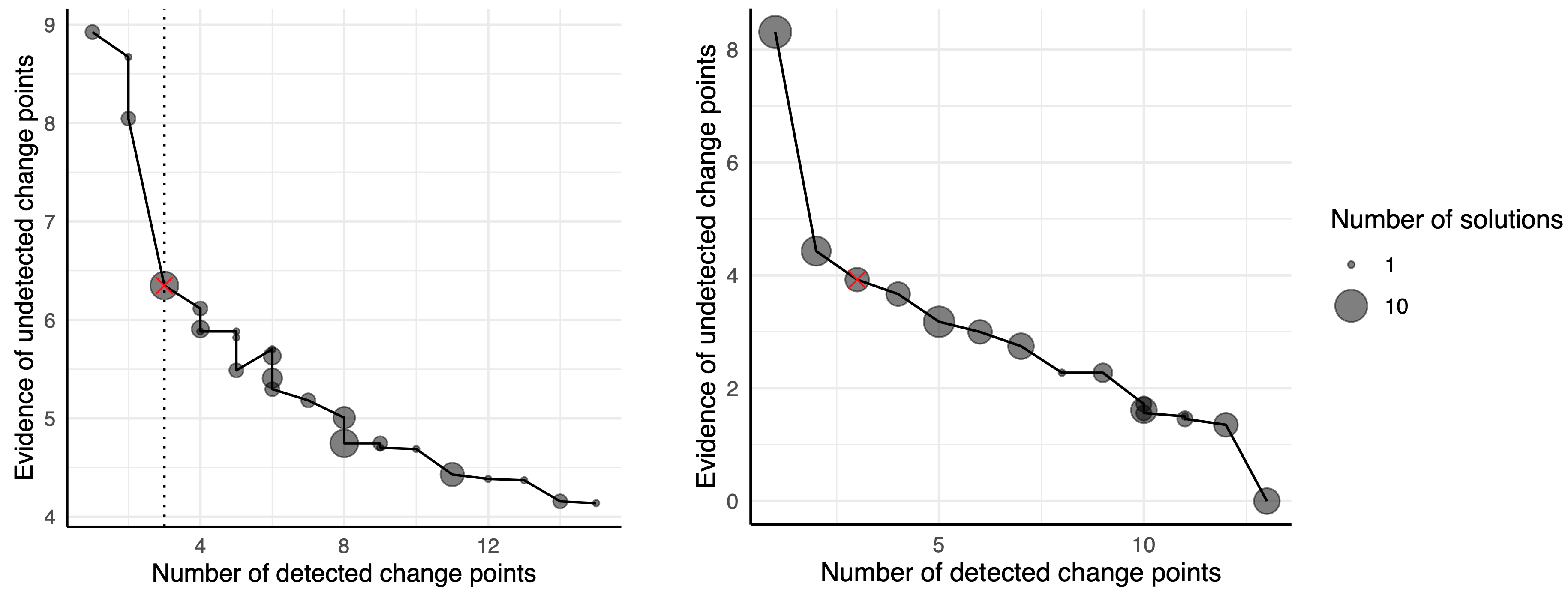}
	\caption{Solution paths of McScan. Left is for  a realisation from the model (M3) with $n=800$, where the true number of change points is marked by the vertical dotted line. Right is for the FRED-MD data (\cref{sec:data}). Each solution path consists of the sets of change point estimators ordered according to their cardinality. Multiple distinct solutions can have the same number of detected change points and the same score (measuring the evidence of undetected change points as described in the main text), which are visualised by a bubble plot. This score is not always monotone with respect to the number of detected change points. Our defined \enquote{elbow} point is marked by a red cross.}
	\label{fig:solution_path}
\end{figure}

As for the trimming parameter $\varpi_{n, p}$, we observe that the performance of McScan varies little for a range of values empirically.
As rule of thumb, we recommend $\varpi_{n, p} = 2 \log(np)$ which is used throughout our numerical experiments.
To examine the performance of McScan with both fixed and automatic thresholds, we adopt a simulation setup from \citet[(M1)]{cho2022high} but with $p = 900$ instead of $p = 100$. Specifically: 

\paragraph{(M3) Multiple change points with sparse coefficients.} We generate the data according to the model~\eqref{eq:model} with $\mbf x_t \sim_{\iid} \mc N_p(\mbf 0, \mbf I_p )$, independently $\varepsilon_t \sim_{\iid} \mathcal{N}(0, 1)$, $p = 900$, $\{\cp_j = jn/4\,:\,j \in [3] \}$ (i.e.\ $q = 3$) and $n \in \{480,560,640,720,800\}$. We fix $\mathcal{S} = [\mathfrak s]$ with $\mathfrak s = 4$, and set $\beta_{0,i} = 0.4 \cdot (-1)^{i-1}$ if $i \in \mathcal{S}$ and $\beta_{0,i} = 0$ otherwise, and $\bm\beta_j = (-1)^j\cdot \bm\beta_0$.
\medskip

Among the methods previously considered, we only include MOSEG as a competitor due to the high computational complexity of VPWBS and DPDU (see also \cref{fig:can_tm}); CHARCOAL is not applicable as $p > n$.
We refer to Table~1 of \citet{cho2022high} for the comparison between MOSEG, DPDU and VPWBS when $p = 100$, where MOSEG performs as well as, or even better than the two other methods in estimating both the total number of the change points and their locations.  
We also include McScan combined with an \enquote{oracle} threshold which is selected as the largest threshold value that leads to the correct number of change points. 
The results are summarised in \cref{fig:mcp_location} and \cref{fig:mcp_hausdorff,fig:mcp_vmeas,fig:mcp_numcp} in the Appendix. For McScan, the data-driven selection of the threshold performs similarly well as the oracle choice, and both slightly outperform the best fixed threshold $\pi_{n,p} = 1.9\sqrt{\log(np)}$. MOSEG is relatively worse when the sample size is smaller, in particular, for $n \le 640$. 

\begin{figure}[h!t!b!]
	\centering
	\includegraphics[width = .8\textwidth]{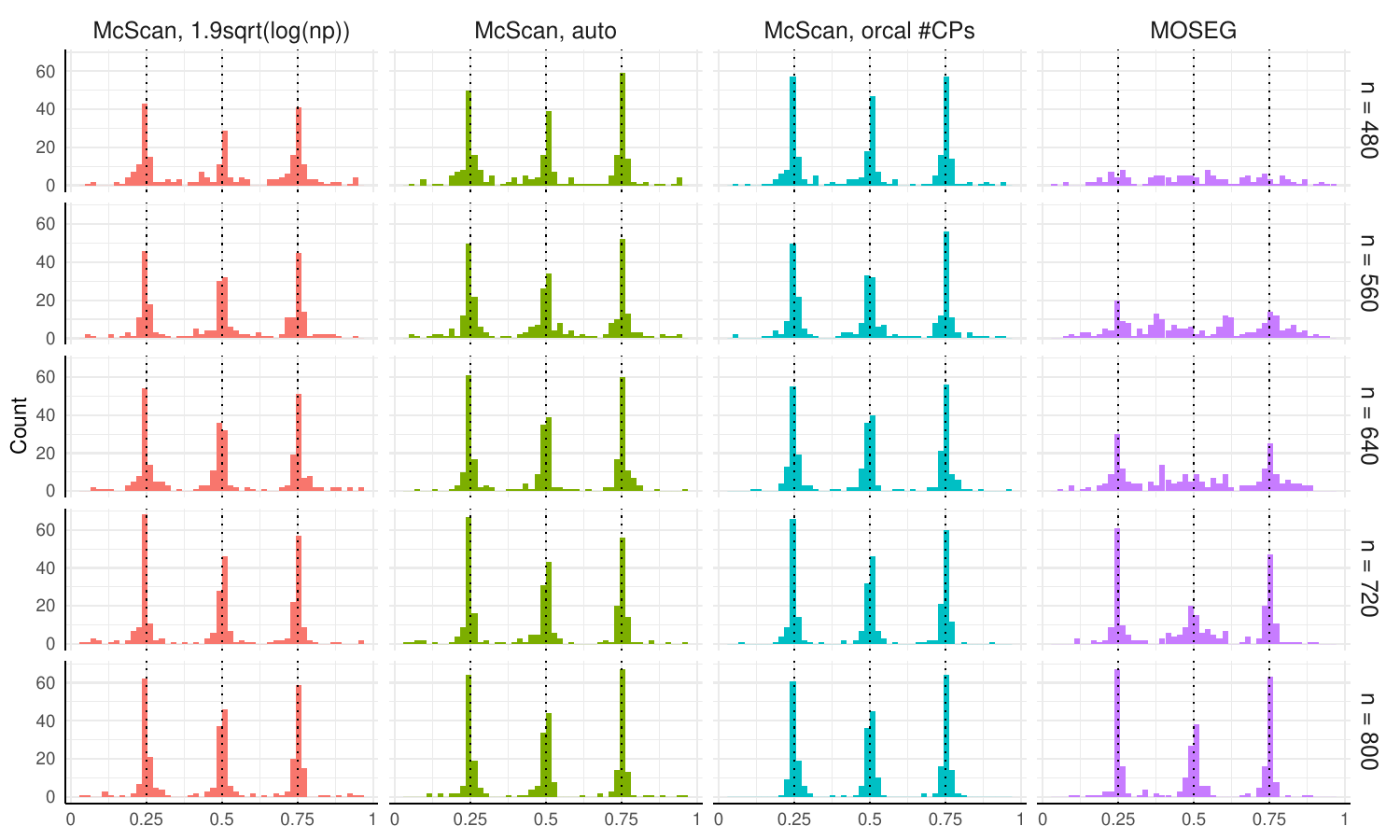}
	\caption{Scaled estimated change points, $\wh\cp_j/n$, in (M3) over 100 repetitions. The three true change points after scaling, $\cp_j/n$, are marked by vertical dotted lines. The first column is McScan with the fixed threshold $\pi_{n,p} = 1.9\sqrt{\log(np)}$ using the standardised data; the second column is McScan with the automatic threshold selection; the third column is McScan with the oracle threshold; the last column is MOSEG \citep{cho2022high} that uses multiple bandwidths and selects the number of change points via sample splitting and cross validation.}
	\label{fig:mcp_location}
\end{figure}

\subsection{Post-segmentation estimation and inference}
\label{sec:sim:post}

In this section, we investigate the performance of post-segmentation estimator of the differential parameters and the simultaneous confidence intervals proposed in \cref{sec:diff}.
In the main text, we continue to consider the scenario (M2) in \cref{sec:sim:single} with the single change point $\cp_1$ regarded as known, to separate the issues arising from change point estimation.
Additionally, we consider the multiple change point scenario~(M3) in Appendix~\ref{sec:add:ci} where we fully take into account the detection step.
 
\subsubsection{Differential parameter estimation}
\label{sec:sim:dpe}

We consider the two estimators $\wh{\bm\delta}_1 = \wh{\bm\delta}_{0, n}(\cp_1)$ obtained as in~\eqref{eq:lasso:est} and \eqref{eq:direct:est}, which we refer to as LOPE and CLOM, respectively.
As a competitor (referred to as NAIVE), we consider the naive estimator obtained by taking the difference of $\wh{\bm\beta}_j, \, j \in \{0, 1\}$, which are separately obtained Lasso estimators from $(Y_t, \mbf x_t), \, 1 \le t \le \cp_1$, and $(Y_t, \mbf x_t), \, \cp_1 + 1 \le t \le n$, respectively.
We compare their performance using the estimation errors in $\ell_\star$-norm, $\vert \wh{\bm\delta}_1 - \bm\delta_1 |_\star$, for $\star \in \{1, 2\}$.
All the estimators in consideration depend on the tuning parameter $\lambda$ and for its selection, we implement two approaches.
First, we select it via cross-validation (referred to as CV).
Secondly, we fix a grid of length $100$ for the possible $\lambda$ values and for each $\lambda$, we obtain $\wh{\bm\delta}_1(\lambda)$, compute the estimation error and present the minimum across the $100$ values of $\lambda$, which may be regarded as the \enquote{oracle} error that serves as a benchmark.
In the case of NAIVE, we separately minimise $\vert \wh{\bm\beta}_j(\lambda) - \bm\beta_j \vert_\star$ and take the difference of the thus-obtained oracle estimators of $\bm\beta_j, \, j \in \{0, 1\}$.

Varying the model parameters as $p \in \{100, 200, 400\}$, $\gamma \in \{0, 0.6, 0.9\}$, $\mathfrak{s} \in \{5, 10, 20\}$ and $\nu \in \{0.5, 1, 2\}$, we report the estimation errors averaged over $1000$ realisations, see Figure~\ref{fig:MSE_ell} as well as \cref{fig:MSE_ell1_g0,fig:MSE_ell1_g00} in the Appendix.
LOPE tends to outperform CLOM in most scenarios, by a greater margin when CV is employed compared to the oracle case. 
We may attribute this to that CLOM is more sensitive to the choice of the grid of~$\lambda$ values.
Comparing the estimation errors in $\ell_1$-norm, we observe the clear advantage of the proposed direct estimators over the NAIVE one in almost all of the cases. This is also the case in terms of the $\ell_2$-norm estimation errors except when $\mathfrak{s}$ is large and $p$ is small. 
In such a scenario, we observe that NAIVE produces a dense estimator with many small non-zero coefficients, which is penalised more heavily by the $\ell_1$-norm than the $\ell_2$-norm.
LOPE and CLOM estimators tend to be far sparser with the estimates of the non-zero coefficients shrunk more towards (although not exactly) zero.

\begin{figure}[h!t!b!p!]
	\centering
\begin{tabular}{cc}
\includegraphics[height = .31\textwidth, trim = 0.1cm 0 2.8cm 0, clip]{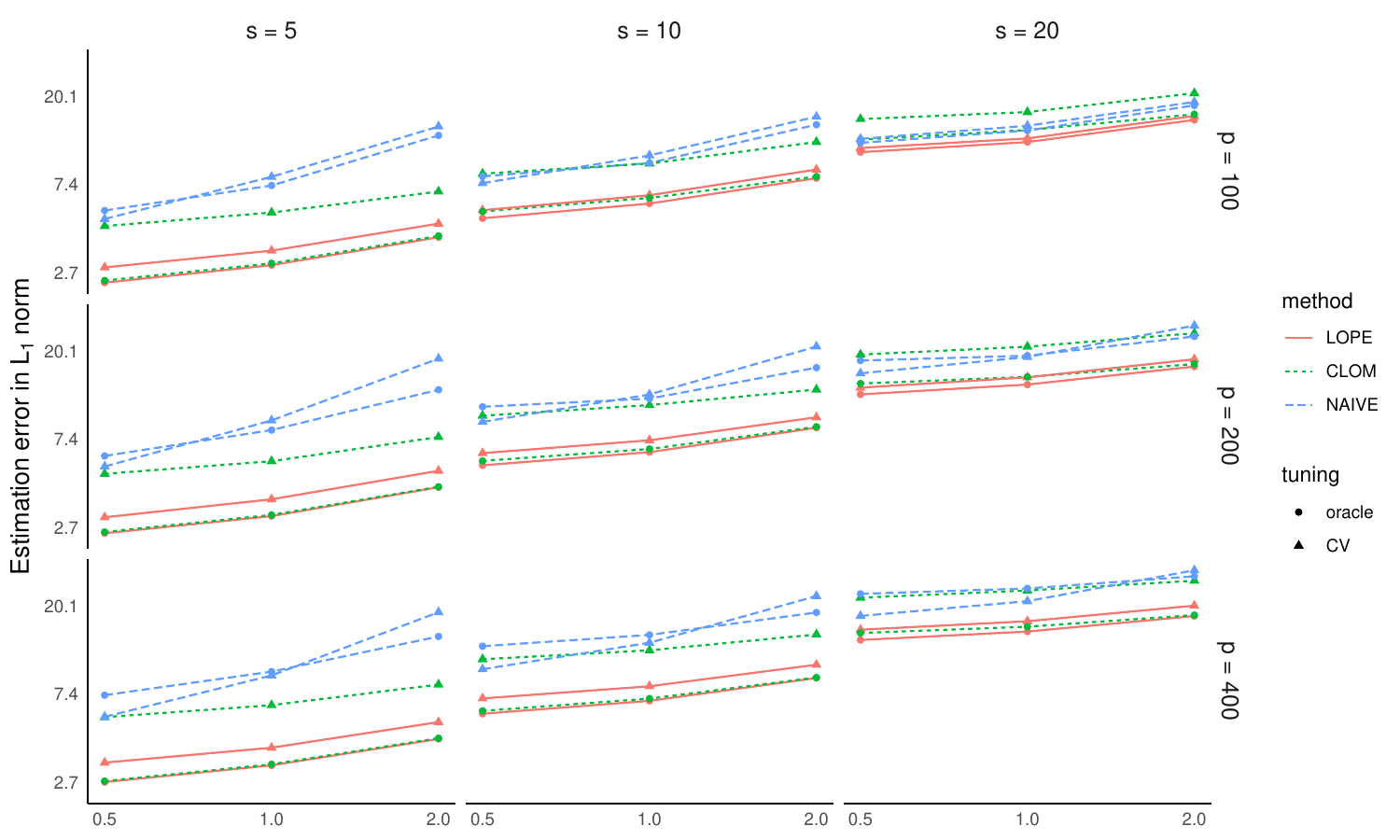}  &
\includegraphics[height = .31\textwidth,trim = 0.1cm 0 0 0.1cm, clip]{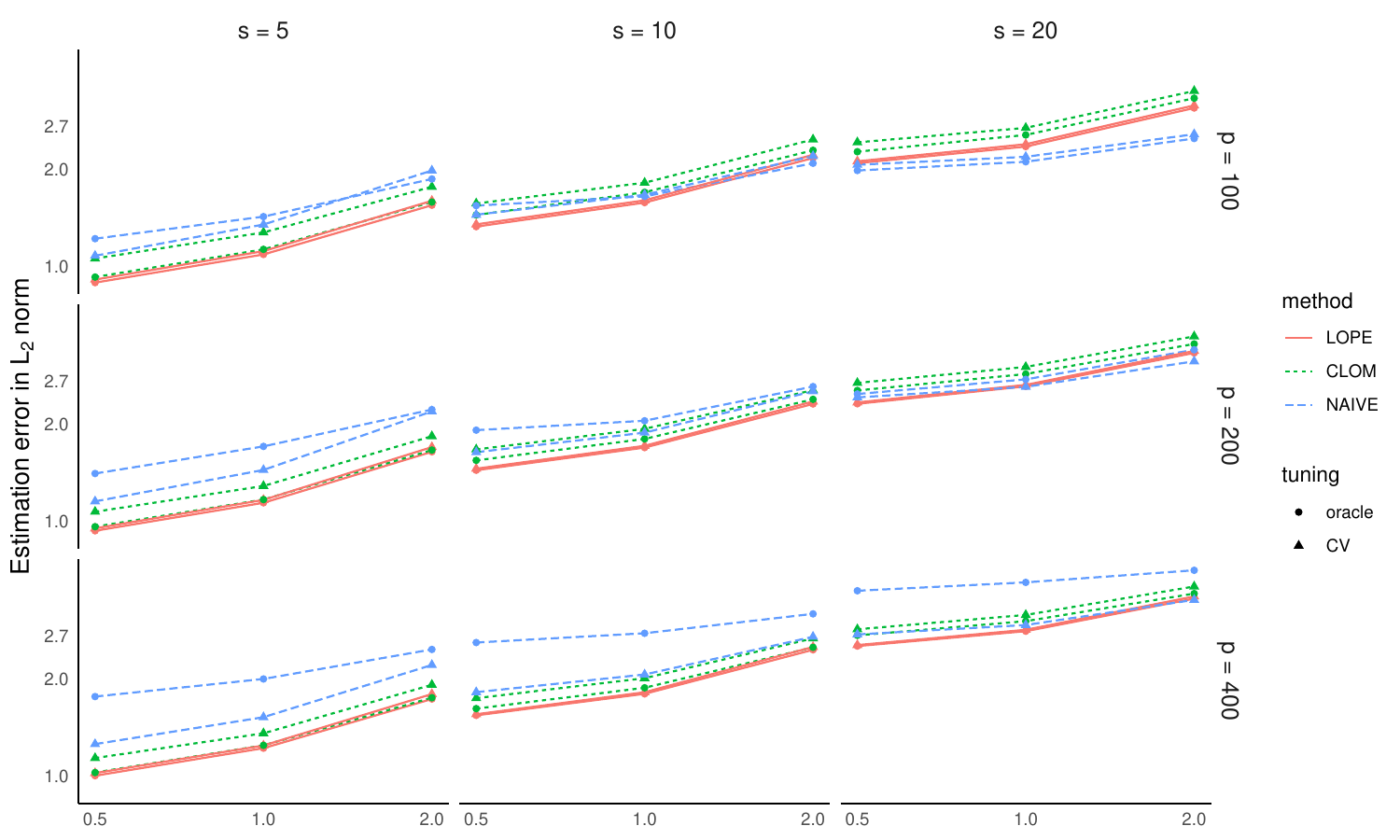}
\end{tabular}
\caption{Estimation errors in $\ell_1$- (left) and $\ell_2$- (right) norms against $\nu \in \{0.5, 1, 2\}$ ($x$-axis), from LOPE~\eqref{eq:lasso:est}, CLOM~\eqref{eq:direct:est} and NAIVE combined with the tuning parameter selected via cross-validation (CV) and the oracle one when $\gamma = 0.6$, averaged over $1000$ realisations. The $y$-axis is in the logarithm scale.}
\label{fig:MSE_ell}
\end{figure}

\subsubsection{Simultaneous confidence intervals}
\label{sec:sim:ci}

We examine the performance of the simultaneous confidence intervals constructed as in~\eqref{eq:ci}.
We investigate the two approaches, one directly sampling from the limit distribution in Corollary~\ref{cor:thm:two}, and the other based on the multiplier bootstrap described in Appendix~\ref{sec:boot}; for both approaches, the (bootstrap) sample size is set to be $B = 999$.
Throughout, we use the LOPE estimator of $\bm\delta_1$, and the CLIME estimator of $\bm\Omega$ obtained with~$\eta$ chosen via cross-validation.
For better finite sample performance, we do not to perform sample splitting and maximally make use of all observations between pairs of adjacent change point estimators, which amounts to setting $a_j = \wh\cp_{j - 1} + 1$, $b_j = \wh\cp_{j + 1}$ and $\epsilon = 0$ with suitable modifications to the construction of $\wc{\bm\delta}_j$ and its limit distribution.
We consider the scenario (M2) in \cref{sec:sim:single} with $(n, \cp_1) \in \{(300, 75), (600, 150)\}$, $p \in \{100, 200, 400\}$ and $\mathfrak{s} \in \{5, 10, 20\}$, and report the following measures of performance:
\begin{align*}
& \mathrm{Coverage} = \prod_{i \in [p]} \mathbb{I}_{\{ \delta_{i1} \in \mc C_{i1}(\alpha) \}}, \quad
\mathrm{Proportion} = \frac{1}{p} \sum_{i \in [p]} \mathbb{I}_{\{ \delta_{i1} \in \mc C_{i1}(\alpha) \}}, \\
& \mathrm{TPR} = \frac{\sum_{i \in \mc S_1} \mathbb{I}_{\{0 \notin \mc C_{i1}(\alpha)\}}}{\vert \mc S_1 \vert} \text{ \ and \ } \mathrm{FDR} = \frac{\sum_{i \notin \mc S_1} \mathbb{I}_{\{0 \notin \mc C_{i1}(\alpha)\}}}{\l(\sum_{i \in [p]} \mathbb{I}_{\{0 \notin \mc C_{i1}(\alpha)\}}\r) \vee 1},
\end{align*}
see Figure~\ref{fig:ci:m2:gamma:two:n:300} in the moderately correlated case with $\gamma = 0.6$ and $\alpha = 0.5$.
Appendix~\ref{sec:add:ci:m2} provides the complete results  with varying $\alpha \in \{0.01, 0.05, 0.1\}$, including Figures~\ref{fig:ci:width:one}--\ref{fig:ci:width:three} reporting the half-width of the confidence intervals. 

Overall, the coverage improves as $n$ increases and $\mathfrak{s}$ decrease and so do TPR and FDR, which confirms \cref{thm:two}. 
As expected, the increase in $n$ also reduces the width of confidence intervals.
While growing $\nu$ and $\gamma$ tends to increase the coverage, particularly in the challenging situations with large $\mathfrak{s}$, this has an adverse effect on TPR due to the increasing width of the confidence intervals.
Within the range in consideration, the dimensionality has little influence on the overall performance.
The multiplier bootstrap improves the coverage, particularly in the more challenging situations with smaller sample size and large $\gamma$ and $\mathfrak{s}$, at the price of marginally increased width.
The coverage tends to be below the nominal level for larger $\alpha$ although on most realisations, nearly all $\delta_{i1}, \, i \in [p]$, are covered by the respective confidence intervals as evidenced by `Proportion' close to one.

\begin{figure}[h!t!b!p!]
	\centering
	\includegraphics[width = .8\textwidth]{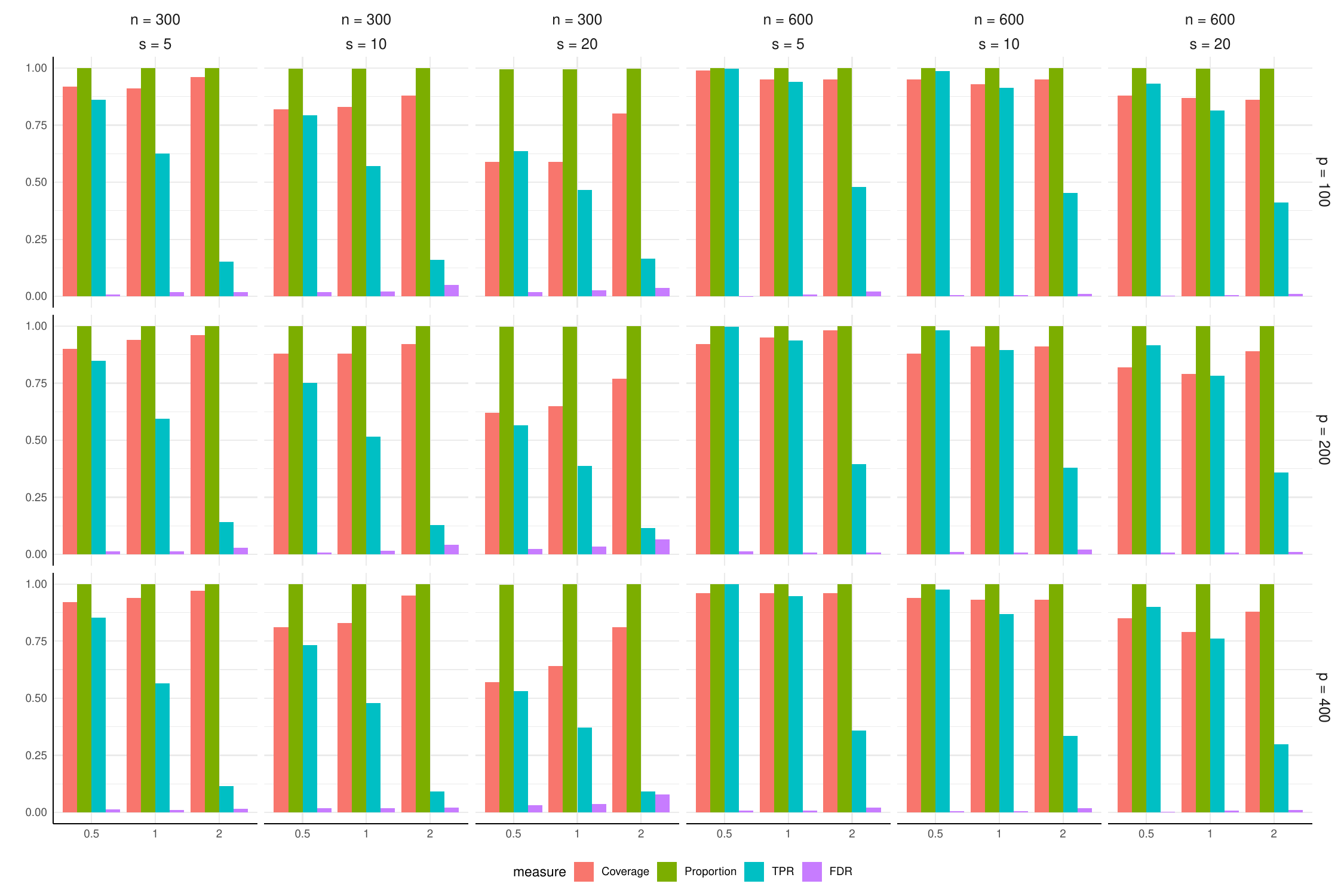} \\
 \includegraphics[width = .8\textwidth]{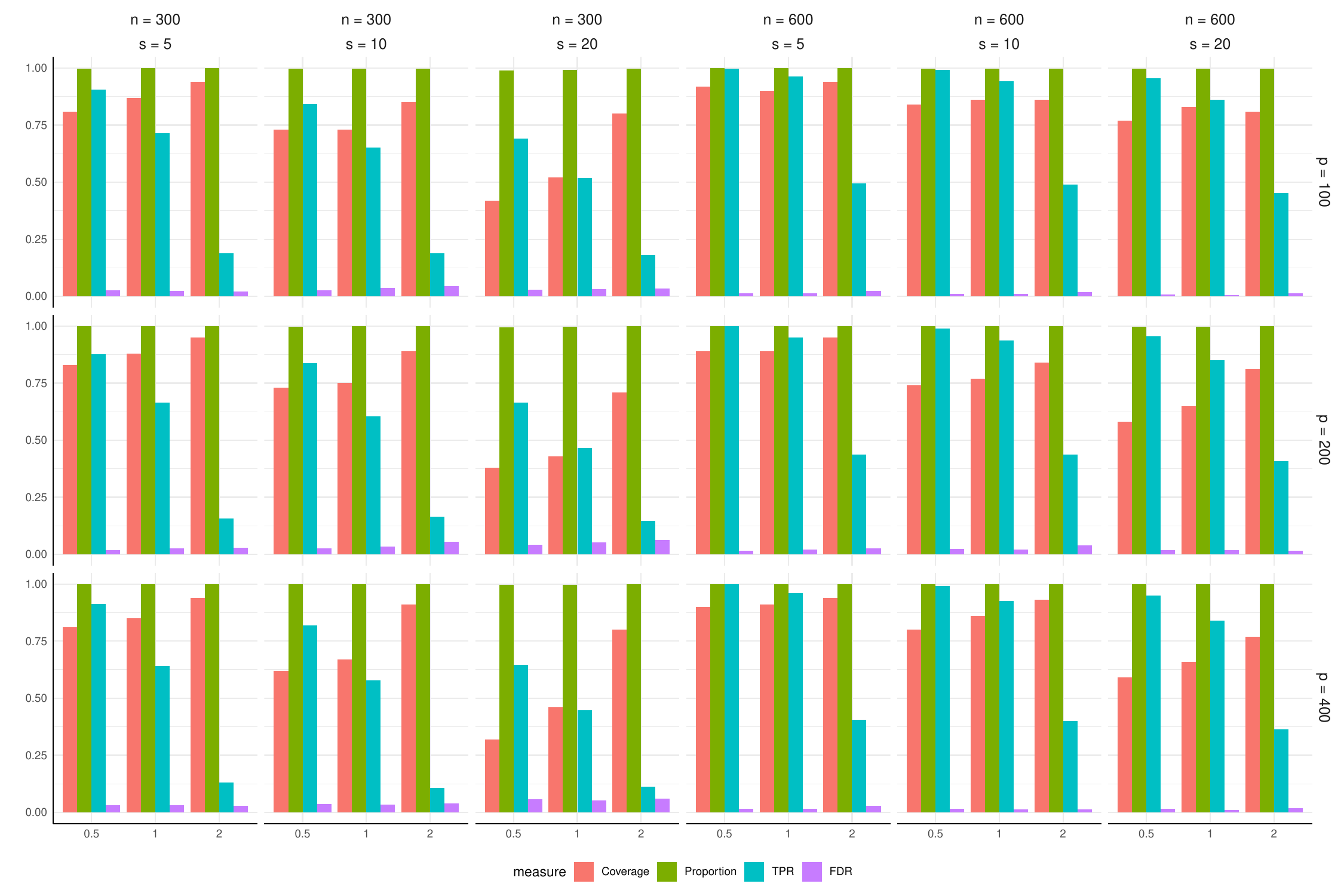}
	\caption{Coverage, Proportion, TPR and FDR of simultaneous $95\%$-confidence intervals with (top) and without (bottom) bootstrapping against $\nu \in \{0.5, 1, 2\}$ ($x$-axis) when $\gamma = 0.6$, averaged over $100$ realisations.
  }
	\label{fig:ci:m2:gamma:two:n:300}
\end{figure}

\section{Application to FRED-MD data}
\label{sec:data}

FRED-MD is a large, monthly frequency, macroeconomic database maintained by the Federal Reserve Bank of St.~Louis \citep{mccracken2016fred}.
We perform change point analysis of FRED-MD spanning the period from 1960-01 to 2024-05, under a factor-augmented forecasting model for the growth rate of the industrial production total index (INDPRO) following \cite{stock2002forecasting} who noted `when macroeconomic forecasts are constructed using many variables over a long period, some degree of temporal instability is inevitable'.
To do so, we first perform variable-specific transformations suggested by \cite{mccracken2016fred}, followed by a data truncation step (with the truncation parameter chosen as the $99.9\%$-percentile of the modulus of transformed data) to address heavy tails commonly observed from economic and financial data \citep{ibragimov2015heavy}.
Then, we form $\mbf x_t = (\mbf f_t^\top, \mbf u_t^\top)^\top$ where $\mbf f_t$ denotes the factors driving the pervasive cross-correlations in the data obtained as leading principal components, and $\mbf u_t$ the remaining idiosyncratic component after removing the presence of factors. Following \cite{barigozzi2024robust}, we set the number of factors at $6$, which leads to the dataset with $p = 114$ and $n = 773$.

Applying McScan with the default tuning parameters and the automatic threshold (see the right panel of \cref{fig:solution_path}), 
we detect three change points at 1974-02, 2010-05 and 2019-05; the last two are, with possible estimation bias, respectively associated with the 2007--2009 financial crisis and the COVID-19 pandemic.
Focusing on the last estimator, the simultaneous bootstrap confidence intervals for the differential parameter $\bm\delta_3$ at $\wh\cp_3$, are visualised in Figure~\ref{fig:fredmd:ci}.
We have enough evidence to reject $\mc H_{0, i}: \, \delta_{i3} = 0$, for $i \in \{ 2, 6 \}$, which correspond to two of the six factors. 
See also Appendix~\ref{sec:fredmd:extra} for the additional results, where we make similar observations about $\bm\delta_1$ and $\bm\delta_2$, regardless of whether the 
confidence intervals are obtained via multiplier bootstrap or sampling from the limit distribution.
This suggests that the `diffusion index forecasting' approach, where a large number of variables are compressed into a handful of factors of predictors, is effective in capturing (time-varying) dependence in large macroeconomic dataset.

\begin{figure}[h!t!b!]
\centering
\includegraphics[width = .8\textwidth]{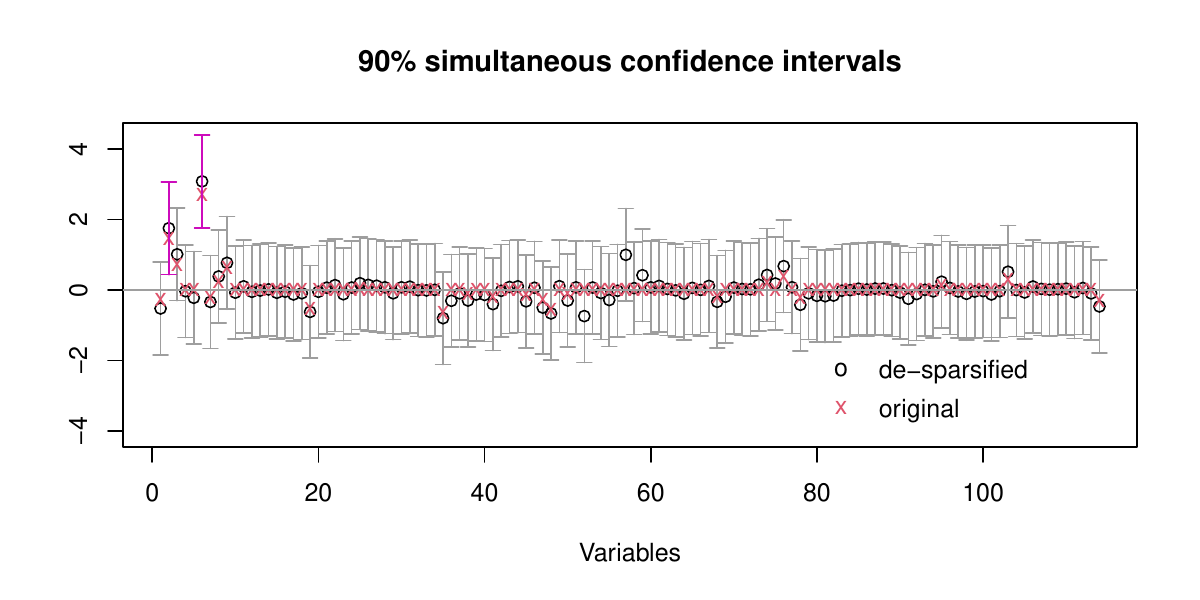}
\caption{Simultaneous $90\%$-confidence intervals obtained via multiplier bootstrap, for the $p = 114$ differential parameter coefficients before and after $\wh\cp_3$ associated with COVID-19, plotted alongside the original LOPE estimator in~\eqref{eq:lasso:est} and its de-sparsified version in~\eqref{eq:debiased}. The interval which does not contain $0$ is highlighted in magenta.}
\label{fig:fredmd:ci}
\end{figure}

\section{Conclusions}
\label{sec:conc}

In this paper, we consider the problem of detecting and inferring about changes under a high-dimensional linear regression model with multiple change points.
Unlike the existing change point detection procedures, the proposed McScan method avoids computationally costly evaluations of $\ell_1$-regularised maximum likelihood-type estimation over a grid, by scanning for large discrepancies in the covariance between $\mbf x_t$ and $Y_t$ over systematically selected intervals.
Consequently, it does not require the {exact} sparsity of either regression or differential parameters for the consistency in data segmentation.
Moreover, McScan achieves better statistical efficiency through adopting the covariance-weighted differential parameter for measuring the size of the changes.
Beyond the detection and estimation of the change points, we provide a first solution to the problem of directly inferring about the differential parameters when the regression coefficients are possibly non-sparse.
The consistency of the proposed $\ell_1$-regularised estimators of the differential parameters is established and, combined with a de-sparsification step, its modification is shown to achieve the asymptotic normality, which enables the construction of simultaneous confidence intervals about the individual components of the differential parameters.
Comparative simulation studies demonstrate the competitiveness of the proposed methods.
Altogether, the paper contributes to the emerging literature on statistical inference under possibly non-sparse high-dimensional models.

%
%
%

\bibliographystyle{apalike}
\bibliography{fbib}

\clearpage

\appendix

\numberwithin{equation}{section}
\numberwithin{figure}{section}
\numberwithin{table}{section}
\numberwithin{thm}{section}

\section{Further discussions on the data segmentation problem}

\subsection{Overview of the literature}
\label{sec:comparison}

\begin{table}[!h!t!b]
\caption{Statistical guarantees on the model~\eqref{eq:model} under (sub-)Gaussianity and serial independence with \emph{multiple} change points ($q \ge 1$), for the proposed McScan, CHARCOAL \citep{gao2022sparse}, MOSEG \citep{cho2022high}, VPWBS \citep{wang2021statistically} and DPDU \citep{xu2022change}. Recall that $\bm\Sigma = \Cov(\mbf x_t)$, $\bm\delta_j = \bm\beta_j - \bm\beta_{j - 1}$ and  $\Delta_j = \min(\cp_j - \cp_{j - 1}, \cp_{j + 1} - \cp_j)$. We omit the subscript $j$ as the conditions hold for every $j$, and $\Psi$ may depend on $n$ and $p$. 
For symmetric matrices $\bm A, \bm B \in \R^{p\times p}$, we write $\bm A\lesssim \bm B$ to indicate that there is a constant $c > 0$ such that $\mbf{x}^\top(\bm A - c\bm B)\mbf{x} \le 0$ for all $\mbf x \in \R^p$, and $\bm A \asymp \bm B$ if $\bm A\lesssim\bm B$ and $\bm B \lesssim \bm A$.}
\label{tab:comp}
\centering
\resizebox{\columnwidth}{!}{
\begin{tabular}{l cc c c}
\toprule
& $\bm\Sigma$ & $\bm\beta$, $\bm\delta$ & Detection boundary & Localisation rate $\vert \wh{\cp} - \cp \vert$ \\
\cmidrule(lr){1-1} \cmidrule(lr){2-3} \cmidrule(lr){4-5}

McScan & $\bm\Sigma \lesssim \mbf I$ & $\vert \bm\beta \vert_2 \lesssim \Psi$ & $\vert \bm\Sigma \bm\delta \vert_\infty^2 \Delta \gtrsim (\Psi^2 \vee 1) \log(pn)$ & $\vert \bm\Sigma \bm\delta \vert_\infty^{-2} (\Psi^2\vee 1) \log(pn)$ \\
\cmidrule(lr){1-1} \cmidrule(lr){2-3} \cmidrule(lr){4-5}

\multirow{2}{*}{CHARCOAL} & \multirow{2}{*}{$\bm\Sigma = \mbf I$} &  \multirow{2}{*}{$\vert \bm\delta \vert_2 \le 1$, $\vert \bm\delta \vert_0 \le \mathfrak{s}$}  & $\vert \bm\delta \vert_2^2 \Delta \gtrsim \mathfrak{s} \log^2(p)$ & \multirow{2}{*}{$\vert \bm\delta \vert_2^{-1} \sqrt{\mathfrak{s}n} \log^2(pn)$} \\

&&& $\Delta \asymp n > p$&\\
\cmidrule(lr){1-1} \cmidrule(lr){2-3} \cmidrule(lr){4-5}

\multirow{2}{*}{MOSEG} & $\bm\Sigma \lesssim \mbf I$ & \multirow{2}{*}{$\vert \bm\beta \vert_2 \lesssim \Psi$, $\vert \bm\beta \vert_0 \le \mathfrak{s}$} & \multirow{2}{*}{$\vert \bm\delta \vert_2^2 \Delta \gtrsim \underline{\sigma}^{-2}(\Psi^2 \vee 1) \mathfrak{s} \log(pn)$} & \multirow{2}{*}{$\underline{\sigma}^{-2} \vert \bm\delta \vert_2^{-2} (\Psi^2 \vee 1) \mathfrak{s} \log(pn)$} \\  

& $\Lambda_{\min}(\bm\Sigma) \ge \underline{\sigma}$ &&& \\ 

\cmidrule(lr){1-1} \cmidrule(lr){2-3} \cmidrule(lr){4-5}

\multirow{2}{*}{VPWBS}& $\bm\Sigma \lesssim \mbf I$ & $\vert \bm\beta \vert_2 \lesssim \Psi$, $\vert \bm\beta \vert_{\infty} \lesssim 1$ & \multirow{2}{*}{$\vert \bm\delta \vert_2^2 \Delta \gtrsim \underline{\sigma}^{-2}(\Psi^2 \vee \mathfrak{s}) \log(pn)$} & \multirow{2}{*}{$\vert \bm\delta \vert_2^{-2} (\Psi^2 \vee 1) \log(n)  (n/\Delta)^2$} \\  

& $\Lambda_{\min}(\bm\Sigma) \ge \underline{\sigma}$ &  $\vert \bm\beta \vert_0 \le \mathfrak{s}$, {$\vert \bm\delta \vert_2 \lesssim 1$} & & \\

\cmidrule(lr){1-1} \cmidrule(lr){2-3} \cmidrule(lr){4-5}

\multirow{2}{*}{DPDU} 
& \multirow{2}{*}{$\bm\Sigma \asymp \mbf I$} & $\vert \bm\beta \vert_0 \le \mathfrak{s}$, $\vert \bm\delta \vert_2 \lesssim 1$ & \multirow{2}{*}{$\vert \bm\delta \vert_2^2 \Delta \gg \mathfrak{s} \log(pn)$} & \multirow{2}{*}{$\vert \bm\delta \vert_2^{-2} \mathfrak{s} \log(pn)$} \\
& & $q \lesssim 1$  & & \\
\bottomrule
\end{tabular}}
\end{table}

We summarise the statistical guarantees on methods developed for model~\eqref{eq:model} with $q\ge 1$ change points under (sub-)Gaussianity in \cref{tab:comp}, complementing Remark~\ref{rem:dlb} in the main text. 
Further details on Table~\ref{tab:comp} are given below. 
\begin{itemize}
    \item For McScan, the results are based on \cref{thm:one} in this paper. 
    \item 
    For CHARCOAL, the results are based on Corollary~7 in \citet{gao2022sparse}.
    \item 
    For MOSEG, the results are based on Corollary~5~(ii) in \citet{cho2022high}.
    \item 
    For VPWBS, the results are based on Theorem~2 in \citet{wang2021statistically}. In its localisation rate, the additional factor {$(n/\Delta)^2$} occurs since we do not know $\Delta$, see the discussion in \citet[bottom of page~12]{wang2021statistically}.
    \item 
    For DPDU, we report the results of the preliminary estimator as the refined rate requires a stronger condition on the size of the change, see Lemma~5 in \citet{xu2022change}.
\end{itemize}

To facilitate comparison, recall that for any $\bm\delta = (\delta_1, \ldots, \delta_p)^\top \in \R^p$,
\begin{align*}
\vert \bm\Sigma \bm\delta \vert_\infty \ge \frac{\Lambda_{\min}(\bm\Sigma) \vert \bm\delta \vert_2}{\sqrt{\vert \bm\delta \vert_0}}
\end{align*}
(Equation~\eqref{eq:size:change} in the main text).
The gap between the left and the right hand sides of the above inequality may increase with $p$ in high dimensions: 
For example, consider the situation where $\bm\Sigma = \mbf I_p$ and $\vert \delta_{1} \vert = 1$ while $\vert \delta_{i} \vert = 1 / \sqrt{p - 1}$, for $2 \le i \le p$. 
Then, provided that $p \ge 2$,
\begin{align*}
\vert \bm\Sigma \bm\delta \vert_\infty = 1 \ge \sqrt{\frac{2}{p}} = \frac{\Lambda_{\min}(\bm\Sigma) \vert \bm\delta \vert_2}{\sqrt{\vert \bm\delta \vert_0}},
\end{align*}
and the ratio between the two measures of change size diverge with $\sqrt p$.
This example demonstrates that theoretical investigation presented in the existing literature, where the RHS is often used to measure the size of change, does not fully reflect the difficulty of the change point detection problem under~\eqref{eq:model}.
Our paper fills this gap by proposing the McScan procedure which, without requiring the sparsity of either $\bm\delta_j$ or $\bm\beta_j$, achieves better statistical and computational efficiency by adopting the LHS as the measure.

\subsection{Minimax lower bounds}
\label{ss:lb}

Complementing Theorem~\ref{thm:one}, we derive the minimax lower bounds on the detection boundary in the following two lemmas. The first one concerns the case of general (potentially degenerate) design matrix $\bm\Sigma$.  

\begin{lem}
\label{lem:lb}
Consider model~\eqref{eq:model} with $q = 1$, $\bm\delta = \bm\beta_1 - \bm\beta_0$, $\bm\mu = \bm\beta_0 + \bm\beta_1$.
Let $\bm \Sigma$ be non-zero and positive semi-definite, $\sigma_{\vep}>0$ and $1 \le \Delta \le n/4$, and define
\begin{multline*}
\mc P_{\bm\Sigma, \sigma_{\vep}}^{\tau, n, p} = \Bigl\{ \{(\mbf x_t, Y_t)\}_{t = 1}^n \text{ from model~\eqref{eq:model}} \;:\; q = 1,\;\mbf x_t \sim_{\iid} \mc N_p(\bm0, \bm \Sigma), \\
\text{ and independently }\vep_t \sim_{\iid}\mc N (0,\sigma_{\vep}^2), \text{ and } \vert\bm\Sigma \bm\delta\vert_{\infty}^2 \Delta \ge \tau \sigma_{\vep}^2\bigl(\vert\bm\delta\vert_2^2 \vee \vert \bm\mu\vert_2^2  \vee 1\bigr)\Bigr\}.   
\end{multline*}

\begin{enumerate}[label = (\roman*)]
\item\label{lem:lb:detect}
If  
\[
0 \;< \;\tau \;\le\; \frac{\max_{i \in [p]}\bm e_i^\top \bm\Sigma^2 \bm e_i}{\max\{54\Vert\bm\Sigma\Vert, \, \sigma_{\vep}^2/\Delta\}}
\] 
with $\bm e_i \in \R^p$ the $i$-th canonical basis, then
$$
\inf_{\wh\theta}\; \sup_{\p \in \mc P_{\bm\Sigma, \sigma_{\vep}}^{\tau, n, p}} \;\p\l(\bigl\vert\wh\theta -\theta\bigr\vert \ge \Delta\r)\; \ge\; \frac{1}{4}.
$$    
\item \label{lem:lb:locate}
Let $\tau \to \infty$ (at an arbitrary rate) as $n\to \infty$. If $\Delta \ge \tau\sigma_{\vep}^2 / (\max_{i \in [p]}\bm e_i^\top \bm\Sigma^2 \bm e_i)$, then
$$
\inf_{\wh\theta}\; \sup_{\p \in \mc P_{\bm\Sigma, \sigma_{\vep}}^{\tau, n, p}} \;\p\l(\bigl\vert\wh\theta -\theta\bigr\vert \ge c\cdot\frac{\sigma_{\vep}^2\bigl(\vert\bm\delta\vert_2^2 \vee \vert \bm\mu\vert_2^2  \vee 1\bigr)}{\vert\bm\Sigma\bm\delta\vert_{\infty}^2}\r)\; \ge\; \frac{1}{4},
$$
where $c = \max_{i \in [p]}\bm e_i^\top \bm\Sigma^2 \bm e_i/(54\Vert\bm\Sigma\Vert)$ is a constant.
\end{enumerate}
\end{lem}
\begin{proof}
{\bf Part~\ref{lem:lb:detect}.}
Introduce \(\alpha > 0\) and \(i_0 \in [p]\) as 
\[
\alpha \;= \;\bm e_{i_0}^\top \bm\Sigma^2 \bm e_{i_0} \;= \; \max_{i\in[p]}\bm e_i^\top \bm\Sigma^2 \bm e_i  \;=\; \max_{\vert\mbf x\vert_2\le 1} \vert\bm\Sigma \mbf x\vert_\infty^2.
\]
Define $\bm \delta = \tfrac{\sigma_{\vep}}{\alpha}\sqrt{\tfrac{\tau}{\Delta}}\bm\Sigma \bm e_{i_0}$ (possibly non-sparse) and $\bm\mu  = \bm 0 \in \R^p$. Then
\[
\vert\bm\delta\vert_2^2 = \frac{\tau \sigma_{\vep}^2}{\Delta \alpha} \le 1\;\text{ and }\; \vert\bm\Sigma\bm\delta\vert_{\infty} = \bm e_{i_0}^\top \bm\Sigma\bm\delta = \sigma_{\vep}\sqrt{\frac{\tau}{\Delta}}.
\]
Introduce $\p_{\Delta,\bm\delta}$ as the distribution of $\{(\mbf x_t, Y_t)\}_{t = 1}^n$ from \eqref{eq:model} with $\theta = \Delta$, $\bm\beta_0 = (\bm\mu -\bm\delta)/2$ and $\bm\beta_1 = (\bm\mu +\bm\delta)/2$, and similarly, $\p_{n - \Delta, -\bm\delta}$ as the distribution of $\{(\mbf x_t, Y_t)\}_{t = 1}^n$ from \eqref{eq:model} with $\theta = n - \Delta$, $\bm\beta_0 = (\bm\mu + \bm\delta)/2$ and $\bm\beta_1 = (\bm\mu -\bm\delta)/2$. Note that $\p_{\Delta, \bm\delta}, \p_{n-\Delta, -\bm\delta} \in \mc P^{\tau, n, p}_{\bm\Sigma, \sigma_{\vep}}$ and $(n-\Delta) - \Delta \ge 2 \Delta$. Then we have
\begin{equation}\label{eq:lbtwo}
\inf_{\wh\theta} \sup_{\p \in \mc P_{\bm\Sigma, \Delta, \sigma_{\vep}}^{\tau, n, p}} \p\l(\bigl\vert\wh\theta -\theta\bigr\vert \ge \Delta\r) \;\ge\; \frac{1}{2}\bigl(1 - \mathrm{TV}(\p_{\Delta,\bm\delta},\p_{n - \Delta, -\bm\delta})\bigr),
\end{equation}
see e.g.\ \citet[Theorem~2.2]{Tsy09}. Further, 
by Theorem~1.1 of \citet{devroye2018total} and similar calculations as in \cref{lem:tv}, we obtain
$$
\mathrm{TV}(\p_{\Delta,\bm\delta},\p_{n - \Delta, -\bm\delta})^2 \le \frac{27\Delta}{2\sigma^2_{\vep}}\bm\delta^\top\bm\Sigma\bm\delta  \le \frac{27\Delta}{2\sigma^2_{\vep}}\vert\bm\delta\vert_2^2\Vert\bm\Sigma\Vert = \frac{27\Vert\bm\Sigma\Vert\tau}{2\alpha} \le \frac{1}{4}.
$$
This together with \eqref{eq:lbtwo} proves the assertion. 

\bigskip

{\bf Part~\ref{lem:lb:locate}.} We use the same notation as in Part~\ref{lem:lb:detect} except that we redefine $\bm \delta = \bm\Sigma \bm e_{i_0}/\sqrt{\alpha}$. Then $\vert\bm \delta\vert_2 = 1$ and $\vert\bm \Sigma \bm\delta\vert_{\infty}^2 = \alpha$.
Introduce 
\[
\gamma \;=\;\frac{\max_{i \in [p]}\bm e_i^\top \bm\Sigma^2 \bm e_i}{54\Vert\bm\Sigma\Vert}\cdot \frac{\sigma_{\vep}^2\bigl(\vert\bm\delta\vert_2^2 \vee \vert \bm\mu\vert_2^2  \vee 1\bigr)}{\vert\bm\Sigma\bm\delta\vert_{\infty}^2} \;=\; \frac{\sigma^2_{\vep}}{54\Vert\bm\Sigma\Vert}.
\]
Since $\Delta \ge \tau \sigma_{\vep}^2/\alpha$, it holds that $\p_{\Delta, \bm\delta}, \p_{\Delta + 2\gamma, \bm\delta} \in \mc P^{\tau, n, p}_{\bm\Sigma, \Delta}$ and further that
$$
\inf_{\wh\theta} \sup_{\p \in \mc P_{\bm\Sigma, \Delta, \sigma_{\vep}}^{\tau, n, p}} \p\l(\bigl\vert\wh\theta -\theta\bigr\vert \ge \gamma\r) \;\ge\; \frac{1}{2}\bigl(1 - \mathrm{TV}(\p_{\Delta,\bm\delta},\p_{\Delta+2\gamma, \bm\delta})\bigr).
$$
Similarly as in Part~\ref{lem:lb:detect}, we have 
\[
\mathrm{TV}(\p_{\Delta,\bm\delta},\p_{\Delta + 2\gamma, \bm\delta})^2 \le \frac{27\gamma}{2\sigma^2_{\vep}}\bm\delta^\top\bm\Sigma\bm\delta  \le \frac{27\gamma}{2\sigma^2_{\vep}}\vert\bm\delta\vert_2^2\Vert\bm\Sigma\Vert = \frac{1}{4},
\]
which concludes the proof. 
\end{proof}

The second lemma concerns the special case of $\bm\Sigma = \mbf I_p$ and $\mathfrak{s}^2\le p$, whose conclusions slightly differ (in terms of log factors) from Lemmas~3 and~4 of \citet{rinaldo2020localizing}. Its proof uses the method of fuzzy hypotheses (a.k.a.\ Le Cam’s method), and the choices of hypotheses have the property that \(\vert\bm \Sigma\bm\delta\vert_{\infty}^2 = \vert\bm\delta\vert_{\infty}^2 = \vert\bm\delta\vert_2^2 / {\vert\bm\delta\vert_0}\).
The complete understanding of minimax optimality in detection and localisation of changes for the whole range of sparse levels up to $p$, even for the simple case of $\bm\Sigma = \mbf I_p$, remains unknown and serves as an interesting avenue for future research.

\begin{lem}\label{lem:lb:id}
Assume $\mathfrak{s}^2 \le p$. Consider model~\eqref{eq:model} with $q = 1$.
Let $\sigma_{\vep}>0$ and $1 \le \Delta \le n/4$, and define
\begin{multline*}
\mc P_{\mathfrak{s}, \sigma_{\vep}}^{\tau, n, p} = \Bigl\{ \{(\mbf x_t, Y_t)\}_{t = 1}^n \text{ from model~\eqref{eq:model}} \;:\; q = 1,\;\mbf x_t \sim_{\iid} \mc N_p(\bm0, \mbf I_p), \text{ and }\\
\text{ independently }  \vep_t \sim_{\iid}\mc N (0,\sigma_{\vep}^2), \; \vert\bm\delta\vert_0 = \mathfrak{s}
\text{ and } \Delta\vert\bm\delta\vert_2^2 \ge \tau \sigma_{\vep}^2\mathfrak{s} \log\bigl(1 + \tfrac{p}{17 \mathfrak{s}^2}\bigr)\Bigr\}. 
\end{multline*} 

\begin{enumerate}[label = (\roman*)]
\item\label{lem:lb:id:detect}
If $\mathfrak{s}\log\bigl(1 + {p}/({17\mathfrak{s}^2})\bigr) \le {2}\sqrt{{\Delta}/{17}}$ and $0< \tau \le {1}/{2}$, then
$$
\inf_{\wh\theta}\; \sup_{\p \in \mc P_{\mathfrak{s}, \sigma_{\vep}}^{\tau, n, p}} \;\p\l(\bigl\vert\wh\theta -\theta\bigr\vert \ge {\Delta}\r)\; \ge\; \frac{1}{4}.
$$   
\item\label{lem:lb:id:locate}
Let $\tau \to \infty$ (at an arbitrary rate) as $n\to \infty$. If $\mathfrak{s}\log\bigl(1 + {p}/({17\mathfrak{s}^2})\bigr) \le \sqrt{{2\Delta}/({17\tau})}$, then
\[
\inf_{\wh\theta}\; \sup_{\p \in \mc P_{\mathfrak{s}, \sigma_{\vep}}^{\tau, n, p}} \;\p\l(\bigl\vert\wh\theta -\theta\bigr\vert \ge \frac{ \sigma_{\vep}^2\mathfrak{s} \log\bigl(1 + \tfrac{p}{17\mathfrak{s}^2}\bigr)}{4\vert\bm\delta\vert_2^2}\r)\; \ge\; \frac{1}{4}.
\]
\end{enumerate}
\end{lem}

\begin{proof}
We start with some auxiliary results. For $\bm u \in \R^p$, let $\p_{\bm u}$ denote the centred Gaussian distribution with covariance matrix
\[
\bmx
    \mbf I_p & \bm u \\
    \bm u^\top & \bm u^\top\bm u + \sigma_{\vep}^2
\emx \in \R^{(p+1)\times(p +1)}.
\]
If $\vert\bm u\vert_2 = \vert\bm v\vert_2 = \kappa \le \sigma_{\vep}/2$, then (see also Lemmas~2 and 3 in \citealp{bradic2022testability})
\begin{align}
\E_{\p_{\bm 0}}\l(\frac{\mathrm{d}\p_{\bm u}}{\mathrm{d}\p_{\bm 0}}\frac{\mathrm{d}\p_{\bm v}}{\mathrm{d}\p_{\bm 0}}\r) & = \E_{\p_{\bm0}}\l(\frac{\mathrm{d}\p_{\bm u}^{\bm Y \vert \mbf x}\mathrm{d}\p_{\bm u}^{\mbf x}}{\mathrm{d}\p_{\bm 0}^{\bm Y \vert \mbf x}\mathrm{d}\p_{\bm 0}^{\mbf x}}\frac{\mathrm{d}\p_{\bm v}^{\bm Y \vert \mbf x}\mathrm{d}\p_{\bm v}^{\mbf x}}{\mathrm{d}\p_{\bm 0}^{\bm Y \vert \mbf x}\mathrm{d}\p_{\bm 0}^{\mbf x}}\r) =\E_{\p_{\bm0}}\l(\frac{\mathrm{d}\p_{\bm u}^{\bm Y \vert \mbf x}}{\mathrm{d}\p_{\bm 0}^{\bm Y \vert \mbf x}}\frac{\mathrm{d}\p_{\bm v}^{\bm Y \vert \mbf x}}{\mathrm{d}\p_{\bm 0}^{\bm Y \vert \mbf x}}\r) \nonumber \\
& = \mathrm{det}\l(\mbf I_p - \frac{1}{\sigma_{\vep}^2}\l(\bm u \bm v^{\top} + \bm v \bm u^{\top}\r)\r)^{-1/2} \nonumber \\
& = \l(\l(1-\frac{\langle\bm u, \bm v\rangle}{\sigma_{\vep}^2}\r)^2 - \l(\frac{\kappa^2}{\sigma_{\vep}^2}\r)^2\r)^{-1/2},\label{eq:uv}
\end{align}
where the last equality can be seen from the fact that 
\[
\l(\mbf I_p - \frac{1}{\sigma_{\vep}^2}\l(\bm u \bm v^{\top} + \bm v \bm u^{\top}\r)\r) (\bm u \pm \bm v) =  \l(1-\frac{\langle\bm u, \bm v\rangle}{\sigma_{\vep}^2} \mp  \frac{\kappa^2}{\sigma_{\vep}^2}\r)(\bm u \pm \bm v).
\]
We recall the method of fuzzy hypotheses (see e.g.\ Section~2.7.4 in \citealp{Tsy09}) using $\chi^2$~divergence in the considered setup. Let $\p_{1,i} \in \mc P_{\mathfrak{s}, \sigma_{\vep}}^{\tau, n, p}$, $i\in[m]$, have the common change point at $\theta$, and $\p_{2,i} \in \mc P_{\mathfrak{s}, \sigma_{\vep}}^{\tau, n, p}$, $i\in[m]$, have the common change point at $\theta + 2\gamma$. Let also $\xi$ and $\eta$ be two independent random variables taking values on $[m]$. Then
\begin{multline}
\inf_{\wh\theta}\; \sup_{\mc P_{\mathfrak{s}, \sigma_{\vep}}^{\tau, n, p}} \;\p\l(\bigl\vert\wh\theta -\theta\bigr\vert \ge {\gamma}\r)\; \ge\; \frac{1 - \mathrm{TV}(\p_{1,\xi}, \p_{2,\eta})}{2} \; \ge\; \frac{1 - \mathrm{TV}(\p_{1,\xi}, \p_{0})-\mathrm{TV}(\p_{2,\eta}, \p_{0})}{2}\\
\ge\; \frac{1- \sqrt{\chi^2(\p_{1,\xi},\p_0)/2} - \sqrt{\chi^2(\p_{2,\eta},\p_0)/2}}{2}, \label{eq:fuzz}    
\end{multline}
where $\p_0$ is an arbitrarily fixed probability measure.

\bigskip

{\bf Part~\ref{lem:lb:id:detect}.} Let $\p_{\theta, \bm\beta_0, \bm\beta_1}$ be the distribution of $\{(\mbf x_t, Y_t)\}_{t=1}^n$ from model~\eqref{eq:model} with $q = 1$, the change point $\theta$ and the regression coefficients $\bm \beta_0$ and $\bm\beta_1$. Let $\p_0$ be the distribution of $\{(\mbf x_t, \vep_t)\}_{t = 1}^n$. 
For $S \subseteq [p]$, let $\mathbf{1}_S \in \R^p$ be a vector with its $i$-th entry $\mathbf{1}_S(i) = 1$ if $i \in S$ and equal to $0$ otherwise.  The construction of fuzzy hypotheses is based on the different choices of the support of changes. Define 
\[
\p_1 = {p \choose \mathfrak{s}}^{-1} \sum_{S\subseteq [p],\; \vert S \vert = \mathfrak{s}}\p_{\Delta, \frac{\kappa}{\sqrt{\mathfrak{s}}}\mathbf{1}_S, \bm 0}\quad\text{and}\quad
\p_2 = {p \choose \mathfrak{s}}^{-1} \sum_{S\subseteq [p],\; \vert S \vert = \mathfrak{s}}\p_{n-\Delta, \bm 0, \frac{\kappa}{\sqrt{\mathfrak{s}}}\mathbf{1}_S}
\]
where we choose $\kappa$ such that $\Delta\kappa^2 = 2^{-1} \sigma_{\vep}^2\mathfrak{s} \log\bigl(1 + {p}/({17\mathfrak{s}^2})\bigr)$. Then, by~\eqref{eq:fuzz}, we obtain 
\begin{align*}
\inf_{\wh\theta}\; \sup_{\mc P_{\mathfrak{s}, \sigma_{\vep}}^{\tau, n, p}} \;\p\l(\bigl\vert\wh\theta -\theta\bigr\vert \ge \Delta\r)\; &
\ge\; \frac{1- \sqrt{\chi^2(\p_{1},\p_0)/2} - \sqrt{\chi^2(\p_{2},\p_0)/2}}{2}\\
& = \; \frac{1-\sqrt{2\alpha}}{2}
\end{align*}
where, with $\mathcal{U} = \bigl\{\bm{u} = ({\kappa}/{\sqrt{\mathfrak{s}}})\mathbf{1}_S \in \R^p:  S \subseteq [p],  \; \vert S \vert = \mathfrak{s} \bigr\}$, 
we have, using \eqref{eq:uv},
\begin{align*}
\alpha &= {p \choose \mathfrak{s}}^{-2}\sum_{\bm u, \bm v \in \mc U} \l[\E_{\p_{\bm0}}\l(\frac{\mathrm{d}\p_{\bm u}}{\mathrm{d}\p_{\bm 0}}\frac{\mathrm{d}\p_{\bm v}}{\mathrm{d}\p_{\bm 0}}\r)\r]^{\Delta} - 1 \\
& = {p \choose \mathfrak{s}}^{-2}\sum_{S, T \subseteq [p],\, \vert S\vert = \vert T\vert = \mathfrak{s}} \l(\l(1-\frac{\vert S \cap T \vert}{\mathfrak{s}}\frac{\kappa^2}{\sigma_{\vep}^2}\r)^2 - \l(\frac{\kappa^2}{\sigma_{\vep}^2}\r)^2\r)^{-\Delta/2} - 1\\
& = {p \choose \mathfrak{s}}^{-1}\sum_{S \subseteq [p],\, \vert S\vert = \mathfrak{s}} \l(\l(1-\frac{\vert S \cap [\mathfrak{s}] \vert}{\mathfrak{s}}\frac{\kappa^2}{\sigma_{\vep}^2}\r)^2 - \l(\frac{\kappa^2}{\sigma_{\vep}^2}\r)^2\r)^{-\Delta/2} - 1.
\end{align*}
Let $\xi_i \sim_{\iid} \mathrm{Ber}(1,\,\mathfrak{s}/p)$, $i = 1,\ldots, \mathfrak{s}$, and $\xi = \sum_{i = 1}^{\mathfrak{s}}\xi_i$, which has distribution $\mathrm{Bin}(\mathfrak{s},\,\mathfrak{s}/p)$. Using the stochastic ordering (in terms of dilatation) of hypergeometric and binomial distributions, $1/(1 -x)\le \exp(2x)$ if $0 \le x \le 3/4$, and $1+x\le \exp(x)$, we obtain
\begin{align*}
 \alpha + 1 \;& \le\; \E_{\xi}\l[\l(\l(1-\frac{\xi}{\mathfrak{s}}\frac{\kappa^2}{\sigma_{\vep}^2}\r)^2 - \l(\frac{\kappa^2}{\sigma_{\vep}^2}\r)^2\r)^{-\Delta/2}\r]   \\
 & \le \; \E_{\xi}\l[\exp\l(\Delta \frac{\kappa^4}{\sigma_{\vep}^4} + 2\Delta\frac{\xi}{\mathfrak{s}}\frac{\kappa^2}{\sigma_{\vep}^2}\r)\r]\\
 & = \; \exp\l(\Delta \frac{\kappa^4}{\sigma_{\vep}^4}\r)\l\{\E_{\xi_1}\l[\exp\l(2\Delta\frac{\xi_1}{\mathfrak{s}}\frac{\kappa^2}{\sigma_{\vep}^2}\r)\r] \r\}^{\mathfrak{s}}\\
 & \le \; \exp\l[\Delta \frac{\kappa^4}{\sigma_{\vep}^4} + \frac{\mathfrak{s}^2}{p}\l(\exp\l(\frac{2\Delta\kappa^2}{\mathfrak{s}\sigma^2_{\vep}}\r) - 1\r)\r] \\
 & = \; \exp\l(\frac{\mathfrak{s}^2}{4\Delta}\log^2\l(1 + \frac{p}{17\mathfrak{s}^2}\r) + \frac{1}{17}\r)\; \le \; \exp\l(\frac{2}{17}\r) \; < \; \frac{9}{8},
\end{align*}
where the last second inequality is due to $\mathfrak{s}\log\bigl(1 + {p}/({17\mathfrak{s}^2})\bigr) \le 2\sqrt{{\Delta}/{17}}$. Thus, 
\[
\inf_{\wh\theta}\; \sup_{\mc P_{\mathfrak{s}, \sigma_{\vep}}^{\tau, n, p}} \;\p\l(\bigl\vert\wh\theta -\theta\bigr\vert \ge \Delta\r)\; \ge \; \frac{1}{4}.
\]

\bigskip

{\bf Part~\ref{lem:lb:id:locate}.} Introduce 
\[
\gamma \;= \;\frac{\sigma_{\vep}^2}{4\kappa^2}\mathfrak{s} \log\bigl(1 + \frac{p}{17\mathfrak{s}^2}\bigr).
\]
We consider the following fuzzy hypotheses 
\[
\p_1 = {p \choose \mathfrak{s}}^{-1} \sum_{S\subseteq [p],\; \vert S \vert = \mathfrak{s}}\p_{\Delta, \frac{\kappa}{\sqrt{\mathfrak{s}}}\mathbf{1}_S, \bm 0}\quad\text{and}\quad
\p_2 = {p \choose \mathfrak{s}}^{-1} \sum_{S\subseteq [p],\; \vert S \vert = \mathfrak{s}}\p_{\Delta+2\gamma, \frac{\kappa}{\sqrt{\mathfrak{s}}}\mathbf{1}_S, \bm 0}
\]
where we choose $\kappa$ such that $\Delta\kappa^2 = \tau \sigma_{\vep}^2\mathfrak{s} \log\bigl(1 + {p}/({17\mathfrak{s}^2})\bigr)$. Then, similar as above, by \eqref{eq:uv} and \eqref{eq:fuzz}, we have
\[
\inf_{\wh\theta}\; \sup_{\mc P_{\mathfrak{s}, \sigma_{\vep}}^{\tau, n, p}} \;\p\l(\bigl\vert\wh\theta -\theta\bigr\vert \ge \gamma\r)\;
\ge\;\frac{1-\sqrt{\alpha/2}}{2},
\]
where $\alpha = \chi^2(\p_{2},\p_{1})$ and thus 
\[
\alpha \; =\; {p \choose \mathfrak{s}}^{-1}\sum_{S \subseteq [p],\, \vert S\vert = \mathfrak{s}} \l(\l(1-\frac{\vert S \cap [\mathfrak{s}] \vert}{\mathfrak{s}}\frac{\kappa^2}{\sigma_{\vep}^2}\r)^2 - \l(\frac{\kappa^2}{\sigma_{\vep}^2}\r)^2\r)^{-\gamma} - 1.
\]
Again by the similar calculation as in Part~\ref{lem:lb:id:detect}, we have
\begin{align*}
 \alpha + 1 \;&  \le \; \exp\l[2\gamma \frac{\kappa^4}{\sigma_{\vep}^4} + \frac{\mathfrak{s}^2}{p}\l(\exp\l(\frac{4\gamma\kappa^2}{\mathfrak{s}\sigma^2_{\vep}}\r) - 1\r)\r] \\
 & = \; \exp\l(\frac{\kappa^2}{2\sigma_{\vep}^2}\mathfrak{s}\log\l(1 + \frac{p}{17\mathfrak{s}^2}\r) + \frac{1}{17}\r)\; \le \; \exp\l(\frac{2}{17}\r) \; < \; \frac{3}{2},
\end{align*}
where the last second inequality is due to $\mathfrak{s}\log\bigl(1 + {p}/({17\mathfrak{s}^2})\bigr) \le \sqrt{{2\Delta}/({17\tau})}$ and the choice of~$\kappa$. This concludes the proof. 
\end{proof}

\section{Multiplier bootstrap}
\label{sec:boot}

Instead of sampling directly from the distribution of $\wh{\mbf V}_j$ in Corollary~\ref{cor:thm:two} for drawing the quantile $C_{\alpha/2, \infty}(\wh{\mbf V}_j)$ in~\eqref{eq:ci}, we can adopt a multiplier bootstrap procedure.
Recall $\wh{\mbf U}_{j, t}$ introduced for the estimation of $\bm\Gamma_j$, which approximates $\mbf U_{j, t} = \mbf x^{\tO}_t (\vep^{\tO}_t + (\mbf x^{\tO}_t)^\top \bar{\mu}^{\tE}_j)$.
This term, in turn, appears in the leading term of the decomposition of $\wc{\bm\delta}_j - \bm\delta_j$ to which Gaussian approximation is performed, see~\eqref{eq:debias:decomp} and also~\eqref{eq:debiased:decomp}.

Thus motivated, we adopt the multipler bootstrap as below: We generate
\begin{multline*}
\mbf W^{(b)}_j = \frac{1}{\sqrt{b^{\tE} - a^{\tE}}} \sum_{t = a^{\tE}_j + 1}^{b^{\tE}_j} \zeta^{(b)}_t \cdot \wh{w}^{\tE}_{j, t} \wh{\bm\Omega}^{\tE} \l( \wh{\mbf U}_{j, t} - \bar{\wh{\mbf U}}_j \r), \text{ \ where \ } \zeta^{(b)}_t \sim_{\iid} \mc N(0, 1),
\\
\wh{w}_{j, t}^{\tE} = \l\{ \begin{array}{ll}
- \sqrt{\frac{b^{\tE}_j - \wh\cp^{\tE}_j}{\wh\cp^{\tE}_j - a^{\tE}_j}} & \text{for \ } t \in \{a^{\tE}_j + 1, \ldots, \wh\cp^{\tE}_j \},  \\
\sqrt{\frac{\wh\cp^{\tE}_j - a^{\tE}_j}{b^{\tE}_j - \wh\cp^{\tE}_j}} & \text{for \ } t \in \{\wh\cp^{\tE}_j + 1, \ldots, b^{\tE}_j \}, 
\end{array}\r.
\text{ \ and \ }
\bar{\wh{\mbf U}}_j = \frac{1}{b^{\tE}_j - a^{\tE}_j} \sum_{t = a^{\tE}_j + 1}^{b^{\tE}} \wh{\mbf U}_{j, t},
\end{multline*}
for $b \in [B]$, with $B$ denoting the bootstrap sample size.
Then, we construct a simultaneous $100(1 - \alpha)\%$ confidence interval about $\delta_{ij}$ as
\begin{align*}
\mc C^{\text{boot}}_{ij}(\alpha) = \l( \wc{\delta}_{ij} - \frac{2}{\wh\Delta^{\tE}_j} C_{\alpha/2, \infty}(\mbf W^{(b)}_j), \, \wc{\delta}_{ij} - \frac{2}{\wh\Delta^{\tE}_j} C_{\alpha/2, \infty}(\mbf W^{(b)}_j) \r),
\end{align*}
with $C_{\alpha/2, \infty}(\mbf W^{(b)}_j)$ denoting the upper $\alpha/2$-quantile of $\{\vert \mbf W^{(b)}_j \vert_\infty, \, b \in [B]\} \cup \{\vert \wc{\bm\delta}_j \vert_\infty\}$.
The performance of the multiplier bootstrap procedure is demonstrated numerically in Section~\ref{sec:sim:ci}.

\section{Proofs}

\subsection{Proof of Lemma~\ref{lem:tv}}
For simplicity, we present the proof for the specific case of $\bm\Sigma$ being invertible. In general, the proof is similar and we need to introduce projection on the range of $\bm\Sigma$, see Theorem~1.1 of \citet{devroye2018total}. 

Define $\mbf Z_t({\bm \beta}):= (\mbf x_t^\top, Y_t)^\top \in \R^{p+1}$ such that $Y_t = \mbf x_t^\top {\bm\beta} + \vep_t$, i.e.\ $\mbf Z_t({\bm \beta}) \sim \mathcal{N}(\bm 0, \bm \Sigma_{{\bm\beta}})$ with 
$$
\bm \Sigma_{{\bm\beta}} := 
\bmx
\bm\Sigma &\;\; \bm\Sigma {\bm\beta}\\
{\bm\beta}^\top\bm\Sigma & \;\;{\bm\beta}^\top \bm\Sigma {\bm\beta} + \sigma_{\vep}^2
\emx.
$$
Let $\mbf I_k$ be the identity matrix in $\R^{k \times k}$ for $k\in\N$. For $j \in \{0,1\}$, we have 
\begin{align*}
\bm\Sigma_{\bm\beta}^{-1} \bm\Sigma_{\bm\beta_j} & = 
\bmx
\bm\Sigma^{-1} + \bm\beta\bm\beta^\top/\sigma_{\vep}^2 & -\bm\beta/\sigma_{\vep}^2\\
-\bm\beta^\top/\sigma_{\vep}^2& 1/\sigma_{\vep}^2
\emx
\bmx
\bm\Sigma & \bm\Sigma {\bm\beta}_j\\
{\bm\beta}_j^\top\bm\Sigma & {\bm\beta}_j^\top\bm\Sigma {\bm\beta}_j + \sigma_{\vep}^2
\emx \\
& =
\bmx
\mbf I_p - \bm\beta(\bm\beta_j - \bm\beta)^\top\bm\Sigma/\sigma_{\vep}^2 & 
(\bm\beta_j - \bm\beta) - \bm\beta(\bm\beta_j - \bm\beta)^\top\bm\Sigma{\bm\beta}_j/\sigma_{\vep}^2\\
({\bm\beta}_j - \bm\beta)^\top\bm\Sigma/\sigma_{\vep}^2 & (\bm\beta_j - \bm\beta)^\top\bm\Sigma{\bm\beta}_j/\sigma_{\vep}^2 + 1
\emx.
\end{align*}
Further, it holds that
\begin{align*}
& \tr\l(\bigl(\bm\Sigma_{\bm\beta}^{-1/2} \bm\Sigma_{\bm\beta_j} \bm\Sigma_{\bm\beta}^{-1/2} - \mbf I_{p+1}\bigr)^2\r) = \tr\l(\bigl(\bm\Sigma_{\bm\beta}^{-1} \bm\Sigma_{\bm\beta_j} - \mbf I_{p+1}\bigr)^2\r)\\
= & \, \tr\l(
\bmx
- \bm\beta(\bm\beta_j - \bm\beta)^\top\bm\Sigma/\sigma_{\vep}^2 & 
(\bm\beta_j - \bm\beta) - \bm\beta(\bm\beta_j - \bm\beta)^\top\bm\Sigma{\bm\beta}_j/\sigma_{\vep}^2\\
({\bm\beta}_j - \bm\beta)^\top\bm\Sigma/\sigma_{\vep}^2 & (\bm\beta_j - \bm\beta)^\top\bm\Sigma{\bm\beta}_j/\sigma_{\vep}^2
\emx^2\r)\\
= & \, \l(\frac{(\bm\beta_j - \bm\beta)^\top \bm\Sigma(\bm\beta_j - \bm\beta)}{\sigma_{\vep}^2} + 1\r)^2 - 1,
\end{align*}
where $\tr(\cdot)$ is the trace operator. Note that $\bm\Sigma_{\bm\beta}^{-1} \bm\Sigma_{\bm\beta_j} - \mbf I_{p+1}$ and $\bm\Sigma_{\bm\beta}^{-1/2} \bm\Sigma_{\bm\beta_j} \bm\Sigma_{\bm\beta}^{-1/2} - \mbf I_{p+1}$ have the same eigenvalues. Then, by the general inequality in Theorem~1.1 of \citet{devroye2018total} and the temporal independence, we obtain
$$
\frac{1}{100} \le \frac{\mathrm{TV} \bigl( \p_0(\bm\beta),\, \p_{\cp_1}(\bm\beta_0, \bm\beta_1) \bigr)}{\min\left\{ 1, \, \sqrt{\cp_1 \l(\frac{(\bm\beta_0 - \bm\beta)^\top \bm\Sigma(\bm\beta_0 - \bm\beta)}{\sigma_{\vep}^2} + 1\r)^2 + (n-\cp_1)\l(\frac{(\bm\beta_1 - \bm\beta)^\top \bm\Sigma(\bm\beta_1 - \bm\beta)}{\sigma_{\vep}^2} + 1\r)^2 - n}\right\}} \le \frac{3}{2}.
$$
It is clear to see that, for any $a,b\ge 0$,
\begin{multline*}
\min\l\{1,\, 2\cp_1 a +2(n-\cp_1)b\r\} \le \min\l\{1,\, \cp_1 a^2 + 2\cp_1 a +(n-\cp_1)b^2 + 2(n-\cp_1)b\r\}\\ \le 
\min\l\{1,\, 3\cp_1 a +3(n-\cp_1)b\r\}.
\end{multline*}
This implies 
$$
\frac{1}{100^2} \le \frac{\mathrm{TV} \bigl( \p_0(\bm\beta),\, \p_{\cp_1}(\bm\beta_0, \bm\beta_1) \bigr)^2}{\min\left\{ 1, \, {\cp_1 (\bm\beta_0 - \bm\beta)^\top \bm\Sigma(\bm\beta_0 - \bm\beta)/\sigma_{\vep}^2 + (n-\cp_1)(\bm\beta_1 - \bm\beta)^\top \bm\Sigma(\bm\beta_1 - \bm\beta)/\sigma_{\vep}^2}\right\}} \le \frac{27}{4}.
$$
Note that 
\begin{align*}
\min_{\substack{\bm\beta, \bm\beta_0, \bm\beta_1: \\ \bm\beta_1 - \bm\beta_0 = \bm\delta}}
\frac{\cp_1(\bm\beta_0 - \bm\beta)^\top \bm\Sigma(\bm\beta_0 - \bm\beta)}{\sigma_{\vep}^2} + \frac{(n-\cp_1)(\bm\beta_1 - \bm\beta)^\top \bm\Sigma(\bm\beta_1 - \bm\beta)}{\sigma_{\vep}^2} 
= \frac{\cp_1(n-\cp_1)}{n\sigma^2_\vep}\bm\delta^\top \bm\Sigma\bm\delta,
\end{align*}
where the minimum is attained when $n \bm\beta = (n-\cp_1)\bm\beta_1 + \cp_1\bm\beta_0$. This, together with the monotonicity of $\min(\cdot)$, concludes the proof. 

\subsection{Preliminary lemmas}

We denote by $\mbf e_\ell \in \R^{p + 1}$ the vector of zeros except for its $\ell$-th element set to be one, and $\mathbb{B}_d(r) = \{\mbf a: \, \vert \mbf a \vert_d \le r\}$ the $\ell_d$-ball of radius $r$ with the dimension of $\mbf a$ determined within the context.

\begin{lem}
	\label{lem:exp}
	\begin{enumerate}[label = (\roman*)]
		\item \label{lem:exp:one} {\bf ({Theorem~3 of \citealp{wu2016performance}}).}
		Suppose that Assumption~\ref{assum:func:dep}~\ref{assum:fd:exp} holds.
		Let $\omega = 2/(1 + 2\kappa)$. Then for all $0 \le s < e \le n$ and $z > 0$, we have
		\begin{align*}
		\sup_{\mbf a, \mbf b \in \mathbb{B}_2(1)} \p\l(
		\frac{1}{\sqrt{e - s}} \l\vert \sum_{t = s + 1}^e \l[ \mbf a^\top \mbf Z_t \mbf Z_t^\top \mbf b - \E\l( \mbf a^\top \mbf Z_t \mbf Z_t^\top \mbf b\r) \r] \r\vert \ge z
		\r) \le C_\kappa \exp\l( - \frac{z^\omega}{2 e \omega \Xi^{\omega} } \r).
		\end{align*}
		
		\item \label{lem:exp:two} {\bf (Theorem~6.6 of \citealp{zhang2021convergence}, Lemma B.5 of \citealp{cho2022high}).} Suppose that Assumption~\ref{assum:func:dep}~\ref{assum:fd:gauss} holds. Then for all $0 \le s < e \le n$ and $z > 0$, there exists a universal constant $C > 0$ such that
		\begin{multline*}
		\sup_{\mbf a, \mbf b \in \mathbb{B}_2(1)} \p\l(
		\frac{1}{\sqrt{e - s}} \l\vert \sum_{t = s + 1}^e \l[ \mbf a^\top \mbf Z_t \mbf Z_t^\top \mbf b - \E\l( \mbf a^\top \mbf Z_t \mbf Z_t^\top \mbf b\r) \r] \r\vert \ge z \r) \\
		\le \; 2 \exp\l(- C \min\l(\frac{z^2}{\Xi^2}, \frac{z \sqrt{e - s} }{\Xi} \r) \r).
		\end{multline*}

        \item \label{lem:exp:three} {\bf (Theorem~3.1 of \citealp{KuCh22})} Suppose that Assumption~\ref{assum:func:dep}~\ref{assum:fd:ind} holds. Then for all $0 \le s < e \le n$ and $z > 0$, we have
		\begin{multline*}
		\sup_{\mbf a, \mbf b \in \mathbb{B}_2(1)} \p\l(
		\frac{1}{\sqrt{e - s}} \l\vert \sum_{t = s + 1}^e \l[ \mbf a^\top \mbf Z_t \mbf Z_t^\top \mbf b - \E\l( \mbf a^\top \mbf Z_t \mbf Z_t^\top \mbf b\r) \r] \r\vert \ge z \r) \\
		\le \; 2 \exp\l(- C_r \min\l(\frac{z^2}{\Xi^2}, \frac{z^{r/2} ({e - s})^{r/4}}{\Xi^{r/2}} \r) \r).
		\end{multline*}
	\end{enumerate}
\end{lem}

\begin{lem}
	\label{lem:bound}
	Suppose that Assumption~\ref{assum:func:dep}
 holds. Then with $\psi_{n, p}$ defined in~\eqref{eq:psi}, there exist some constants $C_0, C_1, c_2, c_3 \in (0, \infty)$ that depend only on $\kappa$ (with $\kappa = 0$ under Assumption~\ref{assum:func:dep}~\ref{assum:fd:gauss} and $\kappa = \max(1/r - 1/2, 0)$ under~\ref{assum:fd:ind}) and $\Xi$, such that 
	\begin{align*}
	\p\l(\mc E_{n, p}^{(1)} \cap \bigcap_{\mbf a, \mbf b \in \mc A} \mc E_{n, p}^{(2)}(\mbf a, \mbf b) \r) \ge 1 - c_2 (p \vee n)^{-c_3},
	\end{align*}
	where 
	\begin{align}
	\mc E_{n, p}^{(1)} &= \l\{ \max_{\substack{0 \le s < e \le n 
	}} \frac{1}{\sqrt{\max(e - s, C_1 \psi_{n, p}^2)}} \l\vert \sum_{t = s + 1}^e \mbf x_t \vep_t \r\vert_\infty \le C_0\psi_{n, p} \r\} \text{ \ and}
	\nn
	\\
	\mc E_{n, p}^{(2)}(\mbf a, \mbf b) &= \l\{ \max_{\substack{0 \le s < e \le n 	}} \frac{1}{\sqrt{\max(e - s, C_1 \psi_{n, p}^2)}} \l\vert \sum_{t = s + 1}^e \mbf a^\top \l(\mbf x_t \mbf x_t^\top - \bm\Sigma\r) \mbf b \r\vert \le C_0 \psi_{n, p} \r\},
	\nn
	\end{align}
	and $\mc A \subset \mathbb{B}_2(1)$ has its cardinality $\vert \mc A \vert \le 2(n + p)$.
\end{lem}

\begin{proof} First, suppose that Assumption~\ref{assum:func:dep}~\ref{assum:fd:exp} holds.
	By Lemma~\ref{lem:exp}~\ref{lem:exp:one}, we have
	\begin{align*}
	\p\l( \mc E_{n, p}^{(1)} \r)
	\ge 1 - p n^2 C_\kappa \exp\l(- \frac{C_0^\omega \log(p \vee n)}{2 e \omega \Xi^\omega} \r)  \ge 1 - c_2 (p \vee n)^{-c_3}
	\end{align*}
	for large enough $C_0$ that depends only on $\kappa$ and $\Xi$, where the inequality follows from setting $\mbf a = \mbf e_\ell, \, \ell \in [p]$ and $\mbf b = \mbf e_{p + 1}$ and then applying the Bonferroni correction.
	Analogously, 
	\begin{align*}
	\p\l( \bigcap_{\mbf a, \mbf b \in \mc A} \mc E_{n, p}^{(2)}(\mbf a) \r) \ge 1 - \vert \mc A \vert^2 n^2 C_\kappa \exp\l(- \frac{C_0^\omega \log(p \vee n)}{2 e \omega \Xi^\omega} \r) \ge 1 - c_2 (p \vee n)^{-c_3}.
	\end{align*}
	When either Assumption~\ref{assum:func:dep}~\ref{assum:fd:gauss} or~\ref{assum:fd:ind} holds, we adopt similar arguments using Lemma~\ref{lem:exp}~\ref{lem:exp:two} and~\ref{lem:exp:three} with $C_0$ depending only on $\Xi$ and $C$ and $C_1$ on $\Xi$, $C_0$ and $\kappa$.
\end{proof}

\subsection{Proofs for the results in Section~\ref{sec:cp}}

\subsubsection{Proof of Theorem~\ref{thm:one}}
\label{sec:pf:thm:one}

Supposing that $\cp_j - \cp_{j - 1} \ge 12$, we define for all $j \in [q]$,
\begin{align*}
\underline{\mc I}_j &= \l\{ \cp_j - \l\lfloor \frac{\Delta_j}{3} \r\rfloor, \ldots,  \cp_j - \l\lceil \frac{\Delta_j}{12} \r\rceil \r\} \text{ \ and \ } 
\bar{\mc I}_j = \l\{ \cp_j + \l\lceil \frac{\Delta_j}{12} \r\rceil, \ldots,   \cp_j + \l\lfloor \frac{\Delta_j}{3} \r\rfloor \r\},
\end{align*}
where $\Delta_j = \min(\cp_j - \cp_{j - 1}, \cp_{j + 1} - \cp_j)$.
Then for large enough sample size $n$, we always have the following event hold: 
\begin{align}
\label{eq:cond:M}
\mc M_n = \Bigl\{\text{For each} \, j \in [q], \; \text{there is some } \, (a, b] \in \mathbb{M} \text{ \ such that \ } (a, b) \in \underline{\mc I}_j \times \bar{\mc I}_j \Bigr\}.
\end{align}

Define 
\begin{align}
\mc E_{n, p} = \mc E_{n, p}^{(1)} \cap \bigcap_{\mbf a, \mbf b \in \mc A} \mc E_{n, p}^{(2)}(\mbf a, \mbf b) \text{ \ where \ } \mc A = \l\{ \frac{\bm\delta_j} {\vert \bm\delta_j \vert_2}, \frac{\bm\mu_j}{\vert \bm\mu_j \vert_2}, \, j \in [q] \r\}  \cup \{ \mbf e_\ell, \, \ell \in [p + 1] \}.
\label{eq:set:e}
\end{align}
By Lemma~\ref{lem:bound}, we have $\p(\mc E_{n, p}) \ge 1 - c_2 (p \vee n)^{-c_3}$.
In what follows, we show that conditional on $\mc E_{n, p}$, the claims in $\mc S_{n, p}$ hold, i.e.\ $\mc E_{n, p} \subset \mc S_{n, p}$ and thus $\p(\mc S_{n, p}) \ge 1 - c_2 (p \vee n)^{-c_3}$.

Throughout the proof, we consider some $(s, e), \, 0 \le s < e \le n$, which satisfies:
\begin{enumerate}[label = (S\arabic*)]
	\item \label{not:se:one} The set $\mathbb{C}_{s, e} \ne \emptyset$, where 
	\begin{align*}
	\mathbb{C}_{s, e} = \bigl\{ j\in [q]\;:\;\text{ there is } (a, b] \in \mathbb{M} \text{ such that } s \le a < b \le e \text{ and } (a, b)\in \underline{\mc I}_j \times \bar{\mc I}_j\bigr\}.
	\end{align*}
	
	\item \label{not:se:two} There exist some $j \in \{0, \ldots, q\}$ and $j' \in \{1, \ldots, q + 1\}$ such that
	\begin{align*}
	& \vert \bm\Sigma \bm\delta_j \vert_\infty^2 \vert s - \cp_j \vert \le c_1 \Psi_j^2 \psi_{n, p}^2, \text{ \ and \ }
	\vert \bm\Sigma \bm\delta_{j'} \vert_\infty^2 \vert e - \cp_{j'} \vert \le c_1 \Psi_{j'}^2 \psi_{n, p}^2 \text{ \ (with $\Psi_0 = \Psi_{q + 1} = 0$).}
	\end{align*}
\end{enumerate}
We show that for $(s, e)$ meeting~\ref{not:se:one}--\ref{not:se:two}:
\begin{enumerate}[label = (R\arabic*)]
	\item \label{not:one} There exists at least one $\ell \in \mathbb{L}_{s, e}$ for which $T_\ell > \pi_{n, p}$, where 
	$$ \mathbb{L}_{s, e} := \{ \ell' \in [\vert \mathbb{M} \vert]: \, s \le a_{\ell'} < b_{\ell'} \le e\}. $$
	\item \label{not:two} Step~2 identifies $\wh\cp = k_{\ell^\circ}$ which, for some $j \in \mathbb{C}_{s, e}$, satisfies
	$\vert \bm\Sigma \bm\delta_j \vert_\infty^2 \vert \wh\cp - \cp_j \vert \le c_1 \Psi_j^2 \psi_{n, p}^2$.
\end{enumerate}

At the beginning of the algorithm, we have $(s, e) = (0, n)$ meet~\ref{not:se:one}--\ref{not:se:two} (since $s = \cp_0 = 0$ and $e = \cp_{q + 1} = n$ {and thanks to the definition of $\mathbb{M}$, see~\eqref{eq:cond:M}}) such that by~\ref{not:one}--\ref{not:two}, we add $\wh\cp$ to $\wh{\Cp}$ which, for some $j \in \mathbb{C}_{0, n} = [q]$, estimates the location $\cp_j$.
Then, we no longer have such $j$ in $\mathbb{C}_{s, e}$ for the subsequently considered $(s, e)$ since either $\cp_j \notin \{s + 1, \ldots, e - 1\}$ or, even so, it has been detected by either $s$ or $e$ such that $\min(\cp_j - s, e - \cp_j) \le c_1 \Psi_j^2 \vert \bm\Sigma \bm\delta_j \vert_\infty^{-2} \psi_{n, p}^2 < \Delta_j/12$ for $c_0$ large enough in Assumption~\ref{assum:size}.
This rules out the possibility of such $j$ belonging to $\mathbb{C}_{s, e}$ and hence a duplicate estimator of $\cp_j$ added to $\wh{\Cp}$.
More specifically, for any $\ell \in \mathbb{L}_{s, e}$ satisfying $\cp_j \in \{ a_\ell + 1, \ldots, b_\ell - 1\}$, we have $\min_{k \in \{a_\ell, b_\ell \}} \vert \bm\Sigma \bm\delta_j \vert_\infty^2 \vert k - \cp_j \vert \le c_1 \Psi_j^2 \psi_{n, p}^2$, such that
\begin{align}
T_\ell =& \, \sqrt{ \frac{(k_\ell - a_\ell) (b_\ell - k_\ell)}{b_\ell - a_\ell} } \l\vert \wh{\bm\gamma}_{k_\ell, b_\ell} - \wh{\bm\gamma}_{a_\ell, k_\ell} \r\vert_\infty \le 2 \max_{j: \, \cp_j \in \{a_\ell + 1, \ldots, b_\ell - 1\}} \sqrt{ \frac{(\cp_j - a_\ell) (b_\ell - \cp_j)}{b_\ell - a_\ell} } \l\vert \bm\Sigma \bm\delta_j \r\vert_\infty 
\nn \\
& + \sqrt{ \frac{(k_\ell - a_\ell) (b_\ell - k_\ell)}{b_\ell - a_\ell} } \l\vert \wh{\bm\gamma}_{k_\ell, b_\ell} - \wh{\bm\gamma}_{a_\ell, k_\ell} - \l( \bm\gamma_{k_\ell, b_\ell} - \bm\gamma_{a_\ell, k_\ell} \r) \r\vert_\infty 
\nn \\
\le & \, 2 \sqrt{c_1} \Psi_j \psi_{n, p} + 2 C_0 \Psi_j \psi_{n, p} < \pi_{n, p},
\label{eq:bounadary}
\end{align}
where the first inequality follows from Lemma~\ref{lem:cusum}, the second from Lemma~\ref{lem:gamma}.

Once $\vert \wh{\Cp} \vert = q$, for any $(s, e)$ defined by two neighbouring points in $\{0, n\} \cup \wh{\Cp}$, we have all $\ell \in \mathbb{L}_{s, e}$ satisfy $\vert \Cp \cap \{ a_\ell + 1, \ldots, b_\ell - 1\} \vert \le 2$.
Then, by the identical arguments leading to~\eqref{eq:bounadary}, we do not have any $T_\ell, \, \ell \in \mathbb{L}_{s, e}$, exceed $\pi_{n, p}$, thus the algorithm is terminated.

It remains to show that~\ref{not:one}--\ref{not:two} hold.

\begin{proof}[Proof of~\ref{not:one}]
	The set $\mathbb{C}_{s, e}$ is not empty  by~\ref{not:se:one}.
	For every $j \in \mathbb{C}_{s, e}$, we have $\ell = \ell(j) \in \mathbb{L}_{s, e}$ and
	$\Delta_j/12 \le \min(\cp_j - a_{\ell},\; b_{\ell} - \cp_j) \le \max(\cp_j - a_{\ell}, b_{\ell} - \cp_j) \le \Delta_j/3$. Then, by the definition of $T_\ell$ and Lemma~\ref{lem:gamma},
	\begin{align*}
	T_{\ell} \ge & \, T_{a_{\ell}, \cp_j, b_{\ell}}
	\\
	\ge & \, 
	\sqrt{\frac{(\cp_j - a_{\ell})(b_{\ell} - \cp_j)}{b_{\ell} - a_{\ell}}} \l( \l\vert \bm\Sigma \bm\delta_j \r\vert_\infty - \l\vert \wh{\bm\gamma}_{\cp_j, b_{\ell}} - \wh{\bm\gamma}_{a_{\ell}, \cp_j} - \l( \bm\gamma_{\cp_j, b_{\ell}} - \bm\gamma_{a_{\ell}, \cp_j} \r) \r\vert_\infty \r)
	\\
	\ge & \, \sqrt{\frac{\Delta_j}{15}} \l\vert \bm\Sigma \bm\delta_j \r\vert_\infty  - 2 C_0 \Psi_j \psi_{n, p} > \pi_{n, p},
	\end{align*}
	provided that $c_0$, in \cref{assum:size}, is sufficiently large.
\end{proof}

\begin{proof}[Proof of~\ref{not:two}] 
	If $\vert \wh\cp - \cp_j \vert \le C_1 \psi_{n, p}^2$, the statement holds provided that $c_1$ is large enough to meet $C_1 < c_1 \max_{j \in [q]} \vert \bm\Sigma \bm\delta_j \vert_\infty^{-2} \Psi_j^2$, since $ \vert \bm\Sigma \bm\delta_j \vert_\infty \lesssim \Psi_j$. 
	Therefore, we consider the case where $\vert \wh\cp - \cp_j \vert \ge C_1 \psi_{n, p}^2$ below.
	Without loss of generality, we consider the case $\wh\cp \le \cp_j$; the case where $\wh\cp \ge \cp_j + 1$ is handled analogously.
	
	From the arguments adopted in the proof of~\ref{not:one}, we have $b_{\ell^\circ} - a_{\ell^\circ} \le \min_{j \in \mathbb{C}_{s, e}} 2\Delta_j/3$.
	In particular, this implies that $\vert (a_{\ell^\circ}, b_{\ell^\circ}) \cap \Cp \vert = 1$.
	We first establish that for $j \in [q]$ satisfying $\{\cp_j\} = (a_{\ell^\circ}, b_{\ell^\circ}) \cap \Cp$, we do not have either $s$ or $e$ estimate $\cp_j$ in the sense that $\min_{k \in \{a_{\ell^\circ}, b_{\ell^\circ}\}} \vert \bm\Sigma \bm\delta_j \vert_\infty^2 \vert k - \cp_j \vert \le c_1 \Psi_j^2 \psi_{n, p}^2$ by~\ref{not:se:two}. If so, by the arguments analogous to those given in~\eqref{eq:bounadary}, we have $T_{\ell^\circ} < \pi_{n, p}$.
	Further, we have that 
	\begin{align}
	& \sqrt{\frac{(\cp_j - a_{\ell^\circ}) (b_{\ell^\circ} - \cp_j) }{ b_{\ell^\circ} - a_{\ell^\circ} }} \l\vert \bm\Sigma \bm\delta_j \r\vert_\infty 
	\ge T_{\ell^\circ} - 2 C_0 \Psi_j \psi_{n, p} > \pi_{n, p} - 2 C_0 \Psi_j \psi_{n, p} \ge \l( 1 - \frac{2 C_0}{c^\prime} \r) \pi_{n, p} \label{eq:not:min:dist}
	\end{align}
	by Lemmas~\ref{lem:gamma} and~\ref{lem:cusum}.
	Since $T_{\ell^\circ} = T_{a_{\ell^\circ}, \wh\cp, b_{\ell^\circ}} \ge T_{a_{\ell^\circ}, \cp_j, b_{\ell^\circ}}$, we have
	\begin{align}
	& \sqrt{\frac{(\cp_j - a_{\ell^\circ}) (b_{\ell^\circ} - \cp_j) }{ b_{\ell^\circ} - a_{\ell^\circ} }} \l\vert \bm\gamma_{\cp_j, b_{\ell^\circ}} - \bm\gamma_{a_{\ell^\circ}, \cp_j} \r\vert_\infty
	- 
	\sqrt{\frac{(\wh\cp - a_{\ell^\circ}) (b_{\ell^\circ} - \wh\cp) }{ b_{\ell^\circ} - a_{\ell^\circ} }} \l\vert \bm\gamma_{\wh\cp, b_{\ell^\circ}} - \bm\gamma_{a_{\ell^\circ}, \wh\cp} \r\vert_\infty
	\nn \\
	\le & \, 
	\sqrt{\frac{(\wh\cp - a_{\ell^\circ}) (b_{\ell^\circ} - \wh\cp) }{ b_{\ell^\circ} - a_{\ell^\circ} }} \l\vert \wh{\bm\gamma}_{\wh\cp, b_{\ell^\circ}} - \wh{\bm\gamma}_{a_{\ell^\circ}, \wh\cp} - \l( \bm\gamma_{\wh\cp, b_{\ell^\circ}} - \bm\gamma_{a_{\ell^\circ}, \wh\cp} \r) \r\vert_\infty
	\nn \\
	& \, +
	\sqrt{\frac{(\cp_j - a_{\ell^\circ}) (b_{\ell^\circ} - \cp_j) }{ b_{\ell^\circ} - a_{\ell^\circ} }} \l\vert \wh{\bm\gamma}_{\cp_j, b_{\ell^\circ}} - \wh{\bm\gamma}_{a_{\ell^\circ}, \cp_j} - \l( \bm\gamma_{\cp_j, b_{\ell^\circ}} - \bm\gamma_{a_{\ell^\circ}, \cp_j} \r) \r\vert_\infty \le 4 C_0 \Psi_j \psi_{n, p}
	\label{eq:prem:rate}
	\end{align}
	by Lemma~\ref{lem:gamma}.
	From this, it follows that $\vert \wh\cp - \cp_j \vert \le \min( \cp_j - a_{\ell^\circ}, b_{\ell^\circ} - \cp_j)/4$; otherwise, by~\eqref{eq:cusum:diff},
	\begin{align}
	& \sqrt{\frac{(\cp_j - a_{\ell^\circ}) (b_{\ell^\circ} - \cp_j) }{ b_{\ell^\circ} - a_{\ell^\circ} }} \l\vert \bm\gamma_{\cp_j, b_{\ell^\circ}} - \bm\gamma_{a_{\ell^\circ}, \cp_j} \r\vert_\infty
	- 
	\sqrt{\frac{(\wh\cp - a_{\ell^\circ}) (b_{\ell^\circ} - \wh\cp) }{ b_{\ell^\circ} - a_{\ell^\circ} }} \l\vert \bm\gamma_{\wh\cp, b_{\ell^\circ}} - \bm\gamma_{a_{\ell^\circ}, \wh\cp} \r\vert_\infty
	\nn \\
	= & \, \sqrt{\frac{(\cp_j - a_{\ell^\circ}) (b_{\ell^\circ} - \cp_j) }{ b_{\ell^\circ} - a_{\ell^\circ} }} \l( 1 - \sqrt{\frac{1 - \vert \wh\cp - \cp_j \vert/(\cp_j - a_{\ell^\circ})}{1 + \vert \wh\cp - \cp_j \vert/(b_{\ell^\circ} - \cp_j)}} \r)  \vert \bm\Sigma \bm\delta_j \vert_\infty 
	\nn \\
	\ge & \, \l( 1 - \sqrt{\frac{3}{5}} \r) \sqrt{\frac{(\cp_j - a_{\ell^\circ}) (b_{\ell^\circ} - \cp_j) }{ b_{\ell^\circ} - a_{\ell^\circ} }} \vert \bm\Sigma \bm\delta_j \vert_\infty\nn\\
	\ge &\l(1 - \sqrt{\frac{3}{5}}\r) (\pi_{n, p} - 2 C_0 \Psi_j \psi_{n, p}) > 4 C_0 \Psi_j \psi_{n, p},
	\label{eq:cusum:deriv}
	\end{align}
	which violates~\eqref{eq:prem:rate}, where the penultimate inequality makes use of~\eqref{eq:not:min:dist}.
	Therefore,
	\begin{align}
	\vert \wh\cp - \cp_j \vert \le \frac{1}{4} \Delta^\circ \text{ \ with \ } \Delta^\circ = \min( \cp_j - a_{\ell^\circ}, b_{\ell^\circ} - \cp_j). 
	\label{eq:one:fourth}
	\end{align}
	
	Let $\mbf e^\circ \in \R^p$ denote a vector containing $p - 1$ zeros and a single one which satisfies
	\begin{align*}
	T_{\ell^\circ} = \sqrt{\frac{(\wh\cp - a_{\ell^\circ})(b_{\ell^\circ} - \wh\cp)}{b_{\ell^\circ} - a_{\ell^\circ}}} \l\vert (\mbf e^\circ)^\top \l(\wh{\bm\gamma}_{\wh\cp, b_{\ell^\circ}} - \wh{\bm\gamma}_{a_{\ell^\circ}, \wh\cp}\r) \r\vert
	\end{align*}
	and without loss of generality, let $(\mbf e^\circ)^\top \bm\Sigma \bm\delta_j > 0$.
	Then, from Lemmas~\ref{lem:gamma} and~\ref{lem:cusum} and the arguments adopted in the proof of the latter, 
	\begin{align*}
	& \sqrt{\frac{(\cp_j - a_{\ell^\circ})(b_{\ell^\circ} - \cp_j)}{b_{\ell^\circ} - a_{\ell^\circ}}} (\mbf e^\circ)^\top \bm\Sigma \bm\delta_j + 2 C_0 \Psi_j \psi_{n, p} \ge T_{\ell^\circ} \ge T^*_{a_{\ell^\circ}, \cp_j, b_{\ell^\circ}} - 2 C_0 \Psi_j \psi_{n, p}, \text{ \ such that}
	\\
	& \sqrt{\frac{(\cp_j - a_{\ell^\circ})(b_{\ell^\circ} - \cp_j)}{b_{\ell^\circ} - a_{\ell^\circ}}} (\mbf e^\circ)^\top \bm\Sigma \bm\delta_j \ge T^*_{a_{\ell^\circ}, \cp_j, b_{\ell^\circ}}  - 4 C_0 \Psi_j \psi_{n, p} \ge \l( 1 - \frac{4 C_0}{c^\prime - 2 C_0} \r) T^*_{a_{\ell^\circ}, \cp_j, b_{\ell^\circ}} 
	\end{align*}
	with the last inequality followed by~\eqref{eq:not:min:dist}.
	The definition of $T^*_{s, k, e}$ is in Lemma~\ref{lem:cusum}; in particular, $T^*_{a_{\ell^\circ}, \cp_j, b_{\ell^\circ}} = \sqrt{(\cp_j - a_{\ell^\circ})(b_{\ell^\circ} - \cp_j)/(b_{\ell^\circ} - a_{\ell^\circ})} \vert \bm\Sigma \bm\delta_j \vert_\infty$.
	This implies that 
	\begin{align}
	\label{eq:e:circ}
	(\mbf e^\circ)^\top \bm\Sigma \bm\delta_j \ge \l( 1 - \frac{4 C_0}{c^\prime - 2 C_0} \r) \vert \bm\Sigma \bm\delta_j \vert_\infty.
	\end{align}
	Also, using Lemma~\ref{lem:gamma}, the results in \eqref{eq:not:min:dist}, \eqref{eq:one:fourth} and~\eqref{eq:e:circ} and the arguments analogous to those adopted in~\eqref{eq:cusum:deriv}, we have
	\begin{align*}
	& \sqrt{\frac{(\wh\cp - a_{\ell^\circ})(b_{\ell^\circ} - \wh\cp)}{b_{\ell^\circ} - a_{\ell^\circ}}} (\mbf e^\circ)^\top \l( \wh{\bm\gamma}_{\wh\cp, b_{\ell^\circ}} - \wh{\bm\gamma}_{a_{\ell^\circ}, \wh\cp} \r) 
	\\
	\ge & \,
	\sqrt{\frac{(\wh\cp - a_{\ell^\circ})(b_{\ell^\circ} - \wh\cp)}{b_{\ell^\circ} - a_{\ell^\circ}}} (\mbf e^\circ)^\top \l( \bm\gamma_{\wh\cp, b_{\ell^\circ}} - \bm\gamma_{a_{\ell^\circ}, \wh\cp} \r) - 2 C_0 \Psi_j \psi_{n, p}
	\\
	\ge & \, \sqrt{\frac{3}{5}} \l( 1 - \frac{4 C_0}{c^\prime - 2 C_0} \r) \sqrt{\frac{(\cp_j - a_{\ell^\circ})(b_{\ell^\circ} - \cp_j)}{b_{\ell^\circ} - a_{\ell^\circ}}} \vert \bm\Sigma \bm\delta_j \vert_\infty - 2 C_0 \Psi_j \psi_{n, p} > 0,
	\end{align*}
	i.e.\ $T_{a_{\ell^\circ}, \wh\cp, b_{\ell^\circ}} = \sqrt{(\wh\cp - a_{\ell^\circ})(b_{\ell^\circ} - \wh\cp)/(b_{\ell^\circ} - a_{\ell^\circ})} (\mbf e^\circ)^\top ( \wh{\bm\gamma}_{\wh\cp, b_{\ell^\circ}} - \wh{\bm\gamma}_{a_{\ell^\circ}, \wh\cp} )$.
	Similarly, 
	\begin{align*}
	T_{a_{\ell^\circ}, \cp_j, b_{\ell^\circ}} = \sqrt{\frac{(\cp_j - a_{\ell^\circ})(b_{\ell^\circ} - \cp_j)}{b_{\ell^\circ} - a_{\ell^\circ}}} (\mbf e^\circ)^\top \l( \wh{\bm\gamma}_{\cp_j, b_{\ell^\circ}} - \wh{\bm\gamma}_{a_{\ell^\circ}, \cp_j} \r) > 0.
	\end{align*}
	
	Then, since $T_{a_{\ell^\circ}, \wh\cp, b_{\ell^\circ}} \ge T_{a_{\ell^\circ}, \cp_j, b_{\ell^\circ}}$, we have
	\begin{align}
	0 \le& \, \frac{1}{\sqrt{b_{\ell^\circ} - a_{\ell^\circ}}} \l( \sqrt{\frac{\wh\cp - a_{\ell^\circ}}{b_{\ell^\circ} - \wh\cp}} \sum_{t = \wh\cp + 1}^{b_{\ell^\circ}} Y_t \mbf x_t^\top -
	\sqrt{\frac{b_{\ell^\circ} - \wh\cp}{\wh\cp - a_{\ell^\circ}}} \sum_{t = a_{\ell^\circ} + 1}^{\wh\cp} Y_t \mbf x_t^\top \r. 
	\nn \\
	& \quad \l. - \sqrt{\frac{\cp_j - a_{\ell^\circ}}{b_{\ell^\circ} - \cp_j}} \sum_{t = \cp_j + 1}^{b_{\ell^\circ}} Y_t \mbf x_t^\top +
	\sqrt{\frac{b_{\ell^\circ} - \cp_j}{\cp_j - a_{\ell^\circ}}} \sum_{t = a_{\ell^\circ} + 1}^{\cp_j} Y_t \mbf x_t^\top \r) \mbf e^\circ 
	\nn \\
	=& \, \frac{1}{\sqrt{b_{\ell^\circ} - a_{\ell^\circ}}} \l( \sqrt{\frac{\wh\cp - a_{\ell^\circ}}{b_{\ell^\circ} - \wh\cp}} \sum_{t = \wh\cp + 1}^{b_{\ell^\circ}} \vep_t \mbf x_t^\top -
	\sqrt{\frac{b_{\ell^\circ} - \wh\cp}{\wh\cp - a_{\ell^\circ}}} \sum_{t = a_{\ell^\circ} + 1}^{\wh\cp} \vep_t \mbf x_t^\top \r. 
	\nn \\ 
	& \quad \l. - \sqrt{\frac{\cp_j - a_{\ell^\circ}}{b_{\ell^\circ} - \cp_j}} \sum_{t = \cp_j + 1}^{b_{\ell^\circ}} \vep_t \mbf x_t^\top +
	\sqrt{\frac{b_{\ell^\circ} - \cp_j}{\cp_j - a_{\ell^\circ}}} \sum_{t = a_{\ell^\circ} + 1}^{\cp_j} \vep_t \mbf x_t^\top \r) \mbf e^\circ 
	\nn \\
	& + \frac{1}{\sqrt{b_{\ell^\circ} - a_{\ell^\circ}}} \l[ \sqrt{\frac{\wh\cp - a_{\ell^\circ}}{b_{\ell^\circ} - \wh\cp}} \sum_{t = \wh\cp + 1}^{\cp_j} \bm\beta_{j - 1}^\top \l( \mbf x_t \mbf x_t^\top - \bm\Sigma \r) + 
	\sqrt{\frac{\wh\cp - a_{\ell^\circ}}{b_{\ell^\circ} - \wh\cp}} \sum_{t = \cp_j + 1}^{b_{\ell^\circ}} \bm\beta_j^\top \l( \mbf x_t \mbf x_t^\top - \bm\Sigma \r) \r.
	\nn \\
	& \quad - \sqrt{\frac{b_{\ell^\circ} - \wh\cp}{\wh\cp - a_{\ell^\circ}}} \sum_{t = a_{\ell^\circ} + 1}^{\wh\cp} \bm\beta_{j - 1}^\top \l( \mbf x_t \mbf x_t^\top - \bm\Sigma \r) -
	\sqrt{\frac{\cp_j - a_{\ell^\circ}}{b_{\ell^\circ} - \cp_j}} \sum_{t = \cp_j + 1}^{b_{\ell^\circ}} \bm\beta_j^\top \l( \mbf x_t \mbf x_t^\top - \bm\Sigma \r) 
	\nn \\ 
	& \quad + \l. \sqrt{\frac{b_{\ell^\circ} - \cp_j}{\cp_j - a_{\ell^\circ}}} \sum_{t = a_{\ell^\circ} + 1}^{\cp_j} \bm\beta_{j - 1}^\top \l( \mbf x_t \mbf x_t^\top - \bm\Sigma \r) \r] \mbf e^\circ 
	\nn \\
	& - \l(\sqrt{\frac{(\cp_j - a_{\ell^\circ})(b_{\ell^\circ} - \cp_j) }{b_{\ell^\circ} - a_{\ell^\circ}}} - \sqrt{\frac{(\wh\cp - a_{\ell^\circ})(b_{\ell^\circ} - \wh\cp) }{b_{\ell^\circ} - a_{\ell^\circ}}} \cdot \frac{b_{\ell^\circ} - \cp_j}{b_{\ell^\circ} - \wh\cp} \r) (\mbf e^\circ)^\top \bm\Sigma \bm\delta_j
	\nn \\
	=:& \, T_1 + T_2 - T_3.
	\label{eq:r:ineq}
	\end{align}
	
	First, note that
	\begin{align*}
	\sqrt{b_{\ell^\circ} - a_{\ell^\circ}} \vert T_1 \vert \le& \, \l\vert \sqrt{\frac{\wh\cp - a_{\ell^\circ}}{b_{\ell^\circ} - \wh\cp}} - \sqrt{\frac{\cp_j - a_{\ell^\circ}}{b_{\ell^\circ} - \cp_j}} \r\vert \; \l\vert \sum_{t = \cp_j + 1}^{b_{\ell^\circ}} \vep_t \mbf x_t^\top \r\vert_\infty + 
	\sqrt{\frac{\wh\cp - a_{\ell^\circ}}{b_{\ell^\circ} - \wh\cp}} \l\vert \sum_{t = \wh\cp + 1}^{\cp_j} \vep_t \mbf x_t^\top \r\vert_\infty
	\\ 
	& + \l\vert \sqrt{\frac{b_{\ell^\circ} - \wh\cp}{\wh\cp - a_{\ell^\circ}}} -  \sqrt{\frac{b_{\ell^\circ} - \cp_j}{\cp_j - a_{\ell^\circ}}} \r\vert \; \l\vert \sum_{t = a_{\ell^\circ} + 1}^{\cp_j} \vep_t \mbf x_t^\top \r\vert_\infty
	+ \sqrt{\frac{b_{\ell^\circ} - \wh\cp}{\wh\cp - a_{\ell^\circ}}} \l\vert \sum_{t = \wh\cp + 1}^{\cp_j} \vep_t \mbf x_t^\top \r\vert_\infty
	\\
	=:& \, T_{1, 1} + T_{1, 2} + T_{1, 3} + T_{1, 4}.
	\end{align*}
	From Lemmas~\ref{lem:bound} and~\ref{lem:sqrt:diff},
	\begin{align*}
	\frac{T_{1, 1}}{\sqrt{b_{\ell^\circ} - a_{\ell^\circ}}} \le \frac{\sqrt{b_{\ell^\circ} - a_{\ell^\circ}} \vert \wh\cp - \cp_j \vert}{\sqrt{\cp_j - a_{\ell^\circ}}(b_{\ell^\circ} - \cp_j)} C_0 \psi_{n, p} \le \frac{\sqrt{2} \vert \wh\cp - \cp_j \vert C_0 \psi_{n, p} }{\Delta^\circ},
	\end{align*}
	and $T_{1, 3}$ is similarly bounded.
	By Lemma~\ref{lem:bound},
	\begin{align*}
	\frac{T_{1, 2}}{\sqrt{b_{\ell^\circ} - a_{\ell^\circ}}} \le \sqrt{\frac{ \vert \wh\cp - \cp_j \vert (\wh\cp - a_{\ell^\circ})}{(b_{\ell^\circ} - a_{\ell^\circ})(b_{\ell^\circ} - \wh\cp)}} C_0 \psi_{n, p} \le \frac{C_0 \psi_{n, p} \sqrt{\vert \wh\cp - \cp_j \vert}}{\sqrt{\Delta^\circ}},
	\end{align*}
	and $T_{1, 4}$ is handled analogously. 
	Putting together the bounds on $T_{1, k}, \, k = 1, \ldots, 4$, and using that $(\Delta^\circ)^{-1} \vert \wh\cp - \cp_j \vert \le 1/4$ from~\eqref{eq:one:fourth}, we have
	\begin{align*}
	\vert T_1 \vert \le C_0 ( \sqrt{2} + 2 ) \psi_{n, p} \sqrt{\frac{\vert \wh\cp  - \cp_j \vert}{\Delta^\circ}}.
	\end{align*}
	Next, we bound $T_2$ as
	\begin{align*}
	\sqrt{b_{\ell^\circ} - a_{\ell^\circ}} \vert T_2 \vert \le & \,
	\l\vert \sqrt{\frac{\wh\cp - a_{\ell^\circ}}{b_{\ell^\circ} - \wh\cp}} - \sqrt{\frac{\cp_j - a_{\ell^\circ}}{b_{\ell^\circ} - \cp_j}} \r\vert \; \l\vert \sum_{t = \cp_j + 1}^{b_{\ell^\circ}} \bm\beta_j^\top \l( \mbf x_t \mbf x_t^\top - \bm\Sigma \r)\r\vert_\infty 
	\\
	& + \sqrt{\frac{\wh\cp - a_{\ell^\circ}}{b_{\ell^\circ} - \wh\cp}} \l\vert \sum_{t = \wh\cp + 1}^{\cp_j} \bm\beta_{j - 1}^\top \l( \mbf x_t \mbf x_t^\top - \bm\Sigma \r) \r\vert_\infty
	\\ 
	& + \l\vert \sqrt{\frac{b_{\ell^\circ} - \wh\cp}{\wh\cp - a_{\ell^\circ}}} -  \sqrt{\frac{b_{\ell^\circ} - \cp_j}{\cp_j - a_{\ell^\circ}}} \r\vert \; \l\vert \sum_{t = a_{\ell^\circ} + 1}^{\cp_j} \bm\beta_{j - 1}^\top \l( \mbf x_t \mbf x_t^\top - \bm\Sigma \r) \r\vert_\infty
	\\
	& + \sqrt{\frac{b_{\ell^\circ} - \wh\cp}{\wh\cp - a_{\ell^\circ}}} \l\vert \sum_{t = \wh\cp + 1}^{\cp_j} \bm\beta_{j - 1}^\top \l( \mbf x_t \mbf x_t^\top - \bm\Sigma \r) \r\vert_\infty
	=: T_{2, 1} + T_{2, 2} + T_{2, 3} + T_{2, 4}.
	\end{align*}
	Using the same arguments as those adopted in bounding $T_{1, 1}$, we have for $k = 1, 3$,
	\begin{align*}
	\frac{T_{2, k}}{\sqrt{b_{\ell^\circ} - a_{\ell^\circ}}} \le \frac{\sqrt{2} \vert \wh\cp - \cp_j \vert C_0 \max_{j' \in \{j - 1, j\}} \vert \bm\beta_{j'} \vert_2 \psi_{n, p}}{\Delta^\circ}.
	\end{align*}
	Also, using the arguments involved in bounding $T_{1, 2}$, we have for $k = 2, 4$,
	\begin{align*}
	\frac{T_{2, k}}{\sqrt{b_{\ell^\circ} - a_{\ell^\circ}}} \le \frac{ C_0 \sqrt{\vert \wh\cp - \cp_j \vert}  \max_{j' \in \{j - 1, j\}} \vert \bm\beta_{j'} \vert_2 \psi_{n, p}}{\sqrt{\Delta^\circ}}.
	\end{align*}
	Therefore, we have
	\begin{align*}
	\vert T_2 \vert \le C_0 ( \sqrt{2} + 2 ) \max_{j' \in \{j - 1, j\}} \vert \bm\beta_{j'} \vert_2 \psi_{n, p} \sqrt{\frac{\vert \wh\cp - \cp_j \vert}{\Delta^\circ}}.
	\end{align*}
	Finally, by~\eqref{eq:e:circ} and Lemma~7 of \cite{wang2018high}, we have
	\begin{align*}
	T_3 \ge \frac{2(1 - 8C_0/c^{\prime})}{3\sqrt{6}} \frac{\vert \wh\cp - \cp \vert}{\sqrt{\Delta^\circ}} \vert \bm\Sigma \bm\delta_j \vert_\infty.
	\end{align*}
	Then, from~\eqref{eq:r:ineq}, we have
	\begin{align}
	\frac{2(1 - 8C_0/c^\prime)}{3\sqrt{6}} \frac{\vert \wh\cp - \cp \vert}{\sqrt{\Delta^\circ}} \vert \bm\Sigma \bm\delta_j \vert_\infty & \le C_0 (2 + \sqrt{2}) \l( 1 + \max_{j' \in \{j - 1, j\}} \vert \bm\beta_{j'} \vert_2 \r) \psi_{n, p} \sqrt{\frac{\vert \wh\cp - \cp_j \vert}{\Delta^\circ}}, \text{ \ such that} \nn \\
	\vert \bm\Sigma \bm\delta_j \vert_\infty^2 \vert \wh\cp - \cp_j \vert & \le \l(\frac{3\sqrt{6}(1 + \sqrt{2})C_0}{1 - 8C_0/c^\prime}\r)^2 \Psi_j^2 \psi_{n, p}^2, \nn
	\end{align}
	from which the conclusion follows with a large enough constant $c_1$.
\end{proof}

\subsubsection{Supporting lemmas}

\begin{lem}
	\label{lem:gamma}
	Suppose that Assumption~\ref{assum:func:dep} holds. 
	Then, for all $(s, k, e) \in \mc I^\prime$ with
	\begin{align*}
	\mc I^\prime =& \, \Bigl\{0 \le s < k < e \le n \; : \; \bigl\vert \{s + 1, \ldots, e - 1\} \cap \Cp \bigr\vert \le 1 
	\\ 
	& \qquad \text{and} \quad \min(k - s, e - k) > C_1 \psi_{n, p}^2 \Bigr\},
	\end{align*}
	we have, uniformly over $(s, k, e) \in \mc I^\prime$,
	\begin{align*}
    \sqrt{\frac{(k - s)(e - k)}{e - s}} \bigl\vert \wh{\bm\gamma}_{k, e} - \wh{\bm\gamma}_{s, k} - \l( \bm\gamma_{k, e} - \bm\gamma_{s, k} \r) \bigr\vert_\infty \;
	\le \; \sum_{j = 0}^q \Psi_j \mathbb{I}_{\{\cp_j < e \le \cp_{j + 1}\}} \cdot 2 C_0 \psi_{n, p} 
	\end{align*}
	conditional on $\mc E_{n, p}$ in~\eqref{eq:set:e}.
\end{lem}

\begin{proof} 
	Below we condition all our arguments on $\mc E_{n, p}$.
	We first consider the case $\{s + 1, \ldots, e - 1\} \cap \Cp \ne \emptyset$ and $s <  \cp_j < e$.
	We prove the case where $k \le \cp_j$; the case with $k \ge \cp_j + 1$ is similarly handled. Note that
	\begin{align*}
	& \sqrt{\frac{(k - s)(e - k)}{e - s}} \l\vert \wh{\bm\gamma}_{k, e} - \wh{\bm\gamma}_{s, k} - \l( \bm\gamma_{k, e} - \bm\gamma_{s, k} \r) \r\vert_\infty 
 \\
	\le & \,
	\sqrt{\frac{(k - s)(e - k)}{e - s}} \Biggl\vert \frac{1}{e - k} \sum_{t = k + 1}^e \mbf x_t \mbf x_t^\top \bm\beta_j  - \frac{1}{e - k} \sum_{t = k + 1}^{\cp_j} \mbf x_t \mbf x_t^\top \bm\delta_j
	\\
	& \qquad \qquad \qquad 
	- \frac{1}{k - s} \sum_{t = s + 1}^k \mbf x_t \mbf x_t^\top \bm\beta_{j - 1}
	- \frac{e - \cp_j}{e - k} \bm\Sigma \bm\delta_j \Biggr\vert_\infty
	\\
	& + \l\vert \frac{1}{\sqrt{k - s}} \sum_{t = s + 1}^k \mbf x_t \vep_t \r\vert_\infty
	+ \l\vert \frac{1}{\sqrt{e - k}} \sum_{t = k + 1}^e \mbf x_t \vep_t \r\vert_\infty
	=: T_1 + T_2 + T_3.
	\end{align*}
	By Lemma~\ref{lem:bound}, we have $T_2 + T_3 \le 2 C_0\psi_{n, p}$ on $\mc E_{n, p}^{(1)}$ uniformly for all $(s, k, e)$ in consideration.
	As for $T_1$, recalling that $\bm\beta_j = (\bm\mu_j + \bm\delta_j)/2$ and $\bm\beta_{j - 1} = (\bm\mu_j - \bm\delta_j)/2$, 
	\begin{align*}
	T_1 \le & \, \frac{1}{\sqrt{e - k}} \l\vert \sum_{t = k + 1}^e \l( \mbf x_t \mbf x_t^\top - \bm\Sigma \r) \frac{\bm\mu_j}{2} \r\vert_\infty
	+ \frac{1}{\sqrt{k - s}} \l\vert \sum_{t = s + 1}^k \l( \mbf x_t \mbf x_t^\top - \bm\Sigma \r) \frac{\bm\mu_j}{2} \r\vert_\infty
	\\
	& + \frac{1}{\sqrt{e - k}} \l\vert  \sum_{t = k + 1}^e \l( \mbf x_t \mbf x_t^\top - \bm\Sigma \r) \frac{\bm\delta_j}{2} \r\vert_\infty
	+ \frac{1}{\sqrt{k - s}} \l\vert \sum_{t = s + 1}^k \l( \mbf x_t \mbf x_t^\top - \bm\Sigma \r) \frac{\bm\delta_j}{2} \r\vert_\infty
	\\
	& + \sqrt{\frac{\max(\cp_j - k, C_1\psi_{n, p}^2)}{e - k}} \cdot \frac{1}{\sqrt{\max(\cp_j - k, C_1\psi_{n, p}^2)}} \l\vert  \sum_{t = k + 1}^{\cp_j} \l(\mbf x_t \mbf x_t^\top - \bm\Sigma \r) \bm\delta_j \r\vert_\infty
	\\
	\le & \, C_0 (\vert \bm\mu_j \vert_2 + 2\vert \bm\delta_j \vert_2 ) \psi_{n, p},
	\end{align*}
	by Lemma~\ref{lem:bound} and~\eqref{eq:set:e}.
	In the case where $\{s + 1, \ldots, e - 1\} \cap \Cp = \emptyset$, we have $T_1 \le C_0 \vert \bm\mu_j \vert_2 \psi_{n, p}$.
	Combining the bounds on $T_1$, $T_2$ and $T_3$, the conclusion follows.
\end{proof}

\begin{lem}
	\label{lem:cusum}
	Suppose that $0 \le s < e \le n$ satisfy $\{\cp_j\} \subset \{s + 1, \ldots, e - 1\} \cap \Cp \subset \{ \cp_j, \cp_{j + 1} \}$ for some $j \in [q]$. Then, writing $T^*_{s, k, e} = \sqrt{(k - s)(e - k)/(e - s)} \l\vert \bm\gamma_{k, e} - \bm\gamma_{s, k} \r\vert_\infty$, we have $\mathop{\arg\max}_{s < k < e} T^*_{s, k, e}  \in \{ \cp_j, \cp_{j + 1} \}$ and
	\begin{align*}
	\max_{s < k < e} T^*_{s, k, e} = \max\l\{ \sqrt{\frac{(\cp_j - s)(e - \cp_j)}{e - s}} \l\vert \bm\Sigma \l( \bm\delta_j + \frac{e - \cp_{j + 1}}{e - \cp_j} \bm\delta_{j + 1} \mathbb{I}_{\{\cp_{j + 1} < e\}} \r) \r\vert_\infty, \r. \\
	\l. \sqrt{\frac{(\cp_{j + 1} - s)(e - \cp_{j + 1})}{e - s}} \l\vert \bm\Sigma \l( \bm\delta_{j + 1} + \frac{\cp_j - s}{\cp_{j + 1} - s} \bm\delta_j \r) \r\vert_\infty 
	\mathbb{I}_{\{\cp_{j + 1} < e\}} \r\}.
	\end{align*}
\end{lem}

\begin{proof} 
	The first statement follows from Lemma~8 of \cite{wang2018high}.
	Note that
	\begin{align*}
	T^*_{s, \cp_j, e} &= \sqrt{\frac{(\cp_j - s)(e - \cp_j)}{e - s}} \l\vert \frac{1}{e - \cp_j} \bm\Sigma \l( (\cp_{j + 1} \wedge e - \cp_j) \bm\beta_j + (e - \cp_{j + 1}) \bm\beta_{j + 1} \mathbb{I}_{\{ \cp_{j + 1} < e\}} \r) - \bm\Sigma \bm\beta_{j - 1} \r\vert_\infty
	\\
	&= \sqrt{\frac{(\cp_j - s)(e - \cp_j)}{e - s}} \l\vert \bm\Sigma \l( \bm\delta_j + \frac{e - \cp_{j + 1}}{e - \cp_j} \bm\delta_{j + 1} \mathbb{I}_{\{\cp_{j + 1} < e\}} \r) \r\vert_\infty.
	\end{align*}
	We can analogously derive $T^*_{s, \cp_{j + 1}, e}$ and from that $\max_{s < k < e} T^*_{s, k, e} = \max\{ T^*_{s, \cp_j, e}, T^*_{s, \cp_{j + 1}, e} \}$, the second statement follows.
	In particular, when $\{s + 1, \ldots, e - 1\} \cap \Cp = \{\cp_j\}$, we have
	\begin{align}
	T^*_{s, k, e} = \frac{ \vert \bm\Sigma \bm\delta_j \vert_\infty }{\sqrt{e - s}} \l( (e - \cp_j) \sqrt{\frac{k - s}{e - k}} \mathbb{I}_{\{k \le \cp_j\}} + (\cp_j - s) \sqrt{\frac{e - k}{k - s}} \mathbb{I}_{\{k > \cp_j\}} \r).
	\label{eq:cusum:diff}
	\end{align}
\end{proof}

\begin{lem}
	\label{lem:sqrt:diff}
	Suppose that $0 \le s < e \le n$ and $e - s > 2$. For any $(a, b)$ satisfying $s < a < b < e$ and $b - a \le \min(b - s, e - b)/4$, we have 
	\begin{align*}
	\l\vert \sqrt{\frac{b - s}{e - b}} - \sqrt{\frac{a - s}{e - a}} \r\vert &\le \sqrt{\frac{b - s}{e - b}} \cdot \frac{(b - a)(e - s)}{(b - s)(e - b)}, \text{ \ and}
	\\
	\l\vert \sqrt{\frac{e - a}{a - s}} - \sqrt{\frac{e - b}{b - s}} \r\vert &\le \sqrt{\frac{e - b}{b - s}} \cdot \frac{(b - a)(e - s)}{(b - s)(e - b)}.
	\end{align*}
\end{lem}
\begin{proof}
	By Taylor expansion, we have
	\begin{align*}
	\l\vert \sqrt{\frac{b - s}{e - b}} - \sqrt{\frac{a - s}{e - a}} \r\vert &= \sqrt{\frac{b - s}{e - b}} \l( 1 - \sqrt{\frac{1 - \frac{b - a}{b - s}}{1 + \frac{b - a}{e - b}}} \r)
	\\
	&\le \sqrt{\frac{b - s}{e - b}} \l( 1 - \l(1 - \frac{1}{2} \cdot \frac{b - a}{b - s} - \frac{1}{8} \l( \frac{b - a}{b - s} \r)^2 \r) \l( 1 - \frac{1}{2} \cdot \frac{b - a}{e - b} \r) \r)
	\\
	&\le \sqrt{\frac{b - s}{e - b}} \l( \frac{17}{32} \cdot \frac{b - a}{b - s} + \frac{1}{2} \cdot \frac{b - a}{e - b}  \r) \le \sqrt{\frac{b - s}{e - b}} \frac{(b - a)(e - s)}{(b - s) (e - b)}.
	\end{align*}
	Analogously,
	\begin{align*}
	\l\vert \sqrt{\frac{e - a}{a - s}} - \sqrt{\frac{e - b}{b - s}} \r\vert &= \sqrt{\frac{e - b}{b - s}} \l( \sqrt{\frac{1 + \frac{b - a}{e - b}}{1 - \frac{b - a}{b - s}}} - 1 \r)
	\\
	&\le \sqrt{\frac{e - b}{b - s}} \l(\l( 1 + \frac{1}{2} \cdot \frac{b - a}{e - b} \r) \l(1 + \frac{1}{2} \cdot \frac{b - a}{b - s} \r) - 1\r)
	\\
	&\le \sqrt{\frac{e - b}{b - s}} \l( \frac{9}{16} \cdot \frac{b - a}{b - s} + \frac{1}{2} \cdot \frac{b - a}{e - b}  \r) \le \sqrt{\frac{e - b}{b - s}} \frac{(b - a)(e - s)}{(b - s) (e - b)}.
	\end{align*}
\end{proof}

\subsection{Proofs for the results in Section~\ref{sec:diff}}

\subsubsection{Proof of Proposition~\ref{prop:dp:est}}
\label{pf:prop:lasso}

We first prove the following more general result.

\begin{prop}
	\label{prop:lasso}
	Suppose that Assumptions~\ref{assum:xe}, \ref{assum:func:dep}  and~\ref{assum:xe:iii} hold.
	For some constant $\crsc$ depending on $\kappa$ (with $\kappa = 0$ under Assumption~\ref{assum:func:dep}~\ref{assum:fd:gauss} and $\kappa = \max(1/r - 1/2, 0)$ under~\ref{assum:fd:ind}), define
	\begin{multline*}
	\mc I = \, \Bigl\{0 \le s < k < e \le n\; : \; \{s + 1, \ldots, e - 1\} \cap \Cp = \{ \cp_j \} \text{ \ for some \ } j \in [q],
	\\ 
	\quad  \min(k - s, e - k) > C_1\psi_{n, p}^2 \text{ and } e - s \ge \underline{\sigma}^{-2}\bigl(64 \mathfrak{s}_j \crsc \log(p) \bigr)^{1 + 2\kappa}
	\Bigr\}.
	\end{multline*}
	Then conditional on $\mc E_{n, p} \cap \mc R_{n, p}$, where $\mc E_{n, p}$ is defined in~\eqref{eq:set:e} and $\mc R_{n, p}$ in Lemma~\ref{lem:rsc} below, we have $\wh{\bm\delta}_{s, e}(k)$ obtained with $\lambda = C_\lambda \Psi_j \psi_{n, p}$ (where $j$ is the index of $\cp_j \in (s, e)$) and $C_\lambda \ge 5 C_0$, satisfy the following uniformly for all $(s, k, e) \in \mc I$: $\vert \wh{\bm\delta}_{s, e}(k) \vert_1 \lesssim \vert \bm\delta_{s, e}(k) \vert_1$ and
	\begin{align*}
	& 
 \sqrt{\min(k - s, e - k)} \l\vert \wh{\bm\delta}_{s, e}(k) - \bm\delta_{s, e}(k) \r\vert_2 \lesssim \frac{\Psi_j \sqrt{\mathfrak{s_j}} \psi_{n, p}}{\underline{\sigma}}, 
	\\
	&  
 \sqrt{\min(k - s, e - k)} \l\vert \wh{\bm\delta}_{s, e}(k) - \bm\delta_{s, e}(k) \r\vert_1 \lesssim \frac{\Psi_j \mathfrak{s}_j \psi_{n, p}}{\underline{\sigma}},
	\end{align*}
	where the unspecified constants depend only on $\Xi$ and $\kappa$.
\end{prop}

\begin{proof} 
Let us suppose that $\{s + 1, \ldots, e - 1\} \cap \Cp = \{ \cp_j \}$. 

	\paragraph{LOPE, $\ell_1$-penalised estimator in~\eqref{eq:lasso:est}.}
	By Lemma~\ref{lem:gamma} and the identity that $\bm\Sigma \bm\delta_{s, e}(k) = \bm\gamma_{k, e} + \bm\gamma_{s, k}$, we have conditional on $\mc E_{n, p}$,
	\begin{align}
	& \sqrt{\frac{(k - s)(e - k)}{e - s}} \l\vert \wh{\bm\Sigma}_{s, e} \bm\delta_{s, e}(k) - \wh{\bm\gamma}_{k, e} + \wh{\bm\gamma}_{s, k} \r\vert_\infty \nn\\
	\le & \sqrt{\frac{(k - s)(e - k)}{e - s}} \l[ \l\vert \l( \wh{\bm\Sigma}_{s, e} - \bm\Sigma \r) \bm\delta_{s, e}(k) \r\vert_\infty +
	\l\vert \wh{\bm\gamma}_{k, e} - \wh{\bm\gamma}_{s, k} - (\bm\gamma_{k, e} - \bm\gamma_{s, k}) \r\vert_\infty
	\r]
	\nn \\
	\le & C_0 \l( \frac{\sqrt{(k - s)(e - k)}}{e - s} \max_{j' \in \{j - 1, j\}} \vert \bm\delta_{j'} \vert_2 + 2 \Psi_j \r) \psi_{n, p} \le \frac{5}{2}C_0 \Psi_j \psi_{n, p}
	\label{eq:feasible}
	\end{align}
	for all such $(s, k, e) \in \mc I$.
	Also by the Karush--Kuhn--Tucker condition, we have 
	\begin{align*}
	\mbf 0 \in \wh{\bm\Sigma}_{s, e} \wh{\bm\delta}_{s, e}(k) - \wh{\bm\gamma}_{k, e} + \wh{\bm\gamma}_{s, k} + \lambda \sqrt{\frac{e - s}{(k - s)(e - k)}} \mathrm{sgn} \l(\wh{\bm\delta}_{s, e}(k) \r)
	\end{align*}
	where, defined as 
	\begin{align*}
	\mathrm{sgn}(x) = \begin{cases}
	1 & \text{ if } x > 0,\\
	-1 & \text{ if } x < 0,\\
	[-1, 1] & \text{ if } x = 0,
	\end{cases}
	\end{align*}
	the operator $\mathrm{sgn}(\cdot)$ also applies to vectors element-wise. It implies that 
	\begin{align}
	\label{eq:lasso:kkt}
	\l\vert \wh{\bm\Sigma}_{s, e} \wh{\bm\delta}_{s, e}(k) - \wh{\bm\gamma}_{k, e} + \wh{\bm\gamma}_{s, k} \r\vert_\infty \le \lambda \sqrt{\frac{e - s}{(k - s)(e - k)}}.
	\end{align}
	By the definition of $\wh{\bm\delta}_{s, e}(k)$, we have
	\begin{align}
	& \frac{1}{2} \bigl(\wh{\bm\delta}_{s, e}(k)\bigr)^\top \wh{\bm\Sigma}_{s, e} \wh{\bm\delta}_{s, e}(k) - \bigl(\wh{\bm\delta}_{s, e}(k)\bigr)^\top (\wh{\bm\gamma}_{k, e} - \wh{\bm\gamma}_{s, k}) + \lambda \sqrt{\frac{e - s}{(k - s)(e - k)}} \vert \wh{\bm\delta}_{s, e}(k) \vert_1
	\nn \\
	\le & \, \frac{1}{2} \bigl({\bm\delta}_{s, e}(k)\bigr)^\top \wh{\bm\Sigma}_{s, e} {\bm\delta}_{s, e}(k) - \bigl({\bm\delta}_{s, e}(k)\bigr)^\top (\wh{\bm\gamma}_{k, e} - \wh{\bm\gamma}_{s, k}) + \lambda \sqrt{\frac{e - s}{(k - s)(e - k)}} \vert {\bm\delta}_{s, e}(k) \vert_1, 
	\nn 
	\end{align}
	such that with $\wh{\mbf v} = \wh{\bm\delta}_{s, e}(k) - \bm\delta_{s, e}(k)$,
	\begin{align}
	\frac{1}{2} \wh{\mbf v}^\top \wh{\bm\Sigma}_{s, e} \wh{\mbf v}  \le & \, \lambda \sqrt{\frac{e - s}{(k - s)(e - k)}} \l( \vert \bm\delta_{s, e}(k) \vert_1 - \vert \bm\delta_{s, e}(k) + \wh{\mbf v} \vert_1 \r)
	\nn \\
	&  \qquad - \wh{\mbf v}^\top \l( \wh{\bm\Sigma}_{s, e} {\bm\delta}_{s, e}(k) - \wh{\bm\gamma}_{k, e} + \wh{\bm\gamma}_{s, k} \r)\nn\\
	\le & \, \lambda \sqrt{\frac{e - s}{(k - s)(e - k)}} \l( \vert \bm\delta_{s, e}(k) \vert_1 - \vert \bm\delta_{s, e}(k) + \wh{\mbf v} \vert_1 \r)
	\nn \\
	& \qquad +\vert \wh{\mbf v} \vert_1 \l\vert \wh{\bm\Sigma}_{s, e} {\bm\delta}_{s, e}(k) - \wh{\bm\gamma}_{k, e} + \wh{\bm\gamma}_{s, k} \r\vert_{\infty}.
	\label{eq:lasso:one}
	\end{align}
	By~\eqref{eq:feasible}, we have
	\begin{align}
	\label{eq:feasible:two}
	\l\vert \wh{\bm\Sigma}_{s, e} {\bm\delta}_{s, e}(k) - \wh{\bm\gamma}_{k, e} + \wh{\bm\gamma}_{s, k} \r\vert_\infty \le \frac{5}{2} C_0 \Psi_j \psi_{n, p} \sqrt{\frac{e - s}{(k - s)(e - k)}} \le \frac{1}{2} \cdot \lambda \sqrt{\frac{e - s}{(k - s)(e - k)}},
	\end{align}
	where the last inequality holds by setting $\lambda = C_\lambda \Psi \psi_{n, p}$ with $C_\lambda \ge 5 C_0$. Also, writing $\mbf a^{\mc S_j} = (a_i, \, i \in \mc S_j)^\top$ for any vector $\mbf a \in \R^p$, we have
	\begin{align*}
	\vert \bm\delta_{s, e}(k) \vert_1 - \vert \bm\delta_{s, e}(k) + \wh{\mbf v} \vert_1
	\le \vert \bm\delta_{s, e}^{\mc S_j}(k) \vert_1 - \l( \vert \bm\delta_{s, e}^{\mc S_j}(k) \vert_1 - \vert \wh{\mbf v}^{\mc S_j} \vert_1 + \vert \wh{\mbf v}^{\mc S_j^c} \vert_1 \r) = \vert \wh{\mbf v}^{\mc S_j} \vert_1 - \vert \wh{\mbf v}^{\mc S_j^c} \vert_1.
	\end{align*}
	This together with~\eqref{eq:lasso:one}  and the positive definiteness of $\wh{\bm\Sigma}_{s, e}$ implies that
	\begin{align}
	\label{eq:lasso:two}
	0 &\le \frac{1}{2} \wh{\mbf v}^\top \wh{\bm\Sigma}_{s, e} \wh{\mbf v} \le \lambda \sqrt{\frac{e - s}{(k - s)(e - k)}} \l( \frac{3}{2} \vert \wh{\mbf v}^{\mc S_j} \vert_1 - \frac{1}{2} \vert \wh{\mbf v}^{\mc S_j^c} \vert_1 \r),
	\end{align}
	which further implies that
	\begin{align}
	\label{eq:lasso:three}
	\vert \wh{\mbf v}^{\mc S_j^c} \vert_1 \le 3 \vert \wh{\mbf v}^{\mc S_j} \vert_1 
	\text{ \ and \ }
	\vert \wh{\mbf v} \vert_1 \le 4\vert \wh{\mbf v}^{\mc S_j} \vert_1 \le 4 \sqrt{\mathfrak{s}_j} \vert \wh{\mbf v} \vert_2. 
	\end{align}
	Conditional on $\mc R_{n, p}$, we have
	\begin{align*}
	\frac{1}{2} \wh{\mbf v}^\top \wh{\bm\Sigma}_{s, e} \wh{\mbf v} \ge \frac{\underline{\sigma}}{4} \vert \wh{\mbf v} \vert_2^2 - \frac{{\underline{\sigma}}^{\frac{2\kappa - 1}{2\kappa + 1}} \crsc \log(p)}{2 (e - s)^{\frac{1}{1 + 2\kappa}}} \vert \wh{\mbf v} \vert_1^2 \ge \frac{\underline{\sigma}}{4} \vert \wh{\mbf v} \vert_2^2 \l( 1 - \frac{32 \crsc \mathfrak{s}_j \log(p) }{\bigl(\underline{\sigma}^2(e - s)\bigr)^{\frac{1}{1 + 2\kappa}}} \r) \ge \frac{\underline{\sigma}}{8} \vert \wh{\mbf v} \vert_2^2,
	\end{align*}
	from the requirement on $e - s$.
	Combining this with~\eqref{eq:lasso:two} and~\eqref{eq:lasso:three}, we have
	\begin{align*}
	& \frac{\underline{\sigma}}{8} \vert \wh{\mbf v} \vert_2^2 
	\le \frac{3 \lambda \sqrt{\mathfrak{s}_j}}{\sqrt{2 \min(k - s, e - k)}} \vert \wh{\mbf v} \vert_2,
	\text{ \ from which it follows that}
	\\
	& \vert \wh{\mbf v} \vert_2 \le \frac{12 \sqrt{2} \lambda \sqrt{\mathfrak{s}_j}}{\underline{\sigma} \sqrt{\min(k - s, e - k)}},
	\quad
	\vert \wh{\mbf v} \vert_1 \le \frac{48 \sqrt{2} \lambda \mathfrak{s}_j}{\underline{\sigma} \sqrt{\min(k - s, e - k)}} \text{ \ and }
	\\
	&  \wh{\mbf v}^\top \wh{\bm\Sigma}_{s, e} \wh{\mbf v} \le \frac{3 \sqrt{2} \lambda \sqrt{\mathfrak{s}_j} \vert \wh{\mbf v} \vert_2}{\sqrt{\min(k - s, e - k)}} \le \frac{72 \lambda^2 \mathfrak{s}_j}{\underline{\sigma} {\min(k - s, e - k)}}.
	\end{align*}
	Besides, from~\eqref{eq:lasso:one} and~\eqref{eq:feasible:two}, we have
	\begin{align*}
	0 \le \vert \bm\delta_{s, e}(k) \vert_1 - \vert \wh{\bm\delta}_{s, e}(k) \vert_1 + \frac{1}{2} \l( \vert \bm\delta_{s, e}(k) \vert_1 + \vert \wh{\bm\delta}_{s, e}(k) \vert_1\r), \text{ \ i.e. \ }
	\vert \wh{\bm\delta}_{s, e}(k) \vert_1 \le 3 \vert \bm\delta_{s, e}(k) \vert_1.
	\end{align*}
	Then, by \cref{lem:bound}, we have conditional on $\mc E_{n, p}$,
	\begin{align}
	\l\vert \l( \wh{\bm\Sigma}_{s, e} - \bm\Sigma \r)  \wh{\bm\delta}_{s, e}(k) \r\vert_{\infty} &\le \sum_{i = 1}^p \l\vert \l( \wh{\bm\Sigma}_{s, e} - \bm\Sigma \r)  \mbf e_i \r\vert_{\infty}\l\vert\bigl(\wh{\bm\delta}_{s, e} (k)\bigr)_i\r\vert \nn\\
	&\le C_0 \psi_{n,p}\l\vert\wh{\bm\delta}_{s, e}(k)\r\vert_1 \le 3C_0\psi_{n,p} \l\vert{\bm\delta}_{s, e}(k)\r\vert_1.
	\label{eq:lasso:lone}
	\end{align}
	Finally, note that 
	\begin{align*}
	& \sqrt{\frac{(k - s)(e - k)}{e - s}} \vert \wh{\mbf v} \vert_\infty \le 
	\sqrt{\frac{(k - s)(e - k)}{e - s}} \l\Vert \bm\Omega \r\Vert_1 \l\vert \bm\Sigma \wh{\bm\delta}_{s, e}(k) - {\bm\gamma}_{k, e} + {\bm\gamma}_{s, k} \r\vert_\infty
	\end{align*}
	where, by~\eqref{eq:lasso:kkt}, \eqref{eq:lasso:lone} and \cref{lem:gamma},
	\begin{align*}
	& \sqrt{\frac{(k - s)(e - k)}{e - s}} \l\vert \bm\Sigma \wh{\bm\delta}_{s, e}(k) - {\bm\gamma}_{k, e} + {\bm\gamma}_{s, k} \r\vert_\infty \le \sqrt{\frac{(k - s)(e - k)}{e - s}} \l[ \l\vert \l( \wh{\bm\Sigma}_{s, e} - \bm\Sigma \r)  \wh{\bm\delta}_{s, e}(k) \r\vert_{\infty} + \r.
	\\
	& \l. \l\vert \wh{\bm\gamma}_{k, e} - \wh{\bm\gamma}_{s, k} - {\bm\gamma}_{k, e} + {\bm\gamma}_{s, k} \r\vert_\infty +  \l\vert \wh{\bm\Sigma}_{s, e} \wh{\bm\delta}_{s, e}(k) - \wh{\bm\gamma}_{k, e} + \wh{\bm\gamma}_{s, k} \r\vert_{\infty}
	\r] \le C_0 \l( 3\vert \bm\delta_{s, e}(k) \vert_1 + 2 \Psi_j \r) \psi_{n, p} + \lambda
	\\
	& \le 3 C_0 \vert \bm\delta_j \vert_1 \psi_{n, p} + 2 \lambda,
	\end{align*}
	Combined with the trivial bound $\vert \wh{\mbf v} \vert_\infty \le \vert \wh{\mbf v} \vert_2$, the bound on $\vert \wh{\mbf v} \vert_\infty$ in \cref{rem:linf} follows.
	
	\paragraph{CLOM, constrained $\ell_1$-minimisation estimator in~\eqref{eq:direct:est}.}
	Note that from~\eqref{eq:feasible}, we have $\bm\delta_{s, e}(k)$ feasible for the constraint in~\eqref{eq:direct:est}, from which it follows that $\vert \wh{\bm\delta}_{s, e}(k) \vert_1 \le \vert \bm\delta_{s, e}(k) \vert_1$.
	Then, by~\eqref{eq:feasible} and the definition of $\wh{\bm\delta}_{s, e}(k)$, we have
	\begin{align*}
	\wh{\mbf v}^\top \wh{\bm\Sigma}_{s, e} \wh{\mbf v} &\le \vert\wh{\mbf v} \vert_1 \vert \wh{\bm\Sigma}_{s, e} \wh{\mbf v} \vert_\infty  \le \vert\wh{\mbf v} \vert_1 \l(\bigl\vert\wh{\bm\Sigma}_{s, e}\wh{\bm \delta}(k) - \wh{\bm\gamma}_{k, e} + \wh{\bm\gamma}_{s, k}\bigr\vert_{\infty} + \bigl\vert\wh{\bm\Sigma}_{s, e}\delta_{s, e}(k) - \wh{\bm\gamma}_{k, e} + \wh{\bm\gamma}_{s, k}\bigr\vert_{\infty}\r)\\
	& \le \sqrt{\frac{e - s}{(k - s) (e - k)}} \l(\lambda + \frac{5}{2}C_0 \Psi_j \psi_{n,p}\r) \vert\wh{\mbf v} \vert_1 \le 
	2\lambda \sqrt{\frac{e - s}{(k - s) (e - k)}} \vert\wh{\mbf v} \vert_1.
	\end{align*}
	By splitting the coordinates into $\mc S_j$ and its complement, we obtain
	\begin{align*}
	\vert \wh{\mbf v}^{\mc S_j^c} \vert_1 &= \vert \wh{\bm \delta}^{\mc S_j^c}_{s, e}(k) \vert_1 = \vert \wh{\bm \delta}_{s, e}(k) \vert_1 - \vert \wh{\bm \delta}^{\mc S_j}_{s, e}(k) \vert_1 \le \vert \bm \delta_{s, e}(k) \vert_1 - \vert \wh{\bm \delta}^{\mc S_j}_{s, e}(k) \vert_1 
	\\
	&\le \vert \wh{\mbf v}^{\mc S_j} \vert_1 + \vert \wh{\bm \delta}^{\mc S_j}_{s, e}(k) \vert_1 - \vert \wh{\bm \delta}^{\mc S_j}_{s, e}(k) \vert_1 = \vert \wh{\mbf v}^{\mc S_j} \vert_1.
	\end{align*}
	Thus, similarly to the proof of the $\ell_1$-penalised estimator in~\eqref{eq:lasso:est}, the claims follow up to constants, i.e.\ 
	\begin{align*}
	\vert \wh{\mbf v} \vert_2 \le \frac{32 \lambda \sqrt{\mathfrak{s}_j}}{\underline{\sigma} \sqrt{\min(k - s, e - k)}},
	\text{ \ and \ }
	\vert \wh{\mbf v} \vert_1 \le \frac{64 \lambda \mathfrak{s}_j}{\underline{\sigma} \sqrt{\min(k - s, e - k)}}.
	\end{align*}
\end{proof}

\begin{proof}[Proof of Proposition~\ref{prop:dp:est}]
	In what follows, we show that conditional on $\mc E_{n, p} \cap \mc R_{n, p}$, the claim follows where $\mc E_{n, p}$ is defined in~\eqref{eq:set:e} and $\mc R_{n, p}$ in Lemma~\ref{lem:rsc} below. 
	In doing so, we will also use that from the proof of Theorem~\ref{thm:one}, we have $\mc E_{n, p} \subset \mc S_{n, p}$ defined therein.
	Then, since $\p(\mc E_{n, p} \cap \mc R_{n, p}) \ge 1 - 2c_2 (p \vee n)^{-c_3}$ by Lemmas~\ref{lem:bound} and~\ref{lem:rsc}, the proof is completed.
	
	By Theorem~\ref{thm:one}, conditional on $\mc S_{n, p}$, the pairs $(a_j, b_j)$ defined in~\eqref{eq:ab} satisfy
	\begin{enumerate}[label = (\alph*)]
		\item \label{se:one} $\min(\cp_j - a_j, b_j - \cp_j) \ge \Delta_j/6$ and $\min(\wh\cp_j - a_j, b_j - \wh\cp_j) \ge \Delta_j/3$, and
		\item \label{se:two} $\{a_j + 1, \ldots, b_j - 1\} \cap \Cp = \{\cp_j\}$.
	\end{enumerate}
	To see~\ref{se:one}, note that, conditional on $\mc S_{n, p}$, 
	\begin{align*}
	\frac{\vert \wh\cp_j - \cp_j \vert}{\Delta_j} \le \frac{c_1\Psi_j^2 \psi_{n, p}^2}{ \Delta_j \vert \bm\Sigma \bm\delta_j \vert_\infty^{2}} \le \frac{c_1}{c_0}\le \frac{1}{6}
	\end{align*}
	for all $j \in [q]$ by \cref{assum:size}, such that
	\begin{align*}
	\wh\cp_j - \l(\frac{2}{3} \wh\cp_{j - 1} + \frac{1}{3} \wh\cp_j \r) &\ge \cp_j - \l(\frac{2}{3} \cp_{j - 1} + \frac{1}{3} \cp_j \r) - \l( \frac{2}{3} \l\vert \wh\cp_{j - 1} - \cp_{j - 1} \r\vert + \frac{1}{3} \l\vert \wh\cp_j - \cp_j \r\vert + \vert \wh\cp_j - \cp_j \vert \r)
	\\
	&\ge \l(\frac{2}{3} - \frac{2c_1}{c_0}\r) \Delta_j
	\end{align*}
	and similarly, $(2\wh\cp_{j + 1}/3 + \wh\cp_j/3) - \wh\cp_j \ge (2/3 - 2c_1/c_0)\Delta_j$.
	Thus, we have $\wh{\Delta}_j \ge (2/3 - 2c_1/c_0) \Delta_j \ge \Delta_j/3$.
	Then,
	\begin{align*}
	\cp_j - a_j \ge \wh\cp_j - a_j - \vert \wh\cp_j - \cp_j \vert \ge \l( \frac{1}{3} - \frac{c_1}{c_0} \r) \Delta_j \ge \frac{\Delta_j}{6} \text{ \ and analogously, \ } b_j - \cp_j \ge \frac{\Delta_j}{6}.
	\end{align*}
	Further, since $a_j \ge \lfloor 2\wh\cp_{j - 1}/3 + \wh\cp_j / 3 \rfloor$,
	\begin{align*}
	\cp_{j - 1} - a_j &\le \cp_{j - 1} - \l\lfloor\frac{2}{3} \wh\cp_{j - 1} + \frac{1}{3} \wh\cp_j \r\rfloor \le \l\lceil - \frac{1}{3} (\cp_j - \cp_{j - 1}) + \frac{2}{3} \vert \wh\cp_{j - 1} - \cp_{j - 1} \vert + \frac{1}{3} \vert \wh\cp_j - \cp_j \vert \r\rceil
	\\
	&\le \l\lceil\l( -\frac{1}{3} + \frac{c_1}{c_0} \r) \Delta_j \r\rceil \le \l\lceil-\frac{\Delta_j}{6}\r\rceil < 0
	\end{align*}
	and similarly, $b_j - \cp_{j + 1} < 0$; Hence~\ref{se:two} follows.
	Combining~\ref{se:one} and~\eqref{eq:prop:lasso:two}, we can apply Proposition~\ref{prop:lasso} and obtain
	\begin{align*}
	\l\vert \wh{\bm\delta}_j - \bm\delta_{s_j, e_j}(\wh\cp_j) \r\vert_1 \lesssim \frac{\Psi_j \mathfrak{s}_j \psi_{n, p}}{\underline{\sigma} \sqrt{\Delta_j}}.
	\end{align*}
	Note that
	$\underline\sigma \mathfrak{s}_j^{-1}\vert \bm \delta_j\vert_1^2\le\underline\sigma \vert \bm \delta_j\vert_2^2\le(\bm\delta_j)^\top \bm\Sigma \bm\delta_j \le \vert \bm\delta_j \vert_1\vert\bm\Sigma \bm\delta_j \vert_\infty,$
	which implies that $\vert\bm\delta_j\vert_1 \le {\underline\sigma}^{-1}\mathfrak{s}_j\vert\bm\Sigma\bm\delta_j\vert_{\infty}$. 
	Then, as $\bm\delta_{a_j, b_j}(\cp_j) = \bm\delta_j$ by~\ref{se:one}, we have conditional on $\mc S_{n,p}$, 
	\begin{align*}
	\l\vert \bm\delta_{a_j, b_j}(\wh\cp_j) - \bm\delta_j \r\vert_1 
	= & \frac{\vert \wh\cp_j - \cp_j \vert \vert \bm\delta_j \vert_1}{(\cp_j - a_j) \mathbb{I}_{\{\wh\cp_j \le \cp_j\}} + (b_j - \cp_j) \mathbb{I}_{\{\wh\cp_j > \cp_j\}} } \\
	\le & \frac{6 c_1 \vert \bm\delta_j \vert_1  \Psi_j^2 \psi_{n, p}^2 }{ \Delta_j \vert \bm\Sigma \bm\delta_j \vert_\infty^{2}}
	\le \frac{6c_1\mathfrak{s}_j \Psi_j^2\psi_{n, p}^2}{\underline{\sigma} \Delta_j\vert\bm\Sigma\bm\delta_j\vert_\infty} \le \frac{6c_1\Psi_j \mathfrak{s}_j \psi_{n, p}}{\underline{\sigma} \sqrt{c_0\Delta_j}},
	\end{align*}
	where the last inequality is due to \cref{assum:size}.  
	Combining the above two displayed inequalities, we get
	\begin{align*}
	\l\vert \wh{\bm\delta}_j - \bm\delta_j \r\vert_1 \le \l\vert \wh{\bm\delta}_j - \bm\delta_{a_j, b_j}(\wh\cp_j) \r\vert_1 + \l\vert \bm\delta_{a_j, b_j}(\wh\cp_j) - \bm\delta_j \r\vert_1 
	\lesssim \frac{\Psi_j \mathfrak{s}_j \psi_{n, p}}{\underline{\sigma} \sqrt{\Delta_j}}.
	\end{align*}
	Analogously, we obtain, by Proposition~\ref{prop:lasso},
	\begin{align*}
	\l\vert \wh{\bm\delta}_j - \bm\delta_j \r\vert_2 &\le 
	\l\vert \wh{\bm\delta}_j - \bm\delta_{a_j, b_j}(\wh\cp_j) \r\vert_2 + \l\vert \bm\delta_{a_j, b_j}(\wh\cp_j) - \bm\delta_j \r\vert_2 \\
	& \lesssim 
	\frac{ \Psi_j \sqrt{\mathfrak{s}_j} \psi_{n, p} }{\underline{\sigma} \sqrt{\Delta_j}}
	+ \frac{ \vert \bm\delta_j \vert_2 \Psi_j^2 \psi^2_{n, p}}{\Delta_j \vert \bm\Sigma \bm\delta \vert_\infty^{2}} 
	\lesssim \frac{ \Psi_j \sqrt{\mathfrak{s}_j} \psi_{n, p}}{\underline{\sigma} \sqrt{\Delta_j}}
	\end{align*}
	conditional on $\mc S_{n, p}$ where in the last inequality, we use that $\vert\bm\delta_j\vert_2 \le \underline{\sigma}^{-1}\sqrt{\mathfrak{s}_j}\vert\bm\Sigma\bm\delta_j\vert_{\infty}$ and \cref{assum:size}.
\end{proof}

\subsubsection{Proof of Theorem~\ref{thm:two}}
\label{sec:pf:thm:two}

We prove the results under a general version of \cref{assum:omega:w}, which allows one to consider a wide range of estimators for precision matrices. 
\begin{massum}{\ref{assum:omega:w}$^\prime$}[Precision matrix estimator]
	\label{assum:omega}
	Suppose that there exists an estimator $\wh{\bm\Omega}^{\tE}$ obtained from $\mc D^{\tE}$ such that $\p(\mc O_{n, p}) \to 1$ as $n \to \infty$ where on $\mc O_{n, p}$, the following event holds: Let $\mathfrak{s}_{\varrho} = \max_{i \in [p]} \vert \bm\Omega_{i\cdot} \vert_{\varrho}^{\varrho} = \max_{i \in [p]} \sum_{i' \in [p]} \vert \Omega_{ii'} \vert^{\varrho}$ and $\wh{\bm\Sigma}^{\tE}$ the sample covariance of $\mbf x^{\tE}_t$'s.
	Then, there exist some $\omega^{(\ell)}_{n, p} > 0$ for $\ell = 1, 2, 3$, and $\varrho \in [0, 1)$ such that
\begin{enumerate}[label = (\roman*)]
    \item \label{assum:omega:one} $\Vert \wh{\bm\Omega}^{\tE} \Vert_\infty \le \omega^{(1)}_{n, p}$, 
    \item \label{assum:omega:two} $\sqrt{n} \vert \mbf I_p - \wh{\bm\Omega}^{\tE} \wh{\bm\Sigma}^{\tE} \vert_\infty \le \omega^{(2)}_{n, p} \lesssim \omega^{(1)}_{n, p} \sqrt{\log(p 
    \vee n)}$, and
    \item \label{assum:omega:three} $\Vert \wh{\bm\Omega}^{\tE} - \bm\Omega \Vert \le \mathfrak{s}_\varrho (\omega^{(3)}_{n, p}/ \sqrt{n})^{1 - \varrho} \le 1/(2\bar{\sigma})$.
\end{enumerate}
\end{massum}
Note that \cref{assum:omega:w} is a particular case of \cref{assum:omega} with $\omega^{(1)}_{n,p} \asymp \Vert\bm\Omega\Vert_1$, $\omega^{(2)}_{n,p} \asymp \Vert\bm\Omega\Vert_1 \sqrt{\log(p \vee n)}$ and $\omega^{(3)}_{n,p} \asymp \Vert\bm\Omega\Vert_1^\omega \sqrt{\log(p \vee n)}$.
We define $\mc E^\ell_{n, p}$, $\mc R^\ell_{n, p}$ and $\mc S^\ell_{n, p}$, $\ell \in \{\tE, \tO\}$, analogously as in~\eqref{eq:set:e}, Lemma~\ref{lem:rsc} and Theorem~\ref{thm:one} with the corresponding dataset $\mc D^{\ell}$, respectively.
We regard $\psi_{n, p} = \sqrt{\log(p \vee n)}$ when applying the preceding results.
We also frequently use that $\mc E^\ell_{n, p} \subset \mc S^\ell_{n, p}$ which follows from the proof of Theorem~\ref{thm:one}.
We condition the subsequent arguments on $\mc Q^{\tE}_{n, p} \cap \mc E^{\tO}_{n, p}$, where $\mc Q^{\tE}_{n, p} = \mc E^{\tE}_{n, p} \cap \mc R^{\tE}_{n, p} \cap \mc O^{\tE}_{n, p}$, noting that the condition~\eqref{eq:prop:lasso:two} holds under Assumptions~\ref{assum:func:dep:two} and~\ref{assum:size:two} with $\kappa = 0$.
Throughout, all constants unspecified in $\lesssim$ depend only on $\Xi$.

Recall that we denote by $\cp_j$ the locations of the change points in the joint distribution of $(Y^{\ell}_t, \mbf x^{\ell}_t), \, t \in [n_0], \, \ell \in \{\tE, \tO\}$.
By the arguments analogous to those adopted in the proof of Proposition~\ref{prop:lasso} (see~\ref{se:one}--\ref{se:two} in Appendix~\ref{pf:prop:lasso}), we have 
\begin{multline}
\min(\wh\cp^{\tE}_j - a^{\tE}_j, b^{\tE}_j - \wh\cp^{\tE}_j) \gtrsim \Delta_j, \quad 
\min(\cp_j - a^{\tE}_j, b^{\tE}_j - \cp_j) \gtrsim \Delta_j, \\ 
\text{ and \ } \{a^{\tE}_j + 1, \ldots, b^{\tE}_j - 1 \} \cap \Cp = \{\cp_j\}.
\label{eq:ab:prop}
\end{multline}
By the definition of $\wc{\bm\delta}_j$ in~\eqref{eq:debiased},
\begin{align}
& \sqrt{\frac{ (\wh\cp^{\tE}_j - a^{\tE}_j) (b^{\tE}_j - \wh\cp^{\tE}_j) }{ b^{\tE}_j - a^{\tE}_j }}
\l( \wc{\bm\delta}_j - \bm\delta_j \r) 
\nn \\
=& \, \sqrt{\frac{ (\wh\cp^{\tE}_j - a^{\tE}_j) (b^{\tE}_j - \wh\cp^{\tE}_j) }{ b^{\tE}_j - a^{\tE}_j }} \l( \mbf I_p - \wh{\bm\Omega}^{\tE} \wh{\bm\Sigma}^{\tO}_{a^{\tE}_j, b^{\tE}_j} \r) \l( \wh{\bm\delta}^{\tE}_j - \bm\delta_j \r)
\nn \\
& \, + \sqrt{\frac{ (\wh\cp^{\tE}_j - a^{\tE}_j) (b^{\tE}_j - \wh\cp^{\tE}_j) }{ b^{\tE}_j - a^{\tE}_j }} \wh{\bm\Omega}^{\tE} \l( \wh{\bm\gamma}^{\tO}_{\wh\cp^{\tE}_j, b^{\tE}_j} - \wh{\bm\gamma}^{\tO}_{a^{\tE}_j, \wh\cp^{\tE}_j} - \wh{\bm\gamma}^{\tO}_{\cp_j, b^{\tE}_j} + \wh{\bm\gamma}^{\tO}_{a^{\tE}_j, \cp_j} \r)
\nn \\
&
- \sqrt{\frac{ (\wh\cp^{\tE}_j - a^{\tE}_j) (b^{\tE}_j - \wh\cp^{\tE}_j) }{ b^{\tE}_j - a^{\tE}_j }} \wh{\bm\Omega}^{\tE} \l( \wh{\bm\Sigma}^{\tO}_{a^{\tE}_j, b^{\tE}_j} \bm\delta_j - \wh{\bm\gamma}^{\tO}_{\cp_j, b^{\tE}_j} + \wh{\bm\gamma}^{\tO}_{a^{\tE}_j, \cp_j} \r)
=: T_{j, 1} + T_{j, 2} + T_{j, 3}.
\label{eq:debiased:decomp}
\end{align}
For simplicity, we omit the subscript $j$ where there is no confusion.
By Assumption~\ref{assum:omega} and Lemma~\ref{lem:bound}, we have 
\begin{multline*}
\sqrt{b^{\tE} - a^{\tE}} \l\vert \mbf I_p - \wh{\bm\Omega}^{\tE} \wh{\bm\Sigma}^{\tO}_{a^{\tE}, b^{\tE}} \r\vert_\infty \le 
\sqrt{b^{\tE} - a^{\tE}} \l( \l\vert \mbf I_p - \wh{\bm\Omega}^{\tE} \wh{\bm\Sigma}^{\tE} \r\vert_\infty + 
\l\vert \wh{\bm\Omega}^{\tE} \l( \wh{\bm\Sigma}^{\tE} - \wh{\bm\Sigma}^{\tO}_{a^{\tE}, b^{\tE}} \r) \r\vert_\infty \r)
\\
\lesssim \omega^{(2)}_{n, p} + 2 C_0 \omega^{(1)}_{n, p} \sqrt{\log(p \vee n)} \lesssim \omega^{(1)}_{n, p} \sqrt{\log(p \vee n)}.
\end{multline*}
From this, Proposition~\ref{prop:dp:est}, \eqref{eq:ab:prop} and Assumption~\ref{assum:size:two}, we bound $T_1$ as 
\begin{align}
\vert T_1 \vert_\infty & = \sqrt{\frac{ (\wh\cp^{\tE} - a^{\tE}) (b^{\tE} - \wh\cp^{\tE}) }{ b^{\tE} - a^{\tE} }}\l\vert \l( \mbf I_p - \wh{\bm\Omega}^{\tE} \wh{\bm\Sigma}^{\tO}_{a^{\tE}, b^{\tE}} \r) \l( \wh{\bm\delta}^{\tE} - \bm\delta \r)\r\vert_{\infty} \nn \\
& \le \sqrt{\frac{ (\wh\cp^{\tE} - a^{\tE}) (b^{\tE} - \wh\cp^{\tE}) }{ b^{\tE} - a^{\tE} }}\l\vert  \mbf I_p - \wh{\bm\Omega}^{\tE} \wh{\bm\Sigma}^{\tO}_{a^{\tE}, b^{\tE}} \r\vert_{\infty} \l\vert \wh{\bm\delta}^{\tE} - \bm\delta \r\vert_{1} \nn \\
& \lesssim \l\vert \mbf I_p - \wh{\bm\Omega}^{\tE} \wh{\bm\Sigma}^{\tO}_{a^{\tE}, b^{\tE}} \r\vert_{\infty}\cdot\sqrt{\min(\wh\cp^{\tE} - a^{\tE}, b^{\tE} - \wh\cp^{\tE})}\l\vert \wh{\bm\delta}^{\tE} - \bm\delta \r\vert_{1} 
\nn \\
& \lesssim \omega^{(1)}_{n, p} \sqrt{\log(p \vee n)} \cdot \frac{\Psi_j \mathfrak{s}_j \sqrt{\log(p \vee n)}}{\underline{\sigma}\sqrt{\Delta_j}}
\lesssim \frac{\omega^{(1)}_{n, p} \mathfrak{s}_j \Psi_j \log(p \vee n)}{\underline{\sigma}\sqrt{\Delta_j}}.
\label{eq:t1:bound}
\end{align}
Also, by Assumption~\ref{assum:omega} and Lemma~\ref{lem:gamma:two}, we have 
\begin{align}
\vert T_2 \vert_\infty & = \sqrt{\frac{ (\wh\cp^{\tE} - a^{\tE}) (b^{\tE} - \wh\cp^{\tE}) }{ b^{\tE} - a^{\tE} }} \l\vert\wh{\bm\Omega}^{\tE} \l( \wh{\bm\gamma}^{\tO}_{\wh\cp^{\tE}, b^{\tE}} - \wh{\bm\gamma}^{\tO}_{a^{\tE}, \wh\cp^{\tE}} - \wh{\bm\gamma}^{\tO}_{\cp, b^{\tE}} + \wh{\bm\gamma}^{\tO}_{a^{\tE}, \cp} \r)\r\vert_{\infty} 
\nn \\
& \le \sqrt{\frac{ (\wh\cp^{\tE} - a^{\tE}) (b^{\tE} - \wh\cp^{\tE}) }{ b^{\tE} - a^{\tE} }} \Vert \wh{\bm\Omega}^{\tE}\Vert_{\infty}\l\vert  \wh{\bm\gamma}^{\tO}_{\wh\cp^{\tE}, b^{\tE}} - \wh{\bm\gamma}^{\tO}_{a^{\tE}, \wh\cp^{\tE}} - \wh{\bm\gamma}^{\tO}_{\cp, b^{\tE}} + \wh{\bm\gamma}^{\tO}_{a^{\tE}, \cp}\r\vert_{\infty} 
\nn \\
& \lesssim \omega^{(1)}_{n, p} \cdot \sqrt{\frac{ (\wh\cp^{\tE} - a^{\tE}) (b^{\tE} - \wh\cp^{\tE}) }{ b^{\tE} - a^{\tE} }} \l\vert  \wh{\bm\gamma}^{\tO}_{\wh\cp^{\tE}, b^{\tE}} - \wh{\bm\gamma}^{\tO}_{a^{\tE}, \wh\cp^{\tE}} - \wh{\bm\gamma}^{\tO}_{\cp, b^{\tE}} + \wh{\bm\gamma}^{\tO}_{a^{\tE}, \cp}\r\vert_{\infty} 
\nn \\
& \lesssim \frac{\omega^{(1)}_{n, p} \Psi_j^2 \log(p \vee n)}{\sqrt{\Delta_j} \,\vert \bm\Sigma \bm\delta_j \vert_\infty}.
\label{eq:t2:bound}
\end{align}
Next, by the definition of $\bm\Gamma_j$ and from the independence between $\mbf x_t$ and $\vep_t$ (Assumption~\ref{assum:func:dep:two}) we have 
\begin{subequations}
	\label{eq:gamma:eigval}
	\begin{align}
	\Lambda_{\min}(\bm\Gamma_j) &\ge \Lambda_{\min}\bigl(\mathrm{Cov}(\mbf x_t^{\tO} \varepsilon_t^{\tO})\bigr)\ge \sigma_{\varepsilon}^2 \underline{\sigma} \text{ \ and }    
	\\
	\Lambda_{\max}( \bm\Gamma_j) &\le \max_{\mbf a \in \mathbb{B}_2(1)} \E \l(\bigl(\mbf a^\top \mbf x_t^\tO\bigr)^2\bigl(\mbf b^\top\bm Z_t^\tO\bigr)^2\Big\vert \mc D^{\tE} \r) \le 16 \Xi^2 \Psi_j^2,
	\end{align}
\end{subequations}
where the last inequality is a consequence of H\"older's inequality with $\mbf b$ set as $\mbf b = \bigl((\bar{\bm\mu}^{\tE}_j)^{\top},\, 1\bigr)^\top\in\R^{p+1}$ which satisfies $\vert\mbf b\vert_2 \le \Psi_j$.
By Assumptions~\ref{assum:size:two} and~\ref{assum:omega} and Weyl's inequality,
\begin{subequations}
	\label{eq:omega:eigval}
	\begin{align}
	\Lambda_{\min}\l( \wh{\bm\Omega}^{\tE} \r) &\ge \Lambda_{\min}\l(\bm\Omega \r) - \l\Vert \wh{\bm\Omega}^{\tE} - \bm\Omega \r\Vert \ge \frac{1}{\bar{\sigma}} - \mathfrak{s}_\varrho \l( \frac{\omega^{(3)}_{n, p}}{\sqrt{n}} \r)^{1 - \varrho} \ge  \frac{1}{2\bar{\sigma}},\\
	\Lambda_{\max}\l( \wh{\bm\Omega}^{\tE} \r) &\le \Lambda_{\max}\l(\bm\Omega \r) + \l\Vert \wh{\bm\Omega}^{\tE} - \bm\Omega \r\Vert \le \frac{1}{\underline{\sigma}} + \mathfrak{s}_\varrho \l( \frac{\omega^{(3)}_{n, p}}{\sqrt{n}} \r)^{1 - \varrho} \le \frac{3}{2\underline{\sigma}}.
	\end{align}
\end{subequations}
This together with~\eqref{eq:ab:prop} implies
\begin{subequations}
	\label{eq:maxmin}
	\begin{align}
	&& \Lambda_{\min}\l(\tfrac{(\wh\cp^{\tE}_j - a^{\tE}_j)(b^{\tE}_j - \wh\cp^{\tE}_j)}{(\cp_j - a^{\tE}_j)(b^{\tE}_j - \cp_j)} \wh{\bm\Omega}^{\tE} \bm\Gamma_j \l(\wh{\bm\Omega}^{\tE}\r)^\top\r)&  \gtrsim \Lambda_{\min}(\bm\Gamma_j)\Lambda_{\min}^2\l( \wh{\bm\Omega}^{\tE} \r) \gtrsim \frac{\sigma_{\varepsilon}^2\underline{\sigma}}{\bar\sigma^2}, \\
	&& \Lambda_{\max}\l(\tfrac{(\wh\cp^{\tE}_j - a^{\tE}_j)(b^{\tE}_j - \wh\cp^{\tE}_j)}{(\cp_j - a^{\tE}_j)(b^{\tE}_j - \cp_j)} \wh{\bm\Omega}^{\tE} \bm\Gamma_j \l(\wh{\bm\Omega}^{\tE}\r)^\top \r) & \lesssim \Lambda_{\max}^2\l(\wh{\bm\Omega}^{\tE}\r) \Lambda_{\max}(\bm\Gamma_j) \lesssim \frac{\Xi^2 \Psi_j^2}{\underline \sigma^2}.
	\end{align}
\end{subequations}
Next, for each $j \in [q]$, we define
\begin{align*}
& w^{\tE}_{j, t} = \l\{\begin{array}{ll}-
\frac{\sqrt{(\wh\cp^{\tE}_j - a^{\tE}_j) (b^{\tE}_j - \wh\cp^{\tE}_j)}}{\cp_j - a^{\tE}_j} & \text{for \ } t \in \{a^{\tE}_j + 1, \ldots, \cp_j\}, \\
\frac{\sqrt{(\wh\cp^{\tE}_j - a^{\tE}_j) (b^{\tE}_j - \wh\cp^{\tE}_j)}}{b^{\tE}_j - \cp_j} & \text{for \ } t \in \{\cp_j + 1, \ldots, b^{\tE}_j\}.
\end{array}\r.
\end{align*}
Then, by~\eqref{eq:ab:prop}, it holds that 
\begin{align}
\label{eq:weights}
\vert w^{\tE}_{j, t} \vert \asymp 1, \quad 
\frac{1}{b^{\tE}_j - a^{\tE}_j} \sum_{t = a^{\tE}_j + 1}^{b^{\tE}_j} (w^{\tE}_{j, t})^2 \asymp 1 \text{ \ and \ }
\frac{1}{b^{\tE}_j - a^{\tE}_j} \sum_{t = a^{\tE}_j + 1}^{b^{\tE}_j} (w^{\tE}_{j, t})^4 \asymp 1.
\end{align}
Also, let 
\begin{align}
\label{eq:def:ujt}
\mbf U_{j, t} = \mbf x^{\tO}_t ( \vep^{\tO}_t + (\mbf x^{\tO}_t)^\top \bar{\bm\mu}^{\tE}_j )
\text{ and }
\mbf U^{\circ}_{j, t} = \mbf U_{j, t} - \E(\mbf U_{j, t} \vert \mc D^{\tE}) = \mbf U_{j, t} - \bm\Sigma \bar{\bm\mu}^{\tE}_j.
\end{align}
Then, we can write (omitting the subscript $j$)
\begin{align*}
\sqrt{ b^{\tE} - a^{\tE} }\,T_3 = \wh{\bm\Omega}^{\tE} \sum_{t = a^{\tE} + 1}^{b^{\tE}} w^{\tE}_t \mbf x^{\tO}_t ( \vep^{\tO}_t + (\mbf x^{\tO}_t)^\top \bar{\bm\mu}^{\tE} ) = \sum_{t = a^{\tE} + 1}^{b^{\tE}} w^{\tE}_t \wh{\bm\Omega}^{\tE}\mbf U_t
= \sum_{t = a^{\tE} + 1}^{b^{\tE}} w^{\tE}_t \wh{\bm\Omega}^{\tE} \mbf U^{\circ}_t,
\end{align*}
where the last equality is due to that $\sum_{t = a^{\tE} +1}^{b^{\tE}}w^{\tE}_t = 0$. Let $\bm\xi_t = (\xi_{it})_{i \in [p]} := \wh{\bm\Omega}^{\tE} \mbf U^{\circ}_t$. Then,
\begin{align}\label{eq:defT3}
T_3 = \frac{1}{\sqrt{ b^{\tE} - a^{\tE} }} \sum_{t = a^{\tE} + 1}^{b^{\tE}} w^{\tE}_t \bm\xi_t.
\end{align}
Further, with $\mbf a_i = (\wh{\bm\Omega}^{\tE}_{i\cdot}, 0)^\top \in \R^{p+1}$ and $\mbf b = ((\bar{\bm\mu}^{\tE})^\top, 1)^\top \in \mathbb{R}^{p+1}$, we have
\begin{align*}
\xi_{it}\;=\;(\wh{\bm\Omega}^{\tE}\mbf U^\circ_{t})_i \; = \; \mbf a_i^\top \mbf Z^{\tO}_t \l(\mbf Z^{\tO}_t\r)^\top \mbf b - \E\l(\mbf a_i^\top \mbf Z^{\tO}_t \l(\mbf Z^{\tO}_t\r)^\top \mbf b \Big\vert \mc D^{\tE}\r).
\end{align*}
Note that $\vert \mbf a_i \vert_2 \le \max_{i \in [p]} \vert \wh{\bm\Omega}^{\tE}_{i \cdot} \vert_2 \le \Lambda_{\max}\bigl(\wh{\bm\Omega}^{\tE}\bigr) \le 3/(2\underline{\sigma})$ from~\eqref{eq:omega:eigval}, 
and $\vert \mbf b \vert_2 \le 1 + \vert \bar{\bm\mu}^{\tE} \vert_2 \le \Psi_j$. 
Then by~\eqref{eq:weights} and Lemma~\ref{lem:assum:func:two}, it follows that
\begin{align}\label{eq:sum4m}
\sum_{t = a^{\tE} + 1}^{b^{\tE}} (w^{\tE}_t)^4 \E\l( \vert \xi_{it} \vert^4  \big\vert \mc D^{\tE}, \mc Q^{\tE}_{n, p} \r)  \lesssim \l(\frac{\Xi \Psi_j}{\underline\sigma}\r)^4 \sum_{t = a^{\tE} + 1}^{b^{\tE}} (w^{\tE}_t)^4 \asymp 
(b^{\tE} - a^{\tE})\l(\frac{\Xi \Psi_j}{\underline\sigma}\r)^4.
\end{align}
Also, with $C_{\xi} \asymp \Xi\Psi_j/\underline{\sigma}$, we have from the property of sub-exponential random variables (see, e.g.\ Proposition~2.7.1 of \citealp{Ver18}),
\begin{align*}
\p\l( \vert \xi_{it} \vert \ge u \big\vert \mc D^{\tE},  \mc Q^{\tE}_{n, p} \r)\; \le \; 2\exp\l(-\frac{u}{C_\xi}\r) \text{ \ for every \ } u \ge 0.
\end{align*}
The above tail probability bound, together with \eqref{eq:weights}, Fubini's theorem and the union bound, implies that
\begin{align}\label{eq:max4m}
&(w^{\tE}_t)^4 \E\l(\max_{i \in [p]}\max_{a^{\tE} < t \le b^{\tE}}\vert\xi_{it}\vert^4\Big\vert\mc D^{\tE}, \mc Q^{\tE}_{n, p} \r)
\nn \\ 
\asymp & \int_{0}^\infty \p\l(\max_{i \in [p]}\max_{a^{\tE} < t \le b^{\tE}}\vert\xi_{it}\vert^4 \ge u \big\vert \mc D^{\tE}, \mc Q^{\tE}_{n, p} \r) \mathrm{d}u 
\nn \\
\le & \int_0^{(2C_{\xi}\log(p(b^\tE-a^\tE)))^4} \mathrm{d}u + p(b^\tE-a^\tE)\int_{(2C_{\xi}\log(p(b^\tE-a^\tE)))^4}^\infty  \p\l(\vert\xi_{it}\vert^4 \ge u \big\vert \mc D^{\tE}, \mc Q^{\tE}_{n, p} \r)\mathrm{d}u 
\nn \\
\le & \bigl(2C_{\xi}\log(p(b^\tE-a^\tE))\bigr)^4 + 2 p (b^\tE - a^\tE) \int_{(2C_{\xi}\log(p(b^\tE-a^\tE)))^4}^\infty \exp\l(-\frac{u^{1/4}}{2C_{\xi}}\r) \mathrm{d}u
\nn \\
\asymp & \l(\frac{\Xi\Psi_j}{\underline{\sigma}} \log(p\Delta)\r)^4.
\end{align}
Similarly, with $\psi \asymp \Xi\Psi \log(p\Delta) / \underline{\sigma}$, we can show
\begin{align}\label{eq:tail4m}
(w^{\tE}_t)^4 \E\l( \max_{i \in [p]} \vert\xi_{it}\vert^4 \mathbb{I}_{\{\max_{i \in [p]}\vert\xi_{it}\vert > \psi \}} \Big\vert \mc D^{\tE}, \mc Q^{\tE}_{n, p} \r) \lesssim \frac{(\Xi\Psi_j)^4\log^3(p\Delta)}{p\Delta\underline{\sigma}^4}.
\end{align}
Recalling the definition of $T_3$ in~\eqref{eq:defT3}, we have 
\begin{align*}
\mathrm{Cov}(T_3 \big\vert\mc D^{\tE}) & = \frac{1}{b^{\tE} - a^{\tE}} \sum_{t = a^{\tE} + 1}^{b^{\tE}} (w^{\tE}_t)^2\mathrm{Cov}(\bm\xi_t \big\vert \mc D^{\tE}) 
\\
& = \frac{1}{b^{\tE} - a^{\tE}} \sum_{t = a^{\tE} + 1}^{b^{\tE}} (w^{\tE}_t)^2\wh{\bm\Omega}^{\tE} \mathrm{Cov}(\mbf U^{\circ}_t\big\vert \mc D^{\tE}) \bigl(\wh{\bm\Omega}^{\tE}\bigr)^{\top}
\\
& = \frac{(\wh{\cp}^{\tE} - a^{\tE})(b^{\tE} - \wh{\cp}^{\tE})}{(\cp - a^{\tE})(b^{\tE} - \cp)}\wh{\bm\Omega}^{\tE}{\bm\Gamma}\bigl(\wh{\bm\Omega}^{\tE}\bigr)^{\top}.
\end{align*}
Now we apply \cref{lem:gaussapp} to $T_3$ of the form \eqref{eq:defT3}.
Collecting~\eqref{eq:maxmin}, \eqref{eq:sum4m} \eqref{eq:max4m} and~\eqref{eq:tail4m}, we obtain 
\begin{align}
& \sup_{z \in \R} \l\vert \p\l( \vert T_3 \vert_\infty \le z \,\big\vert\, \mc D^{\tE}, \mc Q^{\tE}_{n, p} \r) - \p\l( \vert \mbf V \vert_\infty \le z \,\big\vert\, \mc D^{\tE}, \mc Q^{\tE}_{n, p} \r) \r\vert 
\nn \\
\le \;& C\l( \log(n)\log^{3/2}(p)\frac{\Psi_j^4}{\sqrt{\Delta_j}} + \log(n)\log^2(p)\log^2(p\vee n)\frac{\Psi_j^4}{{\Delta_j}} \r. 
\nn \\
& \qquad \l. + \log^{1/2}(n)\log(p)\log^2(p \vee n)\frac{\Psi_j^4}{\Delta_j\sqrt{p}} + \log^{3/2}(p)\log(p \vee n)\frac{\Psi_j^2}{\sqrt{\Delta_j}} \r) 
\nn \\
\le \; & \frac{C' \Psi_j^4 \log^{3/2}(p) \log(p \vee n)}{\sqrt{\Delta_j}} \l( 1 + \frac{\sqrt{\log(p)}\log^2(p \vee n)}{\sqrt{\Delta_j}} \r),
\label{eq:t3:approx}
\end{align}
where the constant $C, C' \in (0, \infty)$ depend only on $\bar{\sigma}$, $\underline{\sigma}$, $\sigma_\vep$ and $\Xi$.
Finally, note that for any $y \ge 0$,
\begin{align*}
& \p\l( \vert T_1 + T_2 + T_3 \vert_\infty \le z \,\big\vert\, \mc D^{\tE}, \mc Q^{\tE}_{n, p} \r) 
\\
\le & \, \p\l( \vert T_1 + T_2 \vert_\infty \ge y \,\big\vert\, \mc D^{\tE}, \mc Q^{\tE}_{n, p} \r) + \p\l( \vert T_3 \vert_\infty \le z + y \,\big\vert\, \mc D^{\tE}, \mc Q^{\tE}_{n, p} \r),
\text{ \ and}
\\
& \p\l( \vert \mbf V \vert_\infty \le z \,\big\vert\, \mc D^{\tE}, \mc Q^{\tE}_{n, p} \r) = \p\l( \vert \mbf V \vert_\infty \le z + y \,\big\vert\, \mc D^{\tE}, \mc Q^{\tE}_{n, p} \r) - \p\l( z \le \vert \mbf V \vert_\infty \le z + y \,\big\vert\, \mc D^{\tE}, \mc Q^{\tE}_{n, p} \r),
\end{align*}
from which we derive
\begin{align*}
& \sup_{z \in \R} \l\vert \p\l( \vert T_1 + T_2 + T_3 \vert_\infty \le z \,\big\vert\, \mc D^{\tE} \r) - \p\l( \vert \mbf V \vert_\infty \le z \,\big\vert\, \mc D^{\tE} \r) \r\vert 
\\
\le & \, \p\l( \vert T_1 + T_2 \vert_\infty \ge y \,\big\vert\, \mc D^{\tE}, \mc Q^{\tE}_{n, p} \r) + \sup_{z \in \R} \Bigl\vert \p\l( \vert T_3 \vert_\infty \le z \,\big\vert\, \mc D^{\tE}, \mc Q^{\tE}_{n, p} \r) - \p\l( \vert \mbf V \vert_\infty \le z \,\big\vert\, \mc D^{\tE}, \mc Q^{\tE}_{n, p} \r) \Bigr\vert 
\\
& \, + \sup_{z \in \R} \p\Bigl( \bigl\vert \vert \mbf V \vert_\infty - z \,\big\vert\, \le y \Bigm\vert \mc D^{\tE}, \mc Q^{\tE}_{n, p} \Bigr) + \p\l((\mc Q^{\tE}_{n, p})^c\r)
=: T_4 + T_5 + T_6 + \p\l((\mc Q^{\tE}_{n, p})^c\r).
\end{align*}
Setting $y \gtrsim \omega^{(1)}_{n, p} \Psi_j \max(\mathfrak{s}_j, \Psi_j \vert \bm\Sigma \bm\delta_j \vert_\infty^{-1})\log(p \vee n)/\sqrt{\Delta_j}$, we have 
\begin{align*}
T_4 \le \p\l( \vert T_1 + T_2 \vert_\infty \ge y \;\big\vert\; \mc D^{\tE}, \mc Q^{\tE}_{n, p} \cap \mc E^{\tO}_{n, p} \r) \p(\mc E^{\tO}_{n, p}) + \p((\mc E^{\tO}_{n, p})^c) = \p((\mc E^{\tO}_{n, p})^c)
\end{align*}
from~\eqref{eq:t1:bound} and~\eqref{eq:t2:bound}, and $T_5$ is handled by~\eqref{eq:t3:approx}.
Also, by \cite{nazarov2003maximal} (see also Lemma~4 of \citealp{chen2022inference} and \cite{chernozhukov2017detailed}), we obtain $T_6 \le C'' y \sqrt{\log(p)}$, where the constant $C'' > 0$ depends only on $\sigma_\vep$, $\bar{\sigma}$ and $\underline{\sigma}$ involved in the lower bound of $\min_{i \in [p]} \Var(V_i)$, see~\eqref{eq:maxmin}.
We collect the bounds on $T_4$, $T_5$ and $T_6$ and combine these with that $\p(\mc (Q^{\tE}_{n, p})^c) \le 2c_2 (p \vee n)^{-c_3} + \p((\mc O^{\tE}_{n, p})^c)$ and $\p((\mc E^{\tO}_{n, p})^c) \le c_2(p \vee n)^{-c_3}$ by Lemmas~\ref{lem:bound} and~\ref{lem:rsc}. 
Finally, noting the lower bound on $\Delta_j$ and that $c_3$ can be arbitrarily large, the proof is complete. 

\subsubsection{Proof of Proposition~\ref{prop:clime:dataE}}

It follows readily the more general result given in \cref{lem:clime}, where the main claim holds conditional on $\mc E^{\tE}_{n, p}$ defined in~\eqref{eq:set:e} with the dataset $\mc D^{\tE}$.

\subsubsection{Proof of Proposition~\ref{prop:G}}

Recall the definitions of $\mc E^\ell_{n, p}$ and $\mc R^\ell_{n, p}$, $\ell \in \{\tE, \tO\}$, from the proof of Theorem~\ref{thm:two}. We frequently use that $\mc E^{\tE}_{n, p} \subset \mc S^{\tE}_{n, p}$.
The subsequent arguments are conditional on $\mc E^{\tE}_{n, p} \cap \mc R^{\tE}_{n, p} \cap \mc E^{\tO}_{n, p} \cap \wt{\mc E}^{\tO}_{n, p}$ with $\wt{\mc E}^{\tO}_{n, p}$ defined in~\eqref{eq:set:e:tilde} below, which fulfils $\p(\wt{\mc E}^{\tO}_{n, p}) \ge 1 - c_2(p \vee n)^{-c_3}$.
Hence, $\p(\mc E^{\tE}_{n, p} \cap \mc R^{\tE}_{n, p} \cap \mc E^{\tO}_{n, p} \cap \wt{\mc E}^{\tO}_{n, p}) \ge 1 - 4 c_2 (p \vee n)^{-c_3}$ by Lemmas~\ref{lem:bound} and~\ref{lem:rsc}.

By~\eqref{eq:ab:prop}, which follows from Theorem~\ref{thm:one} conditional on $\mc E^{\tE}_{n, p}$, and by Assumption~\ref{assum:size}, we have $\cp_j \notin \{a^{\tE}_j + 1, \ldots, \bar{b}^{\tE}_j\}$ and $\cp_j \notin \{\bar{a}^{\tE}_j + 1, \ldots, b^{\tE}_j\}$ for small enough constant $\epsilon$ involved in the definition of $\bar{a}^{\tE}_j$ and $\bar{b}^{\tE}_j$.
Then, we have $\E(Y^{\tO}_t) = (\mbf x_t^{\tO})^\top \bm\beta_{j - 1}$ for $a^{\tE}_j + 1 \le t \le \bar{b}^{\tE}_j$, and $\E(Y^{\tO}_t) = (\mbf x_t^{\tO})^\top \bm\beta_j$ for $\bar{a}^{\tE}_j + 1 \le t \le b^{\tE}_j$.
Recalling the definitions of $\mbf U_{j, t}$ and $\mbf U^{\circ}_{j, t}$ from~\eqref{eq:def:ujt}, we have $\bm\Gamma_j = \Cov(\mbf U_{j, t} \vert \mc D^{\tE})$.
Noting that
\begin{align*}
\wh{\bm\Gamma}_j = \frac{1}{\bar{b}^{\tE}_j - a^{\tE}_j + b^{\tE}_j - \bar{a}^{\tE}_j} \Bigl[ (\bar{b}^{\tE}_j - a^{\tE}_j) \underbrace{\wh{\Cov}_{a^{\tE}_j, \bar{b}^{\tE}_j} (\wh{\mbf U}_{j, t})}_{=: \wh{\bm\Gamma}_{j, (\ell)}} + (b^{\tE}_j - \bar{a}^{\tE}_j) \underbrace{\wh{\Cov}_{\bar{a}^{\tE}_j, b^{\tE}_j} (\wh{\mbf U}_{j, t})}_{=: \wh{\bm\Gamma}_{j, (r)}} \Bigr],
\end{align*}
we derive a bound on $\vert \wh{\bm\Gamma}_{j, (\ell)} - \bm\Gamma_j \vert_\infty$. 
Analogous arguments carry over to $\vert \wh{\bm\Gamma}_{j, (r)} - \bm\Gamma_j \vert_\infty$, which completes the proof. 
We write
\begin{align}
& \l\vert \wh{\bm\Gamma}_{j, (\ell)} - \bm\Gamma_j \r\vert_\infty \le \l\vert \frac{1}{\bar{b}^{\tE}_j - a^{\tE}_j} \sum_{t = a^{\tE}_j + 1}^{\bar{b}^{\tE}_j} \wh{\mbf U}^\circ_{j, t} (\wh{\mbf U}^\circ_{j, t})^\top - \wt{\bm\Gamma}_{j, (\ell)}\r\vert_\infty
\nn \\
& \qquad + \l\vert \bar{\mbf U}^\circ_{j, (\ell)} (\bar{\mbf U}^\circ_{j, (\ell)})^\top \r\vert_\infty
+ \l\vert \wt{\bm\Gamma}_{j, (\ell)} - \bm\Gamma_j \r\vert_\infty =: T_{j, 1} + T_{j, 2} + T_{j, 3}, \text{ where} \label{eq:prop:G}
\\
& \wh{\mbf U}^\circ_{j, t} = \wh{\mbf U}_{j, t} - \E(\wh{\mbf U}_{j, t} \vert \mc D^{\tE}), \  
\bar{\mbf U}^\circ_{j, (\ell)} = \frac{1}{\bar{b}^{\tE}_j - a^{\tE}_j} \sum_{t = a^{\tE}_j + 1}^{\bar{b}^{\tE}_j} \wh{\mbf U}^\circ_{j, t} \text{ \ and}
\nn \\
& \wt{\bm\Gamma}_{j, (\ell)} = \frac{1}{\bar{b}^{\tE}_j - a^{\tE}_j} \sum_{t = a^{\tE}_j + 1}^{\bar{b}^{\tE}_j} \E\l( \wh{\mbf U}^\circ_{j, t} (\wh{\mbf U}^\circ_{j, t})^\top \big \vert \mc D^{\tE}\r).
\nn
\end{align}
Under Assumption~\ref{assum:func:dep:two}, we have $\wh{\mbf U}_{j, t}$ distributed independently over $t$ conditional on $\mc D^{\tE}$.
Further, we can write $\wh{U}_{j, it} = \mbf e_i^\top \mbf Z_t^{\tO} (\mbf Z_t^{\tO})^\top \mbf b_{j, t}$ with
\begin{align*}
\mbf b_{j, t} = \l\{ \begin{array}{ll}
\mbf b_{j, -} = \bm\beta_{j - 1} + \frac{1}{2} \wh{\bm\delta}^{\tE}_j, & t \in \{a^{\tE}_j + 1, \ldots, \bar{b}^{\tE}_j \},  \\
\mbf b_{j, +} = \bm\beta_j - \frac{1}{2} \wh{\bm\delta}^{\tE}_j, & t \in \{\bar{a}^{\tE}_j + 1, \ldots, b^{\tE}_j \}.
\end{array}\r.
\end{align*}
Then, we have 
\begin{align*}
\l\vert \bmx \bm\beta_{j - 1} + \frac{1}{2} \wh{\bm\delta}^{\tE}_j \\ 1 \emx \r\vert_2 &\le 1 + \vert \bm\beta_{j - 1} \vert_2 + \frac{1}{2} \vert \wh{\bm\delta}^{\tE}_j \vert_2 \lesssim 1 + \vert \bm\beta_{j - 1} \vert_2 + \vert \bm\delta_j \vert_2,
\\
\l\vert \bmx \bm\beta_j - \frac{1}{2} \wh{\bm\delta}^{\tE}_j \\ 1 \emx \r\vert_2 &\le 1 + \vert \bm\beta_j \vert_2 + \frac{1}{2} \vert \wh{\bm\delta}^{\tE}_j \vert_2 \lesssim 1 + \vert \bm\beta_j \vert_2 + \vert \bm\delta_j \vert_2,
\end{align*}
by Proposition~\ref{prop:dp:est} and Assumption~\ref{assum:size:two}. 
Combined with Lemma~\ref{lem:assum:func:two}, the above observations give
\begin{align}
\label{eq:uhat:subexp}
\sup_{\nu \ge 2} \nu^{-1} \l(\E\l( \vert \wh{U}_{j, it}^\circ \vert^\nu \big\vert \mc D^{\tE} \r)\r)^{1/\nu} \lesssim \Xi \Psi_j
\end{align}
for all $i \in [p]$ and $t \in \{a^{\tE}_j + 1, \ldots, \bar{b}^{\tE}_j\} \cup \{\bar{a}^{\tE}_j + 1, \ldots, b^{\tE}_j\}$. 
Then by~\eqref{eq:ab:prop}, \eqref{eq:uhat:subexp} and Lemma~\ref{lem:bound}, we have
\begin{align*}
\vert \bar{\mbf U}^\circ_{j, (\ell)} \vert_\infty \lesssim \frac{\Xi \Psi_j \sqrt{\log(p \vee n)}}{\sqrt{\Delta_j}}
\end{align*}
conditional on $\mc E^{\tE}_{n, p} \cap \mc R^{\tE}_{n, p} \cap \mc E^{\tO}_{n, p}$, which leads to 
\begin{align*}
\Delta_j T_{j, 2} \le \max_{j \in [q]} \Delta_j \vert \bar{\mbf U}^\circ_{j, (\ell)} \vert_\infty^2 \lesssim \Xi^2 \Psi_j^2 \log(p \vee n), \text{ \and }
\\
\sqrt{\Delta_j} \; T_{j, 2} \lesssim \frac{ \Xi^2 \Psi_j^2 \log(p \vee n) } {\sqrt{\Delta_j}} \lesssim \frac{\Psi_j^2}{\sqrt{\log(p \vee n)}}
\end{align*}
for all $j \in [q]$ under Assumption~\ref{assum:size:two}.
Similarly as in~\eqref{eq:uhat:subexp}, we have
\begin{align}
& \sup_{\nu \ge 2} \nu^{-2} \l(\E\l( \vert \wh{U}_{j, it}^\circ \wh{U}_{j, i' t}^\circ - \wt{\Gamma}_{j, (\ell), ii'} \vert^\nu \big \vert \mc D^{\tE} \r)\r)^{1/\nu} 
\le 2 \sup_{\nu \ge 2} \nu^{-2} \l(\E\l( \vert \wh{U}_{j, it}^\circ \wh{U}_{j, i't}^\circ \vert^\nu  \vert \mc D^{\tE}\r)\r)^{1/\nu}\nn \\
\le & \,
2 \sup_{\nu \ge 2} \nu^{-2} \l(\E\l( \vert \wh{U}_{j, it}^\circ \vert^{2\nu} \big \vert \mc D^{\tE} \r)\r)^{1/(2\nu)} \l(\E\l( \vert \wh{U}_{j, i't}^\circ \vert^{2\nu} \big \vert \mc D^{\tE} \r)\r)^{1/(2\nu)} 
\lesssim \Xi^2 \Psi_j^2 
\label{eq:uhat:subexp:two}
\end{align}
for all $i, i' \in [p]$, where the first inequality follows from Jensen's inequality and the second from H\"{o}lder's inequality.
Then, e.g.\ by Lemma~5 of \cite{wong2020lasso}, we have
\begin{align}
\label{eq:uu:tail:prob}
    \p\l( \Xi^{-2} \Psi_j^{-2} \vert \wh{U}_{j, it}^\circ \wh{U}_{j, i' t}^\circ - \wt{\Gamma}_{j, (\ell), ii'} \vert \ge z \r) \lesssim \exp( - z^{1/2}).
\end{align}
Theorem~4 of \cite{xu2022change} shows that for a sequence of independent centred random variables $\{W_t\}_{t \in \Z}$ satisfying $\sup_{t \in \Z} \p( \vert W_t \vert \ge z) \le \exp(1 - z^{1/2})$, there exist absolute constants $C_2, C_3 \in (0, \infty)$ such that for any integer $n \ge 3$,
\begin{align}
\label{eq:xu}
\p\l( \frac{1}{\sqrt n} \l\vert \sum_{t = 1}^n W_t \r\vert \ge z \r) \le n \exp\l( - C_2 (z\sqrt{n})^{1/2} \r) + 2 \exp\l( - C_3 z^2 \r).
\end{align}
Let us define with some constants $C_4, C_5 \in (0, \infty)$,
\begin{align}
\wt{\mc E}^{\tO}_{n, p} &= \bigcup_{i, i' \in [p]} \bigcup_{j \in [q]} \bigcup_{\ell \in \{\pm\}} \wt{\mc E}^{\tO}_{n, p}(\mbf e_i, \mbf e_{i'}, \vert \mbf b_{j, \ell} \vert_2^{-1} \mbf b_{j, \ell}), \text{\ where}
\label{eq:set:e:tilde}
\\
\wt{\mc E}^{\tO}_{n, p}(\mbf a, \mbf a', \mbf b) &= \l\{ 
\max_{\substack{0 \le s < e \le n \\  e - s \ge C_4 \log^3(p \vee n)}} \frac{1}{\sqrt{e - s}} \l\vert \sum_{t = s + 1}^e \l( U_t U'_t - \E(U_t U'_t) \r) \r\vert \le C_5 \Xi^2 \sqrt{\log(p \vee n)}, 
\r.
\nn \\
& \qquad \qquad \text{with \ } 
U_t = \mbf a^\top \mbf Z_t^{\tO} (\mbf Z_t^{\tO})^\top \mbf b - \E(\mbf a^\top \mbf Z_t^{\tO} (\mbf Z_t^{\tO})^\top \mbf b)
\nn \\
& \qquad \qquad \l. \text{and \ }
U'_t = (\mbf a')^\top \mbf Z_t^{\tO} (\mbf Z_t^{\tO})^\top \mbf b - \E((\mbf a')^\top \mbf Z_t^{\tO} (\mbf Z_t^{\tO})^\top \mbf b) \r\}.
\nn
\end{align}
Combining~\eqref{eq:uu:tail:prob} with~\eqref{eq:xu}, we can show that for large enough $C_5 > 0$,
\begin{align*}
\p(\wt{\mc E}^{\tO}_{n, p}) &\ge 1 - 2n^2p^2 \exp\l( - C_2 (\sqrt{C_4} C_5 \log^2(p \vee n))^{1/2} \r) - 4 np^2 \exp\l( - C_3 \bigl(C_5 \sqrt{\log(p \vee n)}\bigr)^2 \r)
\\
&\ge 1 - c_2 (p \vee n)^{-c_3}.
\end{align*}
This shows that
\begin{align*}
\sqrt{\Delta_j} \; T_{j, 1} \lesssim \Xi^2 \Psi_j ^2\sqrt{\log (p \vee n)}
\end{align*}
conditional on $\mc E^{\tE}_{n, p} \cap \mc R^{\tE}_{n, p} \cap \wt{\mc E}^{\tO}_{n, p}$, 
where we use that under Assumption~\ref{assum:size:two}, $\Delta_j \gtrsim \log^3(p \vee n)$.

Finally, for bounding $T_{j, 3}$, we note that 
\begin{align*}
\l\vert \wt{\bm\Gamma}_{(\ell), ii'} - \bm\Gamma_{ii'} \r\vert
\le \E\l( \l\vert \wh{U}^\circ_{it}\l( \wh{U}^\circ_{i't} - U^\circ_{i't} \r) \r\vert \,\big\vert\, \mc D^{\tE} \r) + \E\l( \l\vert U^\circ_{it}\l( \wh{U}^\circ_{i't} - U^\circ_{i't} \r) \r\vert \,\big\vert\, \mc D^{\tE} \r)  =: T_{4, ii'} + T_{5, ii'}
\end{align*}
Observe that for $t \in \{a^{\tE}_j + 1, \ldots, \bar{b}^{\tE}_j\}$, 
\begin{align*}
\wh{U}_{it} - U_{it} &= \mbf e_i^\top \mbf Z_t^{\tO} (\mbf Z_t^{\tO})^\top 
\bmx
\frac{\wh\cp^{\tE}_j - a^{\tE}_j}{b^{\tE}_j - a^{\tE}_j} \wh{\bm\delta}^{\tE}_j - \frac{\cp_j - a^{\tE}_j}{b^{\tE}_j - a^{\tE}_j} \bm\delta_j \\ 0
\emx \text{ \ and}
\\
\l\vert \frac{\wh\cp^{\tE}_j - a^{\tE}_j}{b^{\tE}_j - a^{\tE}_j} \wh{\bm\delta}^{\tE}_j - \frac{\cp_j - a^{\tE}_j}{b^{\tE}_j - a^{\tE}_j} \bm\delta_j \r\vert_2 
&\le \frac{\wh\cp^{\tE}_j - a^{\tE}_j}{b^{\tE}_j - a^{\tE}_j} \l\vert \wh{\bm\delta}^{\tE}_j - \bm\delta_j \r\vert_2 + \frac{\vert \wh\cp^{\tE}_j - \cp_j \vert }{b^{\tE}_j - a^{\tE}_j} \vert \bm\delta_j \vert_2
\\
&\lesssim \frac{\Psi_j \sqrt{\mathfrak{s}_j \log(p \vee n)}}{\sqrt{\Delta_j}} \l(1 + \frac{\vert \bm\delta_j \vert_2 \Psi_j \sqrt{\log(p \vee n)}}{\sqrt{\mathfrak{s}_j} \vert \bm\Sigma \bm\delta_j \vert_\infty^2 \sqrt{\Delta}_j} \r)
\lesssim \Psi_j \sqrt{\frac{\mathfrak{s}_j \log(p \vee n)}{\Delta_j}},
\end{align*}
where the second inequality follows 
from~\eqref{eq:ab:prop}, Theorem~\ref{thm:one} and Proposition~\ref{prop:dp:est}, and the last one from~\eqref{eq:size:change} and Assumption~\ref{assum:size}.
Then thanks to Lemma~\ref{lem:assum:func:two}, we have for all $i \in [p]$, 
\begin{align*}
\l(\E\l( \l\vert \wh{U}^\circ_{it} - U^\circ_{it} \r\vert^2 \bigg \vert \mc D^{\tE} \r)\r)^{1/2} \lesssim \Xi \l\vert \frac{\wh\cp^{\tE}_j - a^{\tE}_j}{b^{\tE}_j - a^{\tE}_j} \wh{\bm\delta}^{\tE}_j - \frac{\cp_j - a^{\tE}_j}{b^{\tE}_j - a^{\tE}_j} \bm\delta_j \r\vert_2 \lesssim \Xi \Psi_j \sqrt{\frac{\mathfrak{s}_j \log(p \vee n)}{\Delta_j}}.
\end{align*}
Together with that $(\E( \vert U^\circ_{it} \vert^2 \vert \mc D^{\tE}))^{1/2} \lesssim \Xi \Psi$, which follows analogously as in~\eqref{eq:uhat:subexp}, this gives 
\begin{align*}
T_{5, ii'} \lesssim \Xi^2 \Psi_j^2 \sqrt{\frac{\mathfrak{s}_j \log(p \vee n)}{\Delta_j}}
\end{align*}
for all $i, i' \in [p]$ by H\"{o}lder's inequality. 
We similarly bound $T_{4, ii'}$ such that
\begin{align*}
T_{j, 3} \lesssim \Xi^2 \Psi_j^2 \sqrt{\frac{\mathfrak{s}_j \log(p \vee n)}{\Delta_j}}
\end{align*}
for all $j \in [q]$.
Collecting the bounds on $T_{j, 1}$, $T_{j, 2}$ and $T_{j, 3}$ in~\eqref{eq:prop:G}, the proof is complete.

\subsubsection{Proof of Corollary~\ref{cor:thm:two}}

In what follows, we present our arguments conditionally on $\mc Q_{n, p} = \mc E^{\tE}_{n, p} \cap \mc R^{\tE}_{n, p} \cap \mc E^{\tO}_{n, p} \cap \wt{\mc E}^{\tO}_{n, p}$ with $\mc E^\ell_{n, p}$ and $\mc R^\ell_{n, p}$ defined in the proof of Theorem~\ref{thm:two}, and $\wt{\mc E}^{\tO}_{n, p}$ in the proof of Proposition~\ref{prop:G}.
In particular, we have
$\p(\mc Q_n) \ge 1 - 4 c_2 (p \vee n)^{-c_3}$.
We regard $\psi_{n, p} = \sqrt{\log(p \vee n)}$ when applying the preceding results.

Under Assumption~\ref{assum:func:dep:two}, setting $\mbf b = ((\bar{\bm\mu}^{\tE}_j)^\top, 1)^\top$, we have
\begin{align*}
\vert \bm\Gamma_j \vert_\infty \le \max_{i \in [p]} \E\l[ \l( \mbf e_i^\top \mbf Z^{\tO}_t (\mbf Z^{\tO}_t)^\top \mbf b - \E\l( \mbf e_i^\top \mbf Z^{\tO}_t (\mbf Z^{\tO}_t)^\top \mbf b \r) \r)^2  \Big\vert \mc D^{\tE}\r]
\le (2 \Xi)^2 \vert \mbf b \vert_2^2 \le 4 \Xi^2 \Psi_j^2
\end{align*}
for all $j \in [q]$. Then, the (conditional) covariance matrices of $\mbf V_j$ and $\wh{\mbf V}_j$ satisfy
\begin{align*}
& \l\vert \frac{(\wh\cp^{\tE}_j - a^{\tE}_j) (b^{\tE}_j - \wh\cp^{\tE}_j)}{(\cp_j - a^{\tE}_j) (b^{\tE}_j - \cp_j)} \wh{\bm\Omega}^{\tE} \bm\Gamma_j (\wh{\bm\Omega}^{\tE})^\top - \wh{\bm\Omega}^{\tE} \wh{\bm\Gamma}_j (\wh{\bm\Omega}^{\tE})^\top \r\vert_\infty 
\\
\lesssim & \, \l\vert \wh{\bm\Omega}^{\tE} \l( \wh{\bm\Gamma}_j - \bm\Gamma_j\r) (\wh{\bm\Omega}^{\tE})^\top \r\vert_\infty 
+ \frac{\vert \wh{\cp}_j - \cp_j \vert^2 \l\vert \wh{\bm\Omega}^{\tE} \bm\Gamma_j (\wh{\bm\Omega}^{\tE})^\top \r\vert_\infty}{\Delta_j^2} 
\\
\lesssim & \, \l\Vert \wh{\bm\Omega}^{\tE} \r\Vert_{\infty}^2 \l\vert \wh{\bm\Gamma}_j - \bm\Gamma_j\r\vert_{\infty} + \frac{\vert \wh{\cp}_j - \cp_j \vert^2 \Vert \wh{\bm\Omega}^{\tE} \Vert_{\infty}^2 \l\vert \bm\Gamma_j\r\vert_{\infty}}{\Delta_j^2}
\\
\lesssim & \, \Vert \bm\Omega \Vert_1^2 \Psi_j^2 \sqrt{\frac{ \mathfrak{s}_j \log(p \vee n) }{\Delta_j}} + \frac{ \Vert \bm\Omega \Vert_1^2 \Psi_j^6 \log^2(p \vee n)}{ \vert \bm\Sigma \bm\delta_j \vert_\infty^4 \Delta_j^2} 
\\
= & \Vert \bm\Omega \Vert_1^2 \Psi_j^2 \l[ \sqrt{\frac{\mathfrak{s}_j \log(p \vee n)}{\Delta_j}} + \l(\frac{\Psi_j^2 \log(p \vee n)}{\vert \bm\Sigma \bm\delta_j \vert_\infty^2 \Delta_j} \r)^2 \r]
\end{align*}
where the first inequality follows from~\eqref{eq:ab:prop} and  the third inequality from \cref{thm:one} and Propositions~\ref{prop:clime:dataE} and~\ref{prop:G}. Combined with~\eqref{eq:maxmin}, Lemma~\ref{lem:gausscmp} indicates that
\begin{align*}
& \sup_{z \in \R} \l\vert \p\l( \vert \wh{\mbf V} \vert_\infty \le z \big\vert \mc D^{\tE} \cup \mc D^{\tO} \r) - \p\l( \vert \mbf V \vert_\infty \le z \big\vert \mc D^{\tE} \cup \mc D^{\tO} \r) \r\vert
\\
\le & \, \sup_{z \in \R} \l\vert \p\l( \vert \wh{\mbf V} \vert_\infty \le z \big\vert \mc D^{\tE} \cup \mc D^{\tO}, \mc Q_{n, p} \r) - \p\l( \vert \mbf V \vert_\infty \le z \big\vert \mc D^{\tE} \cup \mc D^{\tO}, \mc Q_{n, p} \r) \r\vert + \p(\mc Q_{n, p}^c)
\\
\lesssim & \, 
\frac{ \Vert \bm\Omega \Vert_1^2 \Psi_j^2 \bar{\sigma}^2 }{\sigma_\vep^2 \underline{\sigma}} \log(p)\log(n) \l[ \sqrt{\frac{\mathfrak{s}_j \log(p \vee n)}{\Delta_j}} + \l(\frac{\Psi_j^2 \log(p \vee n)}{\vert \bm\Sigma \bm\delta_j \vert_\infty^2 \Delta_j} \r)^2 \r] + 4 c_2 (p \vee n)^{-c_3}.
\end{align*}
Taking into account the overlap in the sets conditioned on in the proofs of Theorem~\ref{thm:two}, Propositions~\ref{prop:clime:dataE} and~\ref{prop:G}, and noting that 
\begin{multline*}
\min\l(\frac{ \Vert \bm\Omega \Vert_1^2 \Psi_j^6\log(p)\log(n)\log^2(p \vee n)}{\vert\bm\Sigma\bm\delta_j\vert_\infty^4\Delta_j^2}, 1\r)  \\ \le
\min \l(\frac{ \Vert \bm\Omega \Vert_1^2 \Psi_j^6\log(p)\log(n)\log^2(p \vee n)}{\vert\bm\Sigma\bm\delta_j\vert_\infty^4\Delta_j^2}, 1 \r)^{1/4} \le \frac{\Psi_j^2\Vert\bm\Omega\Vert_1\log(n\vee p)\sqrt{\log(p)}}{\vert\bm\Sigma\bm\delta_j\vert_{\infty}\sqrt{\Delta_j}},    
\end{multline*}
the conclusion follows from Theorem~\ref{thm:two} and that $c_3$ can be arbitrarily large.

\subsubsection{Supporting lemmas}

Define $\mathbb{K}(b) = \mathbb{B}_0(b) \cap \mathbb{B}_2(1)$
with some $b \ge 1$.

\begin{lem}[Restricted eigenvalue condition]
	\label{lem:rsc}
	Suppose that Assumptions~\ref{assum:xe}, \ref{assum:func:dep} 
    and~\ref{assum:xe:iii} hold.
	Then, $\p(\mc R_{n, p}) \ge 1 - c_2(p \vee n)^{-c_3}$ where the constants $c_2, c_3 \in (0, \infty)$ are as in Theorem~\ref{thm:one}, and $\mc R_{n, p} = \mc R^{(1)}_{n, p} \cap \mc R^{(2)}_{n, p}$ with
	\begin{align*}
	\mc R^{(1)}_{n, p} &= \Biggl\{ \frac{1}{e - s} \sum_{t = s + 1}^e \mbf a^\top \mbf x_t \mbf x_t^\top \mbf a 
	\ge \frac{\underline{\sigma}}{2} \vert \mbf a \vert_2^2 - {\underline{\sigma}^{\frac{2\kappa -1}{2\kappa+1}}} \crsc \log(p) (e - s)^{-\frac{1}{1 + 2\kappa}} \vert \mbf a \vert_1^2, 
	\\
	& \qquad \quad \text{ for all }\, 0 \le s < e \le n \text{ satisfying } e - s \ge \crsc \underline{\sigma}^{-2} \psi^2_{n, p} \text{ and } \mbf a \in \R^p \Biggr\},
	\\
	\mc R^{(2)}_{n, p} &= \Biggl\{ \frac{1}{e - s}  \sum_{t = s + 1}^e \mbf a^\top \mbf x_t \mbf x_t^\top \mbf a 
	\le \frac{3 \bar{\sigma}}{2} \vert \mbf a \vert_2^2 + {\underline{\sigma}^{\frac{2\kappa -1}{2\kappa+1}}} \crsc \log(p) (e - s)^{-\frac{1}{1 + 2\kappa}} \vert \mbf a \vert_1^2,
	\\
	& \qquad \quad \text{ for all } \, 0 \le s < e \le n \text{ satisfying } e - s \ge \crsc \underline{\sigma}^{-2} \psi^2_{n, p} \text{ and } \mbf a \in \R^p \Biggr\}.
	\end{align*}
	Here $\crsc \in (0, \infty)$ is a constant that depends on $\kappa$ only (with $\kappa = 0$ under Assumption~\ref{assum:func:dep}~\ref{assum:fd:gauss} and $\kappa = \max(1/r - 1/2, 0)$ under~\ref{assum:fd:ind}).
\end{lem}
\begin{proof}
	First suppose that Assumption~\ref{assum:func:dep}~\ref{assum:fd:exp} holds.
	Let $b_{s, e}$ denote an integer that depends on $(e - s)$ for some $0 \le s < e \le n$, and define
	\begin{multline*}
	\wt{\mc R}_{n, p} = \biggl\{ \sup_{\mbf a \in \mathbb{K}(2b_{s, e})} \frac{1}{e - s} \l\vert \sum_{t = s + 1}^e \mbf a^\top\l(\mbf x_t\mbf x_t^\top - \bm\Sigma\r) \mbf a \r\vert \le \frac{\underline{\sigma}}{54},  \\ 
	\text{ for all }\, 0 \le s < e \le n \text{ with } e - s \ge \crsc \underline{\sigma}^{-2}\psi_{n, p}^2 \biggr\}.
	\end{multline*}
	By Lemma~\ref{lem:exp}~\ref{lem:exp:one}, the $\varepsilon$-net argument (see e.g.\ Lemma~F.2 of \citealp{basu2015regularized}) and the union bound,
	we have
	\begin{align}
	\p\l( \wt{\mc R}_{n, p}^c \r) &\le 
	\sum_{\substack{0 \le s < e \le n \\ e - s \ge \crsc \underline{\sigma}^{-2} \psi_{n, p}^2}}
	C_\kappa \exp\l[ - C^\prime \l(\frac{\underline{\sigma}\sqrt{e - s}}{54 \Xi}\r)^{\frac{2}{1 + 2 \kappa}} + 2 b_{s, e} \log(p) \r]
	\nn \\
	&\le C_\kappa n^2 \exp\l[ -\frac{C^\prime}{2} \l(\frac{\crsc^{1/2}}{54 \Xi} \r)^{\frac{2}{1 + 2\kappa}} \log(p \vee n) \r],
	\nn
	\end{align}
	with some constant $C^\prime \in (0, \infty)$ depending only on $\kappa$.
	Here, the last inequality follows with 
	\begin{align*}
	b_{s, e} = \l\lfloor \frac{C^\prime}{4\log(p)}\l(\frac{ \underline{\sigma}\sqrt{e - s} }{54}\r)^{\frac{2}{1 + 2\kappa}} \r\rfloor
	\end{align*}
	and the condition on $e -s$. 
	We can find $\crsc$ that depends on $\kappa$, such that $b_{s, e}  \ge 1$ for large enough $\crsc$ and further $\p(\wt{\mc R}_{n, p}) \ge 1 - c_2(p \vee n)^{-c_3}$.
	Then, by Lemma~12 of \cite{loh2012high}, on $\wt{\mc R}_{n, p}$, we have
	\begin{align*}
	\sum_{t = s + 1}^e \mbf a^\top \mbf x_t \mbf x_t^\top \mbf a 
	& \ge \underline{\sigma} (e - s) \vert \mbf a \vert_2^2 - \frac{\underline{\sigma}}{2} (e - s) \l( \vert \mbf a \vert_2^2 + \frac{4\log(p)}{C^\prime} \l(\frac{54}{ \underline{\sigma} \sqrt{e - s} }\r)^{\frac{2}{1 + 2\kappa}} \vert \mbf a \vert_1^2 \r)
	\\
	& \ge \frac{\underline{\sigma}}{2} (e - s) \vert \mbf a \vert_2^2 - \underline{\sigma}^{\frac{2\kappa - 1}{2\kappa + 1}} \crsc \log(p) (e - s)^{\frac{2\kappa}{1 + 2\kappa}} \vert \mbf a \vert_1^2
	\end{align*}
	for all $\mbf a \in \R^p$, with sufficiently large $\crsc$ depending only on $\kappa$. 	Analogously conditional on $\wt{\mc R}_{n, p}$, we have
	\begin{align*}
	\sum_{t = s + 1}^e \mbf a^\top \mbf x_t \mbf x_t^\top \mbf a 
	\le \frac{3\bar{\sigma}}{2} (e - s) \vert \mbf a \vert_2^2 +  \underline{\sigma}^{\frac{2\kappa - 1}{2\kappa + 1}} \crsc \log(p) (e - s)^{\frac{2\kappa}{1 + 2\kappa}} \vert \mbf a \vert_1^2
	\end{align*}
	for all $\mbf a \in \R^p$.
	The cases under Assumption~\ref{assum:func:dep}~\ref{assum:fd:gauss}--\ref{assum:fd:ind} are handled similarly and thus we omit the proof.
\end{proof}

\begin{lem}[CLIME]
	\label{lem:clime}
	Suppose that Assumptions~\ref{assum:xe}, \ref{assum:func:dep} and~\ref{assum:xe:iii} hold. Define 
	\begin{align}\label{eq:clime}
	\wh{\bm\Omega}_{s, e} = \mathop{\arg\min}_{\mbf M \in \R^{p \times p}} \vert \mbf M \vert_1 \text{ \ subject to \ } 
	\sqrt{e - s} \l\vert \mbf M \wh{\bm\Sigma}_{s, e} - \mbf I_p \r\vert_\infty \le \eta.
	\end{align}
	Then conditionally on $\mc E_{n, p}$ defined in~\eqref{eq:set:e}, the following results hold.
	\begin{enumerate}[label = (\roman*)]
		\item \label{lem:clime:one} Setting $\eta = C_\eta \Vert \bm\Omega \Vert_1 \psi_{n, p}$ with $C_\eta \ge C_0$, we have
		\begin{align*}
		\sqrt{e-s}\l\vert \wh{\bm\Omega}_{s, e} - \bm\Omega \r\vert_\infty \le 4 C_\eta \Vert \bm\Omega \Vert_1^2 \psi_{n, p} \text{ \ and \ } \Vert \wh{\bm\Omega}_{s, e} \Vert_\infty \le \Vert \bm\Omega \Vert_1,
		\end{align*}
		uniformly over $0 \le s < e \le n$ with $e - s \ge C_1\psi_{n, p}^2$.
		
		\item \label{lem:clime:two} 
		For $\mathfrak{s}_{\varrho}$ defined in \cref{assum:omega} with $\varrho \in [0, 1)$, it holds that 
		\begin{align*}
		\l\Vert \wh{\bm\Omega}_{s, e} - \bm\Omega \r\Vert \le \l\Vert \wh{\bm\Omega}_{s, e} - \bm\Omega \r\Vert_\infty \le 12\, \mathfrak{s}_{\varrho}\,
		\l(\frac{4 C_\eta \Vert \bm\Omega \Vert_1^2 \psi_{n, p} }{\sqrt{e-s}}\r)^{1 - \varrho}
		\end{align*}
		uniformly over $0 \le s < e \le n$ with $e - s \ge C_1\psi_{n, p}^2$.
	\end{enumerate}
\end{lem}

\begin{proof}[Proof of Lemma~\ref{lem:clime}~\ref{lem:clime:one}]
	The proof takes analogous steps as those in the proof of Theorem~6 of \cite{cai2011constrained}.
	We present the following arguments conditional on $\mc E_{n, p}$.
	Since $\Vert \bm\Omega \Vert_\infty = \Vert \bm\Omega \Vert_1$, we have
	\begin{align}
	\sqrt{ e - s } \l\vert \bm\Omega \wh{\bm\Sigma}_{s, e} - \mbf I_p \r\vert_\infty
	&= \sqrt{ e - s } \l\vert \bm\Omega \l(\wh{\bm\Sigma}_{s, e} - \bm\Sigma\r) \r\vert_\infty
	\le \sqrt{ e - s } \Vert \bm\Omega \Vert_1 \l\vert \wh{\bm\Sigma}_{s, e} - \bm\Sigma \r\vert_\infty \nn\\
	&\le C_0 \Vert \bm\Omega \Vert_1 \psi_{n, p} \le \eta.\label{ineq:Sighat}
	\end{align}
	Thus, $\bm\Omega$ is feasible for the constraint in~\eqref{eq:clime}. Since solving the problem in~\eqref{eq:clime} is equivalent to solving
	\begin{align*}
	(\wh{\bm\Omega}_{s, e, i \cdot})^\top = \mathop{\arg\min}_{\mbf m \in \R^p} \vert \mbf m \vert_1 
	\text{ \ subject to \ }
	\sqrt{e - s} \l\vert \mbf m^\top \wh{\bm\Sigma}_{s, e} - \mbf e_i^\top \r\vert_\infty \le \eta
	\end{align*}
	for each $i \in [p]$, it follows that 
	$\Vert \wh{\bm\Omega}_{s, e} \Vert_\infty = \max_{i \in [p]} \vert \wh{\bm\Omega}_{s, e, i \cdot} \vert_1 \le \Vert \bm\Omega \Vert_1$.
	Then, by \eqref{ineq:Sighat},
	\begin{align*}
	\sqrt{e - s} \l\vert \l( \wh{\bm\Omega}_{s, e} - \bm\Omega \r)\wh{\bm\Sigma}_{s, e} \r\vert_\infty
	& \le \sqrt{e - s} \l\vert \wh{\bm\Omega}_{s, e} \wh{\bm\Sigma}_{s, e} - \mbf I_p \r\vert_\infty
	+ \sqrt{e - s} \l\vert \bm\Omega \wh{\bm\Sigma}_{s, e} -\mbf I_p   \r\vert_\infty \le 2\eta.
	\end{align*}
	Further,
	\begin{align*}
	\sqrt{e - s} \l\vert \l( \wh{\bm\Omega}_{s, e} - \bm\Omega \r) \bm\Sigma  \r\vert_\infty
	& \le \sqrt{e-s} \l\vert \l( \wh{\bm\Omega}_{s, e} - \bm\Omega \r) \wh{\bm\Sigma}_{s, e} \r\vert_\infty
	+ \sqrt{e - s} \l\vert \l( \wh{\bm\Omega}_{s, e}- \bm\Omega \r) \l(\wh{\bm\Sigma}_{s, e} - \bm\Sigma\r) \r\vert_\infty
	\\
	&\le 2 \eta +
	\sqrt{e - s} \l\Vert \wh{\bm\Omega}_{s, e} - \bm\Omega  \r\Vert_\infty
	\l\vert \wh{\bm\Sigma}_{s, e} - \bm\Sigma\r\vert_\infty
	\le 4 \eta.
	\end{align*}
	Therefore, it follows that
	\begin{align*}
	\sqrt{e - s} \l\vert \wh{\bm\Omega}_{s, e} - \bm\Omega  \r\vert_\infty
	\le \sqrt{e - s} \l\vert \l( \wh{\bm\Omega}_{s, e} - \bm\Omega \r) \bm\Sigma \r\vert_\infty \l\Vert \bm\Omega  \r\Vert_1
	\le 4 \eta \l\Vert \bm\Omega  \r\Vert_1 ,
	\end{align*}
	which shows the assertion, as $\eta = C_\eta \Vert \bm\Omega \Vert_1 \psi_{n, p}$.
\end{proof}

\begin{proof}[Proof of Lemma~\ref{lem:clime}~\ref{lem:clime:two}]
	We continue to present our arguments conditionally on $\mc E_{n, p}$.
	Let $\varphi := \vert \wh{\bm\Omega}_{s, e} - \bm\Omega \vert_\infty$, $\wh{\bm\Omega}_{s, e} = (\wh\omega_{s, e, ii'})_{i,i \in [p]}$, $\mbf h_i := \wh{\bm\Omega}_{s, e, i \cdot} - \bm\Omega_{i \cdot}$,
	$\mbf h^{(1)}_i := (\wh\omega_{s, e, ii'} \mathbb I_{\{\vert \wh\omega_{s, e, ii'} \vert \ge 2\varphi \}}, \, i' \in [p]) - \bm\Omega_{i \cdot}$ and $\mbf h^{(2)}_i := \mbf h_i - \mbf h^{(1)}_i = (\wh\omega_{s, e, ii'} \mathbb I_{\{\vert \wh\omega_{s, e, ii'} \vert < 2\varphi \}}, \, i' \in [p])$. Then,
	\begin{align*}
	\vert \bm\Omega_{i \cdot} \vert_1 - \vert \mbf h^{(1)}_i \vert_1 + \vert \mbf h^{(2)}_i \vert_1
	\le \vert \bm\Omega_{i \cdot} + \mbf h^{(1)}_i \vert_1 + \vert \mbf h^{(2)}_i \vert_1
	= \vert \wh{\bm\Omega}_{s, e, i \cdot} \vert_1 \le \vert \bm\Omega_{i \cdot} \vert_1,
	\end{align*}
	which implies that $\vert \mbf h^{(2)}_i \vert_1 \le \vert \mbf h^{(1)}_i \vert_1$
	and thus $\vert \mbf h_i \vert_1 \le 2 \vert \mbf h^{(1)}_i \vert_1$.
	The latter is bounded as
	\begin{align*}
	\vert \mbf h^{(1)}_i \vert_1 
	=& \, \sum_{i' = 1}^p \l\vert \wh\omega_{s, e, ii'} \mathbb I_{\{\vert \wh\omega_{s, e, ii'} \vert \ge 2\varphi\}} - \omega_{ii'} \r\vert
	\le 
	\sum_{i' = 1}^p \l\vert \omega_{ii'} \mathbb I_{\{\vert \omega_{ii'} \vert < 2\varphi\}} \r\vert
	+
	\\
	& \, \sum_{i' = 1}^p \l\vert \l(\wh\omega_{s, e, ii'} - \omega_{ii'}\r) 
	\mathbb I_{\{\vert \wh\omega_{s, e, ii'} \vert \ge 2\varphi\}}
	+
	\omega_{ii'} \l(\mathbb I_{\{\vert \wh\omega_{s, e, ii'} \vert \ge 2\varphi\}} -
	\mathbb I_{\{\vert \omega_{ii'} \vert \ge 2\varphi\}} \r)\r\vert
	\\
	\le&  \, \mathfrak{s}_\varrho (2\varphi)^{1 - \varrho} + 
	\varphi \sum_{i' = 1}^p \mathbb I_{\{\vert \wh\omega_{s, e, ii'} \vert \ge 2\varphi\}}
	+ \sum_{i' = 1}^p \vert \omega_{ii'} \vert \l\vert
	\mathbb I_{\{\vert \wh\omega_{s, e, ii'} \vert \ge 2\varphi\}} -
	\mathbb I_{\{\vert \omega_{ii'} \vert \ge 2\varphi\}} \r\vert
	\\
	\le&  \, \mathfrak{s}_\varrho (2\varphi)^{1 - \varrho} + 
	\varphi \sum_{i' = 1}^p \mathbb I_{\{\vert \omega_{ii'} \vert \ge \varphi\}}
	+ \sum_{i' = 1}^p \vert \omega_{ii'} \vert
	\mathbb I_{\{\vert \vert \omega_{ii'} \vert  - 2\varphi \vert \le \vert \wh\omega_{s, e, ii'} - \omega_{ii'} \vert\}}
	\\
	\le&  \, \mathfrak{s}_\varrho \varphi^{1 - \varrho} \l(2^{1 - \varrho} + 1 + 3^{1 - \varrho} \r),
	\end{align*}
	from which we drive that
	\begin{align*}
	\l\Vert \wh{\bm\Omega}_{s, e} - \bm\Omega \r\Vert
	\le \l\Vert \wh{\bm\Omega}_{s, e} - \bm\Omega \r\Vert_\infty
	\le 12\mathfrak{s}_\varrho \varphi^{1 - \varrho}.
	\end{align*}
	This concludes the proof, since $\varphi = \vert \wh{\bm\Omega}_{s, e} - \bm\Omega \vert_\infty \le 4 C_\eta \Vert \bm\Omega \Vert_1^2 \psi_{n, p}/\sqrt{e-s}$ by Lemma~\ref{lem:clime}~\ref{lem:clime:one}. 
\end{proof}

\begin{lem}
	\label{lem:gamma:two}
	Suppose that Assumptions~\ref{assum:func:dep} and~\ref{assum:size} hold.
	Recall the notation $a^{\tE}_j$ and $b^{\tE}_j$ given in~\eqref{eq:ab}.
	Conditional on $\mc E^{\tO}_{n, p} \cap \mc S^{\tE}_{n, p}$, we have for all $j \in [q]$,
	\begin{align*}
	\sqrt{\frac{(\wh\cp^{\tE}_j - a^{\tE}_j)(b^{\tE}_j - \wh\cp^{\tE}_j)}{b^{\tE}_j - a^{\tE}_j}} 	\l\vert \wh{\bm\gamma}^{\tO}_{\wh\cp^{\tE}_j, b^{\tE}_j} - \wh{\bm\gamma}^{\tO}_{a^{\tE}_j, \wh\cp^{\tE}_j} - \wh{\bm\gamma}^{\tO}_{\cp^{\tE}_j, b^{\tE}_j} + \wh{\bm\gamma}^{\tO}_{a^{\tE}_j, \cp^{\tE}_j} \r\vert_\infty 
	\lesssim \frac{\Psi_j^2 \psi_{n, p}^2}{\sqrt{\Delta_j} \,\vert \bm\Sigma \bm\delta_j \vert_\infty}.
	\end{align*}
\end{lem}

\begin{proof} We consider the case $\wh\cp^{\tE}_j \le \cp_j$; the case with $\wh\cp^{\tE}_j > \cp_j$ can be handled analogously. 
	We drop the subscript $j$ where there is no confusion. 
	Then,
	\begin{align*}
	& \sqrt{\frac{(\wh\cp^{\tE} - a^{\tE})(b^{\tE} - \wh\cp^{\tE})}{b^{\tE} - a^{\tE}}} 
	\l\vert \wh{\bm\gamma}^{\tO}_{\wh\cp^{\tE}, b^{\tE}} - \wh{\bm\gamma}^{\tO}_{a^{\tE}, \wh\cp^{\tE}} - \wh{\bm\gamma}^{\tO}_{\cp, b^{\tE}} + \wh{\bm\gamma}^{\tO}_{a^{\tE}, \cp} \r\vert_\infty
	\\
	\le & \,
	\frac{\vert \wh\cp^{\tE} - \cp \vert \sqrt{\wh\cp^{\tE} - a^{\tE}}}{(b^{\tE} - \cp)\sqrt{(b^{\tE} - \wh\cp^{\tE})(b^{\tE} - a^{\tE})}} 
	\l\vert \sum_{t = \cp + 1}^{b^{\tE}} \l[ \mbf x_t^{\tO} \vep_t^{\tO} + \l(\mbf x_t^{\tO}{\mbf x_t^{\tO}}^\top - \bm\Sigma\r) \bm\beta_j \r] \r\vert_\infty
	\\
	& + \frac{\vert \wh\cp^{\tE} - \cp \vert \sqrt{ b^{\tE} - \wh\cp^{\tE} }}{(\cp - a^{\tE})\sqrt{(\wh\cp^{\tE} - a^{\tE})(b^{\tE} - a^{\tE})}} 
	\l\vert \sum_{t = a^{\tE} + 1}^{\wh\cp^{\tE}} \l[ \mbf x_t^{\tO} \vep_t^{\tO} + \l(\mbf x_t^{\tO}{\mbf x_t^{\tO}}^\top - \bm\Sigma\r) \bm\beta_{j - 1} \r] \r\vert_\infty
	\\
	& + \l( \sqrt{\frac{\wh\cp^{\tE} - a^{\tE}}{(b^{\tE} - \wh\cp^{\tE})(b^{\tE} - a^{\tE})}} + \frac{\sqrt{(\wh\cp^{\tE} - a^{\tE}) (b^{\tE} - \wh\cp^{\tE})}}{(\cp - a^{\tE})\sqrt{ b^{\tE} - a^{\tE}}} \r) 	\l\vert \sum_{t = \wh\cp^{\tE} + 1}^{\cp} \l[ \mbf x_t^{\tO} \vep_t^{\tO} + \l(\mbf x_t^{\tO}{\mbf x_t^{\tO}}^\top - \bm\Sigma\r) \bm\beta_{j - 1} \r] \r\vert_\infty 
	\\
	& + \frac{\vert \wh\cp^{\tE} - \cp \vert\sqrt{\wh\cp^{\tE} - a^{\tE}}}{\sqrt{(b^{\tE} - \wh\cp^{\tE})(b^{\tE} - a^{\tE})}} 
	\vert \bm\Sigma \bm\delta_j \vert_\infty
	=: T_1 + T_2 + T_3 + T_4.
	\end{align*}
	By Lemma~\ref{lem:bound}, \cref{assum:size} and and~\eqref{eq:ab:prop}, on $\mc E^{\tO}_{n, p} \cap \mc S^{\tE}_{n, p}$,
	\begin{align*}
	T_1 \le \frac{c_1 \Psi_j^2 \psi_{n, p}^2}{\vert \bm\Sigma \bm\delta_j \vert_\infty^2 \sqrt{(b^{\tE} - \wh\cp^{\tE})(b^{\tE} - \cp)}} \cdot \sqrt{\frac{\wh\cp^{\tE} - a^{\tE}}{b^{\tE} - a^{\tE}}} \cdot C_0 \Psi_j \psi_{n, p}
	\lesssim \frac{\Psi_j^3 \psi_{n, p}^3}{\vert \bm\Sigma \bm\delta_j \vert_\infty^2 \Delta_j}
	\lesssim \frac{\Psi_j^2 \psi_{n, p}^2}{\sqrt{\Delta_j} \,\vert \bm\Sigma \bm\delta_j \vert_\infty}. 
	\end{align*}
	The term $T_2$ is similarly bounded, while
	\begin{align*}
	T_3 \le 
	\frac{\sqrt{c_1}\Psi_j\psi_{n,p}}{\vert \bm\Sigma \bm\delta_j \vert_\infty}
	\l( \frac{1}{\sqrt{b^{\tE} - \wh\cp^{\tE}}} + \frac{\sqrt{\wh\cp^{\tE} - a^{\tE}}}{\cp - a^{\tE}} \r) \cdot C_0 \Psi_j \psi_{n, p}
	\lesssim \frac{\Psi_j^2 \psi_{n, p}^2}{\sqrt{\Delta_j} \,\vert \bm\Sigma \bm\delta_j \vert_\infty}.
	\end{align*}
	Finally, on $\mc E^{\tO}_{n, p} \cap \mc S^{\tE}_{n, p}$,
	\begin{align*}
	T_4 \le \frac{c_1 \Psi_j^2 \psi_{n, p}^2}{\vert \bm\Sigma \bm\delta_j \vert_\infty \sqrt{b^{\tE} - \wh\cp^{\tE}}} \lesssim \frac{ \Psi_j^2 \psi_{n, p}^2}{\vert \bm\Sigma \bm\delta_j \vert_\infty\sqrt{\Delta_j}},
	\end{align*}
	which completes the proof.
\end{proof}

\begin{lem}[Comparison of two Gaussian distributions]\label{lem:gausscmp}
	Let $\bm Z \sim \mathcal{N}_p(\bm 0, \bm \Sigma)$ and $\bm Z' \sim \mathcal{N}_p(\bm 0, \bm \Sigma')$ be centered $p$-dimensional Gaussian vectors. Assume that  the smallest singular (or eigen) value of $\bm \Sigma = (\sigma_{ii'})_{i,i'\in [p]}$, denoted by $\sigma_*$, is strictly positive. Then:
	\begin{align*}
	\sup_{A \in \mathfrak{R}}\l\vert\p(\bm Z \in A) - \p(\bm Z' \in A)\r\vert \;\le\; C\log(p)\frac{\vert\bm\Sigma - \bm\Sigma'\vert_{\infty}}{\sigma_*}\l(\l\vert\log\l(\frac{\vert\bm\Sigma - \bm\Sigma'\vert_{\infty}\min_{i\in[p]}\sqrt{\sigma_{ii}}}{\sigma_*\max_{i\in[p]}\sqrt{\sigma_{ii}}}\r)\r\vert\,\vee\,1\r),
	\end{align*}
	where $\mathfrak{R}$ is the set of rectangles in $\mathbb R^p$, and $C \in (0, \infty)$ is a universal constant. 
\end{lem}

\begin{proof}
	Note that $\bm\Sigma'$ is a Stein kernel for $\bm Z'$. This lemma follows then from Theorem~1.1 in \citet{FaKo21}.
\end{proof}

We require a version of Theorem~2.1 in \citet{CCK23}.
\begin{lem}[Gaussian approximation]\label{lem:gaussapp}
	Let $\mbf Z_1, \ldots, \mbf Z_n$ be a sequence of centered independent random vectors in $\mathbb R^p$. Let $\mbf T := \sum_{i=1}^n\mbf Z_i/\sqrt{n}$ with $\mbf \Sigma_{\mbf T} := \E (\mbf T\mbf T^\top)$, and $\mbf W \sim \mathcal{N}_p (\mbf 0, \, \mbf \Sigma)$. Assume that the smallest singular (or eigen) value of $\mbf \Sigma$, denoted by $\sigma_*$, is strictly positive. Then, for all $\psi > 0$,
	\begin{multline*}
	\sup_{A \in \mathfrak{R}}\l\vert\p(\mbf T \in A) - \p(\mbf W \in A)\r\vert  \le C\log (n) \log(p) \frac{\Vert\mbf\Sigma\Vert}{n\sigma_*^2} \Biggl(n\vert\mbf\Sigma -\mbf\Sigma_{\mbf T}\vert_{\infty}\\
	+  \sqrt{\log (p)\max_{j \in [p]}\sum_{i = 1}^n \E(\vert Z_{ij}\vert^4)} + \log(p)\sqrt{\E\l(\max_{j \in [p]}\max_{i \in [n]}\vert Z_{ij}\vert^4\r)} \\
	+ \sqrt{\frac{n \log(p n)}{\log (n)}\max_{i \in [n]}\E\l(\vert\mbf Z_i\vert_{\infty}^4 \mathbb{I}_{\{\vert\mbf Z_i\vert_{\infty} > \psi\}}\r)}\Biggr) + \frac{C\psi}{\sigma_*}\sqrt{\frac{\log^3(p)\Vert\mbf\Sigma\Vert}{n}},
	\end{multline*}
	where $\mathfrak{R}$ is the set of rectangles in $\mathbb R^p$ and $C \in (0, \infty)$ is a universal constant. 
\end{lem}

\begin{proof}
	Let $d_1, \ldots, d_p$ be diagonal entries of $\mbf \Sigma$ and $\mbf D := \mathrm{diag} (d_1, \ldots, d_p)$. The assumption of $\sigma_* >0$ implies that $d_i >0$ and $\mbf D$ is invertible. Thus, we can define $\wt{\mbf T} := \mbf D^{-1/2}\mbf T$ and $\wt{\mbf W} := \mbf D^{-1/2}\mbf W$. Then we have 
	\begin{align*}
	\wt{\mbf \Sigma}_{\wt{\mbf T}} := \E\l(\wt{\mbf T}\wt{\mbf T}^{\top}\r) = \mbf D^{-1/2}\mbf \Sigma_{\mbf T}\mbf D^{-1/2}\quad \text{and} \quad\wt{\mbf \Sigma} := \E\l(\wt{\mbf W}\wt{\mbf W}^\top\r) = \mbf D^{-1/2}\mbf \Sigma\mbf D^{-1/2}.
	\end{align*}
	Since the diagonal entries of $\wt{\mbf \Sigma}$ are all ones, we can apply Theorem~2.1 in \citet{CCK23} and obtain
	\begin{align}\label{eq:gauss_approx}
	&\sup_{A \in \mathfrak{R}}\l\vert\p(\mbf T \in A) - \p(\mbf W \in A)\r\vert  \nn\\
	= \; &\sup_{A \in \mathfrak{R}}\l\vert\p(\wt{\mbf T} \in A) - \p(\wt{\mbf W} \in A)\r\vert\nn \\
	\le\; &\frac{C\log (n )\log (p)}{\wt\sigma_*}\l\vert\wt{\mbf \Sigma}-\wt{\mbf \Sigma}_{\wt{\mbf T}}\r\vert_{\infty} + \frac{C\log(n)\bigl(\log(p)\bigr)^{3/2}}{n\wt\sigma_*}\sqrt{\max_{j \in [p]}\sum_{i =1}^n \E\l(\frac{\vert Z_{ij}\vert^4}{d_j^2}\r)}\nn\\
	& {}\qquad + \frac{C\log(n)\bigl(\log(p)\bigr)^2}{n\wt\sigma_*}\sqrt{\E\l(\max_{j \in [p]}\max_{i \in [n]}\frac{\vert Z_{ij}\vert^4}{d_j^2}\r)}\nn\\
	& {}\qquad + \frac{C\log(p)}{\wt\sigma_*}\sqrt{\frac{\log(n)\log(pn)}{n}\max_{i \in [n]}\E\l(\vert \mbf D^{-1/2}\mbf Z_{i}\vert_{\infty}^4 \mathbb{I}_{\{\vert \mbf D^{-1/2}\mbf Z_{i}\vert_{\infty} > \wt\psi\}}\r)}\nn\\
	& {}\qquad + \frac{C\wt\psi\bigl(\log(p)\bigr)^{3/2}}{\sqrt{n\wt\sigma_*}},
	\end{align}
	where $\wt\sigma_*$ is the smallest singular value of $\wt{\mbf \Sigma}$, and $\wt\psi >0$ is arbitrary. 
	Note that 
	\begin{align*}
	\wt\sigma_* & = \min_{\mbf x}\frac{\mbf x^\top \wt{\mbf \Sigma}\mbf x}{\mbf x^\top\mbf x} = \min_{\mbf x} \frac{\mbf x^\top \wt{\mbf \Sigma}\mbf x}{\mbf x^\top\mbf D^{-1}\mbf x}\frac{\mbf x^\top\mbf D^{-1}\mbf x}{\mbf x^\top\mbf x}\\
	& \ge \min_{\mbf x} \frac{\mbf x^\top \wt{\mbf \Sigma}\mbf x}{\mbf x^\top\mbf D^{-1}\mbf x}\min_{\mbf x}\frac{\mbf x^\top\mbf D^{-1}\mbf x}{\mbf x^\top\mbf x} 
	= \sigma_*\frac{1}{\max_{j \in [p]}d_j}\ge \frac{\sigma_*}{\Vert\mbf \Sigma\Vert},
	\end{align*}
	and
	\begin{align*}
	\l\vert\wt{\mbf \Sigma}-\wt{\mbf \Sigma}_{\wt{\mbf T}}\r\vert_{\infty} & = \l\vert\mbf D^{-1/2}{\mbf \Sigma}\mbf D^{-1/2}-\mbf D^{-1/2}{\mbf \Sigma}_{\mbf T}\mbf D^{-1/2}\r\vert_{\infty} \\
	&\le \frac{1}{\min_{j \in [p]}d_j}\l\vert\mbf\Sigma -\mbf\Sigma_{\mbf T}\r\vert_\infty \le \frac{1}{\sigma_*}\l\vert\mbf\Sigma -\mbf\Sigma_{\mbf T}\r\vert_\infty.
	\end{align*}
	Note also that $\vert Z_{ij}\vert^4/d_j^2\le \vert Z_{ij}\vert^4/\min_{j \in [p]}d_j^2\le \vert Z_{ij}\vert^4/\sigma_*^2$ and 
	\begin{multline*}
	\E\l(\vert \mbf D^{-1/2}\mbf Z_{i}\vert_{\infty}^4 \mathbb{I}_{\{\vert \mbf D^{-1/2}\mbf Z_{i}\vert_{\infty} > \wt\psi\}}\r) = \E\l(\max_{j \in [p]}\frac{\vert Z_{ij}\vert^4}{d_j^2} \mathbb{I}_{\{\max_{j \in [p]}d_j^{-1/2}\vert Z_{ij}\vert > \wt\psi\}}\r) \\
	\le \frac{1}{\min_{j \in [p]}d_j^2}\E\l(\max_{j \in [p]}\vert Z_{ij}\vert^4 \mathbb{I}_{\l\{\vert Z_{ij}\vert > \wt\psi\sqrt{\min_{j \in [p]}d_j}\r\}}\r) \le \frac{1}{\sigma_*^2}\E\l(\max_{j \in [p]}\vert Z_{ij}\vert^4 \mathbb{I}_{\l\{\vert Z_{ij}\vert > \wt\psi\sqrt{\sigma_*}\r\}}\r).
	\end{multline*}
	Setting $\wt\psi = \psi/\sqrt{\sigma_*}$ and using the above estimates, we obtain the assertion from \eqref{eq:gauss_approx}. 
\end{proof}

\begin{lem}
	\label{lem:assum:func:two} 
	Suppose that Assumption~\ref{assum:func:dep:two} holds. Then,
	\begin{align*}
	\sup_{\mbf a, \mbf b \in \mathbb{B}_2(1)} \sup_{\nu \ge 2} \nu^{-1} \l\Vert \mbf a^\top \mbf Z_t\mbf Z_t^\top \mbf b - \E\l( \mbf a^\top \mbf Z_t\mbf Z_t^\top \mbf b \r) \r\Vert_\nu \le 3\Xi.
	\end{align*}
\end{lem}

\begin{proof}
	Since $\max_{\mbf a \in \mathbb{B}_2(1)} \sup_{\nu \ge 2} \nu^{-1/2} \Vert \mbf a^\top \mbf Z_t \Vert_\nu \le \Xi^{1/2}$,
	\begin{align*}
	\nu^{-1} \l\Vert \mbf a^\top \mbf Z_t (\mbf Z_t)^\top \mbf b \r\Vert_\nu 
	\le 2 (2\nu)^{-1/2} \l\Vert \mbf a^\top \mbf Z_t \r\Vert_{2\nu} \cdot (2\nu)^{-1/2} \l\Vert \mbf b^\top \mbf Z_t \r\Vert_{2\nu}
	\le 2 \Xi, 
	\end{align*}
	where the first inequality follows from H\"older's inequality.
	Combining this with that 
	\begin{align*}
	\vert \E(\mbf a^\top \mbf Z_t \mbf Z_t^\top \mbf b) \vert \le \sqrt{\Vert a^\top \mbf Z_t \Vert_2^2 \Vert b^\top \mbf Z_t \Vert_2^2 } \le 2\Xi, 
	\end{align*}
	the claim follows.
\end{proof}

\clearpage

\section{Additional numerical results}

\subsection{Empirical studies in Section~\ref{sec:size:change}}
\label{sec:sim:intro}

We detail the simulation setup of \cref{fig:intro}, which is designed to examine the impact of $\vert\bm\Sigma\bm\delta_j\vert_{\infty}$, $\vert\bm\delta_j\vert_0$ and $\vert\bm\delta_j\vert_2$ on the detectability of a change point.
For $d \in [p]$, let $\bm a_d$ be a uniformly distributed random vector on the unit sphere in $\R^d$, and define $\bm\delta_1 = \sqrt{2} \cdot (1, \bm a_{\mathfrak s - 1}^\top, \bm 0_{p-\mathfrak s}^\top)^\top$ with $\bm 0_{p-\mathfrak s}\in\R^{p-\mathfrak s}$ being the zero vector. It holds almost surely that $\vert\bm\delta_1 \vert_0 = \mathfrak s$ and $\vert\bm\delta_1 \vert_2 = 2$. 
Next, we consider an orthogonal matrix $\mbf U \in \R^{p\times p}$ such that its first column is $\bm\delta_1/2$. In simulation, we generate such $\mbf U$ by applying the Gram--Schmidt orthonormalisation to the collection of $\bm\delta_1$ and $p-1$ standard normal vectors in $\R^p$. Then the observations are generated under the model~\eqref{eq:model} where we set $q = 1$, $\cp_1 = 100$, $n = 300$, $p = 200$, $\bm\beta_0 = \bm\delta_1/2$, $\bm\beta_1 = -\bm\delta_1/2$ and $\mbf x_t \sim_{\iid} \mc N_p(\mbf 0, \bm\Sigma)$. We consider two choices of $\bm\Sigma$, i.e.\
$$
\bm\Sigma_{1} = \frac{1}{d_1}\mbf U\; \mathrm{diag}(d_1,d_2,\ldots,d_p) \;\mbf U^\top \quad\text{ or }\quad
\bm\Sigma_{2} = \frac{1}{d_1}\mbf U\; \mathrm{diag}(d_p,d_{p-1},\ldots,d_1)\; \mbf U^\top
$$
which correspond to the top and the bottom rows of \cref{fig:intro}, respectively. Here, $d_1 > \cdots > d_p$ denote the eigenvalues of the Toeplitz matrix $\left(0.6^{\vert i - j \vert}\right)_{i,j=1}^p \in \R^{p\times p}$. Note that $\vert\bm\Sigma_1\bm\delta_1\vert_{\infty} = \vert\bm\delta_1\vert_{\infty} \approx 1.414$ and  $\vert\bm\Sigma_2\bm\delta_1\vert_{\infty} = \vert d_1^{-1}d_p\bm\delta_1\vert_{\infty} \approx 0.088$.

\clearpage

\subsection{Change point estimation}\label{sec:add:cp}

\subsubsection{Details of comparative methods}\label{sec:add:cmp:detail}

We describe the implementation and choices of tuning parameters for the methods in comparison study as follows. 
\begin{itemize}
	\item 
	MOSEG \citep{cho2022high}: The implementation is provided in the R package \texttt{moseg} on GitHub (\url{https://github.com/Dom-Owens-UoB/moseg}, version~0.1.0). In the single change point scenarios, we use the function \texttt{moseg} with input argument \texttt{n.cps = 1}, and use the output \texttt{refined.cps}, which has a better empirical performance than the output \texttt{cps}. 
	In the multiple change point scenarios, we use the function \texttt{moseg.ms.cv}, which estimates multiple change points using multiple bandwidths and automatically selects the number of change point via cross validation. In these two functions, we set the input argument \texttt{ncores = 1} when recording the runtime, and the rest of the input arguments are set to their default values. 
	
	\item 
	CHARCOAL \citep{gao2022sparse}: The implementation is provided in the R package \texttt{charcoal} on GitHub (\url{https://github.com/gaofengnan/charcoal/}, version~0.13). In the single change point scenarios, we use the function \texttt{cpreg} with default input arguments, which is designed for estimating a single change point \citep[Algorithm~2]{gao2022sparse}. In the multiple change point scenarios, we use \texttt{not\_cpreg} with default input arguments, where the threshold for the narrowest-over-threshold is computed via Monte Carlo repetitions based on the null model (via function \texttt{getTestThreshold}). 
	
	\item 
	DPDU \citep{xu2022change}: The implementation is provided in the R package \texttt{changepoints} on CRAN (\url{https://CRAN.R-project.org/package=changepoints}, version~1.1.0). We first use the function \texttt{CV.search.DPDU.regression} with input arguments \texttt{lambda\_set = c(0.01, 0.1, 1, 2)} and \texttt{zeta\_set = c(10, 15, 20)}, according to the  examples in the R package. The output that attains the minimal test error (\texttt{test\_error}) is further improved by local refinement, via the function \texttt{local.refine.DPDU.regression}. This refined output is used as the final estimate. There is no option to set the number of change points. Thus, in the single change point scenarios, towards a fair comparison, we select among the estimated change points (including $0$ and $n$) the one that is closest to the truth as the final estimate. In this case, the localisation error $\vert\wh\cp_1-\cp_1\vert$ will be upper bounded by $\min\{\cp_1, n-\cp_1\}$, see \cref{fig:can_err,fig:toep_err}.
	
	\item 
	VPWBS \citep{wang2021statistically}: The implementation is provided in the R codes on GitHub (\url{https://github.com/darenwang/VPBS}, version of May 26, 2021). We tune the parameters involved using a cross validation procedure (via the function \texttt{vpcusum}) provided in the example codes on GitHub. There is no option to set the number of change points. Thus similar to DPDU, in the single change point scenarios, we select the estimated change point (including $0$ and $n$) that is closest to the truth as the final estimate, and the resulting localisation error $\vert\wh\cp_1-\cp_1\vert$ is bounded by $\min\{\cp_1, n-\cp_1\}$ from above, see \cref{fig:can_err,fig:toep_err}.  
\end{itemize}

\subsubsection{Additional results under (M1)--(M2)}
\label{sec:add:cp:m1}

Figures~\ref{fig:can_err:two}--\ref{fig:toep_err:two} report the localisation errors under the scenarios (M1)--(M2) when $\cp_1 = 150$, complementing Figures~\ref{fig:can_err} and~\ref{fig:toep_err} obtained with $\cp_1 = 75$.

\begin{figure}[h!t!b!]
	\centering
	\includegraphics[width = \textwidth]{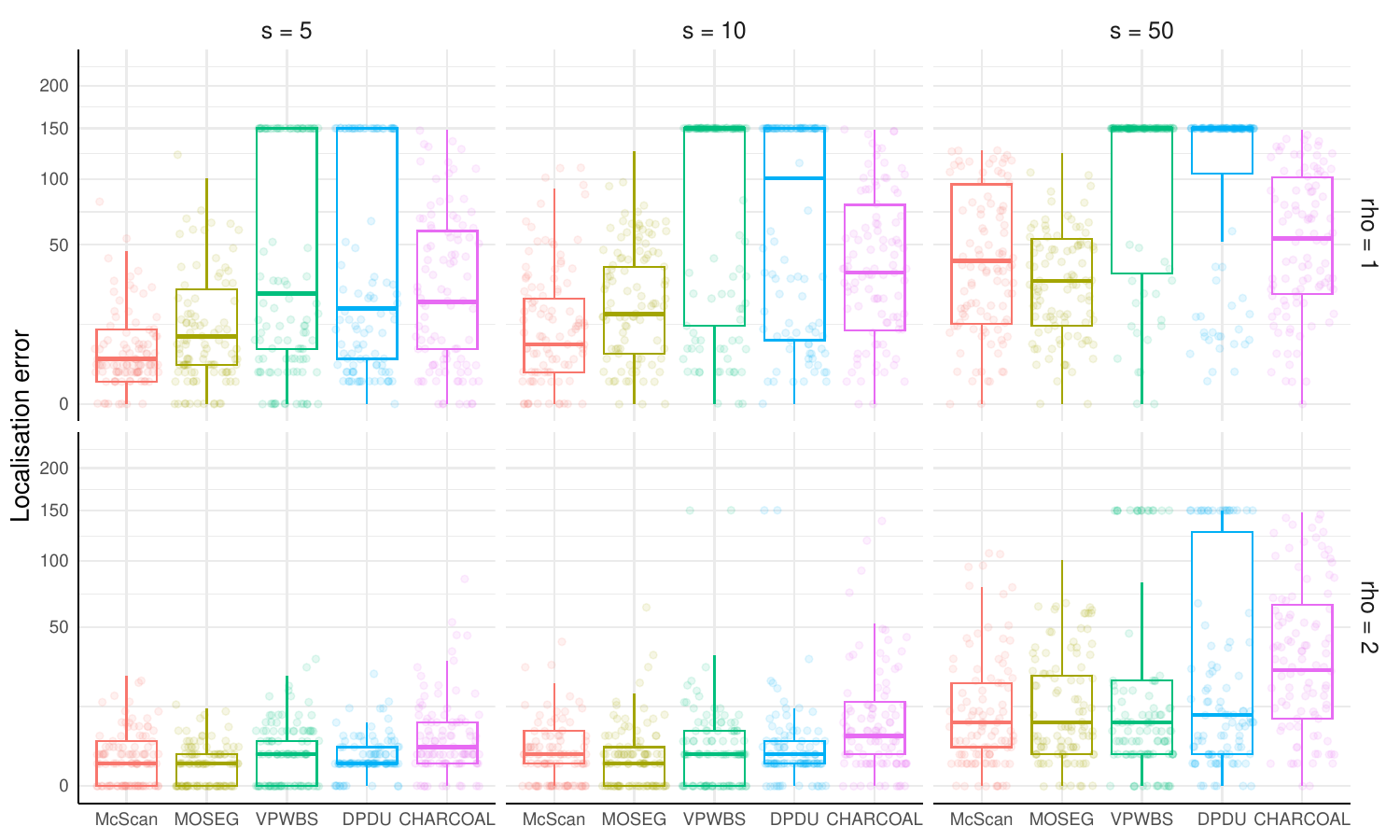}
	\caption{Localisation errors in (M1). For each method, the localisation errors $\vert\wh{\cp}_1 - \cp_1\vert$ over 100 repetitions are jittered in dots with a low intensity, and the overall performance is summarised as a boxplot (obtained from 1000 repetitions, outliers not shown), for $\cp_1 = 150$ and varying $\rho \in \{1, 2\}$. The $y$-axis is in the square root scale.}
	\label{fig:can_err:two}
\end{figure}

\begin{figure}[h!t!b!]
	\centering
	\includegraphics[width = \textwidth]{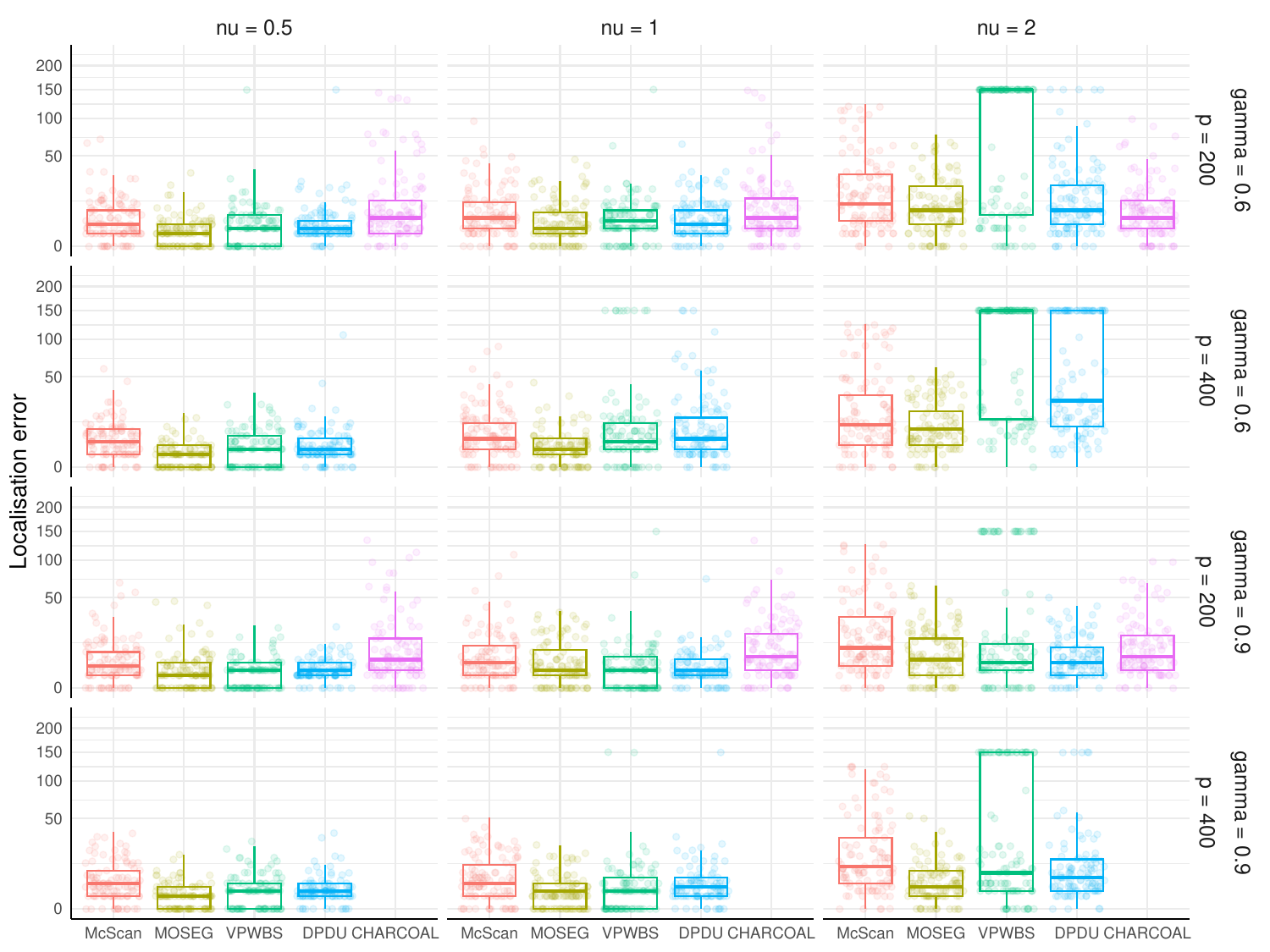}
	\caption{Localisation errors in (M2).  For each method, the localisation errors $\vert\wh{\cp}_1 - \cp_1\vert$ over 100 repetitions are jittered in dots with a low intensity, and the overall performance is summarised as a boxplot (obtained from 1000 repetitions, outliers not shown), for $\cp_1 = 150$ and varying $p \in \{200, 400\}$ and $\gamma \in \{0.6, 0.9\}$. The $y$-axis is in the square root scale. There is no result from CHARCOAL for $p = 400$ as it is not applicable when $p > n = 300$. }
	\label{fig:toep_err:two}
\end{figure}

\clearpage

\subsubsection{Additional results under (M3)}
\label{sec:add:cp:m3}

We complement the comparative simulation study under the scenario (M3) in \cref{sec:sim:multi} by reporting further quantitative performance evaluation measures, including the Hausdorff distances between estimated and true change points, the V-measure \citep{rosenberg2007v} and the number of estimated change points, in \cref{fig:mcp_hausdorff,fig:mcp_vmeas,fig:mcp_numcp}. The V-measure is an entropy-based clustering measure, which takes values in $[0,1]$ with a larger value indicating a higher accuracy. Unlike the other clustering measures, such as (adjusted) Rand index \citep{rand1971objective,hubert1985comparing}, the V-measure satisfies all of the desirable  properties of clustering measures, proposed in \citet{dom2001information}, see \citet{rosenberg2007v} for further discussion.  

\begin{figure}[h!t!b!]
	\centering
	\includegraphics[width = \textwidth]{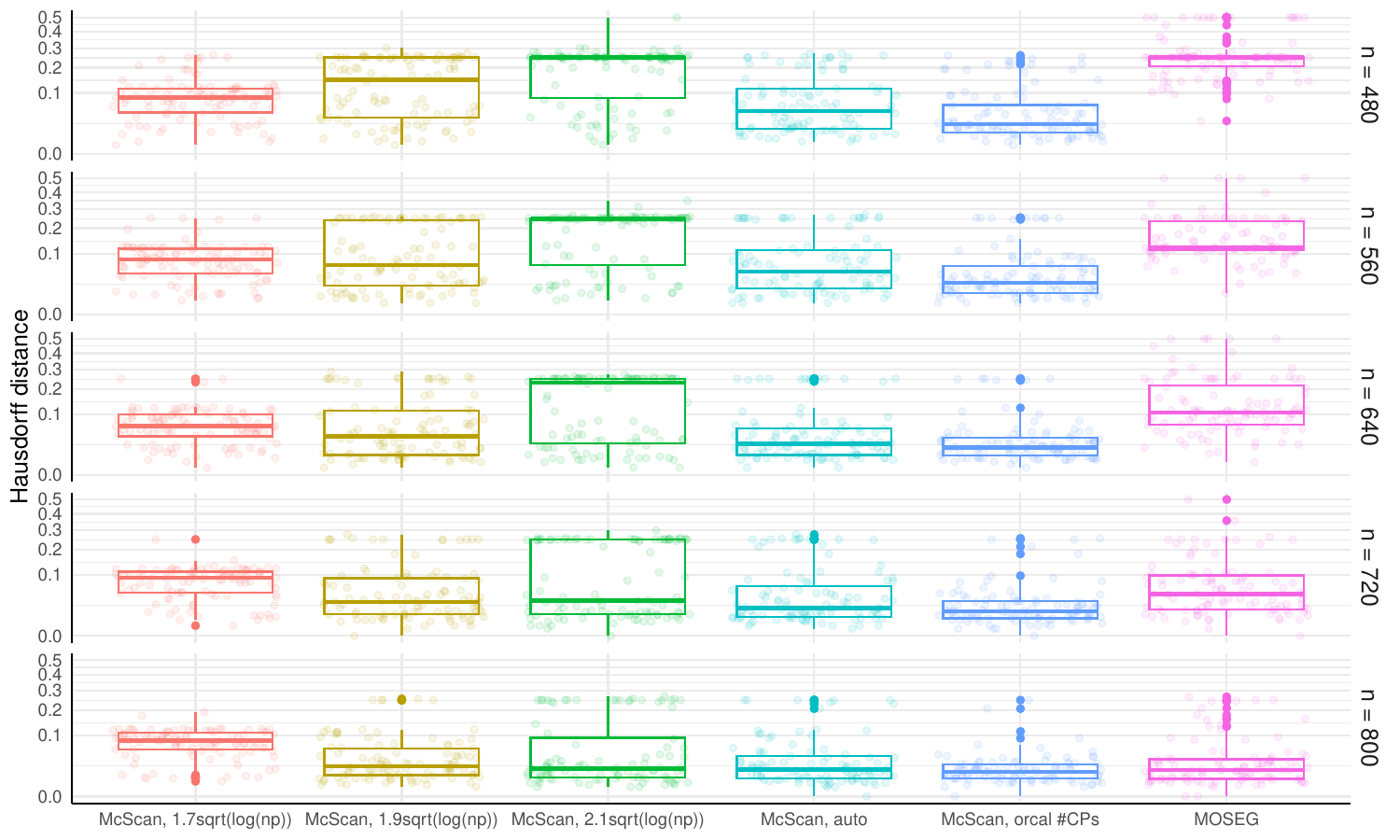}
	\caption{Hausdorff distances of scaled estimated change points $\bigl\{\wh\cp_j/n\,:\,j = 1,\ldots, \wh q\bigr\}$ and scaled true change points $\bigl\{\cp_j/n\,:\,j = 1,\ldots, q\bigr\}$ in (M3). The performance over 100 repetitions is summarised as a boxplot, and that of individual repetition is jittered in dots with a low intensity. The $y$-axis is in the square root scale. Abbreviations of methods are the same as in \cref{fig:mcp_location} in the main text.}
	\label{fig:mcp_hausdorff}
\end{figure}

\begin{figure}[h!t!b!]
	\centering
	\includegraphics[width = \textwidth]{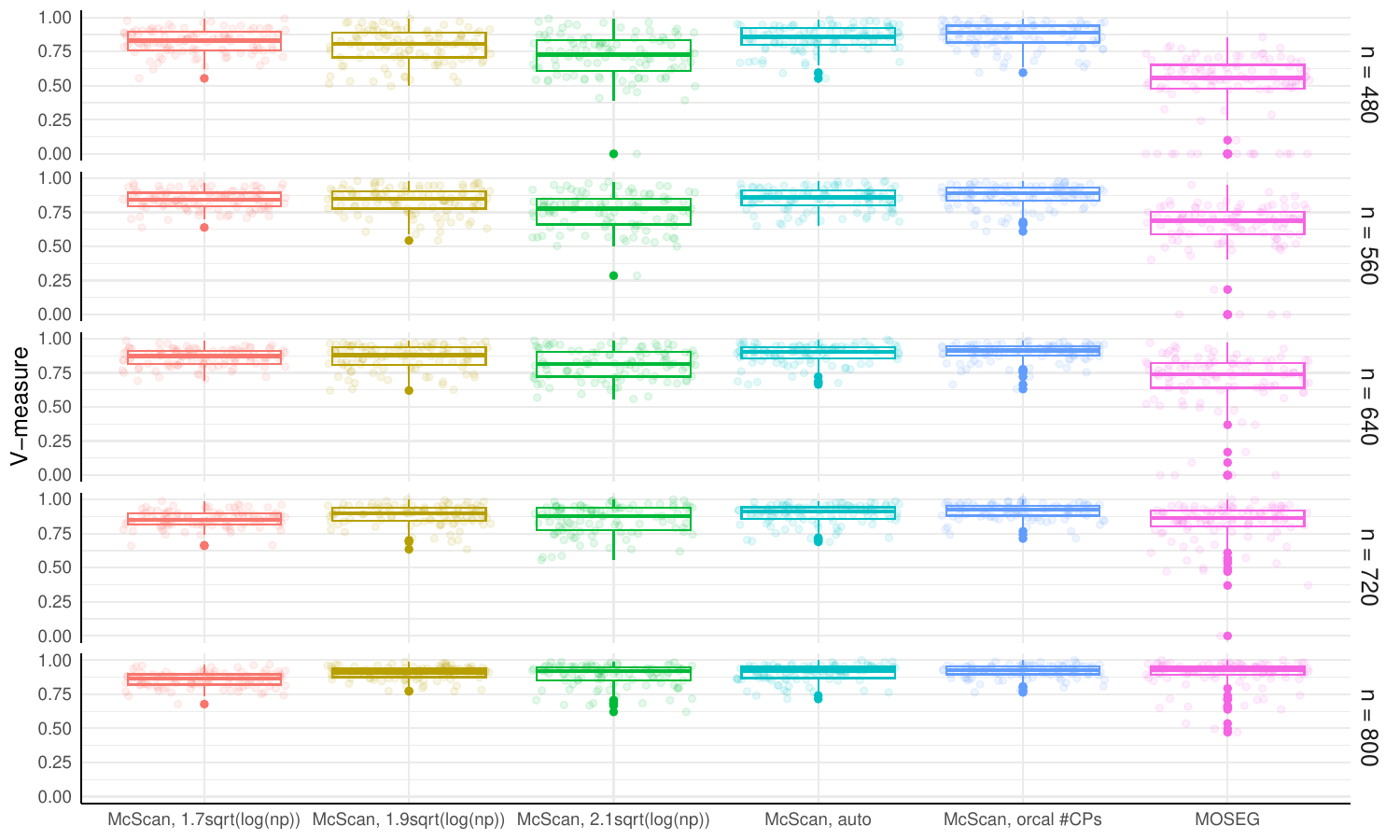}
	\caption{V-measures in (M3). The performance over 100 repetitions is summarised as a boxplot, and that of individual repetition is jittered in dots with a low intensity. Abbreviations of methods are the same as in \cref{fig:mcp_location} in the main text.}
	\label{fig:mcp_vmeas}
\end{figure}

\begin{figure}[h!t!b!]
	\centering
	\includegraphics[width = \textwidth]{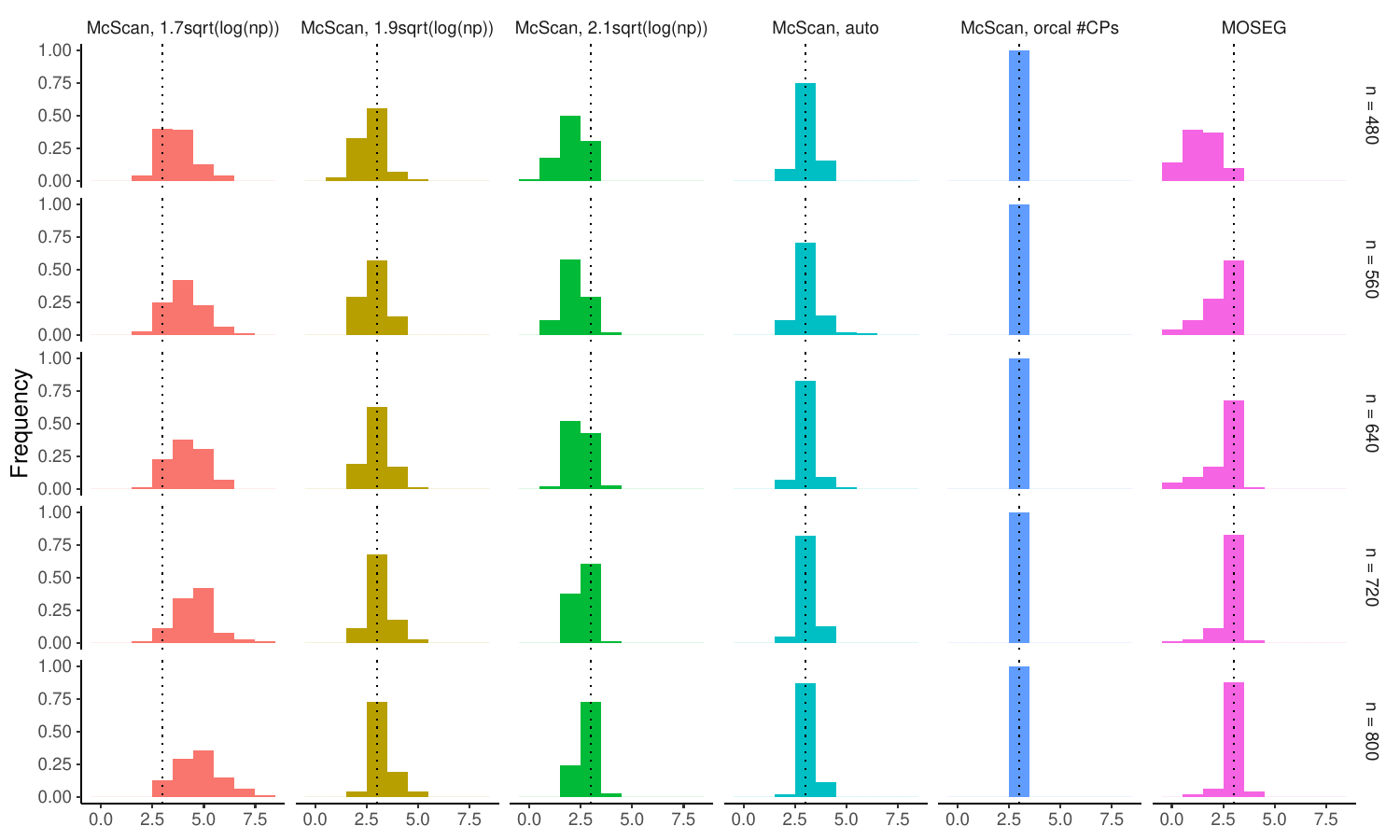}
	\caption{Estimated number of change points in (M3) over 100 repetitions. The true number of change points are marked by the vertical dotted lines. Abbreviations of methods are the same as in \cref{fig:mcp_location} in the main text.}
	\label{fig:mcp_numcp}
\end{figure}

\clearpage

\subsection{{Model mis-specification}}
\label{sec:sim:mask}
As an investigation of robustness to model mis-specification, we consider the scenarios with heterogeneous design, where $\Cov(\mbf x_t)$ undergoes either an abrupt shift or a gradual change. 

Consider particularly the single change model 
\begin{align*}
Y_t = \l\{ \begin{array}{ll}
\mbf x_t^\top \bm\beta_0 + \sigma^{(0)}_\vep \bm\vep_t & \text{for $1 \le t \le \cp_1$},  \\
\mbf x_t^\top \bm\beta_1 + \sigma^{(1)}_\vep \bm\vep_t & \text{for $\cp_1 + 1 \le t \le n$}, 
\end{array}\r.
\end{align*}
where $\mbf x_t \sim \mc N_p(\mbf 0, \bm\Sigma_t)$ and $\bm\vep_t \sim \mc N_p(\mbf 0, \mbf I_p)$ are independently distributed.
\begin{enumerate}[label = (M\arabic*)]
\setcounter{enumi}{3}
\item \label{m:four} We define
$\bm\Sigma_t = \bm\Sigma^{(0)}(\gamma) \cdot \mathbb{I}_{\{t \le \cp_\Sigma\}} + \bm\Sigma^{(1)}(\gamma) \cdot \mathbb{I}_{\{t > \cp_\Sigma\}} \in \R^{p \times p}$, where $\bm\Sigma^{(j)}(\gamma), \, j = 0, 1$, are block-diagonal matrices with $\bm\Sigma^{(0)}(\gamma) = \text{diag}(\bm\Gamma(\gamma), \mbf I_{p - \mathfrak{s}})$, $\bm\Sigma^{(1)}(\gamma) = \text{diag}(\mbf I_{p - \mathfrak{s}}, \bm\Gamma(\gamma))$ and $\bm\Gamma(\gamma) = [\gamma^{\vert i - j \vert}, \, 1 \le i, j \le \mathfrak{s}] \in \R^{\mathfrak{s} \times \mathfrak{s}}$.
We generate $\mbf b = (b_1, \ldots, b_p)^\top \in \R^p$ with $b_i \sim_{\iid} \{-0.3, 0.3\}$ for $i \in \{1, \mathfrak{s} + 1, \ldots, 2\mathfrak{s} - 2, p\}$ and $b_i = 0$ otherwise.
Then, we set $\sigma^{(j)}_\vep = (1 - \mbf b^\top (\bm\Sigma^{(j)}(\gamma))^{-1} \mbf b)^{-1/2}$ and $\bm\beta_j = (\sigma^{(j)}_\vep)^2 \cdot (\bm\Sigma^{(j)}(\gamma))^{-1} \mbf b$, $j = 0, 1$.
We consider $(n, \mathfrak{s}) = (300, 5)$ and $\gamma = 0.9$, and vary $p \in \{400, 900\}$ and $\cp_1, \cp_\Sigma \in \{75, 150, 225\}$.

\item \label{m:five} Using the same model parameters as in~\ref{m:four}, we now define \[\bm\Sigma_t = \frac{n - t}{n - 1}\cdot \bm\Sigma^{(0)}(0.3) +\frac{t - 1}{n - 1} \cdot \bm\Sigma^{(1)}(0.3).\]
Setting $(n, p, \mathfrak{s}) = (300, 900, 5)$, we vary $\cp_1 \in \{75, 150, 225\}$.
\end{enumerate}

Under \ref{m:four}, $\Cov(\mbf x_t)$ undergoes an abrupt shift at $\cp_\Sigma$, whereas under \ref{m:five}, it changes gradually over time. In both cases, the regression parameters remain sparse (with $\mathfrak{s}+1$ nonzero entries), a situation favoured by the existing change point methods that directly estimate the local regression parameter.
As in \cref{sec:sim:cp}, we compare McScan's performance to that of MOSEG \citep{cho2022high}, DPDU \citep{xu2022change} and VPWBS \citep{wang2021statistically}, applying all methods with prior knowledge of the number of change points to isolate change point estimation from model selection. For DPDU and VPWBS, we manually select the estimate nearest to $\cp_1$ as described in \cref{sec:sim:single}. CHARCOAL \citep{gao2022sparse} is excluded since $p > n$ (see \cref{sec:add:cmp:detail} for implementation details). DPDU and VPWBS are omitted when $p = 900$ due to high computational costs.

Figures~\ref{fig:m4:p400}--\ref{fig:m4:p900} display the scaled locations of the change point estimators under~\ref{m:four}.
When $\cp_1 = \cp_\Sigma$ (on the diagonal of the figures), McScan fails to detect the change.
Here, the change in $\bm\Sigma_t$ offsets the change in the regression parameters since by design, $\bm\Sigma^{(0)}(\gamma) \beta_0 = \bm\Sigma^{(1)}(\gamma) \beta_1$.
This poses a particularly challenging set-up for McScan scanning $\Cov(\mbf x_t, Y_t)$ to detect the change, which in turn represents the covariance weighted regression parameter.  
Notably, we observe that the competing methods that utilise the local parameter estimators also fail to detect $\cp_1$. This suggests that the simultaneous shifts in both $\bm\Sigma_t$ and the regression coefficients we have designed pose a difficulty in general for the change point problem.

\begin{figure}[h!t!b!p!]
    \centering
    \includegraphics[width = 1\linewidth]{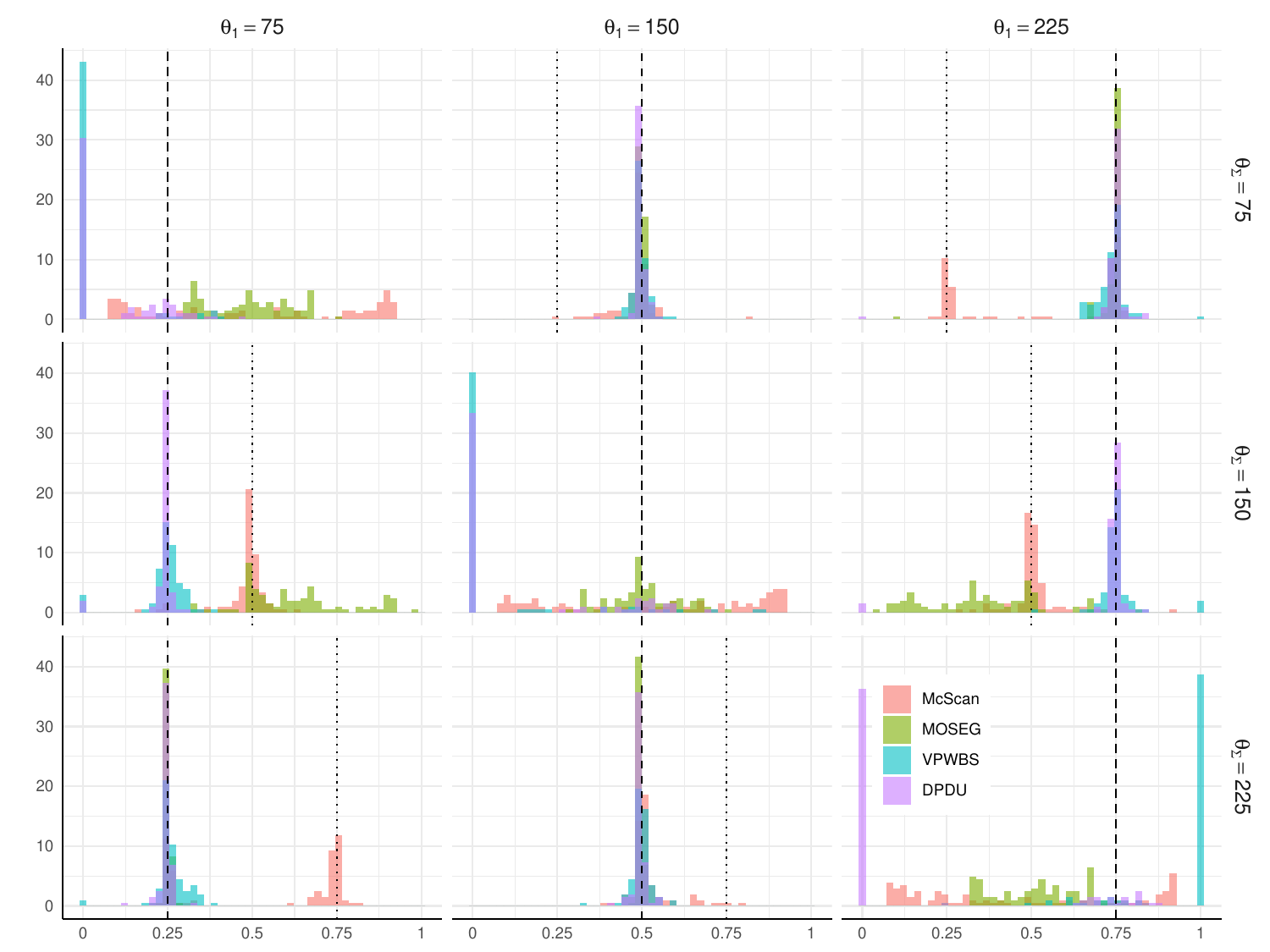}
    \caption{Scaled change point estimators ($\wh\cp_1/n$) returned by McScan, MOSEG, VPWBS and DPDU, under \ref{m:four} for varying $(\cp_1, \cp_\Sigma)$ when $p = 400$; for each setting, 100 repetitions are generated and each method is applied to detect a single change point. We denote $\cp_1$ by a broken line and $\cp_\Sigma$ by a dotted one in each plot. The histograms use a percentage scale.}
    \label{fig:m4:p400}
\end{figure}

\begin{figure}[h!t!b!p!]
    \centering
    \includegraphics[width = 1\linewidth]{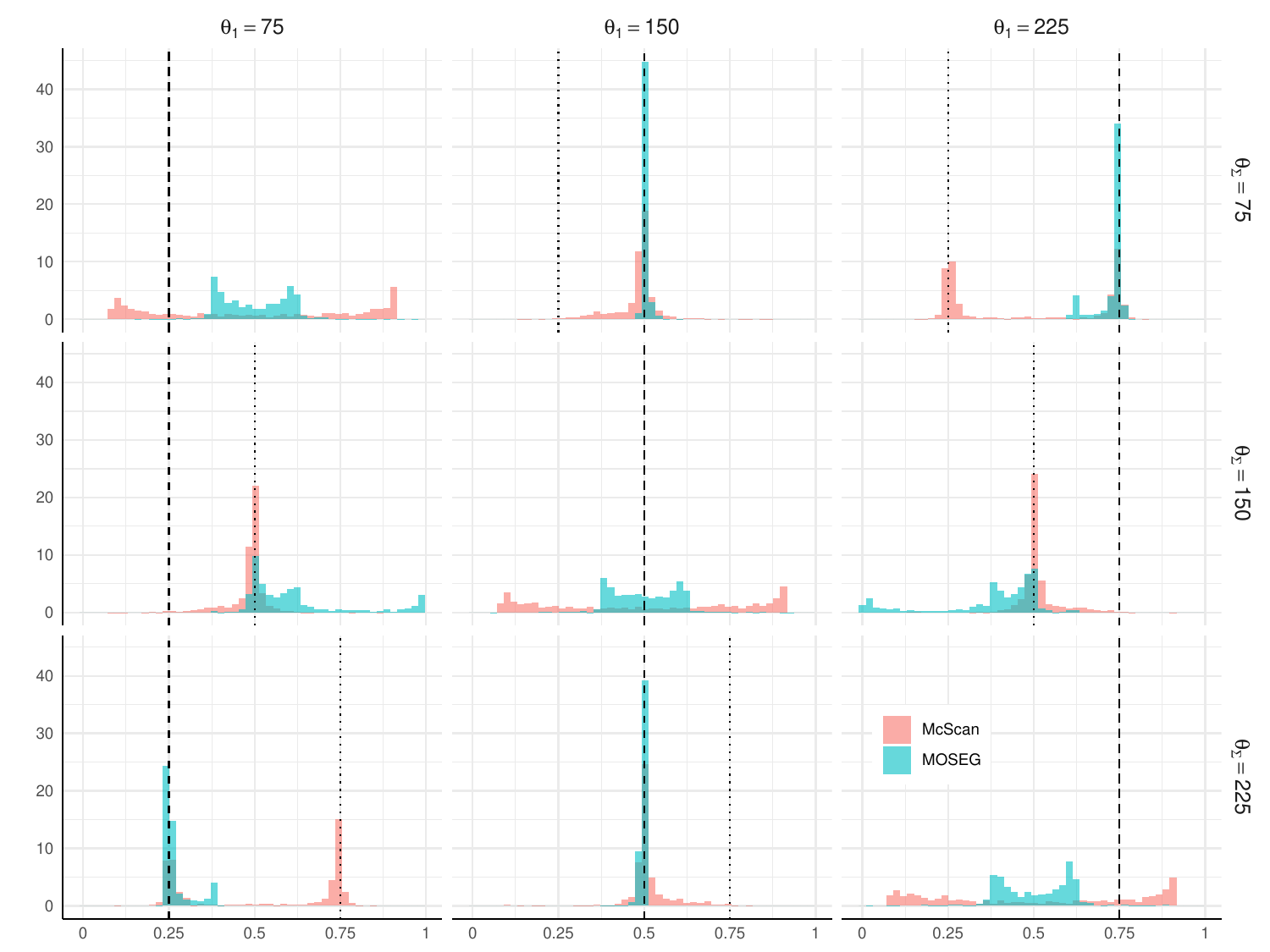}
    \caption{Scaled change point estimators ($\wh\cp_1/n$) returned by McScan and MOSEG, under \ref{m:four} for varying $(\cp_1, \cp_\Sigma)$ when $p = 900$; for each setting, 1000 repetitions are generated and each method is applied to detect a single change point. We denote $\cp_1$ by a broken line and $\cp_\Sigma$ by a dotted one in each plot. The histograms use a percentage scale.}
    \label{fig:m4:p900}
\end{figure}

In the off-diagonal plots,  McScan detects either $\cp_1$ or $\cp_\Sigma$ with high accuracy, often preferring the one closer to $n / 2 = 150$. 
When $\vert \cp_1 - 150 \vert = \vert \cp_\Sigma - 150 \vert$ and $\cp_1 \ne \cp_\Sigma$, the distribution of $\wh\cp_1$ returned by McScan is bi-modal, peaking at $\{\cp_1, \cp_\Sigma\}$.
Surprisingly, when $p = 900$, MOSEG prefers $\cp_\Sigma$ over $\cp_1$ when $\cp_\Sigma = 150$, suggesting that even those methods that are based on local regression parameter estimators, struggle to distinguish regression parameter shifts from covariance shifts.
Additionally, we have replicated the experiments reported in Figure~\ref{fig:m4:p900} with the number of change points set to be two for McScan, see Figure~\ref{fig:m4:p900:two}, which shows that McScan treats the two types of change points on an equal footing and detects them both with high accuracy.

In Figure~\ref{fig:m5}, we present the results for McScan and MOSEG under~\ref{m:five}. 
While generally outperformed by MOSEG, McScan successfully detects the change point despite the gradual evolution of $\Cov(\mbf x_t)$.

In summary, these studies show that the scenarios where $\Cov(\mbf x_t, Y_t)$ is identical pre- and post-change, adversely impact not only McScan but also the other change point methods that directly look for changes in the sparse regression parameters. 
We leave this interesting investigation into the detection lower bound and  optimality for the change point detection problem in a regression setting with time-varying $\Cov(\mbf x_t)$ for future research.

\begin{figure}[h!t!b!p!]
    \centering
    \includegraphics[width = 1\linewidth]{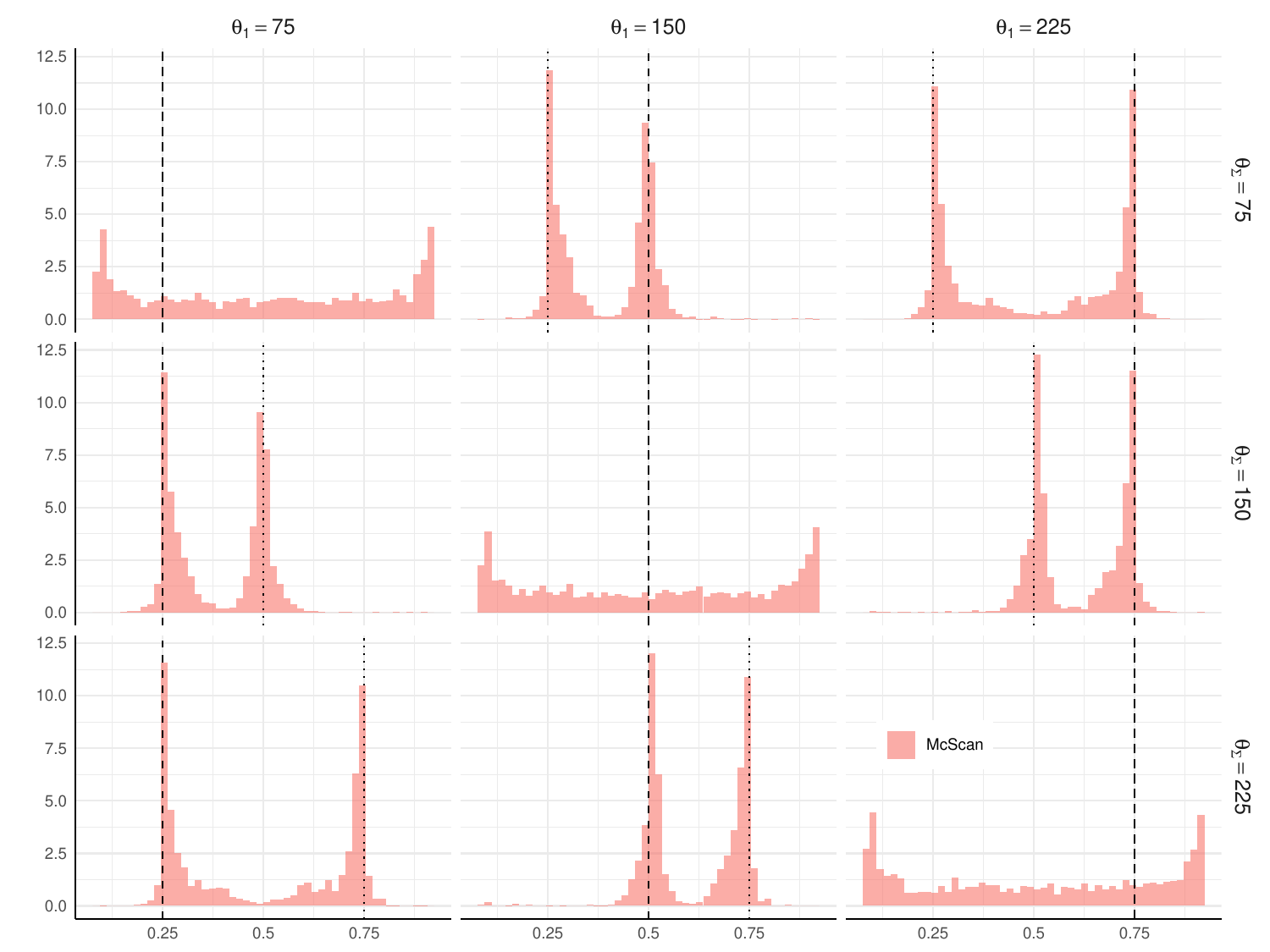}
    \caption{Scaled change point estimators ($\wh\cp_1/n$) returned by McScan applied to detect two change points, under \ref{m:four} for varying $(\cp_1, \cp_\Sigma)$ when $p = 900$, over 1000 repetitions. We denote $\cp_1$ by a broken line and $\cp_\Sigma$ by a dotted one in each plot. The histograms use a percentage scale.}
    \label{fig:m4:p900:two}
\end{figure}

\begin{figure}[h!t!b!p!]
    \centering
    \includegraphics[width = 1\linewidth]{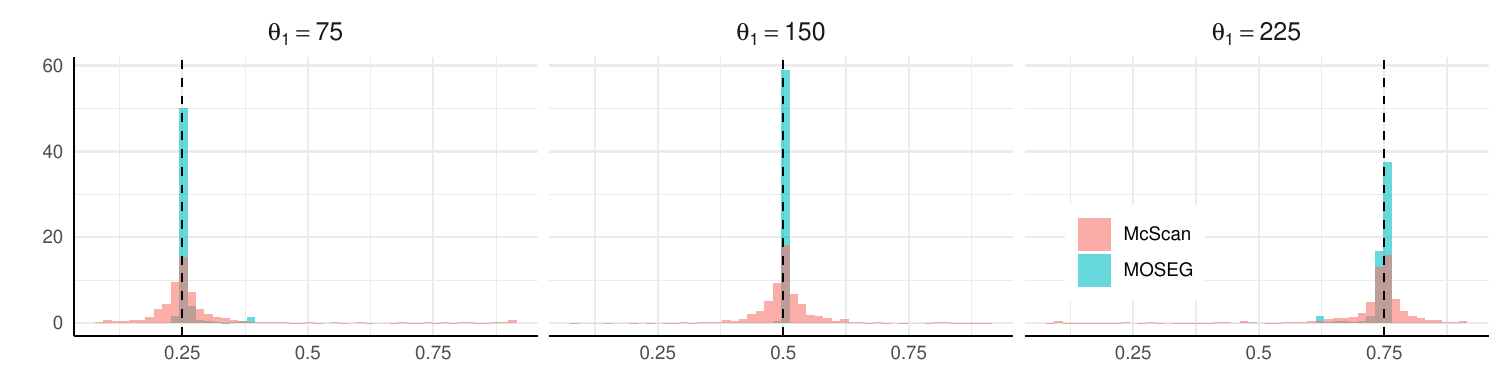}
    \caption{Scaled change point estimators ($\wh\cp_1/n$) returned by McScan and MOSEG, under \ref{m:five} for varying $\cp_1$ when $p = 900$; for each setting, 1000 repetitions are generated and each method is applied to detect a single change point. We denote $\cp_1$ by a broken line. The histograms use a percentage scale.}
    \label{fig:m5}
\end{figure}
\clearpage

\subsection{Differential parameter estimation}
\label{sec:add:dp}

See \cref{sec:sim:dpe} for the details of the experiments.
Figures~\ref{fig:MSE_ell1_g0}--\ref{fig:MSE_ell1_g00} report the estimation errors of different estimators of $\bm\delta_1$ with $\gamma = 0$ and $\gamma = 0.9$, respectively.

\begin{figure}[h!t!b!p!]
	\centering
	\includegraphics[width = .8\textwidth]{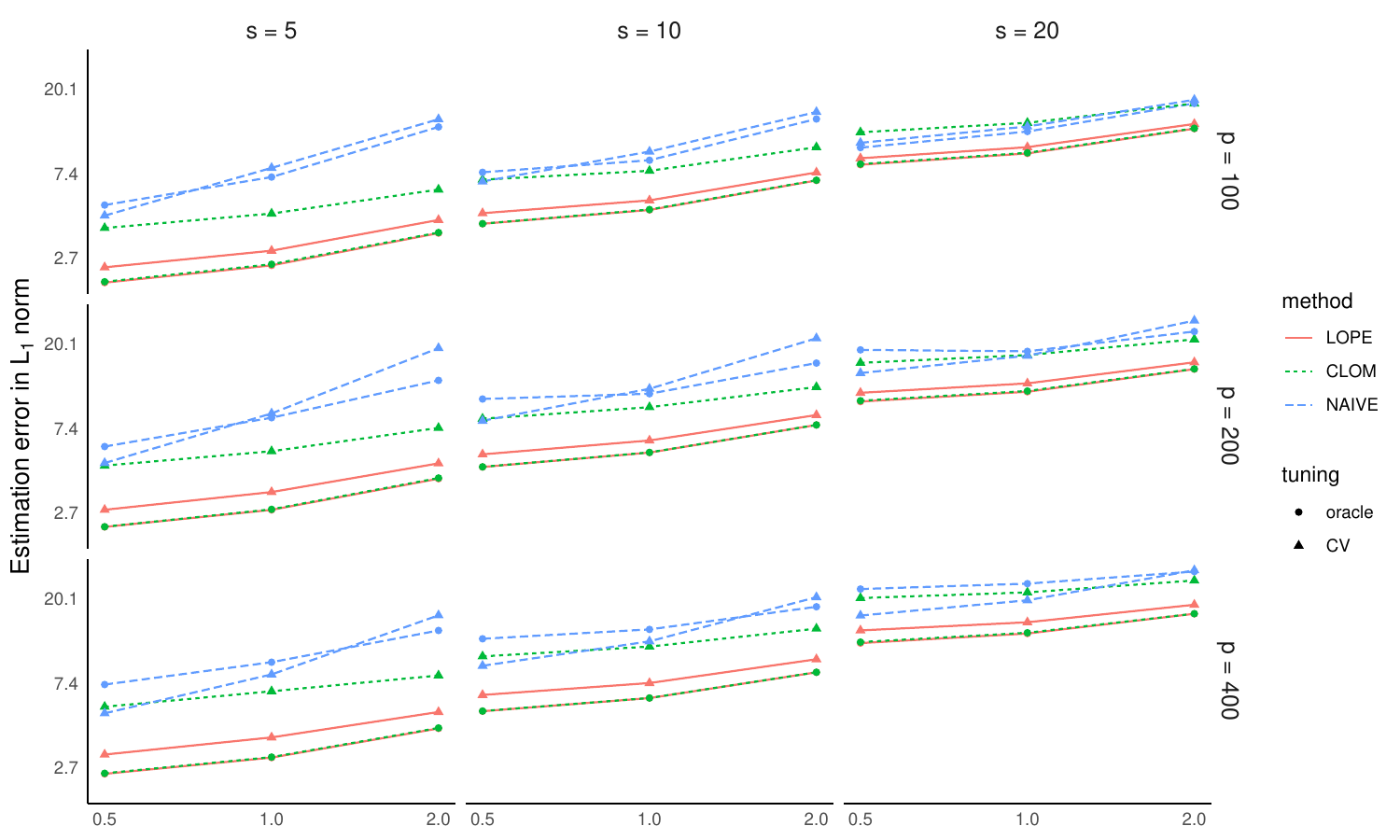}\\
	\includegraphics[width = .8\textwidth]{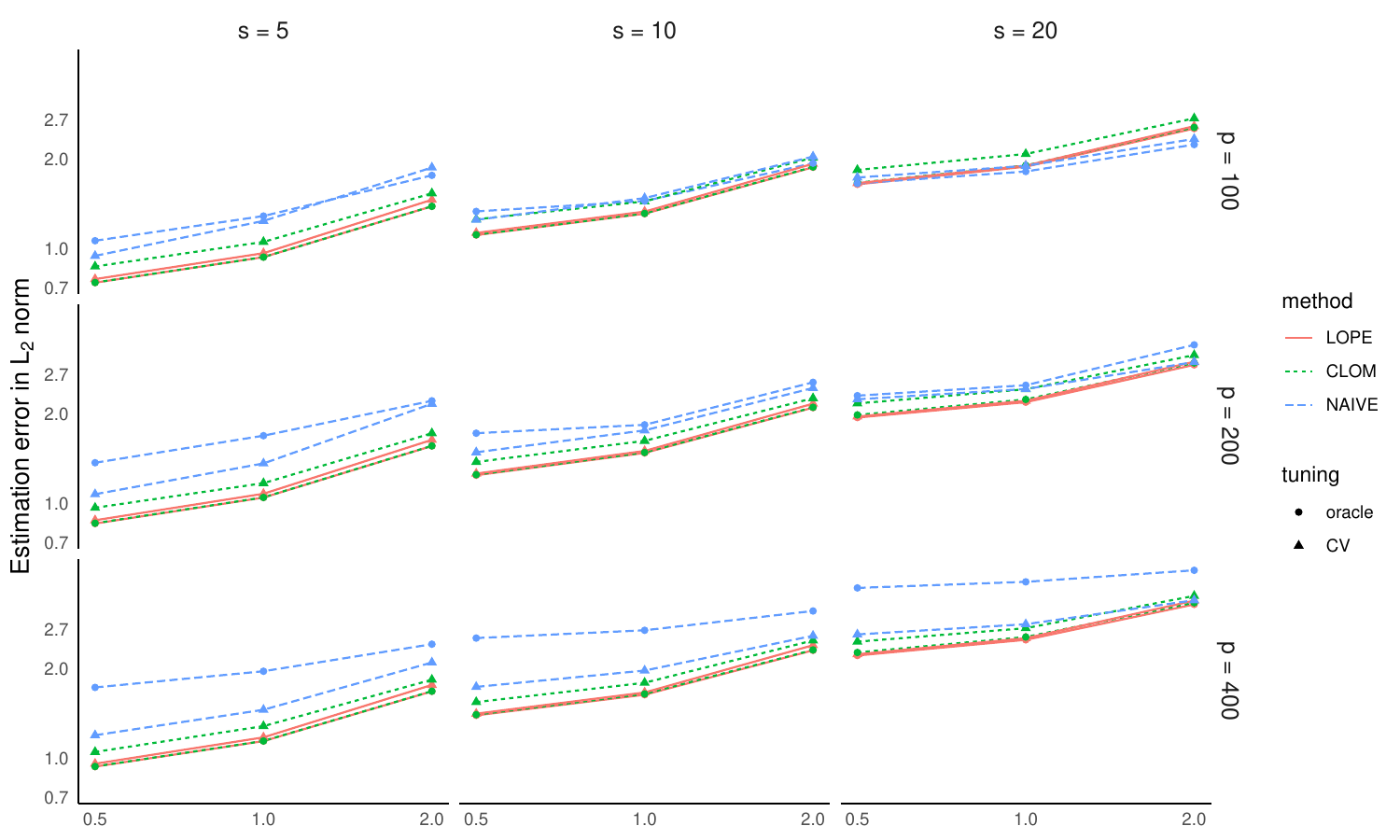}
  	\caption{Estimation errors in $\ell_1$- (top) and $\ell_2$- (bottom) norms against $\nu \in \{0.5, 1, 2\}$ ($x$-axis), from LOPE~\eqref{eq:lasso:est}, CLOM~\eqref{eq:direct:est} and NAIVE combined with the tuning parameter selected via cross-validation (CV) and the oracle one when $\gamma = 0$, averaged over $1000$ realisations. The $y$-axis is in the logarithm scale.}
	\label{fig:MSE_ell1_g0}
\end{figure}

\begin{figure}[h!t!b!p!]
	\centering
	\includegraphics[width = .8\textwidth]{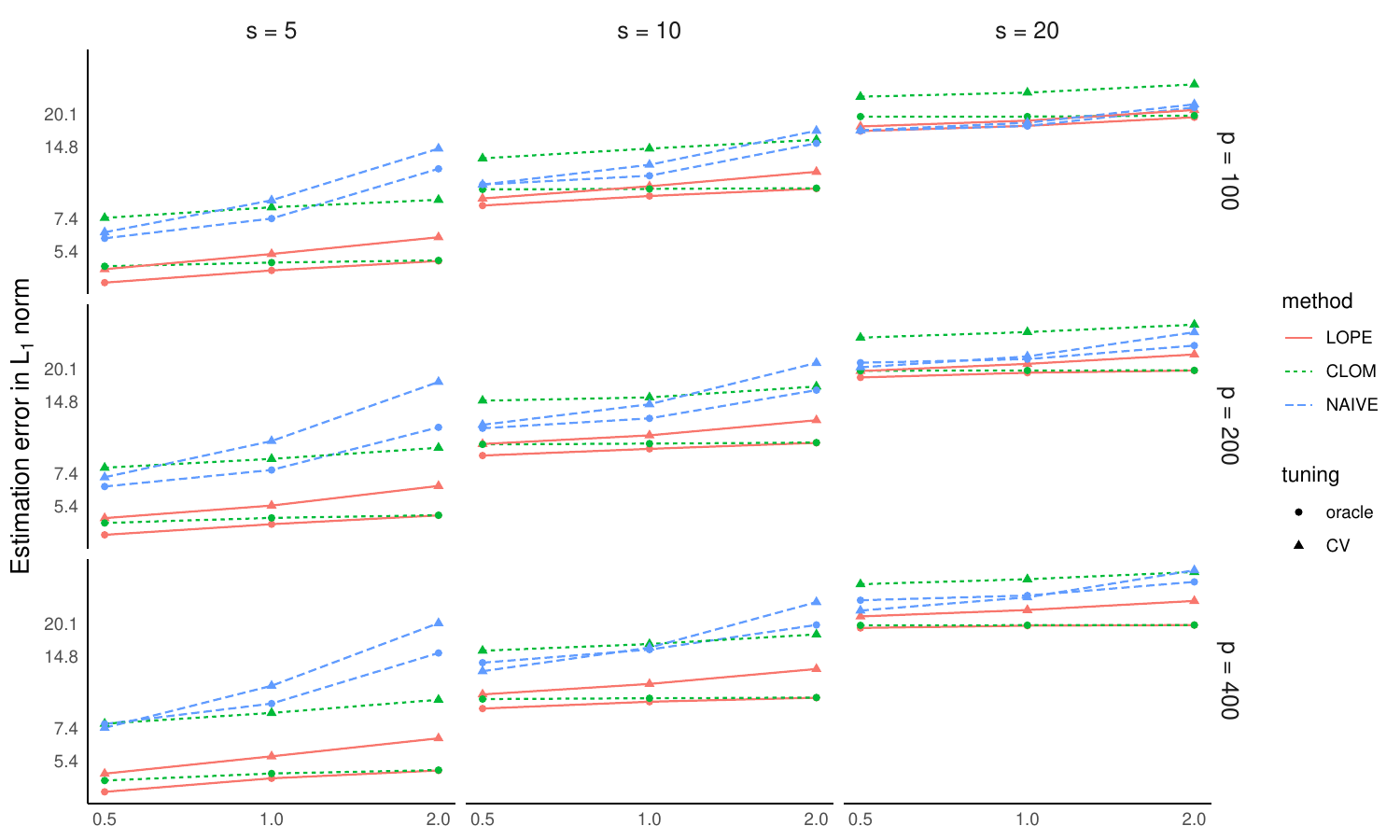}\\
	\includegraphics[width = .8\textwidth]{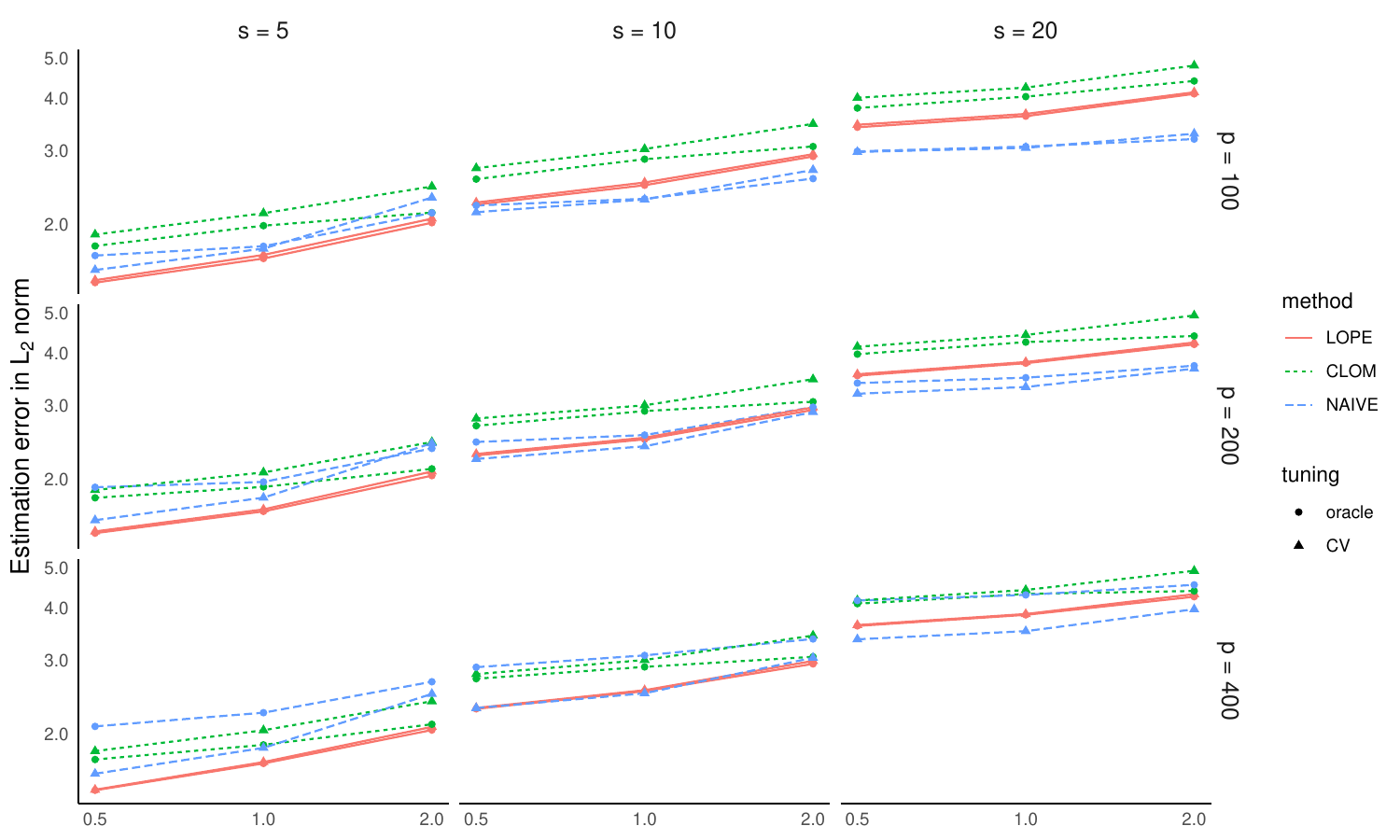}
  	\caption{Estimation errors in $\ell_1$- (top) and $\ell_2$- (bottom) norms against $\nu \in \{0.5, 1, 2\}$ ($x$-axis), from LOPE~\eqref{eq:lasso:est}, CLOM~\eqref{eq:direct:est} and NAIVE combined with the tuning parameter selected via cross-validation (CV) and the oracle one when $\gamma = 0.9$, averaged over $1000$ realisations. The $y$-axis is in the logarithm scale.}
	\label{fig:MSE_ell1_g00}
\end{figure}

\clearpage

\subsection{Simultaneous confidence intervals}
\label{sec:add:ci}

\subsubsection{Additional results under (M2)}
\label{sec:add:ci:m2}

See \cref{sec:sim:ci} for the details of the experiments.
Figures~\ref{fig:ci:cover:1:1}--\ref{fig:ci:cover:3:3} report the additional results on the coverage and other performance metrics of the bootstrap confidence intervals over varying $\alpha \in \{0.1, 0.05, 0.01\}$ and $\gamma \in \{0, 0.6, 0.9\}$.
Also, Figures~\ref{fig:ci:width:one}--\ref{fig:ci:width:three} display the half-width of the confidence intervals averaged over $100$ realisations for each setting.

\begin{figure}[h!t!b!]
\centering
\includegraphics[width = 1\textwidth]{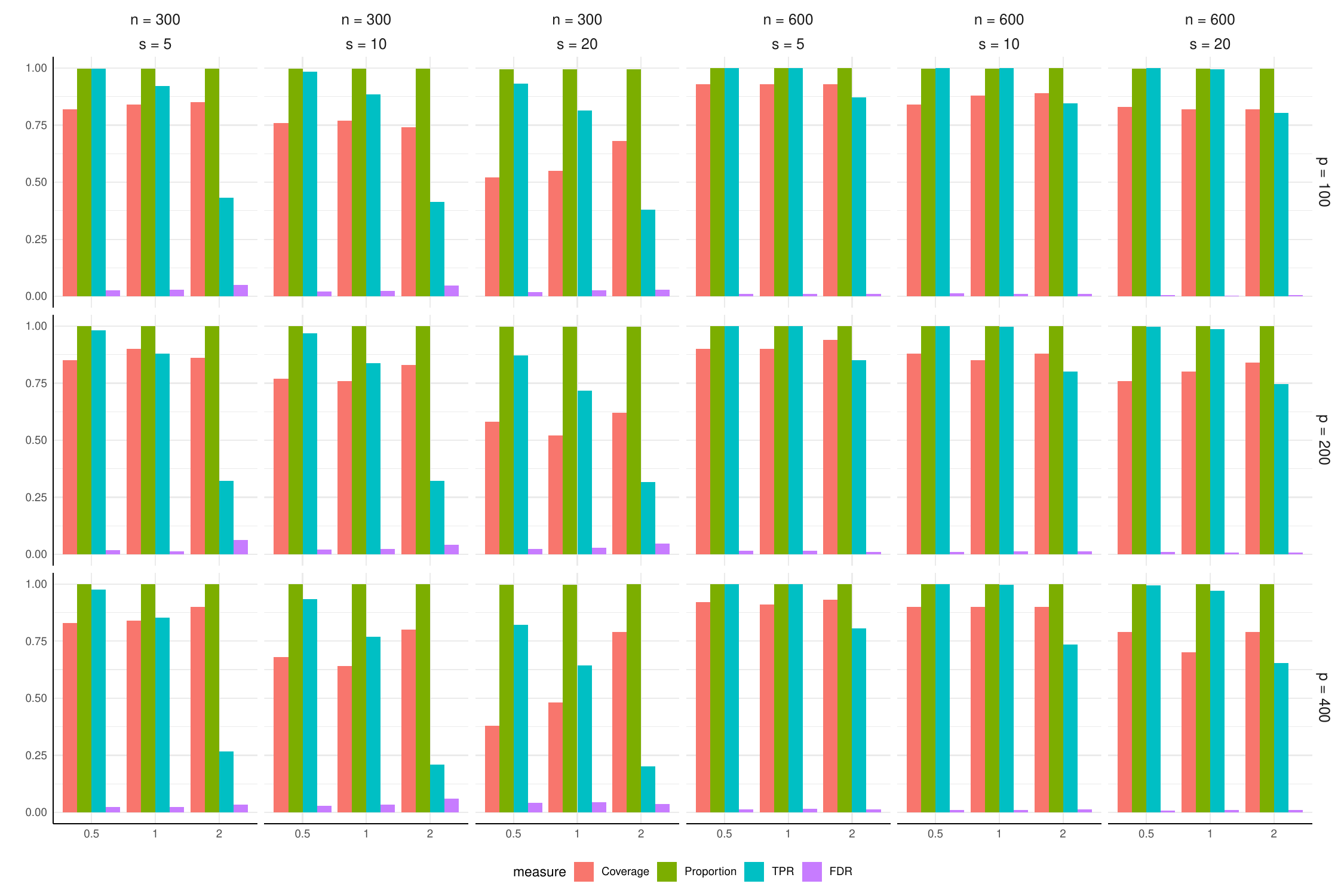}
\includegraphics[width = 1\textwidth]{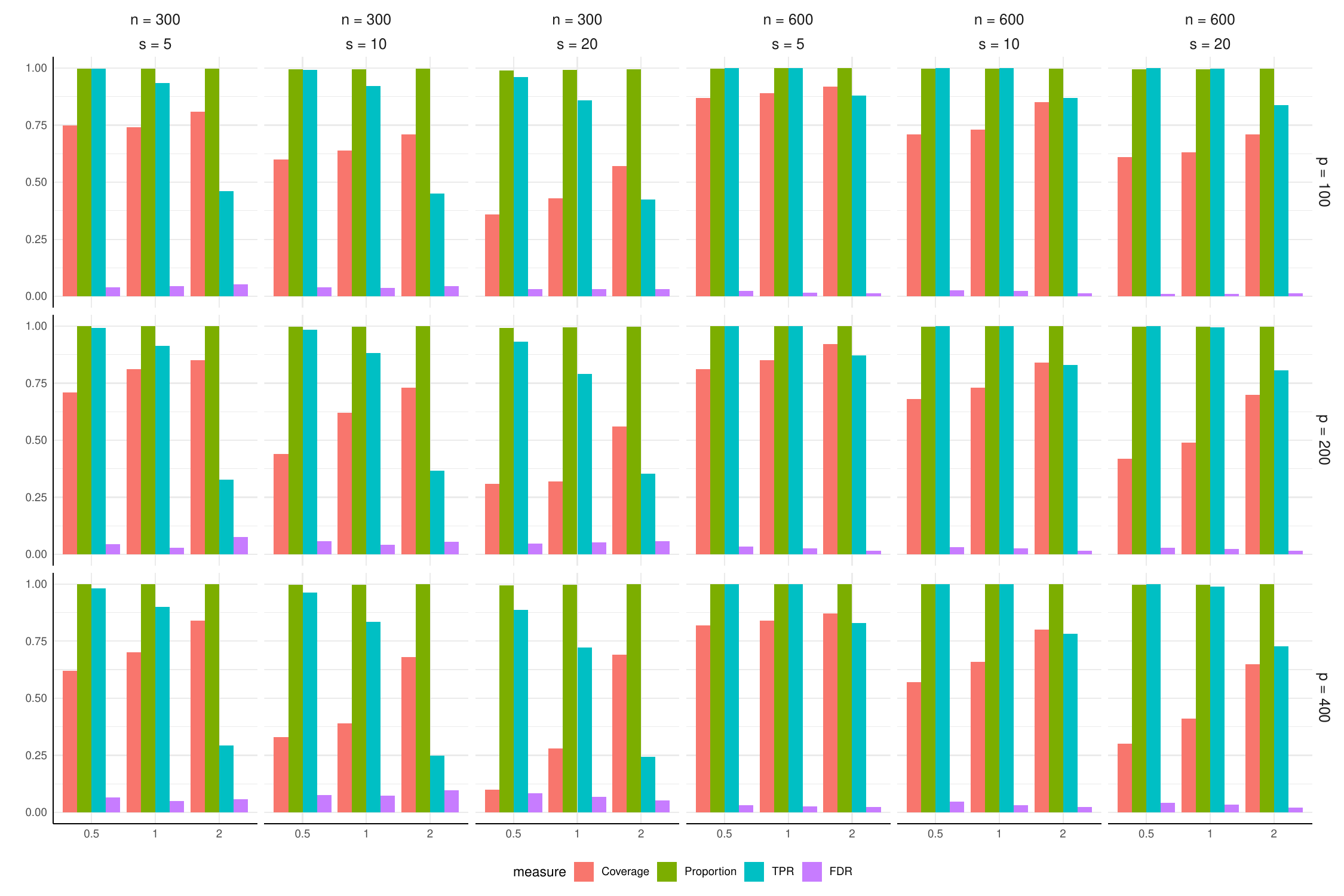}
\caption{Coverage, Proportion, TPR and FDR of simultaneous $90\%$-confidence intervals with (top) and without (bottom) bootstrapping against $\nu \in \{0.5, 1, 2\}$ ($x$-axis) when $\gamma = 0$, averaged over $100$ realisations.}
\label{fig:ci:cover:1:1}
\end{figure}

\begin{figure}[h!t!b!]
\centering
\includegraphics[width = 1\textwidth]{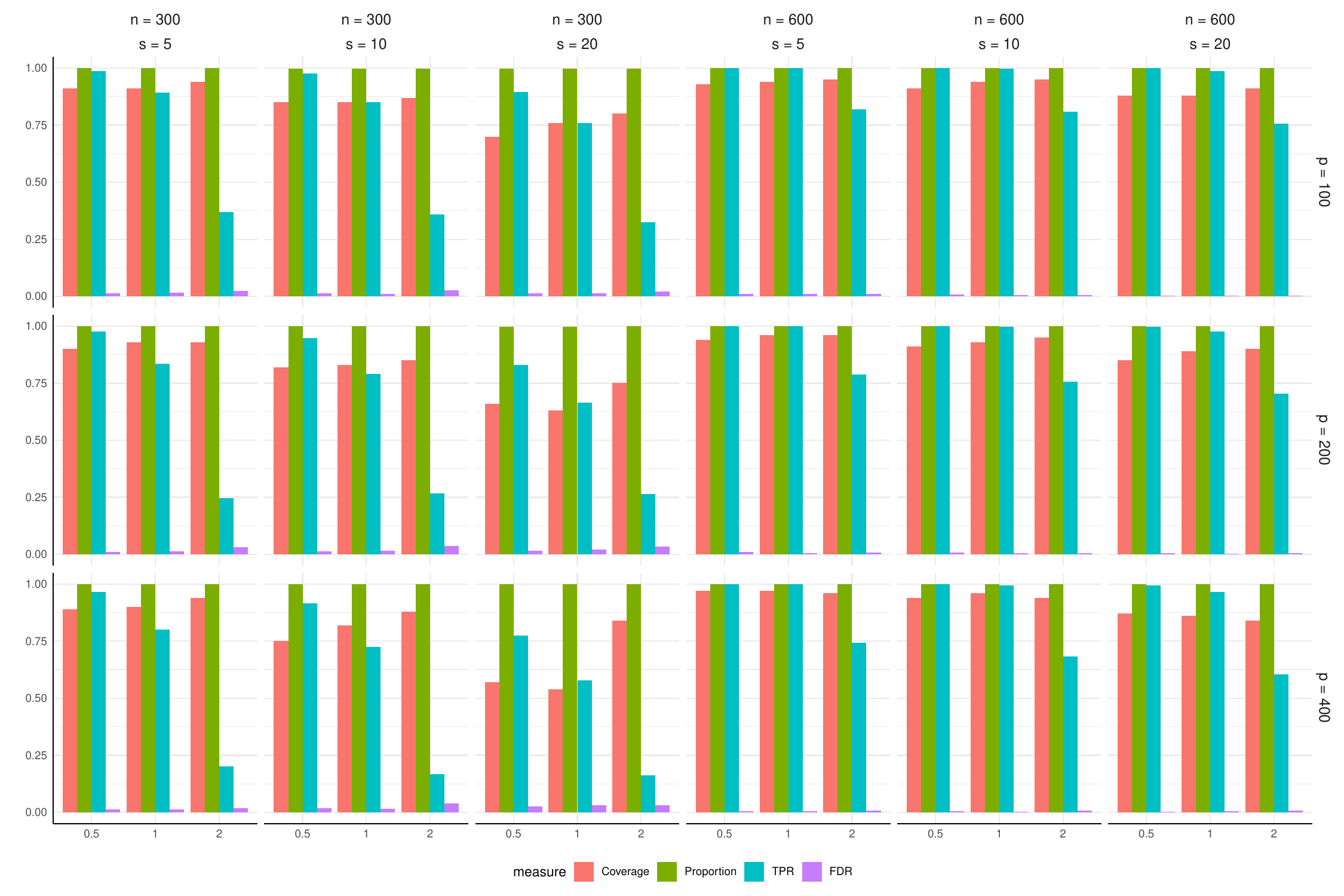}
\includegraphics[width = 1\textwidth]{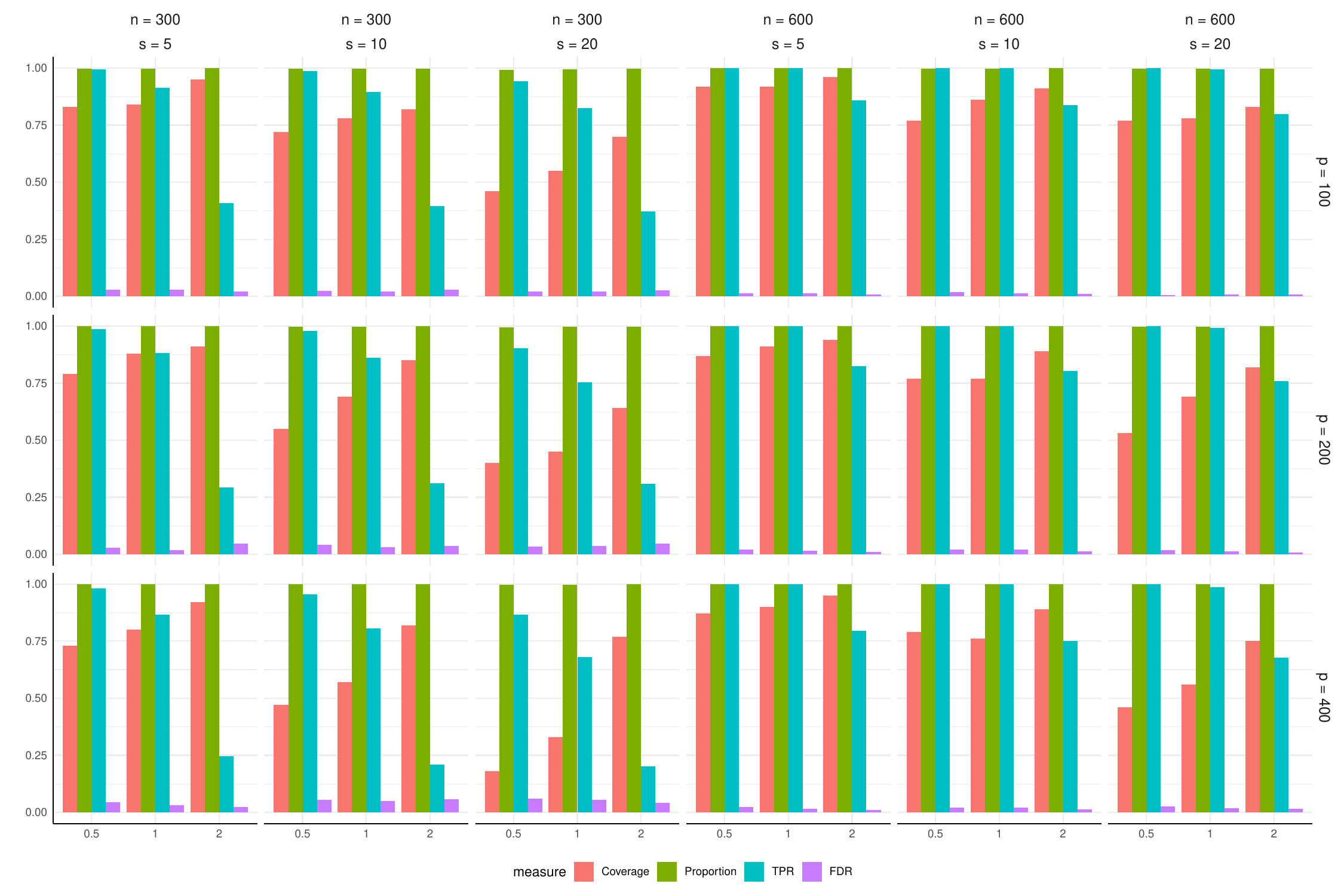}
\caption{Coverage, Proportion, TPR and FDR of simultaneous $95\%$-confidence intervals with (top) and without (bottom) bootstrapping against $\nu \in \{0.5, 1, 2\}$ ($x$-axis) when $\gamma = 0$, averaged over $100$ realisations.}
\label{fig:ci:cover:1:2}
\end{figure}

\begin{figure}[h!t!b!]
\centering
\includegraphics[width = 1\textwidth]{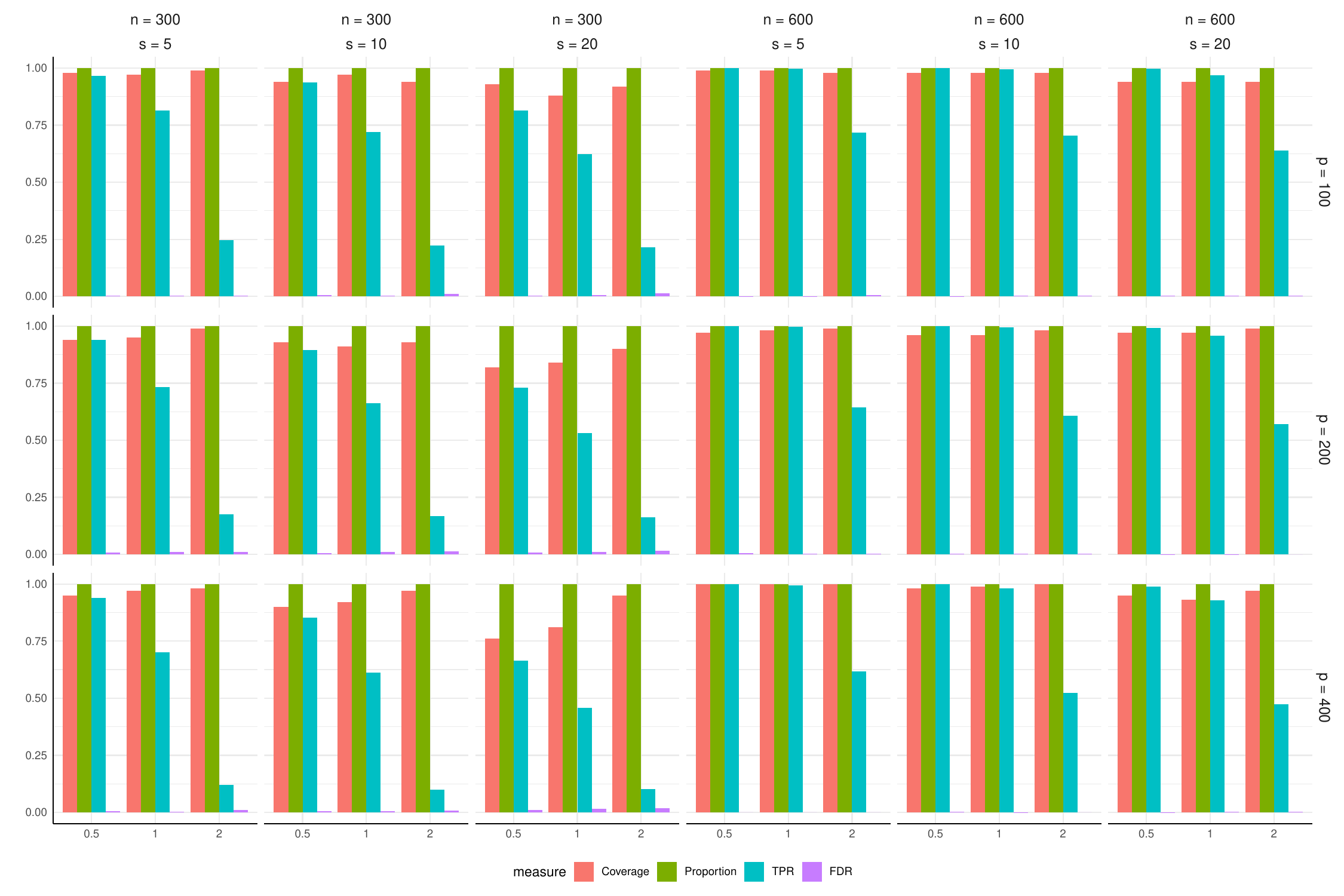}
\includegraphics[width = 1\textwidth]{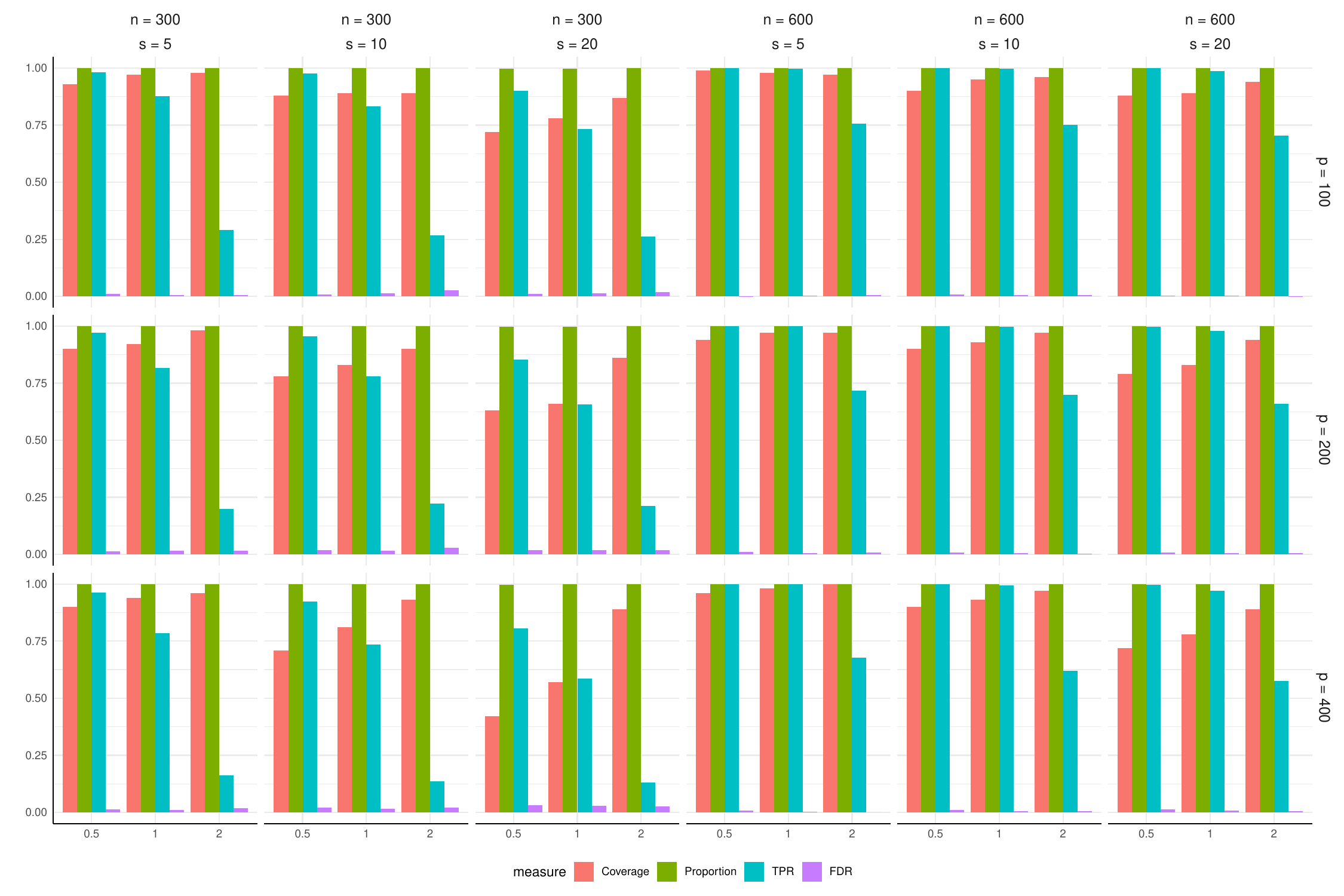}
\caption{Coverage, Proportion, TPR and FDR of simultaneous $99\%$-confidence intervals with (top) and without (bottom) bootstrapping against $\nu \in \{0.5, 1, 2\}$ ($x$-axis) when $\gamma = 0$, averaged over $100$ realisations.}
\label{fig:ci:cover:1:3}
\end{figure}

\begin{figure}[h!t!b!]
\centering
\includegraphics[width = 1\textwidth]{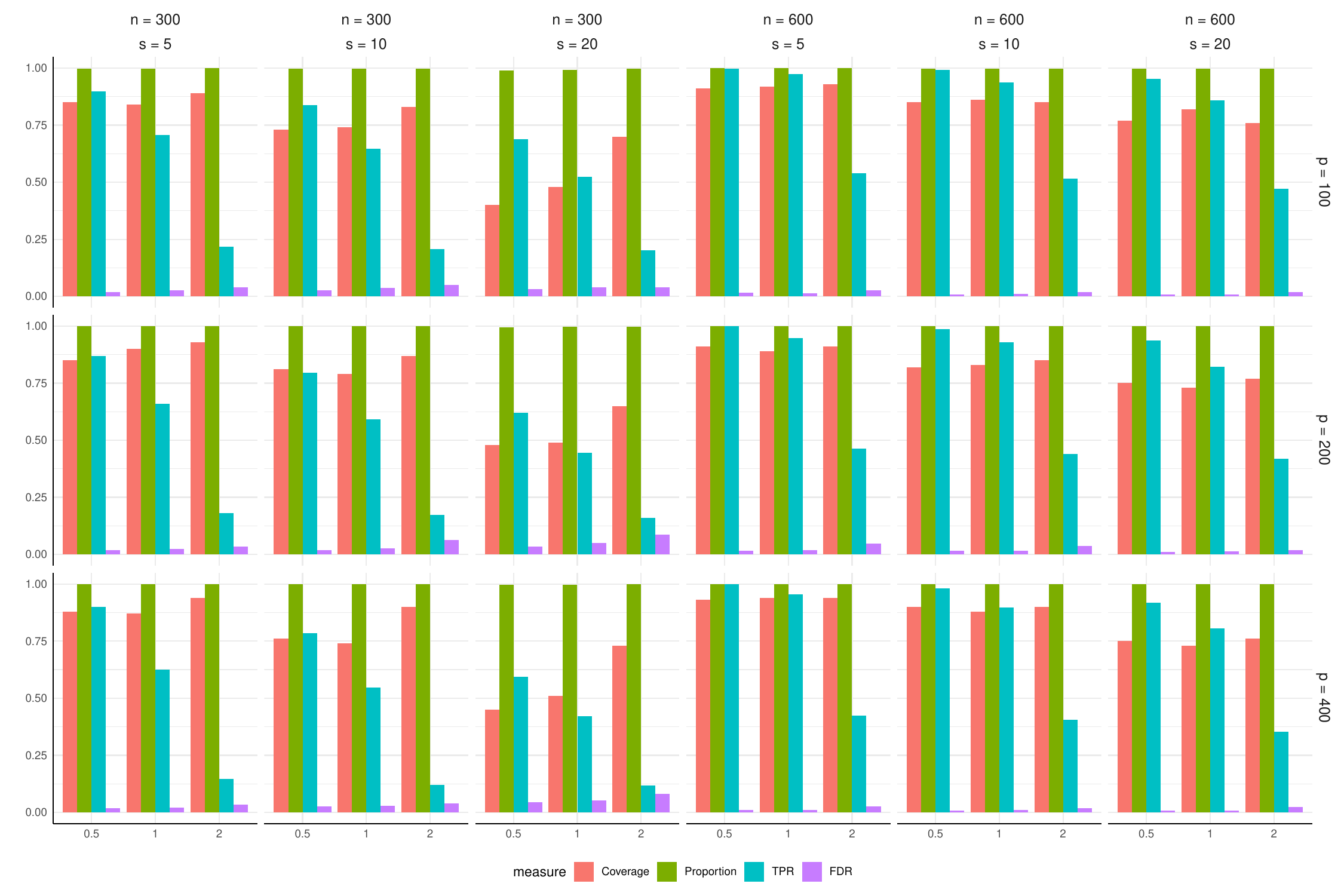}
\includegraphics[width = 1\textwidth]{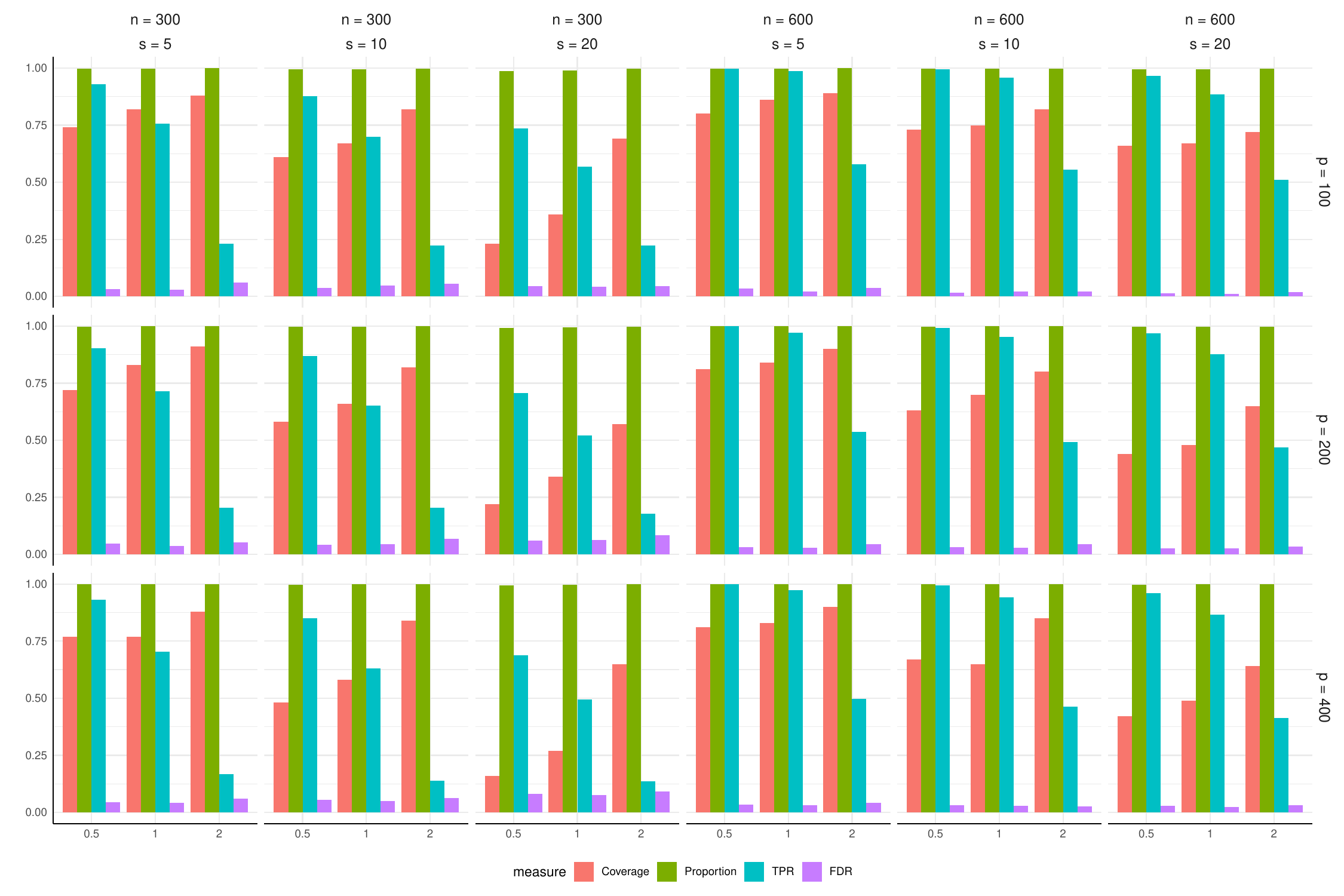}
\caption{Coverage, Proportion, TPR and FDR of simultaneous $90\%$-confidence intervals with (top) and without (bottom) bootstrapping against $\nu \in \{0.5, 1, 2\}$ ($x$-axis) when $\gamma = 0.6$, averaged over $100$ realisations.}
\label{fig:ci:cover:2:1}
\end{figure}

\begin{figure}[h!t!b!]
\centering
\includegraphics[width = 1\textwidth]{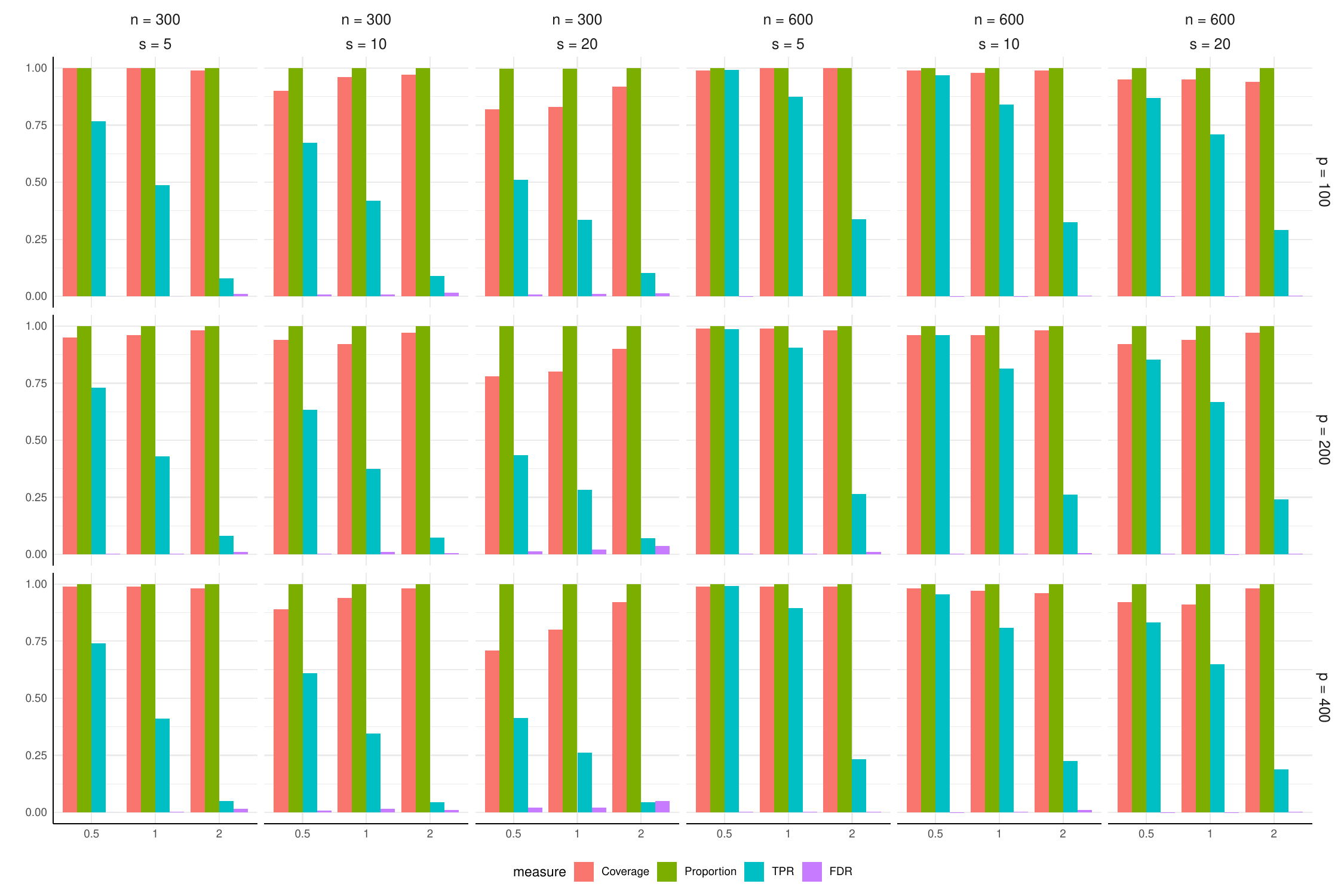}
\includegraphics[width = 1\textwidth]{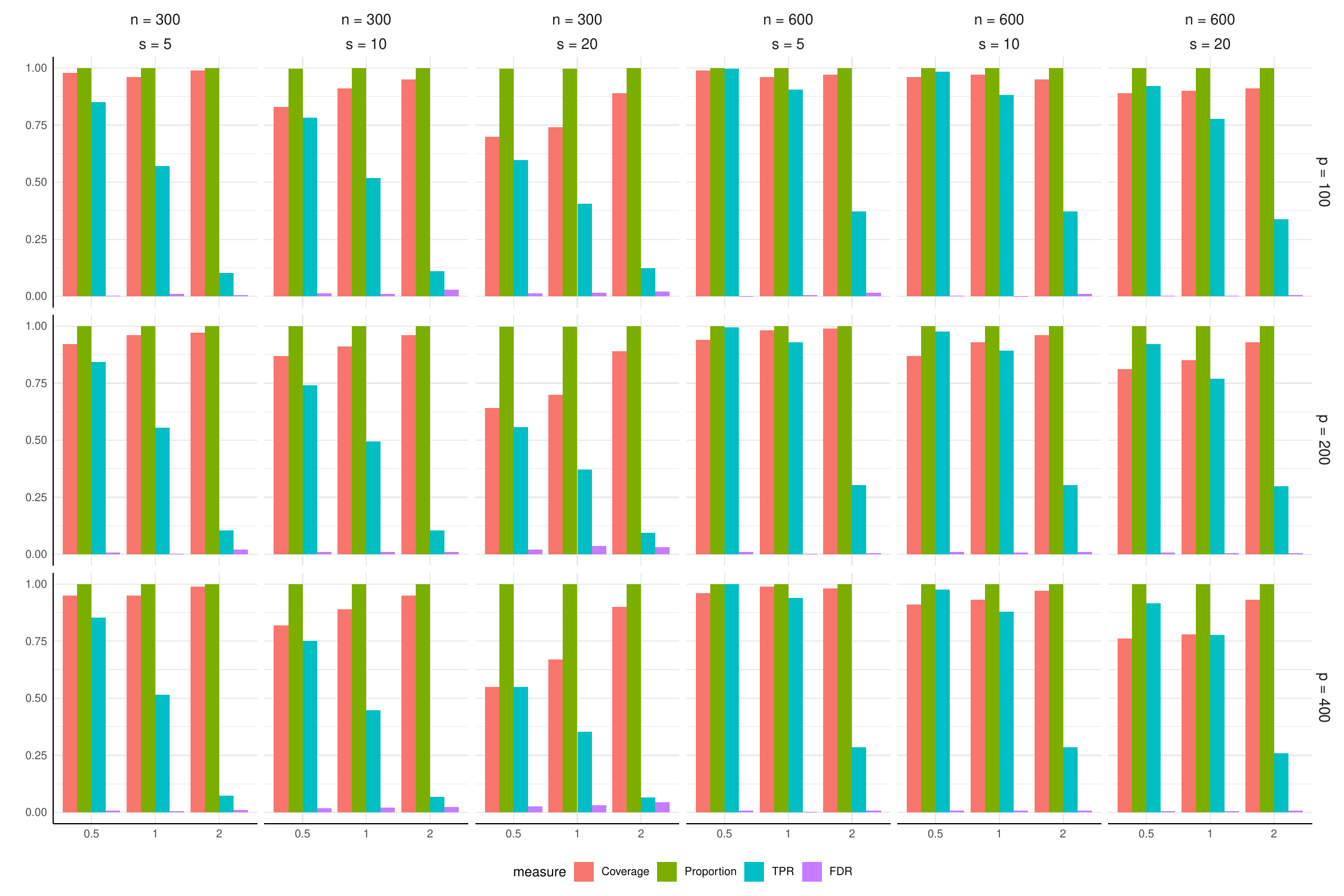}
\caption{Coverage, Proportion, TPR and FDR of simultaneous $99\%$-confidence intervals with (top) and without (bottom) bootstrapping against $\nu \in \{0.5, 1, 2\}$ ($x$-axis) when $\gamma = 0.6$, averaged over $100$ realisations.}
\label{fig:ci:cover:2:3}
\end{figure}

\begin{figure}[h!t!b!]
\centering
\includegraphics[width = 1\textwidth]{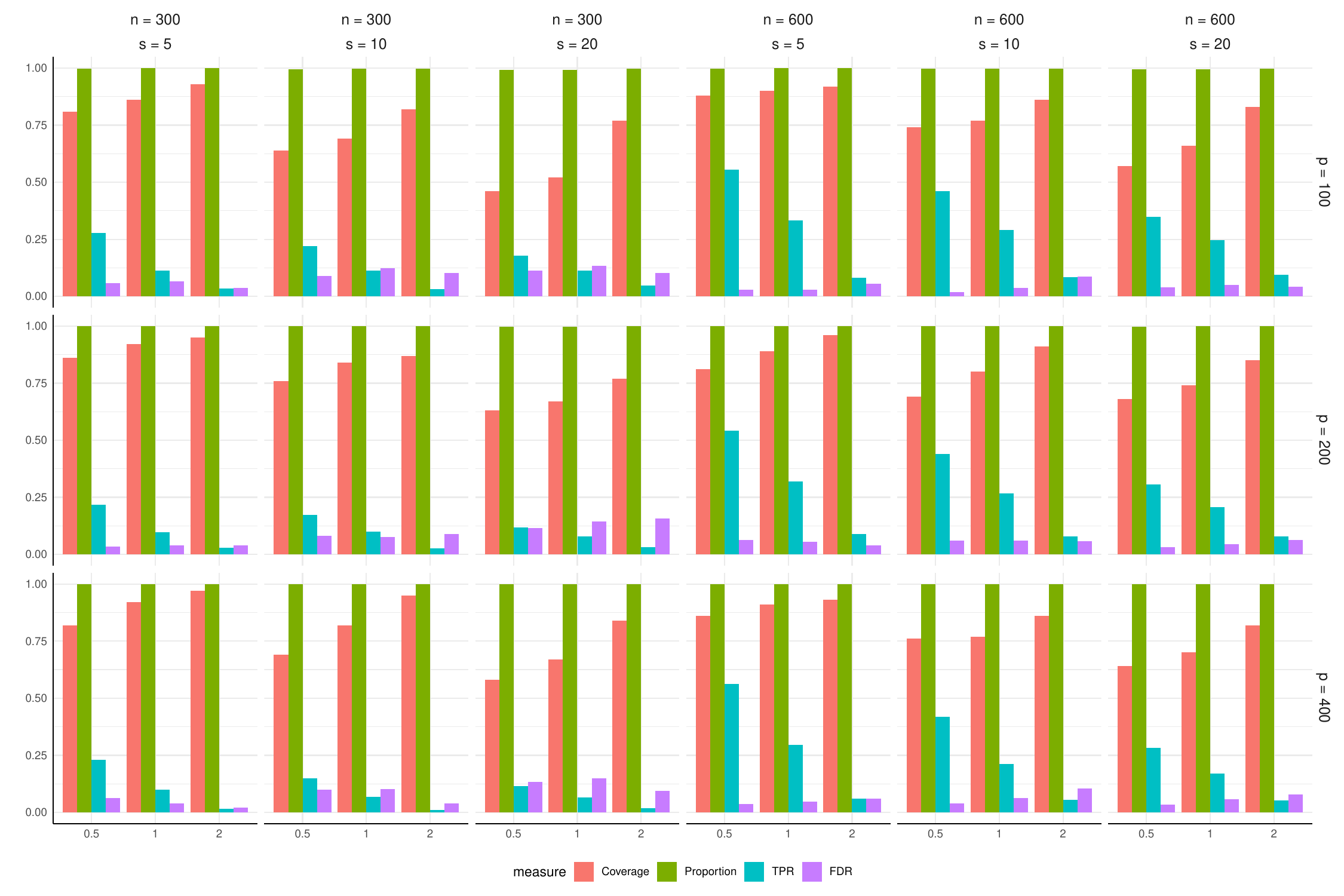}
\includegraphics[width = 1\textwidth]{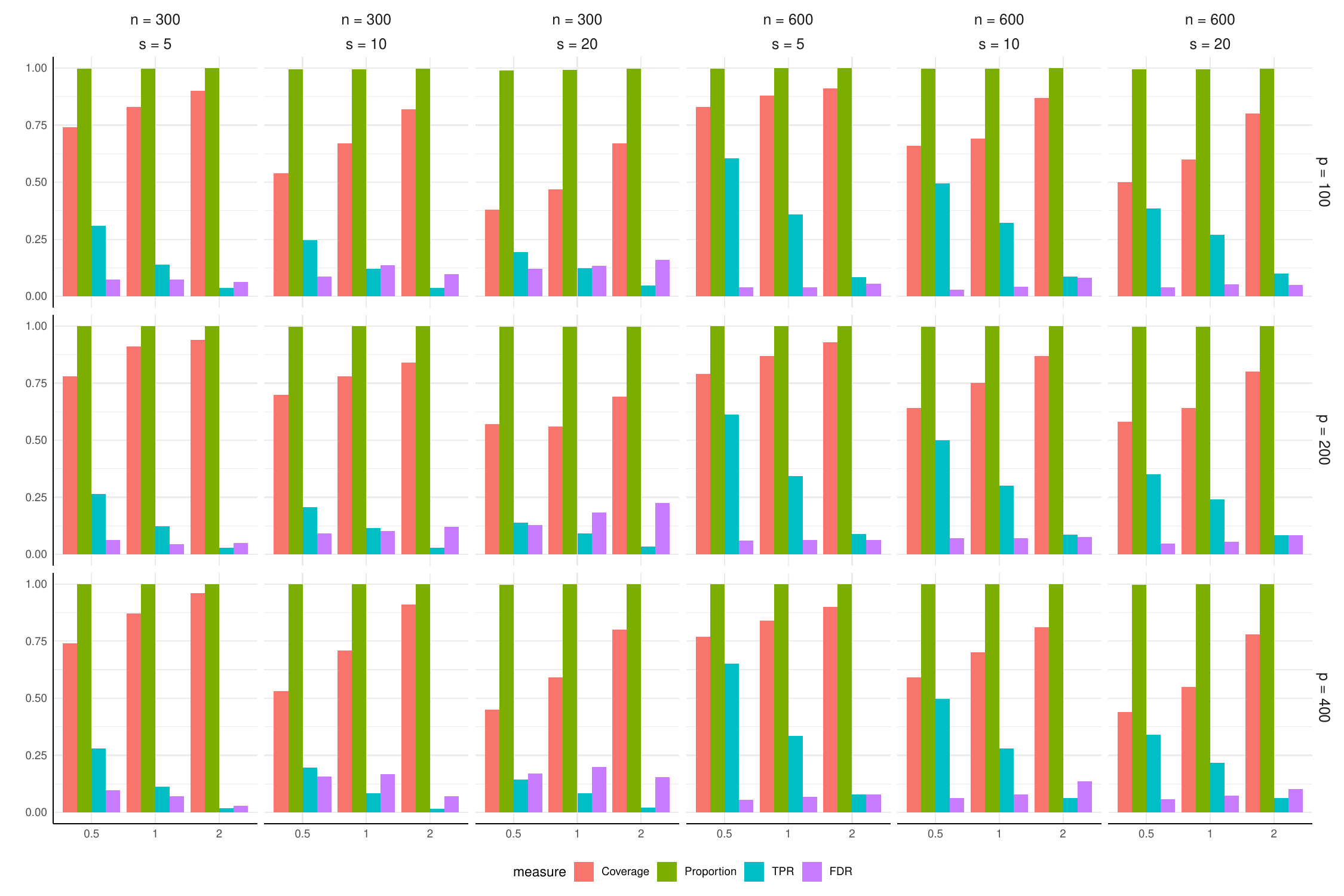}
\caption{Coverage, Proportion, TPR and FDR of simultaneous $90\%$-confidence intervals with (top) and without (bottom) bootstrapping against $\nu \in \{0.5, 1, 2\}$ ($x$-axis) when $\gamma = 0.9$, averaged over $100$ realisations.}
\label{fig:ci:cover:3:1}
\end{figure}

\begin{figure}[h!t!b!]
\centering
\includegraphics[width = 1\textwidth]{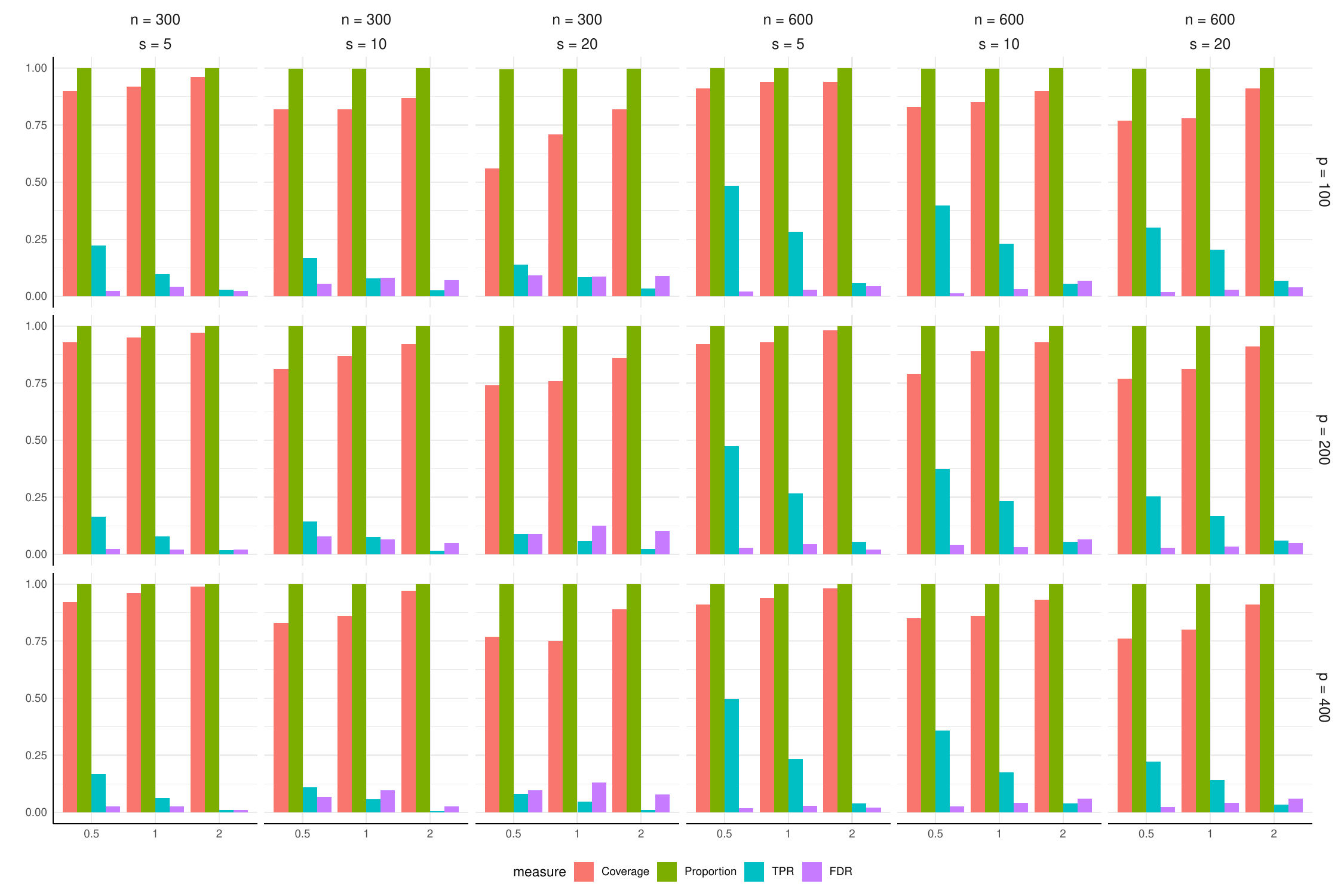}
\includegraphics[width = 1\textwidth]{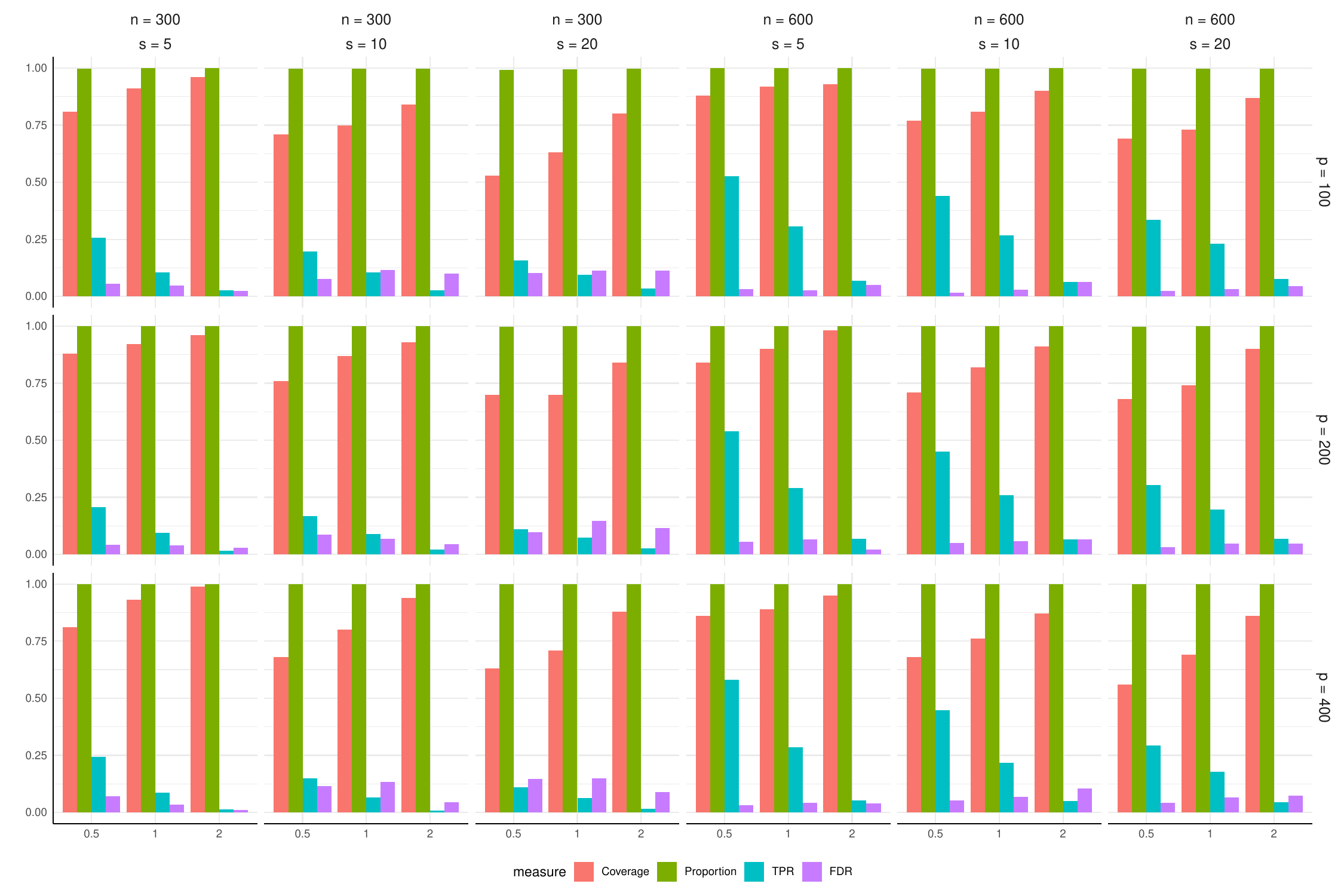}
\caption{Coverage, Proportion, TPR and FDR of simultaneous $95\%$-confidence intervals with (top) and without (bottom) bootstrapping against $\nu \in \{0.5, 1, 2\}$ ($x$-axis) when $\gamma = 0.9$, averaged over $100$ realisations.}
\label{fig:ci:cover:3:2}
\end{figure}

\begin{figure}[h!t!b!]
\centering
\includegraphics[width = 1\textwidth]{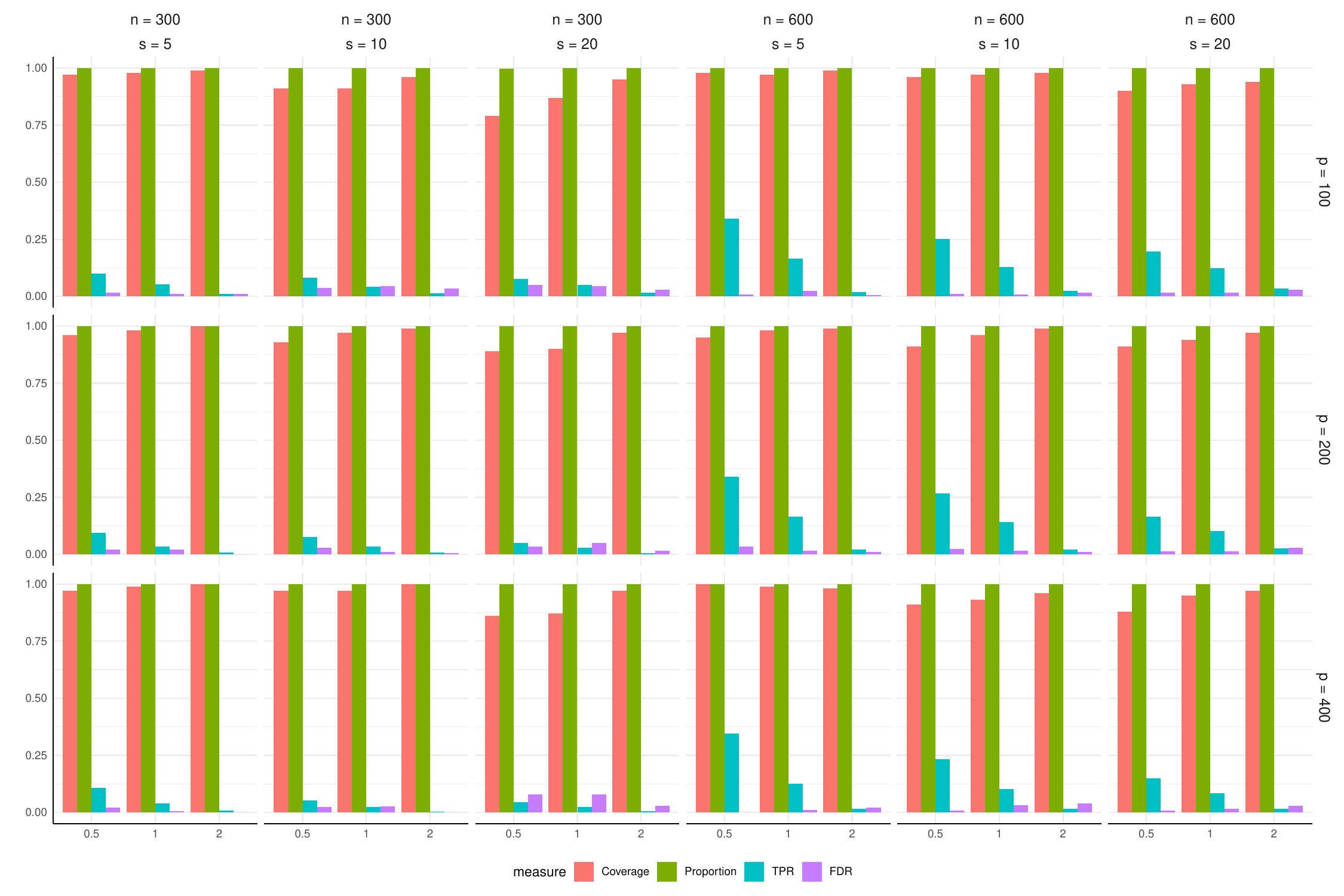}
\includegraphics[width = 1\textwidth]{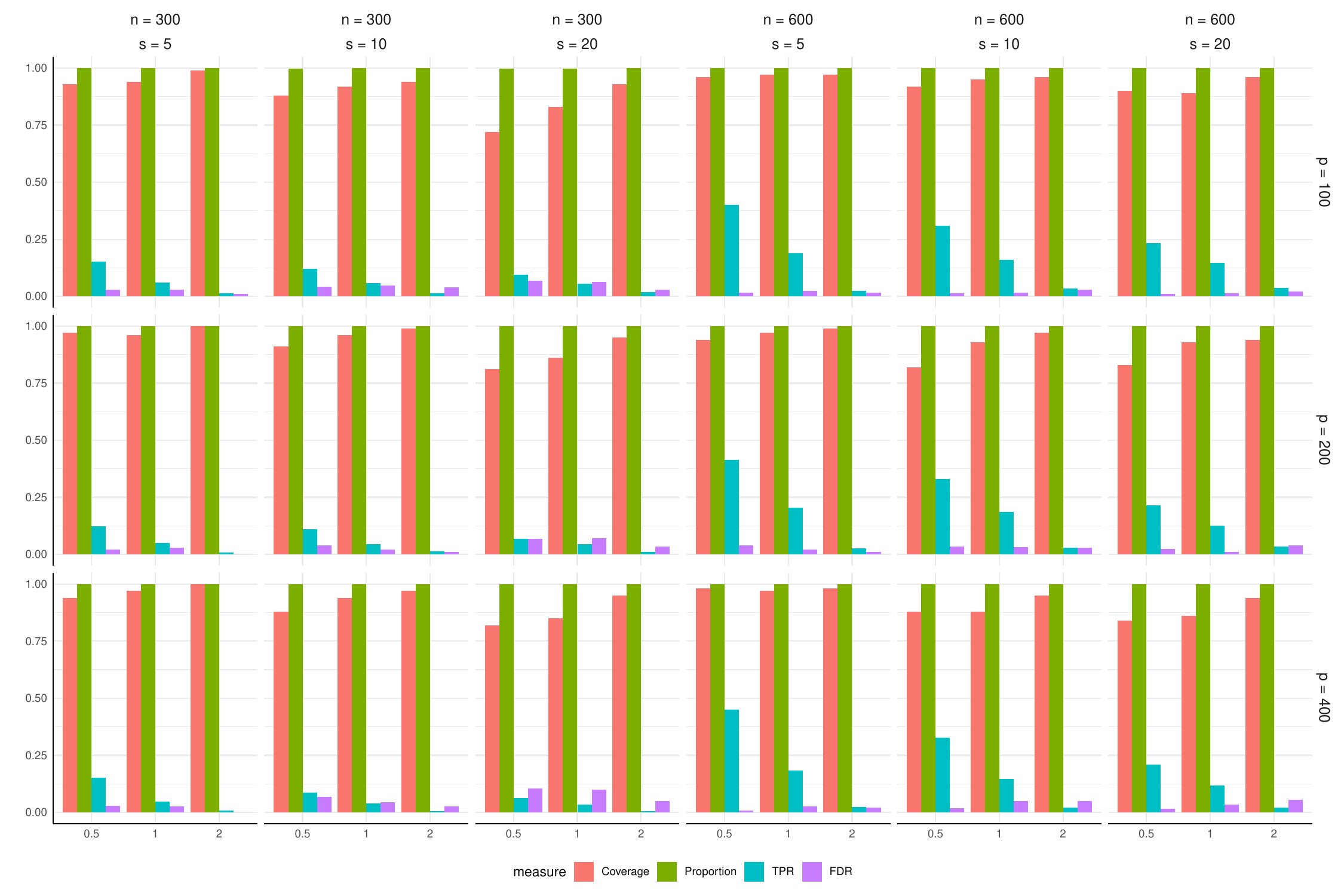}
\caption{Coverage, Proportion, TPR and FDR of simultaneous $99\%$-confidence intervals with (top) and without (bottom) bootstrapping against $\nu \in \{0.5, 1, 2\}$ ($x$-axis) when $\gamma = 0.9$, averaged over $100$ realisations.}
\label{fig:ci:cover:3:3}
\end{figure}

\begin{figure}[h!t!b!]
\centering
\includegraphics[width = 1\textwidth]{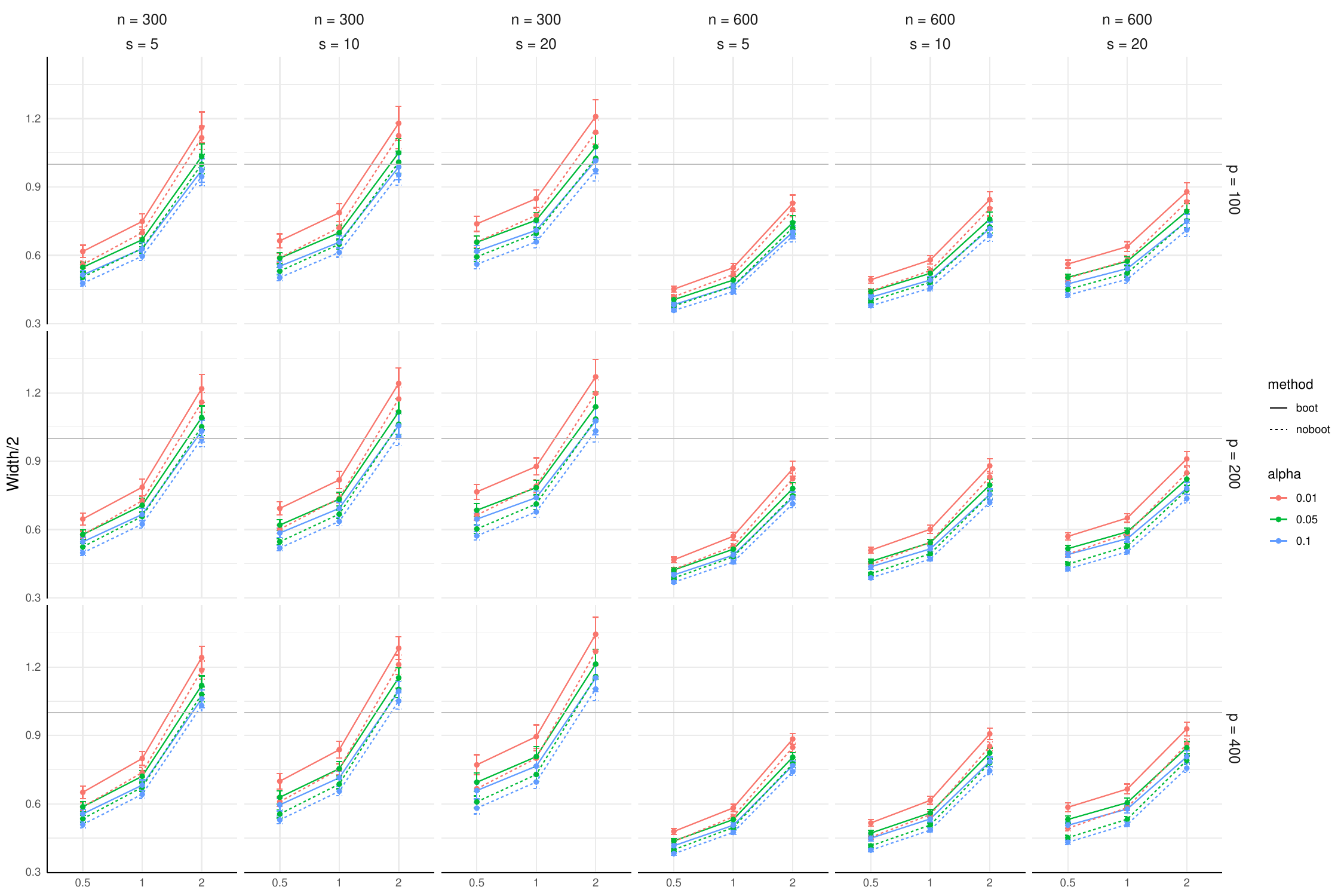}
\caption{Half-width of simultaneous $100(1 - \alpha)\%$-confidence intervals constructed with (solid) and without (dotted) bootstrapping against $\nu \in \{0.5, 1, 2\}$ ($x$-axis) for $\alpha \in \{0.1, 0.05, 0.01\}$ when $\gamma = 0$,
averaged over $100$ realisations with the standard deviations denoted by the error bars. The horizontal line is at $y = 1$ (recall that $\vert \delta_{i1} \vert = 1$ for $i \in \mc S_1$).}
\label{fig:ci:width:one}
\end{figure}

\begin{figure}[h!t!b!]
\centering
\includegraphics[width = 1\textwidth]{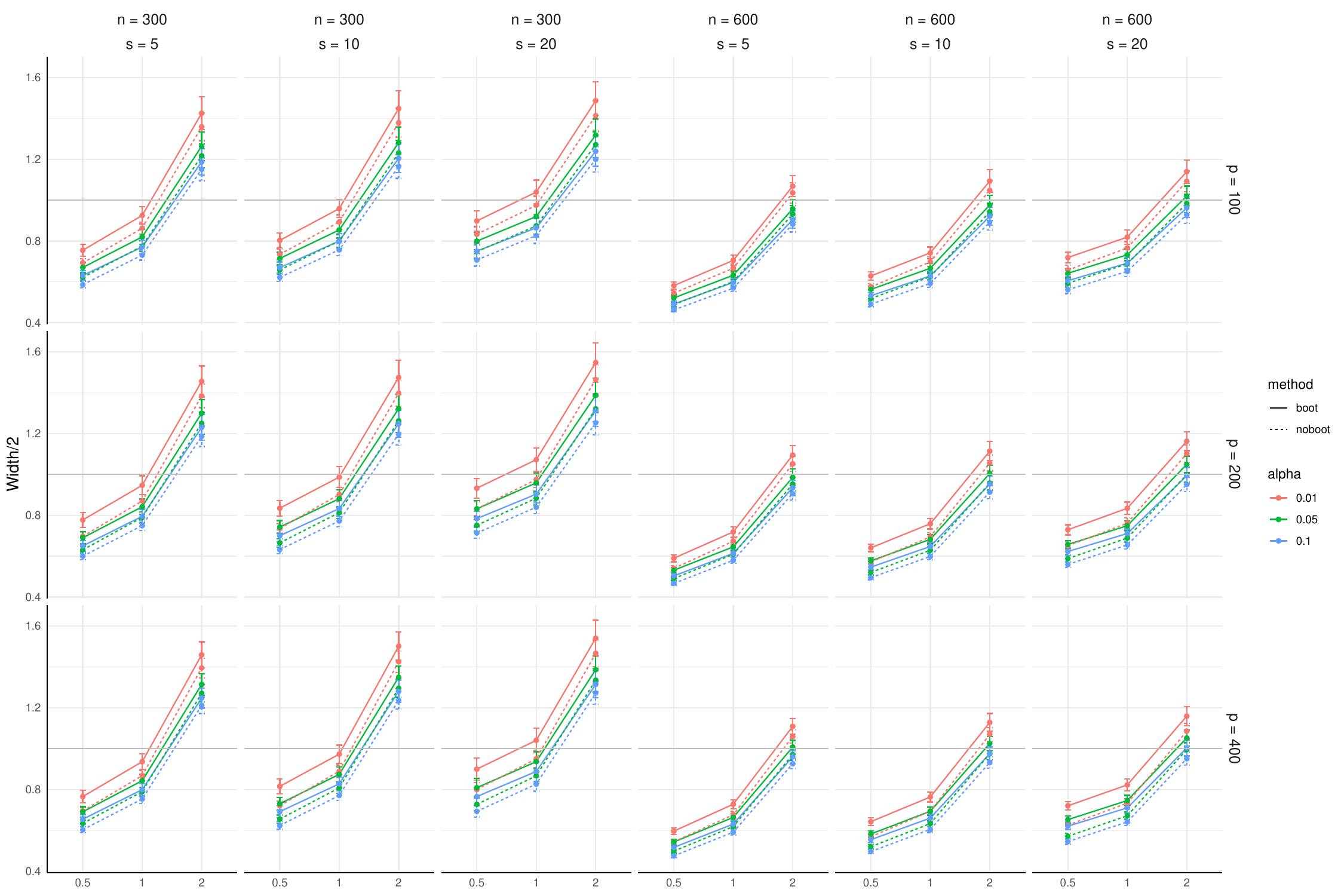}
\caption{Half-width of simultaneous $100(1 - \alpha)\%$-confidence intervals constructed with (solid) and without (dotted) bootstrapping against $\nu \in \{0.5, 1, 2\}$ ($x$-axis) for $\alpha \in \{0.1, 0.05, 0.01\}$ when $\gamma = 0.6$,
averaged over $100$ realisations with the standard deviations denoted by the error bars. The horizontal line is at $y = 1$ (recall that $\vert \delta_{i1} \vert = 1$ for $i \in \mc S_1$).}
\label{fig:ci:width:two}
\end{figure}

\begin{figure}[h!t!b!]
\centering
\includegraphics[width = 1\textwidth]{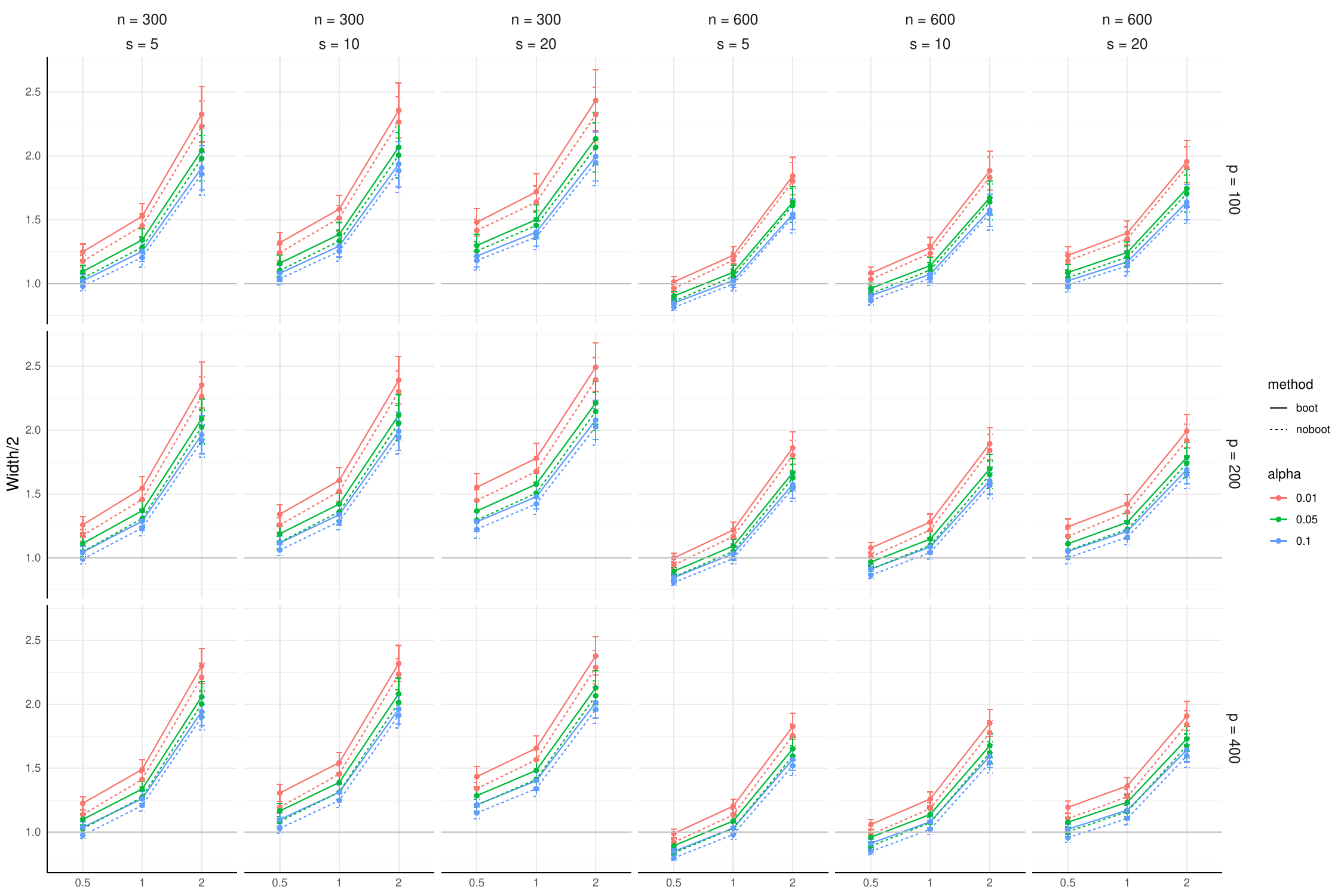}
\caption{Half-width of simultaneous $100(1 - \alpha)\%$-confidence intervals constructed with (solid) and without (dotted) bootstrapping against $\nu \in \{0.5, 1, 2\}$ ($x$-axis) for $\alpha \in \{0.1, 0.05, 0.01\}$ when $\gamma = 0.9$,
averaged over $100$ realisations with the standard deviations denoted by the error bars. The horizontal line is at $y = 1$ (recall that $\vert \delta_{i1} \vert = 1$ for $i \in \mc S_1$).}
\label{fig:ci:width:three}
\end{figure}

\clearpage

\subsubsection{Additional results under (M3)}

We apply McScan with the automatically chosen threshold to the data generated under the multiple change point scenario~(M3) described in Section~\ref{sec:sim:multi} with $p = 900$ and $q = 3$, followed by the simultaneous confidence interval construction for each change point.
Here, we adopt both approaches, the multiplier bootstrap procedure described in Appendix~\ref{sec:boot} and the one based on sampling directly from the distribution of $\vert \wh{\mbf V}_j \vert_\infty$ in Corollary~\ref{cor:thm:two}.
Since the success of the inference is conditional on the consistent estimation of the change points, we report the performance of the confidence intervals on the realisations where all three change points are detected only; this occurs over $80\%$ of realisations in most settings considered.
For comparison, we also include the `oracle' version of McScan where $q = 3$ is known.

Figures~\ref{fig:ci:m3:width} and~\ref{fig:ci:m3:cover} show that generally, with the increase of the sample size, coverage and other performance metrics defined in Section~\ref{sec:sim:ci} improve and the width of the confidence intervals decreases.
Also, the gap between the results attained by the automatic threshold and the oracle one closes with increasing sample size.
For this scenario, there is little difference between the two approaches to confidence interval construction, either in coverage performance or the width of the intervals.
Overall, the proposed inference step handles estimation error from the change point detection step well and produces confidence intervals of small width attaining good coverage.

\begin{figure}[h!t!b!]
	\centering
	\includegraphics[width = .9\textwidth]{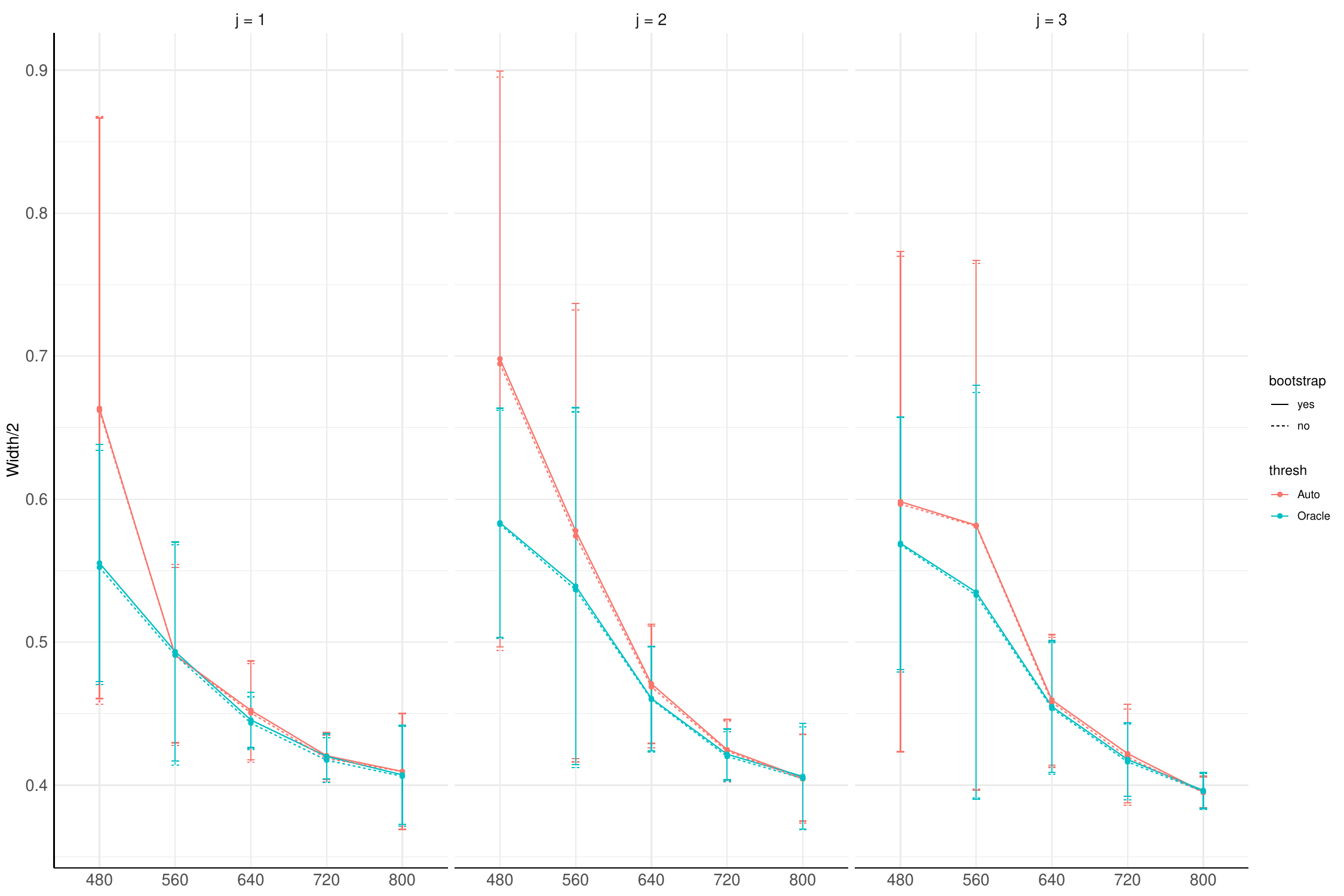}
	\caption{Half-width of simultaneous $95\%$-confidence intervals for $\cp_j, \, j = 1, 2, 3$ (left to right) with varying $n \in \{480, 560, 640, 720, 800\}$ ($x$-axis),
		averaged over $100$ realisations with the standard deviations denoted by the error bars.}
	\label{fig:ci:m3:width}
\end{figure}

\begin{figure}[h!t!b!]
	\centering
	\includegraphics[width = 1\textwidth]{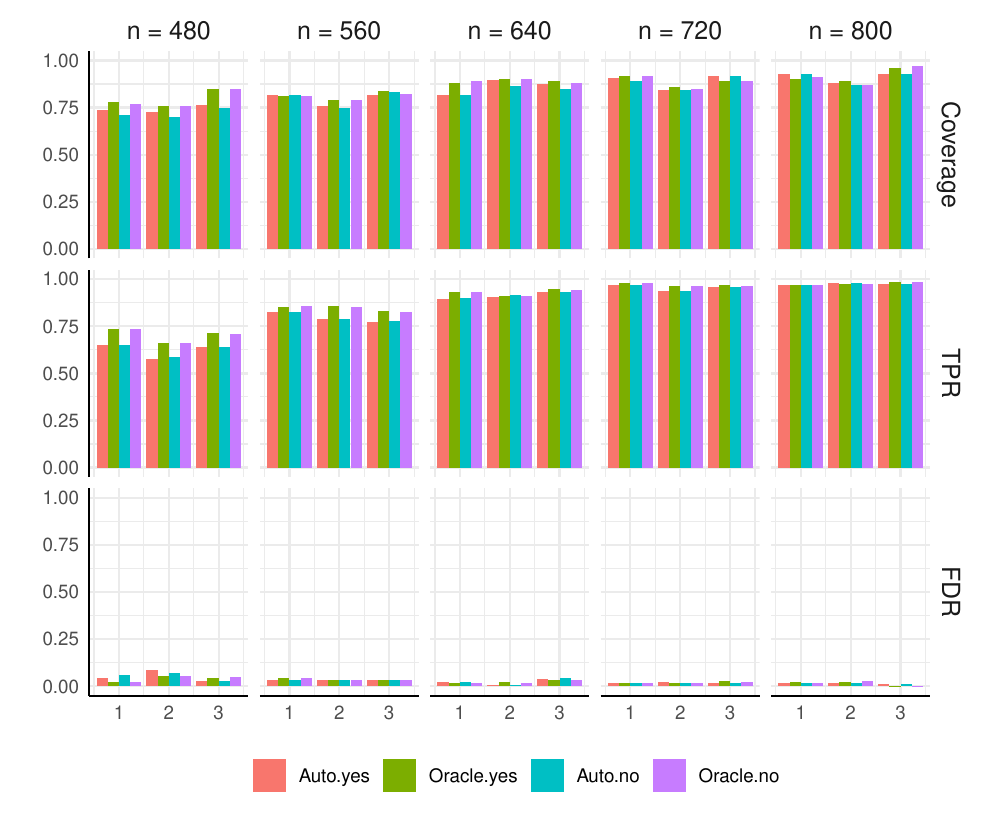}
	\caption{Coverage, TPR and FDR (top to bottom) of simultaneous $95\%$-confidence intervals for $\cp_j, \, j = 1, 2, 3$ ($x$-axis) with varying $n \in \{480, 560, 640, 720, 800\}$ (left to right), averaged over $100$ realisations. Confidence intervals are constructed with the automatic (`Auto') and oracle thresholds, with (`yes') and without ('no') multiplier bootstrap.}
	\label{fig:ci:m3:cover}
\end{figure}

\clearpage

\subsection{Additional results from FRED-MD analysis}
\label{sec:fredmd:extra}

We complement the application in Section~\ref{sec:data} by plotting the additional simultaneous confidence intervals associated with the differential parameters, see Figures~\ref{fig:fredmd:ci:boot:add}--\ref{fig:fredmd:ci:noboot}.

\begin{figure}[h!t!b!]
\centering
\includegraphics[width = .85\textwidth]{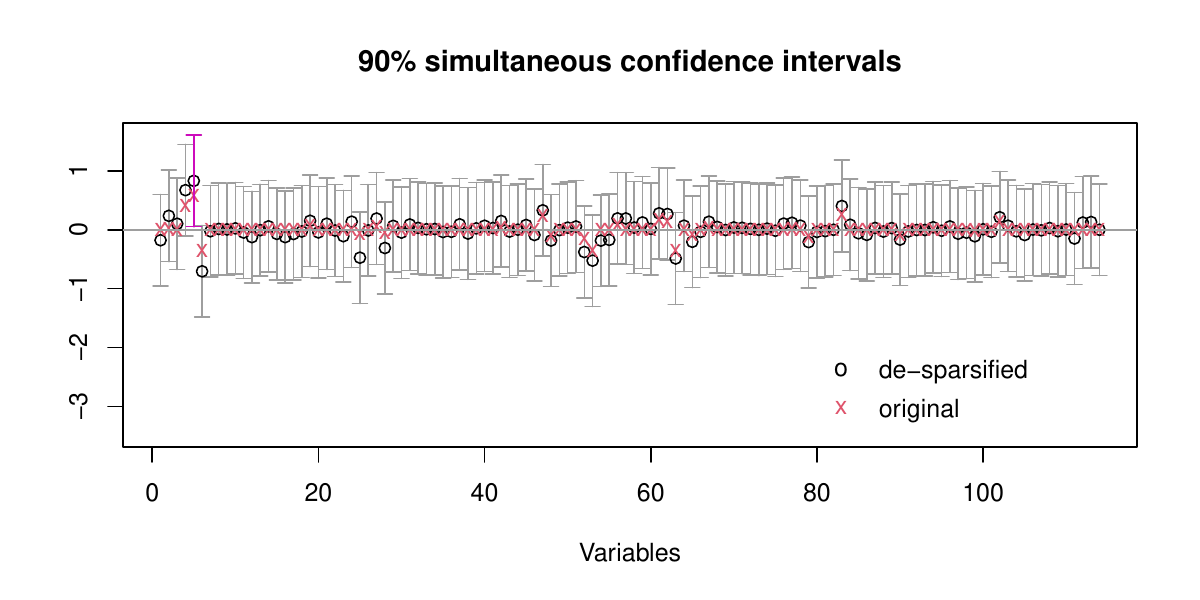} \\
\includegraphics[width = .85\textwidth]{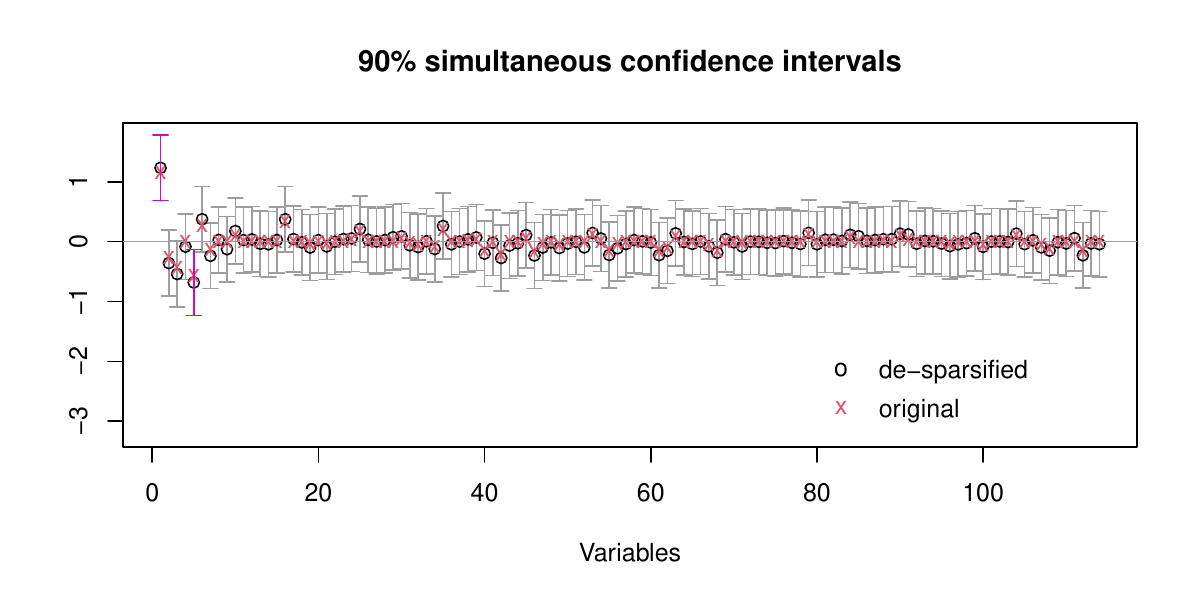}
\caption{Simultaneous $90\%$-confidence intervals obtained via multiplier bootstrap for the $p = 114$ differential parameter coefficients before and after $\wh\cp_1$ (top) and $\wh\cp_2$ (bottom), plotted alongside the original LOPE estimator in~\eqref{eq:lasso:est} and its de-sparsified version in~\eqref{eq:debiased}. The interval which does not contain $0$ is highlighted in magenta.}
\label{fig:fredmd:ci:boot:add}
\end{figure}

\begin{figure}[h!t!b!]
\centering
\includegraphics[width = .85\textwidth]{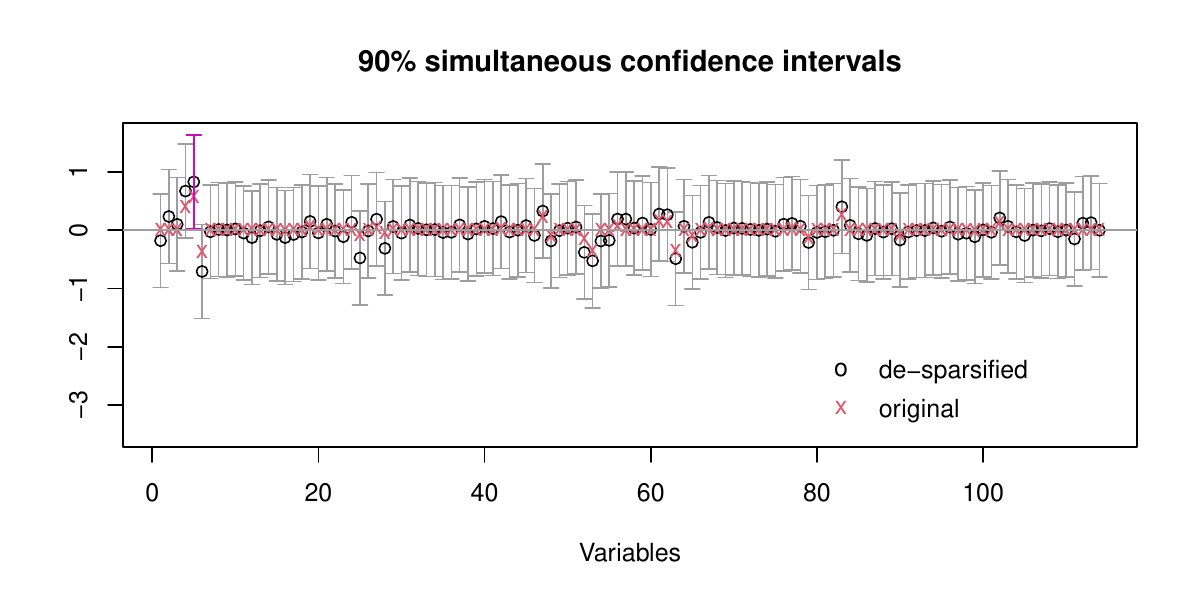} \\
\includegraphics[width = .85\textwidth]{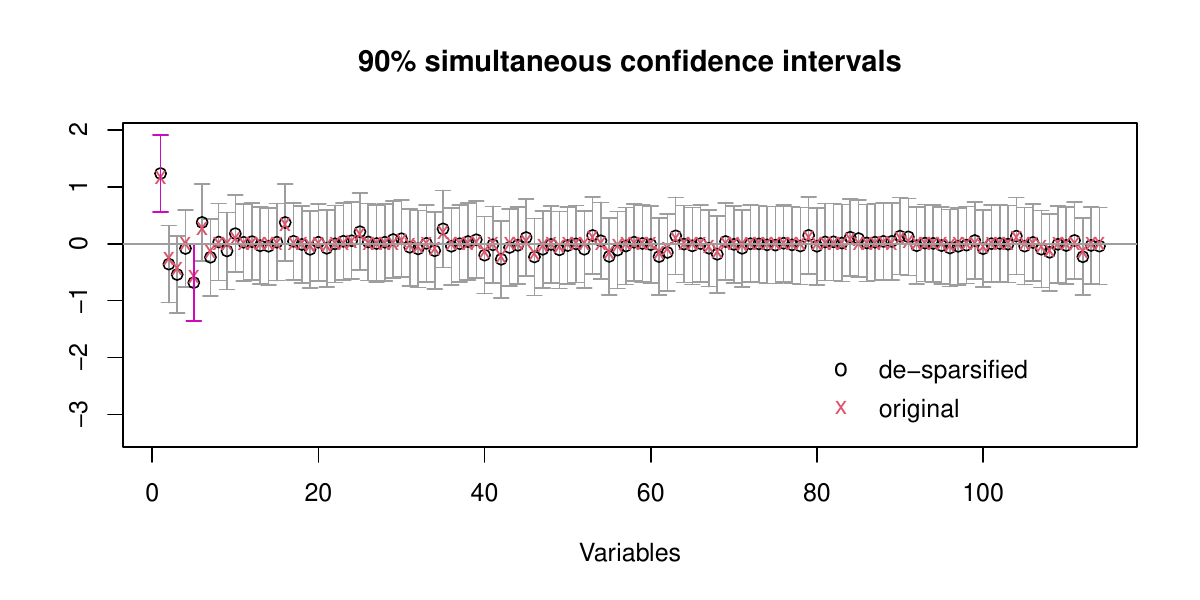} \\
\includegraphics[width = .85\textwidth]{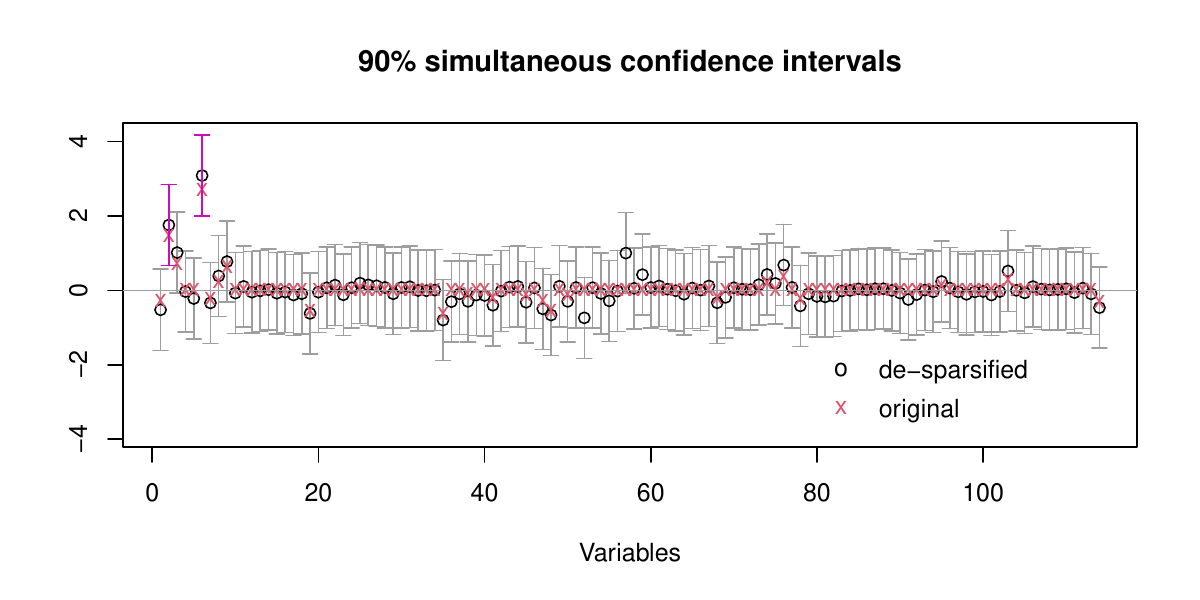}
\caption{Simultaneous $90\%$-confidence intervals obtained via directly sampling from the limit distribution, for the $p = 114$ differential parameter coefficients before and after $\wh\cp_1$ (top), $\wh\cp_2$ (middle) and $\wh\cp_3$ (bottom), plotted alongside the original LOPE estimator in~\eqref{eq:lasso:est} and its de-sparsified version in~\eqref{eq:debiased}. The interval which does not contain $0$ is highlighted in magenta.}
\label{fig:fredmd:ci:noboot}
\end{figure}

\end{document}